\documentclass[11pt]{article}
\usepackage[utf8]{inputenc}
\usepackage{amsthm,amsmath,amssymb}
\usepackage[margin=1in]{geometry}
\usepackage{pgfplots}
    \pgfplotsset{compat=newest}
\usepackage{tikz}
\usepackage{tikz-3dplot}
\usetikzlibrary{petri,patterns,tikzmark}
\usepgflibrary{patterns}
\usepackage{pgfplotstable}
\usepackage{filecontents}
\usepackage{pgfcalendar}
\usetikzlibrary{pgfplots.dateplot}
\usepackage{setspace}
\usepackage{graphicx}
\usepackage{subcaption,placeins}
\usepackage{xcolor}
\usepackage{upgreek}
\usepackage{comment}
\usepackage{clipboard}
\usepackage{hyperref}
\usepackage{bm}
\usepackage{bbm}
\usepackage[normalem]{ulem}
\usepackage{enumerate}
\usepackage{multirow}
\usepackage{rotating}
\hypersetup{
    colorlinks=true,
    urlcolor=blue,
    citecolor=blue,
    linkcolor=blue
}
\PassOptionsToPackage{hyphens}{url}\usepackage{hyperref}
\PassOptionsToPackage{sloppy}{url}\usepackage{hyperref}
\usepackage{cleveref}
\usepackage{ctable}
\usepackage{indentfirst}
\usepackage{makecell}
\usepackage{ragged2e}

\newcommand{\E}{\mathbb{E}}
\newcommand{\V}{\mathbb{V}}
\newcommand{\Hyp}{\mathbb{H}}

\newcommand{\indicator}[1]{\mathbbm{1}\{#1\}}
\newcommand{\T}{\mathcal{T}}
\newcommand{\W}{\mathcal{W}}

\newcommand{\STAB}[1]{\begin{tabular}{@{}c@{}}#1\end{tabular}}

\usepackage[bibstyle=numeric, citestyle=authoryear, natbib=true, doi=false, url=true, backend=biber, maxbibnames=10, maxcitenames=4, uniquelist=false, uniquename=false, sorting=nyt,natbib=true,hyperref=true]{biblatex}
\renewbibmacro{in:}{}
\DeclareNameAlias{default}{family-given/given-family}
\DeclareMathOperator*{\plim}{plim}
\DeclareMathOperator*{\argmin}{argmin}

\newtheorem{assumption}{Assumption}
\newtheorem{proposition}{Proposition}
\newtheorem{definition}{Definition}
\newtheorem{lemma}{Lemma}
\newtheorem{theorem}{Theorem}
\newtheorem{corollary}{Corollary}

\newtheorem{remark}{Remark}
\newtheorem{example}{Example}
\newtheorem{inneruassumption}{Assumption}
\newenvironment{namedassumption}[1]
  {\renewcommand\theinneruassumption{#1}\inneruassumption}
  {\endinneruassumption}

\newtheorem{inneruproposition}{Proposition}
\newenvironment{namedproposition}[1]
  {\renewcommand\theinneruproposition{#1}\inneruproposition}
  {\endinneruproposition}

\newtheorem{innerutheorem}{Theorem}
\newenvironment{namedtheorem}[1]
  {\renewcommand\theinnerutheorem{#1}\innerutheorem}
  {\endinnerutheorem}

\newtheorem{innerucorollary}{Corollary}
\newenvironment{namedcorollary}[1]
  {\renewcommand\theinnerucorollary{#1}\innerucorollary}
  {\endinnerucorollary}

\crefname{figure}{Figure}{Figures}
\crefname{table}{Table}{Tables}
\crefname{assumption}{Assumption}{Assumptions}
\crefname{inneruassumption}{Assumption}{Assumptions}
\crefname{proposition}{Proposition}{Propositions}
\crefname{condition}{Condition}{Conditions}
\crefname{lemma}{Lemma}{Lemmata}
\crefname{example}{Example}{Examples}
\crefname{remark}{Remark}{Remarks}
\numberwithin{lemma}{section}
\numberwithin{remark}{section}
\numberwithin{table}{section}
\numberwithin{condition}{section}
\numberwithin{equation}{section}

    \def\independenT#1#2{\mathrel{\setbox0\hbox{$#1#2$}%
    \copy0\kern-\wd0\mkern4mu\box0}}

\title{Difference-in-differences with as few as two cross-sectional units\\ -- A new perspective to the democracy–growth debate}

\author{Gilles Boevi Koumou\footnote{Chaire Desjardins en Finance Responsable, École de Gestion, Université de Sherbrooke, 2500 Boulevard de l'Université, Sherbrooke, Québec, J1K 2R1, Canada. Email: \href{mailto:nettey.boevi.gilles.koumou@usherbrooke.ca}{nettey.boevi.gilles.koumou@usherbrooke.ca.}, Orcid: \url{https://orcid.org/0000-0003-0466-3904}} \and Emmanuel Selorm Tsyawo\footnote{Corresponding author, email: \href{mailto:estsyawo@gmail.com}{estsyawo@gmail.com}, Department of Economics, Finance and Legal Studies, Culverhouse College of Business, University of Alabama}}

\begin{document}

\maketitle

\abstract{\noindent Pooled panel analyses often mask heterogeneity in unit-specific treatment effects. This challenge, for example, crops up in studies of the impact of democracy on economic growth, where findings vary substantially due to differences in country composition. To address this challenge, this paper introduces the Temporal Difference-in-Differences (T-DiD) estimator that leverages temporal variation in the data to estimate unit-specific average treatment effects on the treated (ATT) with as few as two cross-sectional units. Under asymptotic parallel trends, limited anticipation, and temporal dependence conditions, the proposed DiD estimator is shown to be asymptotically normal. Provided at least two control units are available, the method is further complemented with an identification test that, unlike pre-trends tests, is more powerful and can detect violations of parallel trends in post-treatment periods. Empirical results using the DiD estimator suggest Benin's economy would have been 6.4\% smaller on average over the 1993-2018 period had she not democratised.}

\bigskip

\bigskip

\bigskip

\bigskip

\bigskip


\bigskip


\vspace{250pt}

\linespread{1.5}

\pagebreak

\begin{refsection}

\section{Introduction}

Pooled regression analyses are commonly employed to assess the impact of interventions. However, when unit-specific effects exhibit substantial heterogeneity, estimates from pooled regressions become uninformative. Moreover, pooled analyses are often impractical when only a few cross-sectional units are available. Such analyses also heighten the risk of inadvertently violating identification conditions—such as the parallel trends assumption in the standard Difference-in-Differences (DiD) framework—by failing to carefully select appropriate control units. In certain empirical settings, including the one considered in this paper, there may be as few as a single ideal control unit. Under this two-unit baseline setting, existing methods prove inadequate. The Synthetic Control (SC) approach—the most closely related method—requires multiple candidate control units. Furthermore, unit-level placebo inference, a workhorse for SC-based inference and robustness checks, becomes infeasible with only one control unit or suffers from severe imprecision when the number of controls is small. This paper proposes a Difference-in-Differences estimator (hereafter, T-DiD) that exploits temporal variation in the data to estimate unit-specific average treatment effects on the treated (ATT) with as few as two cross-sectional units. The novelty of this paper lies not in the Difference-in-Differences estimator \emph{per se}, but in the analysis of its theoretical properties—specifically, asymptotic identification and asymptotic normality—together with an identification test in a fixed-$N$, large-$T$ setting.

The DiD quantity, with as few as two units (one treated and one control) and two time periods (one pre-treatment and one post-treatment), can at best be unbiased but not consistent, since there are only two units. However, by leveraging temporal variation in the data, i.e., by taking the (weighted) average across all possible pre- and post-treatment pairs, one obtains the proposed T-DiD. This paper shows that identification holds under asymptotic parallel trends and limited anticipation conditions. The T-DiD framework allows temporal dependence in the outcome series and arbitrary cross-sectional dependence between the treated and control units without any distributional assumptions on the outcome variables themselves. Identification can hold up to an asymptotic bias, which does not interfere with inference under asymptotic parallel trends and limited anticipation assumptions. Under near-epoch dependence (NED) and standard regularity conditions, the T-DiD is asymptotically normal.\footnote{Informally, NED structures allow dependence between time series observations that decays with the time interval -- \Cref{Def:near_epoch} provides a formal definition.} Further, this paper proposes a valid \emph{over-identifying restrictions test} of identification, which has desirable properties such as detecting a vast array of violations of identification, unlike pre-trends tests, when at least two candidate controls are available. For instance, the proposed test detects violations of identification in the post-treatment period, of which pre-tests are completely incapable. Although inference using the T-DiD is robust to cross-sectional dependence, such dependence is not \emph{sine qua non}, unlike for factor and SC models; identification in the T-DiD framework does not require correlations among cross-sectional units driven by, e.g., common factors -- see \citet{hsiao-ching-wan-2012,ferman-pinto-2021-synthetic}. While DiD estimators in general are inconsistent under non-trivial multiplicative factors, the SC generally fails without them. Existing robust methods---such as those by \citet{xu-2017,arkhangelsky-etal-2021,callaway-karami-2023treatment}---require large $N$, a constraint our small-$N$ T-DiD approach avoids. The T-DiD does not involve regressing the outcome of a unit on that of other units; it thus avoids concerns with spurious regressions or non-standard inference in the presence of co-integration -- see, e.g., \citet{masini2022counterfactual,Li-2020-statistical} for discussions.

Thanks to a linear regression-based formulation of the DiD estimator, the T-DiD accommodates observed time-varying heterogeneity through the inclusion of covariates and handles complications often encountered in time series analyses, e.g., auto-regressive processes, unit root processes, moving average processes, potentially unbounded common trends, and idiosyncratic deterministic time trends in untreated potential outcomes. However, the proposed T-DiD requires several pre- and post-treatment periods, unlike the SC, which only requires at least one post-treatment period in typical cases, and the C-DID, which needs as few as two periods. Moreover, exploring cross-sectional heterogeneity by estimating unit-specific effects using the T-DiD involves the trade-off of collapsing temporal heterogeneity through (weighted) averages, in contrast to approaches like the C-DiD and SC, which are better suited for analysing the temporal heterogeneity of treatment effects. In this sense, the T-DiD reverses the roles of time and units relative to the C-DiD: just as an extra pre-treatment period provides a testable implication (pre-test) in the C-DiD framework, an extra control unit provides a testable implication (the proposed test) in the T-DiD framework. Therefore, the T-DiD is correctly viewed as complementary to the SC and C-DiD.

The paper’s empirical application revisits the long-debated question of whether democracy fosters economic growth, a timely issue given current global democratic backsliding \citep{Coppedge2022, Foa2020}. Previous studies yield conflicting evidence --- positive effects, e.g., \citet{Acemoglu2019}, negative effects, e.g., \citet{Gerring2005}, or insignificant effects, e.g., \citet{Murtin2014} --- due to modelling differences, data limitations, and heterogeneity in regime performance, especially among autocracies. To illustrate the T-DiD framework, the paper focuses on Benin, which has undergone significant democratisation since the 1990s, and considers neighbouring countries in West Africa as potential controls. While Burkina Faso and Niger have partially democratic histories and Nigeria’s economy diverges markedly from Benin’s, Togo emerges as the only feasible control in West Africa due to its shared colonial legacy, cultural and institutional similarities, and membership in the same monetary union.\footnote{Further discussion is available in \Cref{App_Sect:Emp_Supplement}.} The empirical results show that Benin's economy would have been 6.4\% smaller on average over the 1993-2018 post-treatment window had she not democratised. This effect is both statistically and economically significant.

The T-DiD estimator is not limited to the empirical example of interest in this paper; it is applicable in contexts with as few as a single treated unit and a single control unit, provided there are many pre- and post-treatment periods. Such settings are common when studying heterogeneous effects, like unit-specific effects. Relevant empirical examples include the North and South Korea experiment within the ensuing democracy-growth debate, the effect of government policy on company share prices \citep{masini2022counterfactual}, the impact of anti-tax evasion laws on macroeconomic indicators \citep{carvalho-Masini-Ricardo-2018arco}, Hong Kong's gross domestic product (GDP) growth after sovereignty reversion \citep{hsiao-ching-wan-2012}, the economic impact of International Monetary Fund (IMF) bailouts \citep{lee-shin-2008imf}, the impact of opening physical showrooms on the sales of a firm \citep{Li-2020-statistical}, and the transition from federal to state management of the Clean Water Act \citep{marcus-santanna-2021role,grooms-2015-enforcing}. See, e.g., \citet[Sect. 1.3]{carvalho-Masini-Ricardo-2018arco} for additional examples.

For the remainder of the paper, \Cref{Sect:Ident} defines the parameter of interest and outlines sufficient conditions for its asymptotic identification. \Cref{Sect:Estimation} introduces the T-DiD estimator, while \Cref{Sect:Asym} develops the asymptotic theory. \Cref{Sect:Tests} proposes a T-DiD-based over-identifying restrictions test of identification, and \Cref{Sect:Emp} applies the method to estimate the effect of democracy on Benin's economic growth. Finally, \Cref{Sect:Con} concludes the paper. The supplementary material contains proofs of technical results, model extensions, supplementary discussions, and a simulation exercise.

\paragraph{Notation:} $ Y_t(1)$ and $Y_t(0)$ denote the treated and untreated potential outcome at period $t$, respectively, $Y_{d,t}$ denotes the outcome $Y_t$ for a unit with treatment status $D=d$. These outcomes are defined on a complete probability space $(\Omega, \mathcal{F}, P)$ and behave as stochastic processes that are Near-Epoch Dependent (NED) on a strictly stationary $\alpha$-mixing underlying process; all conditional expectations are understood to be defined \emph{almost surely}. $||W||_p:= (\E[|W|^p])^{1/p}, \ p\geq 1 $ denotes the $L_p$-norm of the random variable $W$, and $\V[W]$ denotes its variance. $a_{1,n} \lesssim a_{2,n}$ means $a_{1,n} \leq c\,a_{2,n}$ for some finite $ c>0 $, and $a_{1,n} \asymp a_{2,n}$ means both $a_{1,n}\lesssim a_{2,n}$ and $a_{2,n} \lesssim a_{1,n}$ where $\{a_{1,n}:n\geq 1\}$ and $\{a_{2,n}:n\geq 1\}$ are sequences of non-negative numbers. $a\wedge b := \min\{a,b\} $, $a\vee b := \max\{a,b\} $, and $\rho_{\mathrm{min}}(\Sigma)$ denotes the minimum eigenvalue of the positive semi-definite matrix $\Sigma$. $[T]:= \{1,\ldots, T \} $, while $[-\T]:= \{-\T,\ldots,-1\} $. Finally, let $\epsilon > 0$ denote a small constant, the value of which may vary across occurrences without loss of generality.

\section{The Two-unit Baseline Setting}\label{Sect:Ident}
In the baseline case considered in this paper, the researcher has one treated unit and one control unit with large $\T$ periods pre-treatment and large $T$ periods post-treatment. Pre-treatment periods are labelled as $\tau \in [-\T]$, and post-treatment periods as $t \in [T]$. The period $t=0$ is used to denote any period in the intervening transition window between pre-treatment and post-treatment.\footnote{\url{https://freedomhouse.org}, for example, reports 1990-1991 as the transition period to democracy in Benin.} The binary variable $D$ denotes the treatment status of a unit. 

\subsection{Parameter of interest}\label{Sub_Sect:Par_Interest}
The goal of this sub-section is to first define (in the population), the parameter of interest before establishing its asymptotic identification. The ATT at period $t$ is given by 
\[
ATT(t) := \E[Y_t(1)-Y_t(0)\mid D=1], \ t \in [T]
\]
\noindent where potential outcomes $ \big\{ \big(Y_t(1),Y_t(0)\big), \; t \in [T] \big\} $ are random for a single treated unit. A fundamental identification problem is that $Y_t(1)$ and $Y_t(0)$ are not observed simultaneously for each $t\in [T]$. For example, unlike $Y_t(1)$, $Y_t(0)$ is not observed for the treated unit. Thus, a crucial first step in the identification of $ATT(t)$ is the identification of  $\E[Y_t(0)\mid D=1], \, t \geq 1 $. Fix a pre-treatment period $\tau \in [-\T] $, then consider a standard parallel trends assumption $\E[Y_t(0)-Y_\tau(0)\mid D=1]=\E[Y_t(0)-Y_\tau(0)\mid D=0], \ t\in [T] $, e.g., \citet[Assumption 1]{roth-SantAnna-Bilinski-Poe-2023} and a no anticipation assumption $ Y_\tau(1) = Y_\tau(0) $, e.g., \citet[Assumption 2]{roth-SantAnna-Bilinski-Poe-2023}. The identification of $ \big\{ ATT(t), \ t\in [T] \big\} $, namely $ ATT(t) = ATT_{\tau,t} $, where $ATT_{\tau,t} := \E[Y_{1,t}-Y_{1,\tau}] - \E[Y_{0,t}-Y_{0,\tau}]$ denotes the DiD quantity for $(\tau,t) \in [-\T] \times [T] $, follows from both standard identification conditions. There are four random variables in the expression of the identified $ATT(t)$: $Y_{1,t},\ Y_{1,\tau}, \ Y_{0,t}$, and $ Y_{0,\tau}$. In the two-unit baseline setting, one has at most four realisations of data thus there can be only one summand in an estimator of $ATT(t)$, namely $\widehat{ATT}_{\tau,t}:= (Y_{1,t} - Y_{1,\tau}) - (Y_{0,t}-Y_{0,\tau})$. Under the aforementioned standard C-DiD identification conditions, $\widehat{ATT}_{\tau,t}$ is unbiased for each $ATT(t), \ t\in [T] $ for a fixed $\tau \in [-\T] $. However, $\widehat{ATT}_{\tau,t}$ cannot be consistent, let alone lead to any meaningful inference procedure for each $t \in [T] $ since the number of cross-sectional units is fixed. Identifying variation from cross-sectional units, as in C-DiD settings with fixed $\T,T$ and a large number of both treated and control units, cannot be exploited in the setting considered by this paper. 

One can, instead, use averages of $\widehat{ATT}_{\tau,t} $ across all pairwise combinations of pre-treatment and post-treatment periods, namely
\[
\widehat{ATT}_{1,T} = \frac{1}{\T T}\sum_{t=1}^T\sum_{-\tau=1}^{\T} \widehat{ATT}_{\tau,t},
\]to target the parameter
\[
ATT_{1,T} = \frac{1}{T} \sum_{t=1}^T ATT(t).
\]Leveraging temporal variation in the data under weak dependence conditions, the above estimand can be consistently estimated, and one can obtain an asymptotically normal distribution of the estimator. Moreover, the aggregation scheme enables the parallel trends and the limited anticipation identification conditions to only hold on aggregate.

The preceding discussion focuses on the uniformly weighted estimand \( ATT_{1,T} \). In general, however, the treatment effect may vary across post-treatment periods, yielding heterogeneous \( ATT(t) \), \( t \geq 1 \), which are aggregated through averaging over \( t \geq 1 \). As previously argued, such period-specific effects \( ATT(t) \) cannot be consistently estimated individually in the present two-unit setting. To explore potential effect heterogeneity within this framework, the paper introduces a generalisation of \( ATT_{1,T} \) to a class of estimands defined as convex-weighted averages of \(\{ ATT(t), \ t \in [T] \}\), where the researcher specifies the convex weighting scheme \(\{ w_T(t), \ t \in [T] \}\):
\begin{equation}\label{eqn:ATT_estimand}
    ATT_{\omega,T} = \sum_{t=1}^T w_T(t) ATT(t),
\end{equation}
\noindent
with weights satisfying \(\sum_{t=1}^T w_T(t) = 1\) and \(w_T(t) \geq 0\). $ATT(t)$ varies over time due to dynamic treatment effects, meaning that homogeneous treatment effects across time are not assumed. Although similar $ATT_{\omega,T}$-type parameters are explored in the literature (e.g., \citet[eqn. 2.6]{masini2022counterfactual}, \citet[eqn. 2]{carvalho-Masini-Ricardo-2018arco}, \citet[p. 2]{chernozhukov-Wuthrich-Zhu-2024t} and \citet[Sect. 2]{Li-2020-statistical}), identification within the DiD framework and the corresponding asymptotic theory, as developed in this paper, appear novel. Additionally, while the parameters in the cited references rely on the synthetic control (SC) framework, which requires multiple control units, the DiD approach presented here requires only a single control unit. Moreover, the SC framework typically employs sharp null hypotheses (e.g., $\Hyp_o': ATT(t) = 0 \ \forall t \in [T]$), but testing a weaker hypothesis, such as $ATT_{\omega,T} = a$ for some constant $a$, may be more relevant \citep[p. 2]{chernozhukov-Wuthrich-Zhu-2024t}. Similar to \citet{chernozhukov-Wuthrich-Zhu-2024t}, this paper's approach makes a test of the latter straightforward using standard $t$-tests.

Convex weighting schemes have the form $w_T(t) = w(t)/\sum_{t=1}^{T}w(t) $ for some non-negative function $w: \mathbb{R} \mapsto \mathbb{R}_+ $. There are interesting aggregation schemes that a researcher may want to use. For example, a researcher may be interested in assigning more weight to $ATT(t), \ t\geq 0$ for post-treatment periods closer to the treatment window and less weight to those farther from it. Alternatively, the researcher might want to assign zero weight to periods immediately following treatment where treatment effects are expected to be null. As the estimand \eqref{eqn:ATT_estimand} depends on the particular weighting scheme, the results in this paper, e.g., identification, asymptotic normality, and identification testing ought to hold uniformly in a suitable (sub-)class of non-stochastic convex weighting schemes. Let $w(\cdot)$ belong to a general class of non-negative functions $\W $, then the class of weighting schemes is given by
\begin{align}\label{eqn:dfn_W}
    \W := \bigg\{ w: \frac{w(t)}{\sum_{t'=1}^{T}w(t')} =:w_T(t) \geq 0 \text{ and } \max_{t \in [T]} w_T(t)\lesssim T^{-1} \quad \text{ uniformly in }T \geq 1 \bigg\}.
\end{align}

A leading example is the uniform weighting scheme: $w_T(t) = 1/T$. Another example allows linearly decreasing weighting, which puts greater weight on $ATT(t)$ closer to the period of treatment $t=0$: $w_T(t) = \frac{T - 2at}{ \sum_{t'=1}^T (T-2at')} = \frac{1}{1-a}\frac{1 - (2a/T)t}{ T-a/(1-a)}, \ a \in [0,1/2) $. Estimation of $ATT_{\omega,T}$ in this paper exploits both pre-treatment and post-treatment temporal variation for consistency and asymptotic inference. It is thus crucial that the contribution of any $ATT(t), \, t\in [T] $ to $ATT_{\omega,T}$ be asymptotically negligible otherwise inference on the sequence of $ATT_{\omega,T}$ parameters based on the T-DiD estimator cannot be valid. This paper applies the notational convention $w(0)=0$ throughout.

\subsection{Asymptotic identification}
The counterfactual component $\sum_{t=1}^T w_T(t)\E[Y_t(0)\mid D=1] $ of $ATT_{\omega,T}$ in \eqref{eqn:ATT_estimand} is unobservable. Its identification guarantees that of $ATT_{\omega,T}$ since $\sum_{t=1}^T w_T(t)\E[Y_t(1)\mid D=1] = \sum_{t=1}^T w_T(t)\E[Y_{1,t}] $ is identified from the sampling process. Although sufficient, standard parallel trends and no-anticipation assumptions such as those mentioned in \Cref{Sub_Sect:Par_Interest} are stronger than necessary for (asymptotically) identifying $ATT_{\omega,T}$. This paper opts for an asymptotic parallel trends assumption that allows deviations away from standard parallel trends. For the baseline result presented in this section, it is maintained that potential outcomes are \emph{suitably adjusted} for, e.g., time-varying heterogeneity via the inclusion of time-varying covariates, unit roots via first differences, and deterministic trends via detrending. Details on implementation follow in \Cref{SubSect:Estimator}, \Cref{subsec:non_stationarity}, and \Cref{App_Sect:Det_Trend}. The adjustment of potential outcomes ensures the identification conditions \Cref{ass:parallel_trends,ass:limited_anticip} below are robust to such complications often encountered in time series as well as ensuring the conditions hold conditional on observable non-overlapping time-varying characteristics.

\begin{assumption}[Asymptotic Parallel Trends]\label{ass:parallel_trends}

\[
\frac{1}{\T T}\sum_{-\tau=1}^{\T }\sum_{t=1}^T \Big(\E[Y_t(0)-Y_\tau(0)\mid D=1] - \E[Y_t(0) - Y_\tau(0)\mid D=0]\Big) = \mathcal{O}\big((\T \wedge T)^{-(1/2+\gamma)}\big), \ \gamma > 0.
\]

\end{assumption}

\noindent \Cref{ass:parallel_trends} allows some extent of mean dependence of the paths of average untreated potential outcomes on treatment status $D$. Thus, unlike standard parallel trends assumptions which impose \( \E[Y_t(0)-Y_\tau(0)\mid D=1] - \E[Y_t(0) - Y_\tau(0)\mid D=0] = 0 \) for all $(\tau,t) \in [-\T]\times [T] $, \Cref{ass:parallel_trends}  allows for some, albeit controlled, violations. Moreover, \Cref{ass:parallel_trends} applies to the average over all $(\tau,t)$-pairs and \emph{not} each pair.
 
One could replace \Cref{ass:parallel_trends} with an assumption of \emph{average unbiasedness}---for example, \`a la \citet[Assumption 1]{botosaru-giacomini-weidner-2023forecasted},  namely, 
\[
\frac{1}{\T T}\sum_{-\tau=1}^{\T }\sum_{t=1}^T \Big( \E[Y_t(0)-Y_\tau(0)\mid D=1] - \E[Y_t(0)-Y_\tau(0)\mid D=0] \Big) = 0.
\]  
This condition, while weaker than standard parallel trends, can still be restrictive. For example, such a parallel trends assumption may require some form of stationarity in untreated potential outcomes for plausibility. This paper appears to be the first to use an asymptotic form of the parallel trends assumption.

To further shed light on the asymptotic parallel trends assumption, consider the following specification of expected untreated potential outcomes: $ \E[Y_t(0) \mid D=d] = \nu_{0}(t,d;\eta) $ where
\[
\nu_{0}(t,d;\eta) :=
\begin{cases}
\cos(t), & d=0, \\[6pt]
0.5 + \cos(t) \;+\; 0.25\,\bigl|1+0.5\sin(t)\bigr| \cdot \operatorname{sign}(t)\,|t|^{-\eta}, & d=1,\;\;|t|\geq 1.
\end{cases}
\] 
\noindent The expected untreated potential outcomes (against time) for both the treated and untreated units are plotted in \Cref{Fig:PT_illus} with $T=\T=5$ and $T=\T=25$ at $\eta=0.6$. As can be seen, the expected untreated potential outcomes are not parallel in the traditional C-DiD sense. That is, for any pairwise combination of pre- and post-treatment periods, expected untreated potential outcomes are not parallel. However, the T-DiD parallel trends condition, namely \Cref{ass:parallel_trends} is satisfied.

\begin{figure}[!htbp]
\centering 
\caption{Asymptotic Parallel Trends}
\begin{subfigure}{0.38\textwidth} 
\centering
\includegraphics[width=1\textwidth]{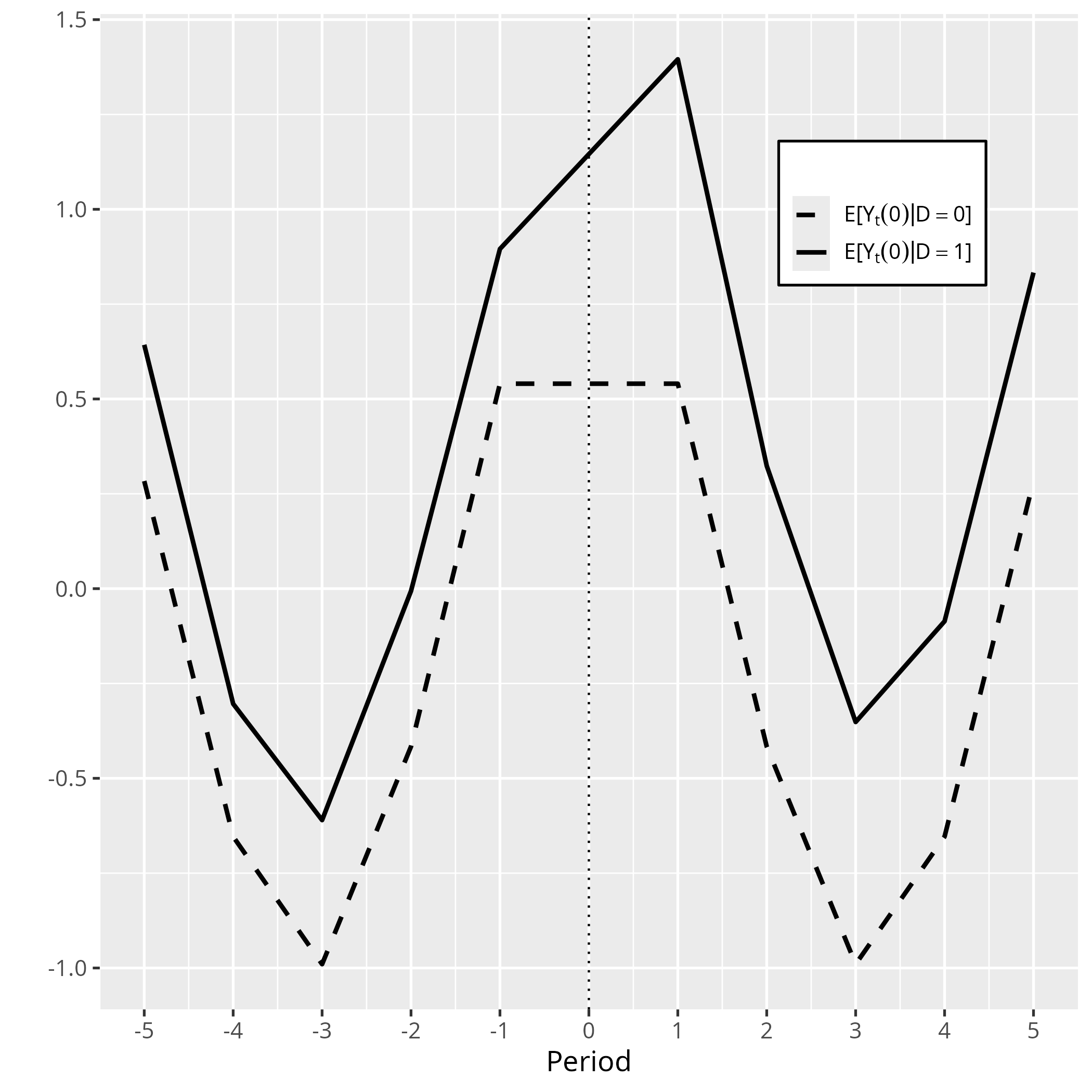} 
\caption{ }
\end{subfigure}
\hspace{0.05\textwidth} 
\begin{subfigure}{0.38\textwidth} 
\centering
\includegraphics[width=1\textwidth]{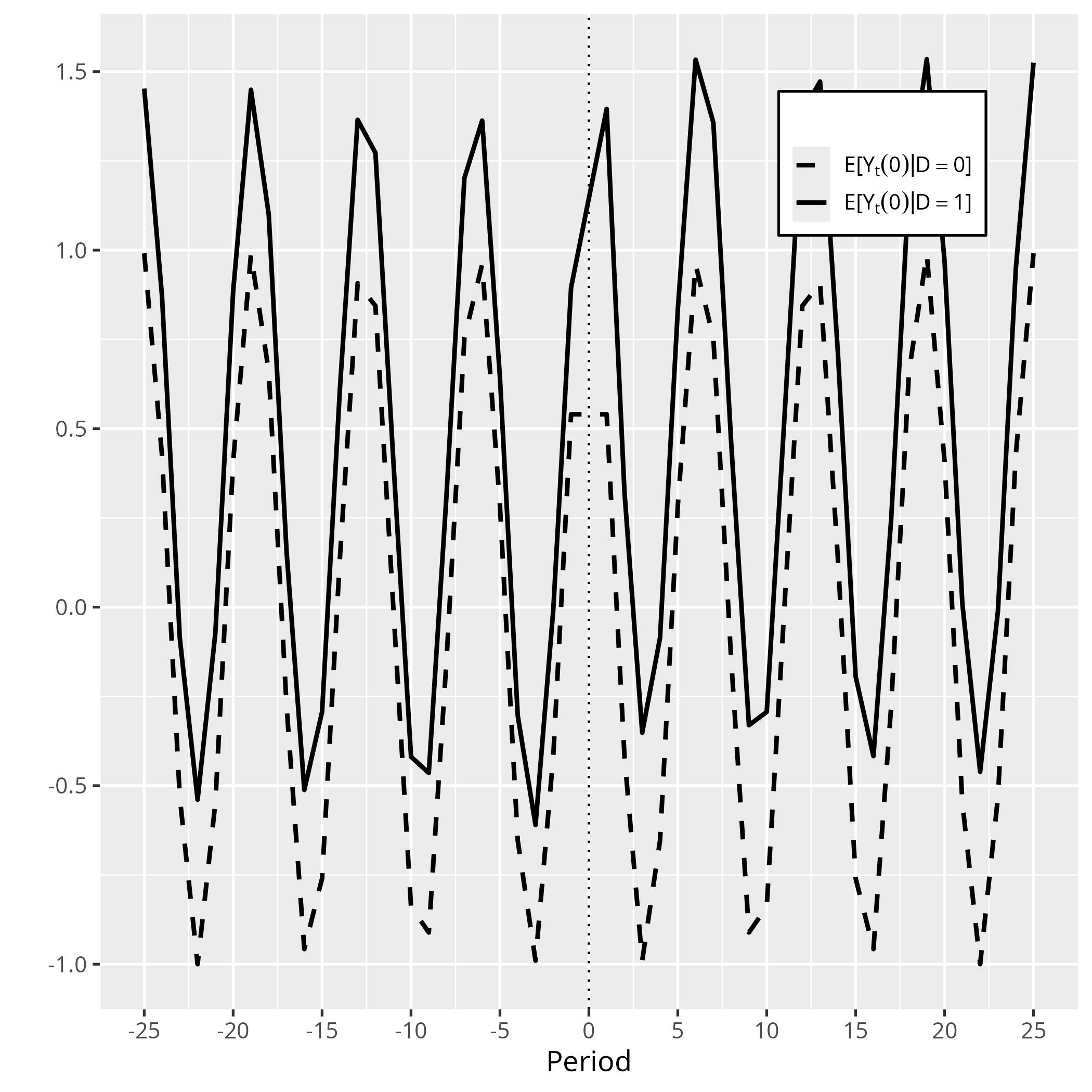} 
\caption{ }
\end{subfigure}
\label{Fig:PT_illus}
\begin{justify}
\footnotesize The plots depict the paths of untreated potential outcomes from \eqref{eqn:PT_DGP1} for the control and treated units with pre- and post-treatment time horizons (a) $\T=T=5$ and (b) $\T=T=25$.
\end{justify}
\end{figure}

\begin{figure}[!htbp]
\centering 
\caption{Asymptotic Parallel Trends}
\begin{subfigure}{0.3\textwidth}
\centering
\includegraphics[width=1\textwidth]{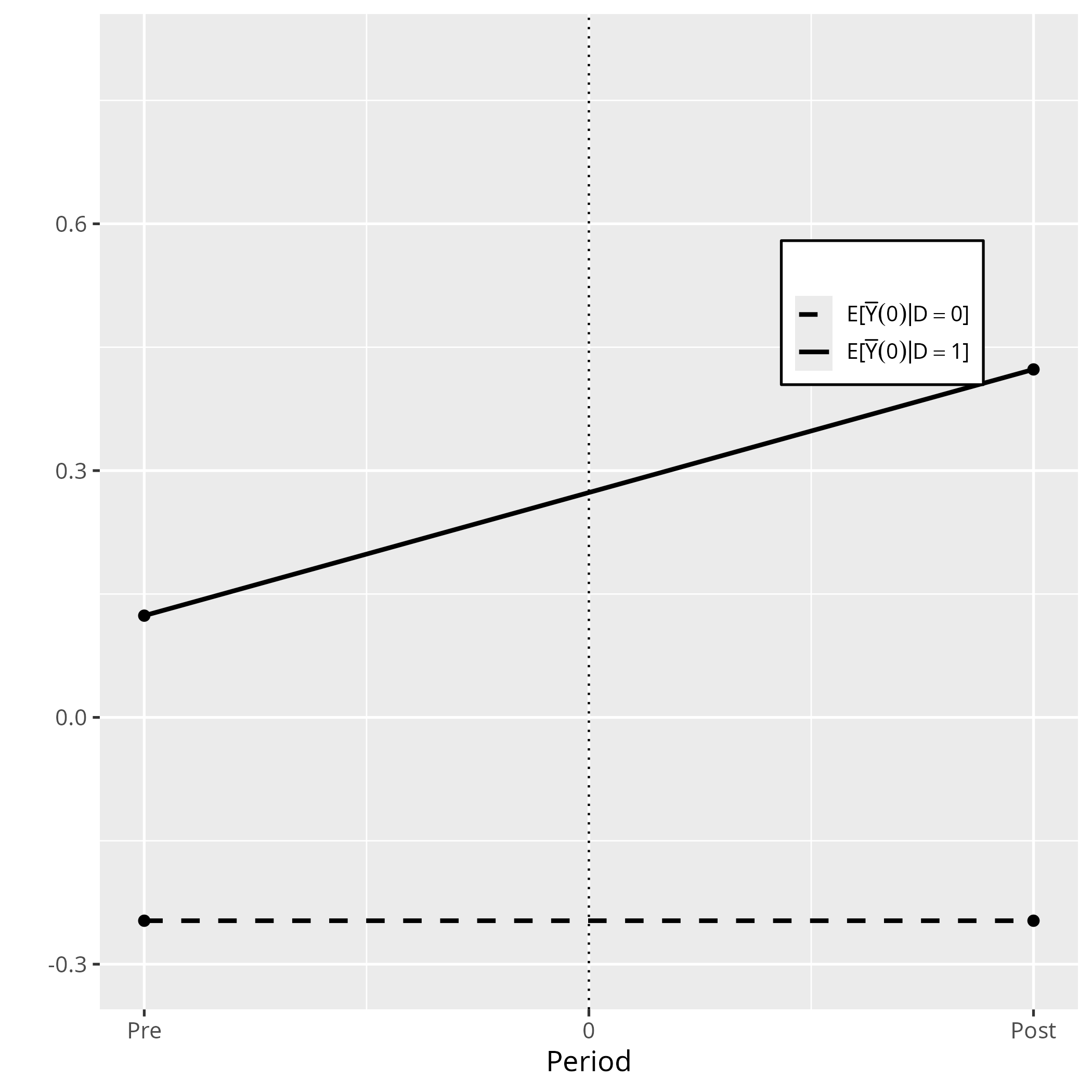}
\caption{ $T=\mathcal{T}=5$ }
\end{subfigure}
\begin{subfigure}{0.3\textwidth}
\centering
\includegraphics[width=1\textwidth]{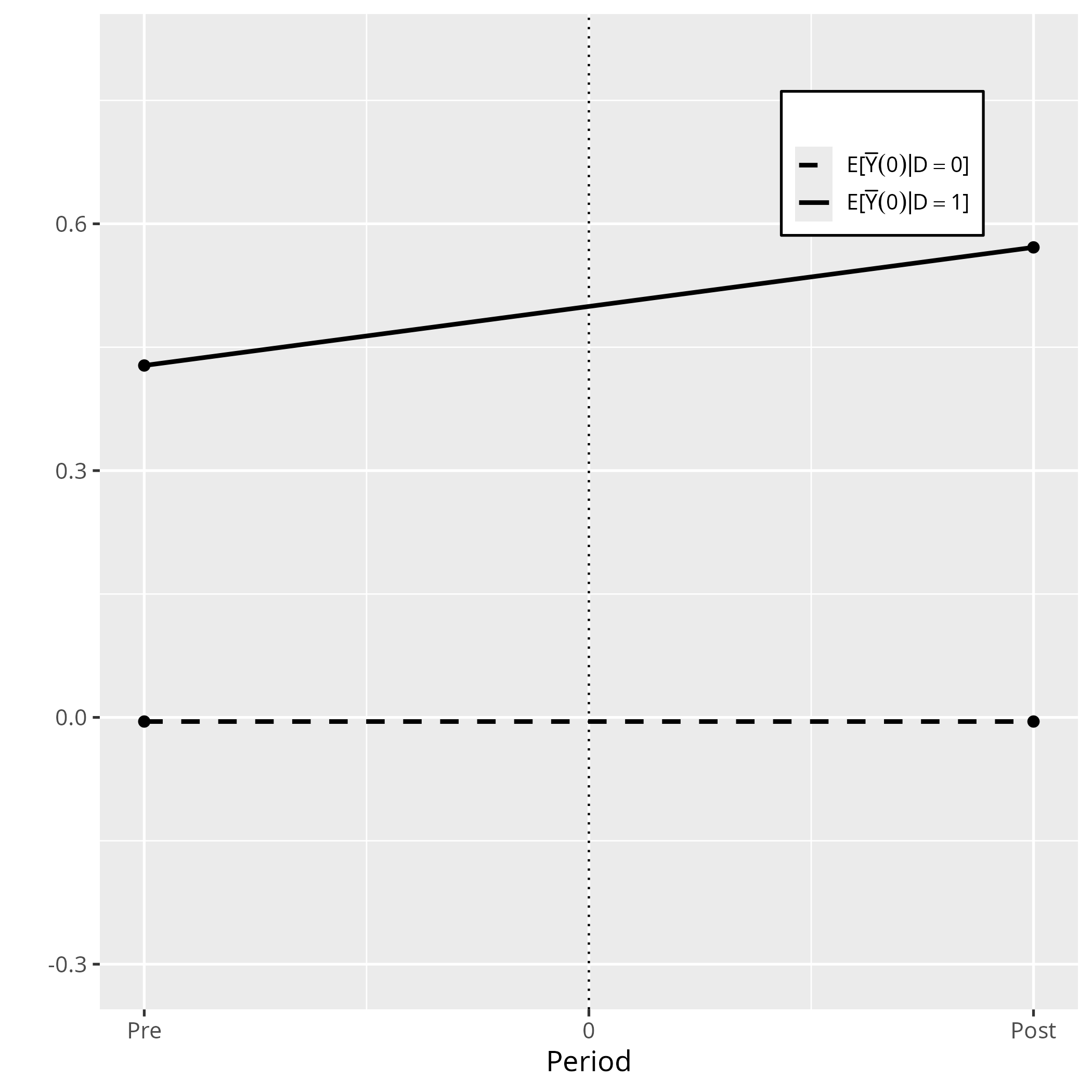}
\caption{ $T=\mathcal{T}=25$ }
\end{subfigure}
\begin{subfigure}{0.3\textwidth}
\centering
\includegraphics[width=1\textwidth]
{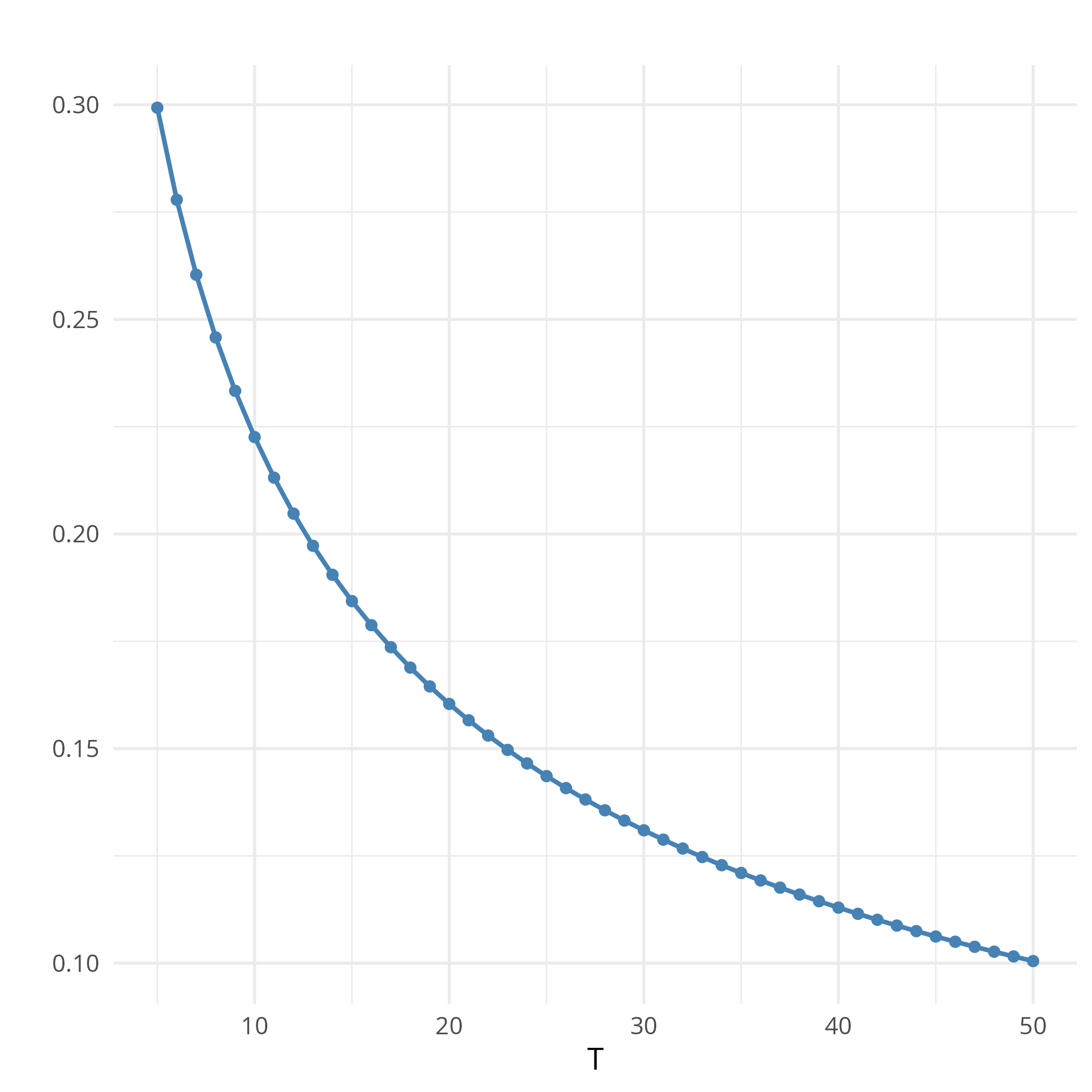}
\caption{Violation of Asymptotic Parallel Trends, $T=\mathcal{T}$ }
\end{subfigure}
\label{Fig:PT_illus2}
\begin{justify}
\footnotesize Plots (a) and (b) depict the paths of pre- and post-treatment averages of untreated potential outcomes from \eqref{eqn:PT_DGP1} for the control and treated units with pre- and post-treatment time horizons (a) $\T=T=5$ and (b) $\T=T=25$. Plot (c) plots violations of the parallel trends condition (Assumption \Cref{ass:parallel_trends}) as a function of the post-treatment time horizon $T$.
\end{justify}
\end{figure}

The first and second panels of \Cref{Fig:PT_illus2} plot the paths of pre-post averages of expected untreated potential outcomes. These are the paths in untreated potential outcomes that matter for the T-DiD. At $T=\mathcal{T}=25$ in Panel (b), relative to $T=\mathcal{T}=5$ in Panel (a), the pre-post paths of average untreated potential outcomes are closer to being parallel. Panel (c) shows that the deviation from parallel paths is decreasing in $T=\T$.\footnote{This argument extends to $T\wedge \mathcal{T}$ in light of \Cref{ass:TTratio_lambda_n} below. } \Cref{ass:parallel_trends} says that although the paths of pre-post averages of expected untreated potential outcomes may not be exactly parallel in finite samples, the violation is shrinking to zero sufficiently fast in $T\wedge\T$. 

Consider the following Data Generating Process (DGP) of untreated potential outcomes
\begin{equation}\label{eqn:PT_DGP1}
     Y_t(0) = \alpha_0 + (\alpha_1-\alpha_0)D + \varphi(t) + \nu_{0}(t,D;\eta)  + e_t,
 \end{equation}where $(\alpha_0,\alpha_1)$ are constant, $ \varphi(t) $ is some possibly unbounded function of $t$, and $ \E[e_t \mid D = d] = 0, \ d\in \{0,1\} $. The following result shows that untreated potential outcomes following \eqref{eqn:PT_DGP1} satisfy \Cref{ass:parallel_trends}.
 
\begin{proposition}\label{Prop:PT_DGP}
\Copy{Key:Prop:PT_DGP}{ Let $\varphi(\cdot)$ be some possibly unbounded function $\varphi: [-\T]\cup [T] \rightarrow \mathbb{R}$ and $\eta > 1/2 $, then untreated potential outcomes following \eqref{eqn:PT_DGP1} satisfy \Cref{ass:parallel_trends}. }
 \end{proposition}
 
 \begin{remark}\label{Rem:Prop_1_DGP}
 First, temporal dependence in $\{ e_t, -\T  \leq t \leq T \}$, e.g., MA processes, is not ruled out by \Cref{ass:parallel_trends}. Second, with no bound restrictions on $\varphi(\cdot)$ in \eqref{eqn:PT_DGP1}, non-stationary common shocks in individual untreated potential outcomes $Y_t(0)$ conditional on $D=d$ do not violate \Cref{ass:parallel_trends}. Third, \eqref{eqn:PT_DGP1} leaves unrestricted the cross-sectional dependence between treated and control unit untreated potential outcomes. Fourth, the restriction on $ \eta > 1/2 $ does not require that $\E[\nu_{0}(t,D;\eta)\mid D=d] =:\nu_{0}(t,d;\eta), \, d\in \{0,1\} $ or some average thereof over time be zero; thus \eqref{eqn:PT_DGP1} effectively allows violations of standard parallel trends assumptions.
 \end{remark}

The third point in \Cref{Rem:Prop_1_DGP} above, in effect, allows correlated shocks to both the treated and control units. For example, neighbouring West African countries Ghana, Togo, and Benin, were all exposed to the same wave of democratisation in the early 1990s. It is thus reasonable not to rule out correlations between country outcomes over time. The T-DiD, however, does not rely on such ``common factors" in pre-treatment outcomes for constructing counterfactuals unlike, e.g., the SC and factor models -- cf. \citet{abadie-gardeazabal-2003,hsiao-ching-wan-2012,ferman-pinto-2021-synthetic,Sun-Ben-Michael-Feller-2023-using}.

\begin{example}[Economic Interpretation of \Cref{ass:parallel_trends}]
    With the outcome defined as $Y_t = \log(GDP_t)$, \Cref{ass:parallel_trends} says that, had Benin not democratised, her average pre- to post-democratisation economic growth rate would have differed ``negligibly" from that of Togo. The $\mathcal{O}\big((\T \wedge T)^{-(1/2+\gamma)}\big)$ term in \Cref{ass:parallel_trends} translates the magnitude of the ``negligibility" required of the violations of the canonical parallel trends assumption.
\end{example}

\noindent To conclude the discussion of \Cref{ass:parallel_trends}, it is worth noting that, although one might argue that requiring a large number of post-treatment periods \( T \) constitutes a restrictive condition, this concern is less pertinent in the context considered here. Crucially, this paper focuses on settings in which the condition of a large \( T \) is credible. For instance, regime changes—such as transitions between autocracy and democracy—often entail significant costs in terms of time and human lives, thereby lending plausibility to the large-\( \T \wedge T \) asymptotics adopted in this paper.

Rational expectations of economic agents, e.g., foreign investors, can increase the uncertainty around the timing of treatment, i.e., when treatment \emph{actually} begins. Per the running example, the anticipation of democratisation by investors can already induce a change in foreign direct investment and private domestic investment before the documented period of transition to democratisation. While a standard no-anticipation assumption, e.g., $Y_\tau(1)=Y_\tau(0)$ or the weaker $\E[Y_\tau(1) - Y_\tau(0)\mid D=1] = 0$ for $\tau \in [-\T] $, together with \Cref{ass:parallel_trends}, should suffice for the asymptotic identification of $ATT_{\omega,T}$, it is not necessary. The assumption below allows for deviations away from standard no-anticipation assumptions.

\begin{assumption}[Asymptotically Limited Anticipation] \label{ass:limited_anticip}

\[
    \frac{1}{\T }\sum_{-\tau=1}^{\T } \E\big[Y_{\tau}(1) - Y_{\tau}(0)\mid D=1 \big] = \mathcal{O}(\T ^{-(1/2+\delta)}), \ \delta > 0.
\]

\end{assumption}
\noindent By definition, \( \E\big[Y_{\tau}(1) - Y_{\tau}(0)\mid D=1 \big] = ATT(\tau), \ \tau \leq -1 \). \Cref{ass:limited_anticip} requires that the average anticipatory effect of treatment on the treated unit is asymptotically negligible. The following illustration draws on the running example to provide further intuition.

\begin{example}[Economic interpretation of \Cref{ass:limited_anticip}]
    In general, economic analysts and prospective investors use past democratisation experiences, among other indicators, to form expectations of a country's future political stability. Thus, the anticipation that the wave of democratisation sweeping across African countries may eventually materialise in Benin could bolster investor confidence. Such a boost in confidence can subsequently lead to a cautious increase in investment before the actual transition to democracy begins. \Cref{ass:limited_anticip} says that a boost in investment (and GDP for that matter) owing to anticipations of democratisation is \emph{negligible} (up to some $\mathcal{O}(\T ^{-(1/2+\delta)})$ term).
\end{example}

\noindent Like \citet{callaway-santanna-2021}, \Cref{ass:limited_anticip} has, as a special case, limited treatment anticipation during a fixed number of periods before treatment. This special case simply involves dropping those periods out of the sample. This is further supplemented by \Cref{ass:limited_anticip}, thus providing a buffer against residual anticipatory effects.

The following remark is useful in explaining the gains from choosing suitable control units for a treated unit.
\begin{remark}\label{Rem:PT_NA_cond_Z}
    Consider a control unit, which shares a set of characteristics $Z_t$ at time $t$. For example, Togo shares several characteristics with Benin: a long border, a similar colonial history, the same monetary policy, the same currency, very similar climate and weather patterns, very similar baskets of export and import goods, and natural resources. Thus, imposing \Cref{ass:parallel_trends,ass:limited_anticip} in the running example implies they hold conditional on shared characteristics.
\end{remark}
\noindent Choosing controls that share characteristics $Z_t$ (both observable and unobservable) with the treated unit serves to control for the shared characteristics and therefore weakens both parallel trends and limited anticipation conditions (\Cref{ass:parallel_trends,ass:limited_anticip}) to conditional ones (on $Z_t$) -- cf. \citet[Assumption 3 and 4]{callaway-santanna-2021} and \citet[Assumption 6]{roth-SantAnna-Bilinski-Poe-2023}. A main takeaway from \Cref{Rem:PT_NA_cond_Z} is that the greater the overlap in characteristics between the treated and control unit, the weaker and more credible are \Cref{ass:parallel_trends,ass:limited_anticip}. Moreover, the rates of decay in \Cref{ass:parallel_trends,ass:limited_anticip} allow controlled degrees of violation of standard parallel trends and no anticipation conditions conditional on shared characteristics.

Fixing a pre-treatment ``base" period $\tau \leq -1 $,
\[
ATT_{\tau,t} = \E[Y_t(1)-Y_\tau(1)\mid D=1] - \E[Y_t(0) - Y_\tau(0)\mid D=0] = \E\big[Y_t - Y_\tau \mid D = 1 \big] - \E\big[Y_t - Y_\tau \mid D = 0 \big]
\]
for a post-treatment period $t\geq 1$ where the second equality holds because treated (resp. untreated) potential outcomes are observed for the treated (resp. untreated) unit. From the proof of \Cref{Theorem:Identification} below,
\begin{align*}
    ATT(t) - ATT_{\tau,t} =&: -R_{\tau,t}\\
    =& -\Big\{\underbrace{ \E[Y_t(0)-Y_\tau(0)\mid D=1] - \E[Y_t(0)-Y_\tau(0)\mid D=0] }_{\text{Trend Bias}(\tau,t)} - \underbrace{\E[Y_\tau(1)-Y_\tau(0)\mid D=1]}_{\text{Anticipation Bias}(\tau)}\Big\}.
\end{align*}
\noindent Although $ATT_{\tau,t}$ is identified, $ATT(t)$ is not identified because the difference between the trend and anticipation bias terms, i.e., $R_{\tau,t}$, is not identified from the data sampling process. To exploit temporal variation from pre-treatment periods, another convex weighting scheme on the same class $ \W $ is introduced, namely, $\psi \in \W  $ with $\psi_{\T}(-\tau) \geq 0$ for all $\tau\leq -1$. Under the uniform weighting scheme, for example, where each period is equally weighted post- and pre-treatment, $ w_T(t) = 1/T $ and $ \psi_{\T}(-\tau) = 1/\T $. Define the effective estimand 
\[
\widetilde{ATT}_n^{w,\psi}:= \sum_{t=1}^{T}\sum_{-\tau=1}^{\T } \psi_{\T }(-\tau) w_T(t)ATT_{\tau,t}.
\]

Let $n: = \T  + T$ and $\lambda_n = T/n$. The following regularity condition controls the ratio $\lambda_n$.
\begin{assumption}\label{ass:TTratio_lambda_n}
    Uniformly in $n$, $ \lambda_n \in [\epsilon, \, (1-\epsilon) ] $ for some constant $ \epsilon \in (0,1/2] $.
\end{assumption}

\noindent \Cref{ass:TTratio_lambda_n} is essential for theoretical results under large $T,\T $ asymptotics -- cf. \citet[Assumption E(iii)]{chan2021pcdid} and \citet[357]{carvalho-Masini-Ricardo-2018arco}. It ensures that the ratio of pre- and post-treatment periods to the total number of periods does not vanish even in the limit while also accommodating the possibility that the limit does not exist.\footnote{For example, $ \lambda_n = \epsilon + (1 - 2\epsilon) \cdot \frac{1 + \sin(n)}{2} \in [\epsilon, \, (1-\epsilon)], \quad \text{for all } n \geq 1 $.} 

The following provides the (asymptotic) identification result.

\begin{theorem}[Asymptotic Identification]\label{Theorem:Identification}
\Copy{Key:Theorem:Identification}{Let \Cref{ass:parallel_trends,ass:limited_anticip,ass:TTratio_lambda_n} hold, then 

\noindent (a) $\displaystyle \sup_{\{w,\psi\} \in \W^2 } |ATT_{\omega,T} - \widetilde{ATT}_n^{w,\psi}| = \mathcal{O}\big(n^{-(1/2+\gamma\wedge \delta)}\big) $ and 

\noindent (b) $\displaystyle \sup_{ \{\psi,\phi\} \in \W^2  } |\widetilde{ATT}_n^{w,\psi} - \widetilde{ATT}_n^{w,\phi}| =  \mathcal{O}\big(n^{-(1/2+\gamma\wedge \delta)}\big) $ uniformly in $w\in\W$, where the constants $\gamma>0$ and $\delta>0$ are defined in \Cref{ass:parallel_trends,ass:limited_anticip}, respectively.}
\end{theorem}

\noindent Part (a) of \Cref{Theorem:Identification} shows that although identification may not hold exactly in finite samples, it does asymptotically as $n \rightarrow \infty $ subject to \Cref{ass:TTratio_lambda_n}. Precisely, it says $ATT_{\omega,T}$ is identified up to a $o\big(n^{-1/2}\big)$ term uniformly in $\W^2 $. Part (b) is particularly important as it shows that identification of $ATT_{\omega,T}$ is asymptotically invariant to the weighting scheme $\psi\in\W $ applied to pre-treatment outcomes. This is reasonable because pre-treatment weighting is simply a methodological device and not an intrinsic part of the $ATT_{\omega,T}$ parameter. Thus, the \emph{effective estimand} $\widetilde{ATT}_n^{w,\psi}$ is asymptotically invariant to the pre-treatment weighting scheme $\psi \in \W $.\footnote{Recall that the post-treatment weighting scheme $w$ is an effective part of the definition of $ATT_{\omega,T}$.} 

The following decomposition 
\begin{equation}\label{eqn:Decomp_ATTn}
    ATT_{\omega,T} = \widetilde{ATT}_n^{w,\psi} - R_n^{w,\psi}
\end{equation}
is used in the proof of \Cref{Theorem:Identification} where $R_n^{w,\psi}$ is the difference between the convex-weighted averages of the trend and anticipation biases, i.e., $\displaystyle R_n^{w,\psi}:=\sum_{-\tau=1}^{\T }\sum_{t=1}^T \psi_{\T }(-\tau) w_T(t)R_{\tau,t}$.

\begin{remark}\label{rem:weak_suff_idcond}
    The identification result in \Cref{Theorem:Identification} is equivalent to $R_n^{w,\psi}=o(n^{-1/2})$ uniformly in $\W^2$. \Cref{ass:parallel_trends,ass:limited_anticip} are imposed separately following conventional DiD identification arguments, mostly for economic interpretability and clarity. They are thus jointly sufficient but not necessary for \Cref{Theorem:Identification} as it suffices that $R_n^{w,\psi}$ be $o(n^{-1/2})$ uniformly in $\W^2$. The latter condition, which is both sufficient and necessary for asymptotic identification while allowing for an asymptotically normal estimator, comes at the cost of less transparency and intuition. Although the rates in \Cref{ass:parallel_trends,ass:limited_anticip} are sufficient for identification, weaker rates that deliver $R_n^{w,\psi}=o(1)$ are also sufficient for asymptotic identification.
\end{remark}

\section{Estimation}\label{Sect:Estimation}
\subsection{The T-DiD estimator}\label{SubSect:Estimator}
In view of \Cref{Theorem:Identification}, the DiD estimator of $ATT_{\omega,T}$, which is the T-DiD, can be stated as $\widehat{ATT}_{\omega,T} = \sum_{-\tau=1}^{\T }\sum_{t=1}^T \psi_{\T }(-\tau) w_T(t)\widehat{ATT}_{\tau,t}$ where $\widehat{ATT}_{\tau,t}:= (Y_{1,t} - Y_{1,\tau}) - (Y_{0,t}-Y_{0,\tau}) $. This suggests the T-DiD can be alternatively expressed as 
\begin{equation}\label{eqn:DiD_estimator}
    \widehat{ATT}_{\omega,T} = \sum_{t=1}^{T}w_T(t)(Y_{1,t} - Y_{0,t}) - \sum_{-\tau=1}^{\T }\psi_{\T }(-\tau)(Y_{1,\tau} - Y_{0,\tau}).
\end{equation} 

\noindent The above expression clearly shows that $\widehat{ATT}_{\omega,T}$ is a difference of post-treatment and pre-treatment differences in the outcomes of the treated and control units; it is thus a \emph{bona fide} difference-in-differences estimator -- cf. \citet[eqn. 7 and 8]{arkhangelsky-etal-2021}. While the above estimator may seem obvious, a formal treatment under identification and sampling conditions (as done in this paper) does not appear to have been done.

\begin{remark}
The main difference between \eqref{eqn:DiD_estimator} and a conventional DiD estimator lies in the variation exploited for identification, estimation, and inference. In a conventional DiD setting with short panels or repeated cross-sections, the number of cross-sectional units (or clusters)—typically independently sampled—in both treated and control groups is large, while the number of time periods remains fixed. By contrast, the framework considered in this paper requires a large number of pre- and post-treatment periods, but a fixed number of cross-sectional units in the treated and control groups.
\end{remark}

From a practical point of view, a regression-based estimator is desirable as existing routines for estimation and heteroskedasticity- and auto-correlation-robust standard errors are readily applicable without modification. Define the random variable $X_t:= Y_{1,t} - Y_{0,t}$. As arbitrary cross-sectional dependence between the treated and the untreated unit is possible, one can cast the sequence of identified parameters $\widetilde{ATT}_n^{w,\psi}$ in the following (weighted) linear model: 
\begin{equation}\label{eqn:ATT_lin_mod}
    X_t = \beta_0 + \widetilde{ATT}_n^{w,\psi}  \indicator{t\geq 1} + U_t
\end{equation}
\noindent using the set of \emph{non-negative weights} $\{\widetilde{w}_n(t), t \in [-\T] \cup [T] \} $ where $\widetilde{w}_n(t):= w_T(t)\indicator{t \geq 1} + \psi_{\T}(t)\indicator{t\leq -1} $.  \Cref{Prop:ATT_equiv_wreg} in the supplement shows the numerical equivalence of \eqref{eqn:DiD_estimator} to a regression-based one using \eqref{eqn:ATT_lin_mod}, namely $ \widehat{ATT}_{\omega,T} = \widehat{B} $ where $\displaystyle (\widehat{\beta}_0,\widehat{B})' = \argmin_{(\beta_0,B)'} S_{\widetilde{w},n}(\beta_0,B)$ and $S_{\widetilde{w},n}(\beta_0,B): = \sum_{t=-\T}^T \widetilde{w}_n(t)(X_t - \beta_0 - B  \indicator{t\geq 1})^2.$ The differencing in $X_t$ eliminates \emph{arbitrarily unbounded} common trends or common shocks in $Y_{1,t}$ and $Y_{0,t}$. 

\Cref{ass:parallel_trends,ass:limited_anticip} hold for suitably transformed potential outcomes. The regression-based form in \eqref{eqn:ATT_lin_mod} renders such adjustments feasible in a straightforward manner by (1) accounting for observed uncommon time-varying heterogeneity: $X_t = \beta_0 + \widetilde{ATT}_n^{w,\psi}  \indicator{t\geq 1} + (\dot{Z}_{1,t}-\dot{Z}_{0,t})\bm{\beta} + U_t$ where $\dot{Z}_{d,t}$ comprises flexible transformations, e.g., polynomials, of observed uncommon time-varying covariates of unit $d \in \{0,1\} $, (2) mitigating persistence in the errors $U_t$ in order to improve efficiency, e.g., $X_t = \beta_0 + \widetilde{ATT}_n^{w,\psi}  \indicator{t\geq 1} + \sum_{l=1}^p X_{t-l}\beta_l + U_t$, (3) taking first differences of $X_t$ to remove unit roots: $(X_t - X_{t-1}) = \beta_0 + \widetilde{ATT}_n^{w,\psi}  \indicator{t\geq 1} + U_t$, (4) fitting a moving average process in \eqref{eqn:ATT_lin_mod} to mitigate serial correlation in $U_t$, and (5) controlling for idiosyncratic time trends in $X_t$ to remove non-stationarity that can otherwise hurt asymptotic identification and inference -- details follow in \Cref{subsec:non_stationarity}. These adjustments are useful. For example, adjusting for time-varying heterogeneity reduces the risk of omitted variable bias in the estimation of $\widetilde{ATT}_n^{w,\psi}$, thereby broadening the applicability of the T-DiD.\footnote{Maintaining the $\widetilde{ATT}_n^{w,\psi}$ notation in the above extensions is to avoid notational clutter.}

\subsection{Relation to alternative estimators}
The Synthetic Control (SC) method is, in many respects, the most closely related to the T-DiD. Other comparable approaches include the Synthetic DiD and Before--After (BA) estimators. This section provides a brief overview of these methods and discusses their connection to the T-DiD.

\subsubsection{Synthetic Controls}
With a single control unit, the convex pre-treatment SC weight is trivially 1. Thus, the SC estimator adapted to the baseline case considered in this paper is simply $ \widehat{ATT}^{SC} = \sum_{t=1}^{T} w_T(t) X_t = \sum_{t=1}^{T} \omega_n(t) X_t $ where $\omega_n(t):= w_T(t)\indicator{t\geq 1} - \psi_{\T }(t)\indicator{t\leq -1} $. Observe for instance that, unlike DiD estimators in general, the $\widehat{ATT}^{SC}$ imposes a stronger condition of a zero mean-difference on the post-treatment untreated potential outcomes, i.e., $\displaystyle \sum_{t=1}^{T} w_T(t)\E[Y_t(0)\mid D=1] = \sum_{t=1}^{T} w_T(t) \E[Y_{0,t}] $. One observes from $\widehat{ATT}^{SC}$ above that no use is made of pre-treatment observations; this suggests that under this strong assumption, no pre-treatment observations are needed.

A clear advantage of the T-DiD estimator emerges. First, the convex-weighted SC is valid under a strong assumption that post-treatment untreated potential outcomes have equal means (uniformly in $\W$). This underlying assumption can fail, and the SC suffers from identification failure as a result. Granted that \Cref{ass:parallel_trends,ass:limited_anticip} are satisfied, the identification of $ATT_{\omega,T}$ holds asymptotically and inference thereon using the T-DiD framework is valid under standard regularity conditions outlined in the next section.

\citet{ferman-pinto-2021-synthetic,tian-lee-panchenko-2024-synthetic} advocate demeaning outcomes before pre-treatment fitting.\footnote{A similar approach is discussed in \citet{doudchenko-imbens-2016}, where an intercept term is included in the pre-treatment fit.} The demeaned convex-weighted SC with a single control unit and convex-weighted ATTs leads to a form of the T-DiD -- see \Cref{App_Sub_Sect:TDiD_SC} for the derivation. However, the \( T \)-DiD framework, taken \emph{en bloc} as a unified approach to identification, estimation, inference, and identification testing, is not a special case of the demeaned synthetic controls in \citet{tian-lee-panchenko-2024-synthetic,ferman-pinto-2021-synthetic}, as those methods require multiple control units and/or multiple outcomes for feasible SC inference.

The synthetic DiD (S-DiD) estimator introduced by \citet{arkhangelsky-etal-2021} seeks to merge desirable features of the SC and the C-DiD using panel data. The estimator is cast in a weighted two-way fixed effects regression framework:
\begin{equation*}
    (\widehat{ATT}^{sdid},\widehat{\mu}_Y,\widehat{a},\widehat{b}) = \argmin_{ATT,\mu_Y,a,b} \Big\{ \sum_{i=1}^N \sum_{t=-\T}^T \big(Y_{it} - \mu_Y - a_i - b_t - D_i\indicator{t\geq 1} ATT \big)^2\hat{\omega}_i^{sdid}\lambda_t^{sdid} \Big\}
\end{equation*}
where the weights are data-driven. Unit-specific weights $\{\hat{\omega}_i^{sdid}, \ i\in [N] \}$ ensure that the average outcome for treated units aligns closely with the weighted average for control units in the pre-treatment periods, while period-specific weights $\{\lambda_t^{sdid}, \ t \in [-\T] \cup [T] \}$ ensure that the average post-treatment outcome for control units differs from the weighted pre-treatment average of the same control units by a constant. In this sense, the S-DiD borrows from the SC the idea of re-weighting control units to approximate pre-treatment trends, and from the DiD the robustness to additive unit-level shifts --- the unit fixed effects $a_i$ absorb permanent level differences that the SC weights need not eliminate exactly. The S-DiD is not suitable for the two-unit baseline setting of primary interest in this paper as it requires several control units; see \citet[Assumption~2(i)]{arkhangelsky-etal-2021}. More fundamentally, its asymptotic theory is driven by a large cross-section: both the number control units and pre-treatment periods ($N_{co}$ and $T_{pre}$, respectively) must diverge at comparable rates, so that identification and inference lean on cross-sectional variation rather than the temporal variation exploited in this paper.

\subsubsection{Before-After (BA)}
Another interesting estimator that emerges in the present context is the before-after estimator \citep[p. 366]{carvalho-Masini-Ricardo-2018arco}. Unlike the SC above which uses post-treatment outcomes of the untreated unit to construct the counterfactual, the BA uses pre-treatment outcomes of the treated unit to construct the counterfactual. The expression of the BA estimator is given by $\widehat{ATT}^{BA}:= \sum_{t=1}^T w_T(t) Y_{1,t} - \sum_{-\tau=1}^{\T } \psi_{\T}(\tau) Y_{1,\tau} = \sum_{t=-\T}^T \omega_n(t) Y_{1,t}$. The BA estimator effectively uses the pre-treatment mean of the outcome of the treated unit to impute the mean post-treatment untreated potential outcome of the treated unit. 

Denote
$ \widehat{ATT}_d^{BA}:= \sum_{t=1}^T w_T(t) Y_{d,t} - \sum_{-\tau=1}^{\T } \psi_{\T}(\tau) Y_{d,\tau} $ for $d\in \{0,1\}$ and note that $\widehat{ATT}_{\omega,T} = \widehat{ATT}_1^{BA} - \widehat{ATT}_0^{BA} = \widehat{ATT}^{BA} - \widehat{ATT}_0^{BA} $. Thus, the T-DiD is the difference between the BA estimators of the treated and control units, while the BA estimator ignores the second term. The BA is therefore not robust to non-negligible correlated or common shocks to which both treated and control units may be exposed. Unlike the T-DiD, the BA is not invariant to the pre-treatment weighting scheme. Thus, using non-uniform weights for the BA estimator may not be meaningful.

\section{Asymptotic Theory}\label{Sect:Asym}

\subsection{Baseline - Two units}\label{Sub_Sect:Baseline_Theory}
For the treatment of the asymptotic theory of the T-DiD estimator \eqref{eqn:DiD_estimator}, it is notationally convenient to opt for an alternative formulation:
\begin{equation}\label{eqn:DiD_estimator2}
    \widehat{ATT}_{\omega,T} = \sum_{t=-\T }^{T} \omega_n(t)X_t,
\end{equation}where $\omega_n(t):= w_T(t)\indicator{t\geq 1} - \psi_{\T }(t)\indicator{t\leq -1} $ and $\omega_n(0)=0$ by notational convention. This paper allows heterogeneity and weak dependence of quite general forms in the sampling process of the data $\{(Y_{1,t},Y_{0,t}), t \in [-\T] \cup [T] \}$. For example, $ \big \{X_t:= Y_{1,t}-Y_{0,t}, \ t \in [-\T] \cup [T] \big \}$ are allowed to be auto-correlated and heterogeneous, thereby accommodating various interesting forms of dependence. To this end, the following statement supplies the definition of near-epoch dependence (NED), which is flexible enough to accommodate several empirically relevant forms of temporal dependence and heterogeneity.

\begin{definition}[Near-Epoch Dependence -- \citet{davidson2021stochastic} Definition 18.2]\label{Def:near_epoch}
    Let $ \{\{V_{nt}\}_{t=-\infty}^{\infty}\}_{n=1}^{\infty} $ be a possibly vector-valued stochastic array defined on the probability space $ (\Omega,\mathcal{F},P) $ and $\mathcal{F}_{n,t-m}^{t+m}:= \sigma(V_{n,t-m},\ldots,V_{n,t+m}) $ where $\sigma(\cdot)$ denotes a sigma-algebra. An integrable array $ \{\{U_{nt}\}_{t=-\infty}^{+\infty}\}_{n=1}^{\infty} $ is $L_p$-NED on $ \{ V_{nt} \} $ if it satisfies $||U_{nt} - \E[U_{nt}|\mathcal{F}_{n,t-m}^{t+m}]||_p \leq d_{nt}\nu_m $ where $\nu_m \rightarrow 0 $ and $ \{d_{nt}\} $ is an array of positive constants.
\end{definition}

\noindent \citet[Chapter 18]{davidson2021stochastic} provides some examples of NED processes: (1) linear processes with absolutely summable coefficients \citep[Example 18.3]{davidson2021stochastic}, (2) bi-linear auto-regressive moving average processes \citep[Example 18.4]{davidson2021stochastic}, (3) a wide class of lag functions subject to dynamic stability conditions \citep[Example 18.5]{davidson2021stochastic}, and GARCH(1,1) processes under given conditions \citep{hansen-1991-garch}.

Define 
\begin{align*}
    s_{\omega,n}^2:= \V\Big[\sum_{t=-\T }^{T} \omega_n(t)X_t \Big] = \sum_{t=-\T }^{T} \omega_n(t)^2\V[X_t] + 2 \sum_{t=-\T }^{T}\sum_{t'=-\T }^{t-1} \omega_n(t)\omega_n(t')\mathrm{cov}[X_t,X_{t'}]
\end{align*}
and $ X_{nt}:= \omega_n(t)(X_t - \E[X_t])/s_{\omega,n} $, then $\big\{X_{nt}, \ t\in [-\T] \cup [T], n \in \mathbb{N} \big\}$ constitutes a triangular array of zero-mean random variables. Also, define $\sigma_{nt}:= \sqrt{\E[X_{nt}^2]} $, $c_{nt}:= \max\{\sigma_{nt}, s_{\omega,n}\}$, and $d_{nt} \leq \bar{d}||X_{nt}||_2$ uniformly in $n$ and $t$ for some constant $\bar{d}>0$. The following set of assumptions is important for establishing the asymptotic normality of the T-DiD estimator. Let $r$ and $C$ be constants such that $r>2$ and $0<C<\infty$.
\begin{assumption}\label{ass:dominance_Y}
    $ \displaystyle \sup_{t \in [-\T] \cup [T]} ||(Y_{1,t} - Y_{0,t}) - \E[(Y_{1,t} - Y_{0,t})]||_r \leq C $ uniformly in $n$.
\end{assumption}
\noindent \Cref{ass:dominance_Y} is weak in two dimensions: (1) it is imposed on the difference $X_t:=Y_{1,t} - Y_{0,t}$ and \emph{not} on the levels $Y_{d,t}$, $(d,t) \in \{0,1\} \times [-\T] \cup [T] $ and (2) it is imposed on deviations of $X_t$ from its mean $\E[X_t]$. \Cref{ass:dominance_Y} thus does not rule out non-stationary common shocks. \Cref{ass:dominance_Y} also does not rule out possibly unbounded deterministic trends in $X_t$. It is weaker than sub-Gaussianity and variance-homogeneity conditions needed for some SC methods -- cf. \citet[Assumption 1]{arkhangelsky-etal-2021} and \citet[Assumption 3]{Sun-Ben-Michael-Feller-2023-using}. 

The following is a standard regularity condition on $s_{\omega,n}$.
\begin{assumption}\label{ass:bound_sn}
$ \sqrt{n}s_{\omega,n} \geq \epsilon $ for some small constant $\epsilon>0$ uniformly in $\W^2$.
\end{assumption}
\noindent \Cref{ass:bound_sn} is a flexible technical condition on $s_{\omega,n}$ that accommodates, e.g., covariance stationarity, i,e., $s_{\omega,n}=\mathcal{O}(n^{-1/2})$. It, however, rules out degeneracy in the limit: $\sqrt{n}s_{\omega,n} \rightarrow 0 $ as $n\rightarrow \infty$. As \Cref{ass:bound_sn} restricts only $s_{\omega,n}$, standard deviations $\sigma_{nt}$ for some $t \in [-\T] \cup [T] $ are allowed to be zero (in the limit). \Cref{ass:bound_sn}, however, restricts this form of degeneracy.

This paper uses the general representation $X_{nt} = g_t(\ldots,V_{n,t-1},V_{nt},V_{n,t+1},\ldots)$ for some mixing array $V_{nt}$ and time-dependent function $g_t(\cdot)$. The following assumption ensures that this representation satisfies the required near epoch dependence and mixing properties.
\begin{assumption}\label{ass:Sampling}
    (a) $\{X_{nt}\} $ is $L_2$-near epoch dependent of size $-1/2$ with respect to a constant array $\{d_{nt}\}$ on $ \{V_{nt}\} $.    
    (b) $ \{V_{nt}\} $ is an $\alpha$-mixing array of size $-r/(r-1)$.
\end{assumption}
\noindent \Cref{ass:Sampling} allows the observed (weighted) data to be heterogeneous and weakly dependent. Besides, $\{ATT(t), \ t \in [T]\} $ and the user-specified weighting scheme $ \{\omega_n(t), \ t\in [-\T] \cup [T] \} $ are also allowed to be heterogeneous in $t$. \Cref{ass:Sampling} is a more general form of weak dependence than commonly imposed mixing conditions in the literature -- cf. \citet{chan2021pcdid,Sun-Ben-Michael-Feller-2023-using,fry-2024-method}.

The next assumption is essential in establishing central limit theorem results under NED.
\begin{assumption}\label{ass:technical_CLT}
\[\max_{1\leq j\leq r_n+1}M_{nj}=o(b_n^{-1/2}) \text{ and }\sum_{j=1}^{r_n} M_{nj}^2 = \mathcal{O}(b_n^{-1}) \]
\noindent where $\displaystyle M_{nj}=\max_{(j-1)b_n+1\leq t \leq jb_n} c_{nt}$, for $j \in [r_n]$, $\displaystyle M_{n,r_n+1} = \max_{r_nb_n+1\leq t \leq n} c_{nt}$, $b_n= \lfloor n^{1-\alpha} \rfloor $, $\alpha\in(0,1]$, and $r_n= \lfloor n/b_n \rfloor $.
\end{assumption}
\noindent While \Cref{ass:Sampling} allows heterogeneity in observed data, \Cref{ass:technical_CLT} serves to restrict the degree of heterogeneity allowed -- see \citet{deJong-1997central}. Also, while \Cref{ass:bound_sn} characterises a lower bound on the growth rate of $s_{\omega,n}$ allowed, \Cref{ass:technical_CLT} characterises an upper bound on its growth rate, i.e., $n^{(1-\alpha)/2}c_{nt}= n^{(1-\alpha)/2}\max\{\sigma_{nt}, s_{\omega,n}\} = o(1), \ \alpha \in (0,1] $. The absence of this upper bound on the growth rate of $s_{\omega,n}$ can result in long-range dependence, and asymptotic normality fails.

With the foregoing assumptions in hand, the asymptotic normality of the T-DiD follows. 
\begin{theorem}[Asymptotic Normality]\label{Theorem:AsympN}
\Copy{Key:Theorem:AsympN}{Under \Cref{ass:parallel_trends,ass:limited_anticip,ass:Sampling,ass:technical_CLT,ass:bound_sn,ass:TTratio_lambda_n,ass:dominance_Y}, $
s_{\omega,n}^{-1}(\widehat{ATT}_{\omega,T} - ATT_{\omega,T}) \xrightarrow{d} \mathcal{N}(0,1) $ uniformly in $\W^2$.}
\end{theorem}
\noindent In practice, $s_{\omega,n}$ is estimated using a heteroskedasticity and auto-correlation robust procedure, e.g., the \citet{newey-west-1987-simple} procedure. Using the self-normalised $t$-test statistic of \citet{chernozhukov-Wuthrich-Zhu-2024t} is an alternative inference procedure. Exploring this path to inference in the current framework is, however, left for future work due to considerations of scope and space. The asymptotic theory for the regression-based estimator of $ATT_{\omega,T}$ in \eqref{eqn:ATT_lin_mod} and its variants is analogous to the preceding analysis; formal details are omitted. \Cref{App_Sect:Det_Trend} formally considers deterministic time trends.

\subsection{Extension -- Multiple control units}\label{SubSect:Ext_Mult_Control}

Beyond the two-unit baseline case, practical cases often involve a fixed number of units with large $\T, T$ (e.g., \citet{marcus-santanna-2021role,grooms-2015-enforcing}), where both SC and T-DiD might be suitable. This section focuses on scenarios with multiple control units, say two or three, and large $\T$ and $T$, which still fall within the T-DiD framework. In particular, the availability of multiple control units is exploited in this paper to propose an over-identifying restrictions test of identification. Two cases of interest arise when multiple units are available: (1) \( J > 1 \) control units for a given treated unit, and (2) \( I > 1 \) treated units, each of which has at least one corresponding control unit. To maintain focus on the primary case of a single treated unit, the generalisation to multiple treated units is deferred to \Cref{App_Sect:Mult_Treat}. It is assumed that the researcher has control units $j\in [J]$ for the treated unit. Whenever there are $J>1$ valid control units for the treated unit, one can, using the theory elaborated in \Cref{Sub_Sect:Baseline_Theory} above, obtain $J$ estimates of $ATT_{\omega,T}$. In what follows, let $\widehat{ATT}_{\omega,jT}$ denote the estimator \eqref{eqn:DiD_estimator} using control unit $j \in [J] $ and the treated unit.

Let $\widehat{ATT}_{\omega,\cdot T}:=(\widehat{ATT}_{\omega, 1 T},\ldots,\widehat{ATT}_{\omega, J T})'$ denote a $J\times 1$ vector of T-DiD estimators of $ATT_{\omega,T}$. Define the $J\times J$ matrix $S_{\omega,n}:= \E[(\widehat{ATT}_{\omega,\cdot T} - \mathbbm{1}_JATT_{\omega,T})(\widehat{ATT}_{\omega,\cdot T} - \mathbbm{1}_J ATT_{\omega,T})']$ where $\mathbbm{1}_J$ denotes a $J \times 1$ vector of ones. The following presents an extension of \Cref{ass:bound_sn} to the multiple-control setting. It requires that the minimum eigenvalue of \( nS_{\omega,n} \) be uniformly bounded below by \( \epsilon^2 > 0 \).

\begin{namedassumption}{\ref{ass:bound_sn}-Ext}\label{ass:bound_sn_ext}
    Uniformly in $n$, $\rho_{\mathrm{min}}(nS_{\omega,n})\geq\epsilon^2$.
\end{namedassumption}

\noindent The following provides the \Cref{Theorem:AsympN}-analogue of the multiple-control case. Define $\mathbb{S}^J:= \{ \uptau \in \mathbb{R}^J: ||\uptau|| = 1 \} $ where $J \in \mathbb{N} $.
\begin{corollary}\label{Theorem:AsympN_ext}
\Copy{Key:Theorem:AsympN_ext}{Suppose \Cref{ass:parallel_trends,ass:limited_anticip,ass:Sampling,ass:technical_CLT,ass:TTratio_lambda_n,ass:dominance_Y} hold for each of the $J$ unique treated-control pairs. Suppose further that \Cref{ass:bound_sn_ext} holds, then $(\uptau_J'S_{\omega,n}\uptau_J)^{-1/2}\uptau_J'(\widehat{ATT}_{\omega,\cdot T} -\mathbbm{1}_J ATT_{\omega,T} ) \xrightarrow{d} \mathcal{N}(0,1)$ uniformly in $\W^2 $ for all $\uptau_J \in \mathbb{S}^J$.}
\end{corollary}

\noindent A useful deduction from \Cref{Theorem:AsympN_ext} using the Cram\'er-Wold device is that $S_{\omega,n}^{-1/2}(\widehat{ATT}_{\omega,\cdot T} - \mathbbm{1}_J ATT_{\omega,T}) \xrightarrow{d} \mathcal{N}(0,\mathrm{I}_J)$ under the imposed conditions. Although theoretically useful in its own right, \Cref{Theorem:AsympN_ext} is not interesting from a practical point of view as no direction vector $\uptau_J \in \mathbb{S}^J $ is specified. With asymptotically ``efficient" linear combinations of $\widehat{ATT}_{\omega,\cdot T}$ and over-identifying restrictions tests of identification in mind, it is useful to cast the problem of finding an optimal linear combination as follows: 
\begin{align}\label{eqn:Eff_ATT}
    \widehat{ATT}_{\omega,T}^{h_J}: = \argmin_{a} \big\{ (\widehat{ATT}_{\omega,\cdot T} - \mathbbm{1}_Ja)'H_n(\widehat{ATT}_{\omega,\cdot T} - \mathbbm{1}_Ja) \big\} = h_J'\widehat{ATT}_{\omega,\cdot T}
\end{align}
\noindent where $h_J:= H_n'\mathbbm{1}_J(\mathbbm{1}_J'H_n\mathbbm{1}_J)^{-1}$, $H_n$ is a sequence of $J\times J$ positive definite matrices, and $\V[\widehat{ATT}_{\omega,T}^{h_J}] = h_J'S_{\omega,n}h_J = (\mathbbm{1}_J'H_n\mathbbm{1}_J)^{-2}\mathbbm{1}_J'H_nS_{\omega,n}H_n\mathbbm{1}_J $. Thus $\widehat{ATT}_{\omega,T}^{h_J}$ entails a specific linear combination of $\widehat{ATT}_{\omega, \cdot T}$, namely $h_J$ which in turn is a function of the positive definite $H_n$. Using standard Generalised Method of Moments (GMM) and minimum distance estimation (MD) arguments, e.g., \citet[Chap. 14]{wooldridge-2010}, the efficient choice of $H_n$ within the MD class of $ATT_{\omega,T}$ estimators $\widehat{ATT}_{\omega,T}^{h_J}$ is $H_n=S_{\omega,n}^{-1}$ whence
\[
    \widehat{ATT}_{\omega,T}^*:= \widehat{ATT}_{\omega,T}^{h_J^*} \text{ with } (s_{\omega,n}^*)^2: = \V[\widehat{ATT}_{\omega,T}^*]=1/(\mathbbm{1}_J'S_{\omega,n}^{-1}\mathbbm{1}_J).
\]
\noindent For completeness, the efficiency result is provided in \Cref{Prop:Eff_ATT}. Let $\Hyp^J:= \{h\in \mathbb{R}^J: h'\mathbbm{1}_J = 1\} $ denote the space of vectors in $\mathbb{R}^J$ whose elements sum to 1, and observe that $h_J\in \Hyp^J$. 

\begin{proposition}\label{Prop:Eff_ATT}
    \Copy{Key:Prop:Eff_ATT}{Under the assumptions of \Cref{Theorem:AsympN_ext}, $h_{J}^*:= S_{\omega,n}^{-1}\mathbbm{1}_J(\mathbbm{1}_J'S_{\omega,n}^{-1}\mathbbm{1}_J)^{-1} $ delivers the most efficient estimator of $ATT_{\omega,T}$ among all $h_{J} \in \Hyp^J$.}
\end{proposition}

\section{Over-identifying Restrictions Test}\label{Sect:Tests}

Whenever all available candidate control units of the treated unit are valid, each element in $\widehat{ATT}_{\omega,\cdot T}$ is a consistent estimator for the same $ATT_{\omega,T}$. This provides the basis for a test of over-identifying restrictions. Thus, the proposed test is similar in spirit to that of \citet{marcus-santanna-2021role}, exploiting over-identification for efficiency and identification testing.\footnote{Such a test is not feasible in the SC framework, when a pre-treatment fit is taken of \emph{all} control outcomes.}

\subsection{Two-candidate controls}\label{Subsub:two_control}
To drive intuition, consider the simple case with two candidate control units. One would expect the same $ATT_{\omega,T}$ up to a negligible $o(n^{-1/2})$ bias term from both $\widetilde{ATT}_{1n}^{w,\psi}$ and $\widetilde{ATT}_{2n}^{w,\psi}$, where $\widetilde{ATT}_{jn}^{w,\psi}$ is the effective estimand $\widetilde{ATT}_n^{w,\psi}$ using control unit $j\in [2]$. Thus from the decomposition \eqref{eqn:Decomp_ATTn},
\begin{equation*}
    \widetilde{ATT}_{1n}^{w,\psi} - \widetilde{ATT}_{2n}^{w,\psi} = R_{1n}^{w,\psi} - R_{2n}^{w,\psi}
\end{equation*} where $R_{jn}^{w,\psi}$ denotes $R_n^{w,\psi}$ using control unit $j$. Let $Y_{0,t}^{(j)}$ denote the outcome of the $j$'th control unit at period $t$, then
$$
\widehat{ATT}_{\omega,1T} - \widehat{ATT}_{\omega,2T} = \sum_{t=1}^{T}w_T(t)(Y_{0,t}^{(1)} - Y_{0,t}^{(2)}) - \sum_{-\tau=1}^{\T }\psi_{\T }(-\tau)(Y_{0,\tau}^{(1)} - Y_{0,\tau}^{(2)}).
$$

The expression is a difference-in-differences estimator akin to \eqref{eqn:DiD_estimator} using two control units since the outcome of the treated unit cancels out. This baseline case is particularly appealing because it involves a straightforward implementation: estimating $ATT_{\omega,T}$ with just two controls and then testing for a null effect using a $t$-test. Suppose \Cref{Theorem:Identification} holds for each candidate control unit --- an implication of \Cref{ass:parallel_trends,ass:limited_anticip}, with \Cref{ass:TTratio_lambda_n} maintained --- then it follows that $R_{jn}^{w,\psi} = o(n^{-1/2})$ for each $j\in [2]$. This further implies $\widetilde{ATT}_{1n}^{w,\psi} - \widetilde{ATT}_{2n}^{w,\psi} = o(n^{-1/2})$. $\widehat{ATT}_{\omega,1T} - \widehat{ATT}_{\omega,2T}$ would not be statistically significant if both controls were valid, i.e., each control guarantees that \Cref{Theorem:Identification} hold for a given treated unit.

\subsection{Multiple candidate controls}

Generally, when more than two candidate control units are available, it follows from the identification conditions in \Cref{ass:parallel_trends,ass:limited_anticip}—see \Cref{Theorem:Identification}—that \( \| R_{\cdot n}^{w,\psi} \| = o(n^{-1/2}) \) uniformly over \( \mathcal{W}^2 \), where \( R_{\cdot n}^{w,\psi} := (R_{1n}^{w,\psi}, \ldots, R_{Jn}^{w,\psi})' \) denotes the \( J \times 1 \) vector of bias indexed by the \( J \) control units for a given treated unit. In view of \Cref{ass:parallel_trends}, \Cref{ass:limited_anticip}, and \Cref{rem:weak_suff_idcond}, one can characterise the test hypotheses using the representation 
$$
||\sqrt{n}R_{\cdot n}^{w,\psi}|| = C_{\omega,n} n^{1/2 - \tilde{\delta}}
$$ where $ \{C_{\omega,n}: n \geq 1 \}$ is a sequence of bounded positive constants. $\tilde{\delta} > 1/2$ if asymptotic identification (\Cref{Theorem:Identification}) holds for each $j \in [J] $, and $\tilde{\delta}\leq 1/2$ otherwise for at least one control unit $j$. The following correspond to the null, local alternative, and fixed alternative hypotheses:
\[ 
\Hyp_o: \tilde{\delta} > 1/2; \quad \Hyp_{an}: \tilde{\delta} = 1/2; \quad \text{and} \quad \Hyp_{a}: \tilde{\delta} < 1/2. 
\]
\noindent If \Cref{ass:parallel_trends,ass:limited_anticip} hold for each control unit \( j \in [J] \), then the null hypothesis \( \Hyp_o \) is satisfied. Exact identification in the C-DiD set-up holds under $\Hyp_o$ with $\tilde{\delta} = \infty$. Thus, the current framework only requires $\tilde{\delta} > 1/2$. Under the local alternative $\Hyp_{an}$ or fixed alternative $\Hyp_a$, either \Cref{ass:parallel_trends} or \Cref{ass:limited_anticip} is violated for at least one control unit.

Define $\widehat{S}_{\omega,n}$ as the estimator of the covariance matrix $\widetilde{S}_{\omega,n}:=\E[(\widehat{ATT}_{\omega,\cdot T} - \E[\widehat{ATT}_{\omega,\cdot T}])(\widehat{ATT}_{\omega,\cdot T} - \E[\widehat{ATT}_{\omega,\cdot T}])']$.\footnote{The separate notation of the covariance matrix is introduced since $\widetilde{S}_{\omega,n}=S_{\omega,n}$ only holds under $\Hyp_o$. $\widetilde{S}_{\omega,n}$, in contrast, remains a valid covariance matrix under $\Hyp_o$, $\Hyp_{an}$, and $\Hyp_a$.} The following consistency condition is imposed on $\widehat{S}_{\omega,n}$.

\begin{assumption}\label{ass:consis_Sn}
    $ \displaystyle \frac{\uptau_J'\widehat{S}_{\omega,n}\uptau_J}{\uptau_J'\widetilde{S}_{\omega,n}\uptau_J} \xrightarrow{p} 1 $ for all $\uptau_J \in \mathbb{S}^J $.
\end{assumption}

\noindent The over-identifying restrictions test statistic for testing $\Hyp_o$ above is
\begin{align*}
    \widehat{Q}_{\omega,n}: = (\widehat{ATT}_{\omega,\cdot T} - \mathbbm{1}_J\widehat{ATT}_{\omega,T}^*)'\widehat{S}_{\omega,n}^{-1}(\widehat{ATT}_{\omega,\cdot T} - \mathbbm{1}_J\widehat{ATT}_{\omega,T}^*).
\end{align*}

While \Cref{ass:parallel_trends,ass:limited_anticip} imply $\Hyp_o$, the converse generally does not hold. There exist violations of \Cref{ass:parallel_trends}, \Cref{ass:limited_anticip}, or both that imply neither $\Hyp_{an}$ nor $\Hyp_{a}$. Technically, the proposed test derives power under the condition that $\uptau_J^R:= ||\sqrt{n}R_{\cdot n}||^{-1} \sqrt{n}R_{\cdot n} $ does not lie in the null space of $(n\widetilde{S}_{\omega,n})^{-1/2}[\mathrm{I}_J - \widetilde{P}_{n}](n\widetilde{S}_{\omega,n})^{-1/2}$ as $n\rightarrow \infty$ where $\widetilde{P}_{n}:= \widetilde{S}_{\omega,n}^{-1/2}\mathbbm{1}_J(\mathbbm{1}_J'\widetilde{S}_{\omega,n}^{-1}\mathbbm{1}_J)^{-1} \mathbbm{1}_J'\widetilde{S}_{\omega,n}^{-1/2} $. Thus, the test has trivial power under $\Hyp_{an}$ and $\Hyp_a$ if the above condition is violated. This constitutes the source of the inconsistency of over-identifying restrictions tests, including the one proposed above.\footnote{\citet{guggenberger2012note} provides a detailed discussion.}

As to whether ${\uptau_J^R}'(n\widetilde{S}_{\omega,n})^{-1/2}[\mathrm{I}_J - \widetilde{P}_{n}](n\widetilde{S}_{\omega,n})^{-1/2}\uptau_J^R=o(n^{-1/2})$ can hold under $\Hyp_{an}$ or $\Hyp_{a}$ is less of a statistical problem and more of the economic question at hand. Thus, from a statistical point of view, the above test is not consistent, i.e., it cannot detect \emph{all} violations of $\Hyp_o$. Thus, a good understanding of the economics of such a violation would be useful in complementing the over-identifying restrictions test of $\Hyp_o$. To illustrate in the context of the empirical example, suppose the asymptotic parallel trends assumption (\Cref{ass:parallel_trends}), with Togo as a control for Benin, were violated. If the path of pre-post log GDP per capita averages of Cameroon, a second control unit, coincided with that of Togo, the proposed test would be incapable of detecting this violation. In other words, the test fails to reject the null if the trend biases across control units are identical up to a $o(n^{-1/2})$ term.

Let $F_n^\omega$ be the distribution from which observed data are drawn, and let $\mathcal{F}_o^\omega \subset \mathcal{F}^\omega$ characterise the set of distributions over which the test is not consistent, i.e., the set of distributions of the data for which $\Hyp_o$ is violated but the test lacks power, i.e., ${\uptau_J^R}'(n\widetilde{S}_{\omega,n})^{-1/2}[\mathrm{I}_J - \widetilde{P}_{n}](n\widetilde{S}_{\omega,n})^{-1/2}\uptau_J^R=o(n^{-1/2})$.\footnote{The superscript $\omega $ in $F_o^\omega $ and $ \mathcal{F}^\omega$ emphasises the dependence on the weighting scheme in $\W^2$.} Thus, $\displaystyle 0 < \theta:= \lim_{n\rightarrow \infty} ||[\mathrm{I}_J - \widetilde{P}_{n}](n\widetilde{S}_{\omega,n})^{-1/2}\uptau_J^R||^2\cdot ||\sqrt{n}R_{\cdot n}||^2$ if $F_n^\omega \not\in \mathcal{F}_o^\omega $. Observe that distributions for which $\Hyp_o$ is true, namely $\mathcal{F}_{\Hyp_o}^\omega$ belong to $\mathcal{F}_o^\omega$. The following shows the validity and non-trivial power of the test over the set of distributions $\mathcal{F}^\omega\setminus\mathcal{F}_o^\omega$.

\begin{theorem}\label{Theorem:Test_Implication}
\Copy{Key:Theorem:Test_Implication}{Let \Cref{ass:Sampling,ass:technical_CLT,ass:bound_sn,ass:TTratio_lambda_n,ass:dominance_Y,ass:consis_Sn} hold, then (a) under $\Hyp_o$, $\widehat{Q}_{\omega,n} \xrightarrow{d} \chi_{J-1}^2$; (b) under $\Hyp_{an}$ and if $F_n^\omega \not\in \mathcal{F}_o^\omega $, $\widehat{Q}_{\omega,n} \xrightarrow{d} \chi_{J-1}^2(\theta)$; and (c) under $\Hyp_{a}$ and if $F_n^\omega \not\in \mathcal{F}_o^\omega $, $\widehat{Q}_{\omega,n} \rightarrow \infty$ uniformly in $\W^2$.}
\end{theorem}
\noindent Simulations in \Cref{Sect:Sim} support the theory developed in this paper by demonstrating the strong performance of the DiD estimator and the identification test in terms of low bias, good empirical size control, and substantial power under alternatives.

\section{Empirical Analysis}\label{Sect:Emp}

This section concerns the empirical results. The first part briefly explains the economic model of the democracy-growth nexus, the second describes the data,  and the third applies the T-DiD to estimate the effect of democracy on Benin's GDP per capita.

\subsection{The economic model}
On the one hand, democracy is a useful indicator of stability in a country. Thus, democracy can be an indicator of low country risk from the perspective of investors. Investment is useful not just in maintaining depreciating capital but also in growing the capital stock. From standard macroeconomic models, e.g., the Solow model, output per worker increases in capital per worker. In this sense, one can expect democracy to drive economic growth. On the other hand, democracy involves elections at sometimes unpredictable frequencies, which can be a source of uncertainty since a change of government can lead to radical changes in policy direction. To the extent that such uncertainty undermines investor confidence, one can expect democracy to negatively impact economic growth. Thus, it may not be \textit{a priori} clear if and in which direction democracy generally impacts economic growth. The answer may therefore be country-specific.

\subsection{Data}\label{Sect:Data}
The democracy measure is sourced from the \textit{Varieties of Democracy} (V-Dem) project. The data comprise annual observations from 1960 to 2018. The measure of democracy is defined as the average of the five high-level democracy indices of the V-Dem project, namely the electoral democracy index, liberal democracy index, participatory democracy index, deliberative democracy index, and egalitarian democracy index. Each index is an aggregate index capturing specific dimensions of the concept of democracy and ranges from 0 (zero democracy) to 1 (full democracy). 

\pgfplotsset{my personal style/.style=
{font=\footnotesize},width=12cm,height=7.0 cm} 

\begin{figure}[h!]
\begin{center}
\caption{Democracy index}
\label{Fig:comp1}

\begin{tikzpicture}[]
\begin{axis}[my personal style,minor x tick num=1,
xlabel=Period,
ymin=,
ymax=1,
xticklabels={{},{},1960,1970,1980,1990,2000,2010,2020},
ylabel=,
title=,legend style={
at={(0.50,0.9)},
anchor=north east}]

\addplot[no markers,color=black,line width=1pt,solid,mark size=0.5pt,smooth] table[col sep=tab,x=Year,y=Benin] {figures/DemocracyIndexAll.txt};

\addplot[no markers,color=black,line width=1pt,densely dotted,mark size=0.5pt,smooth] table[col sep=tab,x=Year,y=Togo] {figures/DemocracyIndexAll.txt};

\addplot[draw=black] coordinates {(1990,0) (1990,0.8) };
  \addplot[draw=gray] coordinates {(1960,0.5) (2016,0.5) };
  
\legend{Benin, Togo, Threshold (0.5)}
\end{axis}
\end{tikzpicture}
\end{center}
{\footnotesize
\textit{Notes:} The plots provide the aggregate democracy indices for Benin and Togo from 1960 through 2018. The thin vertical line at 1990 marks the beginning of Benin's democratisation process, whereas the horizontal line marks the 0.5 threshold for democracy.
}
\end{figure}
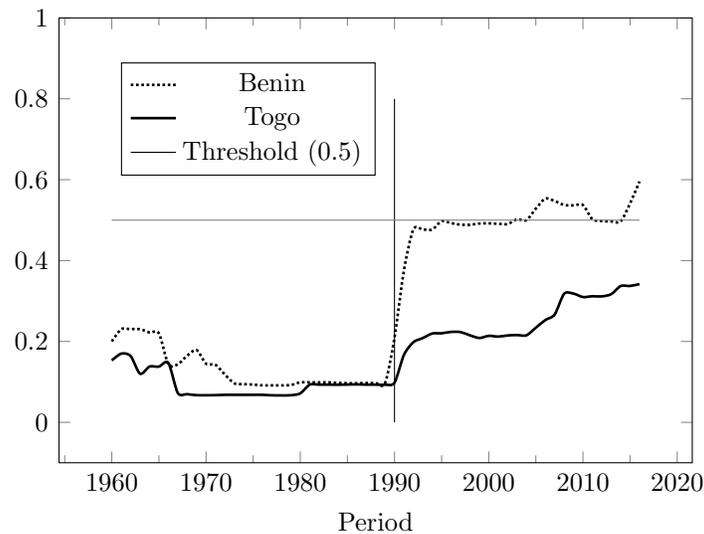

\Cref{Fig:comp1} plots the democracy index for Benin and Togo. As observed, Benin and Togo have similar levels of democracy until 1990; the average gap in democracy between the two countries is around 0.039, with Benin being the more democratic country. Starting in 1990, when the democratisation process begins, the two countries diverge; the gap now stands at approximately 0.269, seven times the pre-1990 gap, with Benin becoming the more democratic country. Overall, \Cref{Fig:comp1} indicates a 600-percent gain in Benin's democracy relative to Togo during the 1990s, 2000s, and 2010s. Although democracy in Togo improves since 1990, it does not reach a level that qualifies as democratic. This is why Togo serves as a control unit for Benin. For the empirical analyses, the main transition window considered is 1990-1992, and treatment remains an absorbing state up to 2018.

The outcome variable is GDP per capita of Benin and Togo \emph{in constant 2015 US dollars} from 1960 to 2018. The data are sourced from the World Bank Development Indicators.\footnote{The preferred measure of the outcome, PPP-adjusted GDP per capita, is not available for the 1960-1989 pre-treatment period.} \Cref{fig_gdp} shows the GDP per capita of Benin and Togo from 1960 to 2018. As can be observed, neither country's GDP per capita dominates until 1990 when Benin begins the democratisation process. After that, Benin's economic performance markedly surpasses Togo's.

\begin{figure}[!htbp]
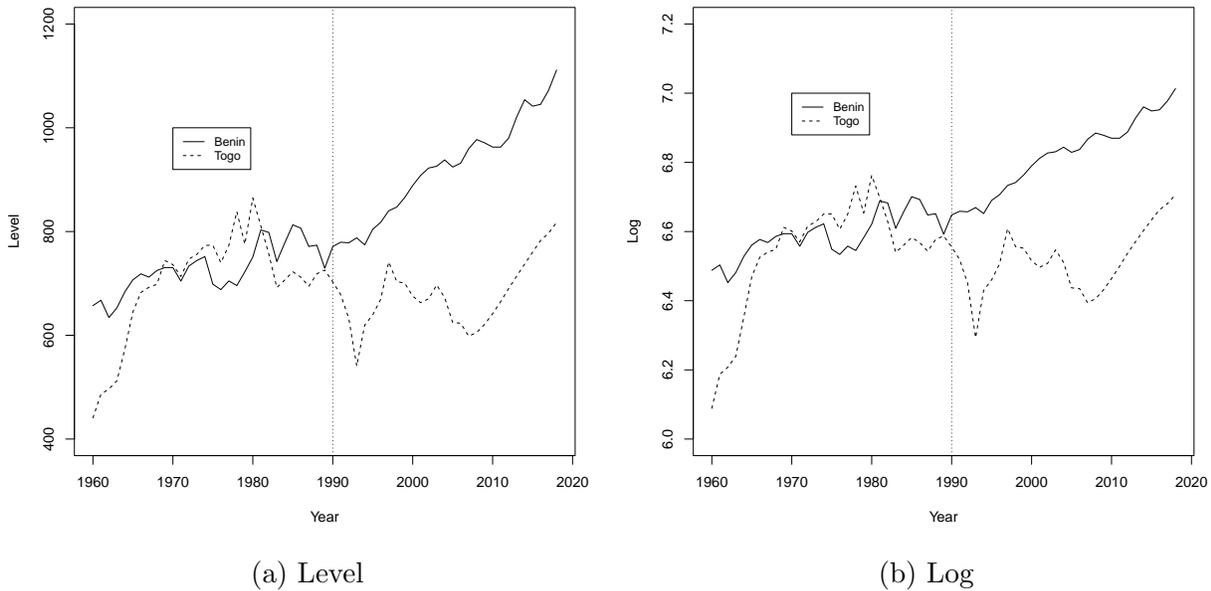

\centering 
\caption{GDP per Capita}
\begin{subfigure}{0.49\textwidth}
\centering
\includegraphics[width=1\textwidth]{figures/out_GDP_per_Capita}
\caption{Level}
\end{subfigure}
\begin{subfigure}{0.49\textwidth}
\centering
\includegraphics[width=1\textwidth]{figures/out_logGDP_per_Capita}
\caption{Log}
\end{subfigure}
\label{fig_gdp}
\begin{justify}
{\footnotesize
\textit{Notes:} The above plots show level and log GDP per capita, measured in constant 2015 US dollars. The vertical dotted lines at 1990 mark the beginning of Benin's democratisation process.
}
\end{justify}
\end{figure}

\subsection{The effect of democracy on growth}
\Cref{Tab:Empirical_Results} presents the empirical results. Column-wise, results in Panel A are based on log GDP per capita, while results in Panel B are based on GDP per capita in levels (in constant 2015 US dollars). For each panel, there are four specifications by the transition window taken as the transition period of Benin from an autocracy to a democracy: (1) 1990-1992, (2) 1990, (3) 1991, and (4) 1992. Row-wise, the first part of \Cref{Tab:Empirical_Results} corresponds to estimates of $ATT_{\omega,T}$ and the parameter $\beta_1$ on $X_{t-1}$, which captures persistence in the series $X_t$. The estimates are accompanied by heteroskedasticity- and auto-correlation-robust standard errors (in parentheses). The second part comprises typical diagnostic tests: (1) the Augmented Dickey-Fuller (ADF) test, (2) the Kwiatkowski-Phillips-Schmidt-Shin (KPSS) test, and (3) the Durbin-Watson (DW) test with respective null hypotheses (1) unit root, (2) trend stationarity, and (3) zero auto-correlation against a two-sided alternative. The uniform weighting scheme $ \omega_n(t) = T^{-1}\indicator{t\geq 1} - \T^{-1}\indicator{t\leq -1}  $ is applied throughout. Estimates of \( ATT_{\omega,T} \) are interpreted as the loss in GDP per capita that would have occurred had Benin not democratised.\footnote{For results in Panel A, note that $ \displaystyle \log(Y_t(1)) - \log(Y_t(0)) = -\log\Big(\frac{Y_t(0)}{Y_t(1)}\Big) = -\log\Big(\frac{Y_t(0)}{Y_t(1)} - 1 + 1\Big) = -\log\Big(\frac{Y_t(0) - Y_t(1)}{Y_t(1)} + 1\Big) \approx - \frac{Y_t(0) - Y_t(1)}{Y_t(1)} $ for small changes.}
 
\begin{table}[!htbp]
\centering
\small
\caption{The Effect of Democracy on Growth - Empirical Estimates}
\setlength{\tabcolsep}{2pt}
\begin{tabular}{@{}rccccccccc@{}}
\toprule
\multicolumn{1}{l}{} & \multicolumn{4}{c}{Panel A}  & & \multicolumn{4}{c}{Panel B}\\ 
\multicolumn{1}{l}{} & \multicolumn{4}{c}{Log GDP per capita}  & & \multicolumn{4}{c}{GDP per capita}\\ \cmidrule(l){2-5} \cmidrule(r){7-10}
\multicolumn{1}{l|}{Transition Window} & \multicolumn{1}{c}{ } & \multicolumn{1}{c}{ } & \multicolumn{1}{c}{ } & \multicolumn{1}{c}{ } & $\quad$ & \multicolumn{1}{c}{ } & \multicolumn{1}{c}{ } & \multicolumn{1}{c}{ } & \multicolumn{1}{c}{ } \\ 
 \multicolumn{1}{l|}{$(t=0)$}  & \multicolumn{1}{c}{1990-1992} & \multicolumn{1}{c}{1990} & \multicolumn{1}{c}{1991} & \multicolumn{1}{c}{1992} & $\quad$ & \multicolumn{1}{c}{1990-1992} & \multicolumn{1}{c}{1990} & \multicolumn{1}{c}{1991} & \multicolumn{1}{c}{1992} \\ \midrule
\multicolumn{1}{l|}{Estimate} & & & & & & & & & \\ \cmidrule(r){1-1}
\multicolumn{1}{r|}{$ATT_{\omega,T}$}
&0.064  &0.078  &0.074  &0.059  &       &46.586 &56.536 &54.101 &42.786 \\ 
\multicolumn{1}{r|}{ } 
&(0.017)   &(0.018)   &(0.019)   &(0.017)   &        &(16.178)  &(14.886)  &(16.040)   &(16.076)  \\ 
\multicolumn{1}{r|}{p-val} 
&0.000     &0.000     &0.000     &0.000     &      &0.004 &0.000     &0.001 &0.008 \\
\multicolumn{1}{r|}{2015 BMK} 
&  &  &  &  &        &  4.5\% &5.4\% &5.2\% &4.1\% \\ 
\multicolumn{1}{r|}{ } 
&  &  &  &  &  &  &  &  &  \\
\multicolumn{1}{r|}{ $X_{t-1}$} 
&0.803 &0.775 &0.775 &0.802 &      &0.835 &0.806 &0.805 &0.835 \\ 
\multicolumn{1}{r|}{ } 
&(0.050)      &(0.060)      &(0.060)      &(0.049)     &          &(0.060)      &(0.060)      &(0.064)     &(0.059)     \\   
\midrule
\multicolumn{1}{l|}{Diagnostic Tests} & & & & & & & & & \\ \cmidrule(r){1-1}
\multicolumn{1}{r|}{ADF p-val} &0.028 &0.025 &0.035 &0.043 &      &0.023 &0.034 &0.033 &0.024 \\  
\multicolumn{1}{r|}{KPSS p-val} &0.100  &0.100  &0.100  &0.100  &     &0.100  &0.100  &0.100  &0.100  \\  
\multicolumn{1}{r|}{DW p-val} &0.159 &0.977 &0.798 &0.153 &      &0.432 &0.822 &0.906 &0.441 \\
\bottomrule
\end{tabular}\label{Tab:Empirical_Results}
\begin{justify}
{\footnotesize
    \textit{Notes:} GDP per capita is measured in constant 2015 US dollars. All estimates are based on quasi-differenced $X_t$ following \eqref{eqn:ATT_lin_mod2} in the supplement. All standard errors (in parentheses) are heteroskedasticity and auto-correlation robust using the \citet{newey-west-1987-simple} procedure. ADF is the Augmented Dickey-Fuller test of a unit root null hypothesis. KPSS denotes the Kwiatkowski-Phillips-Schmidt-Shin (KPSS) test of a trend stationarity null hypothesis. DW is the Durbin-Watson test of a zero auto-correlation null hypothesis against the two-sided alternative.}
\end{justify}
\end{table}

\begin{table}[!htbp]
\centering
\small
\caption{Empirical Estimates with Transition Window: 1990-1992 }
\setlength{\tabcolsep}{2pt}
\begin{tabular}{@{}rccccccccc@{}}
\toprule
\multicolumn{1}{l}{} & \multicolumn{4}{c}{Panel A}  & & \multicolumn{4}{c}{Panel B}\\ 
\multicolumn{1}{l}{} & \multicolumn{4}{c}{Log GDP per capita}  & & \multicolumn{4}{c}{GDP per capita}\\ \cmidrule(l){2-5} \cmidrule(r){7-10}
\multicolumn{1}{l|}{Transition Window} & \multicolumn{1}{c}{ } & \multicolumn{1}{c}{ } & \multicolumn{1}{c}{ } & \multicolumn{1}{c}{ } & $\quad$ & \multicolumn{1}{c}{ } & \multicolumn{1}{c}{ } & \multicolumn{1}{c}{ } & \multicolumn{1}{c}{ } \\ 
 \multicolumn{1}{l|}{  }  & \multicolumn{1}{c}{ T-DiD$^*$ } & \multicolumn{1}{c}{S-DiD} & \multicolumn{1}{c}{SC} & \multicolumn{1}{c}{ASC} & $\quad$ & \multicolumn{1}{c}{ T-DiD$^*$ } & \multicolumn{1}{c}{S-DiD} & \multicolumn{1}{c}{SC} & \multicolumn{1}{c}{ASC} \\ \midrule
\multicolumn{1}{l|}{Estimate} & & & & & & & & & \\ \cmidrule(r){1-1}
\multicolumn{1}{r|}{$ATT_{\omega,T}$}
&0.042       &0.249       &0.090       &0.405 &      &32.493      &101.582     &-43.968     &254.327     \\ 
\multicolumn{1}{r|}{ } 
&(0.011)    &(0.282)    &           &           & &(10.223)   &(1519.637) &           &           \\  
\multicolumn{1}{r|}{p-val} 
&0.000 &      &0.700 &0.417 & &0.001 &      &0.800 &0.316 \\ 
\multicolumn{1}{r|}{2015 BMK} 
&         &         &         &      &   &3.1\%   &9.8\%   &-4.2\%  &24.4\%  \\ \midrule
\multicolumn{1}{l|}{Identification test} 
&  &  &  &  &      &   &    &   & \\ \cmidrule(r){1-1}
\multicolumn{1}{r|}{p-val}&0.882                     &                          &                          &                      &    &0.938                     &                          &                          &                          \\
\bottomrule
\end{tabular}\label{Tab:Empirical_Results_2}
\begin{justify}
{\footnotesize
    \textit{Notes:} Methods compared to the T-DiD include the Synthetic difference-in-differences (S-DiD) of \citet{arkhangelsky-etal-2021}, the Synthetic Control (SC) of \citet{abadie-gardeazabal-2003}, and the Augmented Synthetic Control (ASC) of \citet{ben-feller-rothstein-2021}. The S-DiD, SC, and ASC p-values are computed from permutations. The transition window of 1990-1992 is maintained for all results above. T-DiD uses Togo and Cameroon as controls following \Cref{Prop:Eff_ATT}. }
\end{justify}
\end{table}

\begin{table}[htbp]
\centering
\caption{ Synthetic Control Weights with Transition Window: 1990-1992 }
\begin{tabular}{l|ccc|ccc}
\toprule
 & \multicolumn{3}{c|}{Log GDP per capita} & \multicolumn{3}{c}{Level GDP per capita} \\ \cmidrule(l){2-4} \cmidrule(r){5-7}
Country & S-DiD & SC & ASC & S-DiD & SC & ASC \\ \midrule
Algeria                  & 0.056 & 0.035 & 0.039 & 0.038 & 0.020 & 0.014 \\
Burundi                  & 0.068 & 0.068 & 0.222 & 0.060 & 0.039 & 0.304 \\
Cameroon                 & 0.115 & 0.048 & 0.180 & 0.057 & 0.029 & 0.127 \\
Central African Republic & 0.060 & 0.059 & 0.032 & 0.061 & 0.034 & -0.001 \\
Chad                     & 0.057 & 0.057 & -0.009 & 0.064 & 0.034 & -0.005 \\
Gabon                    & 0.000 & 0.023 & -0.033 & 0.000 & 0.011 & -0.005 \\
Kenya                    & 0.044 & 0.050 & -0.021 & 0.054 & 0.030 & 0.000 \\
Libya                    & 0.000 & 0.016 & -0.016 & 0.000 & 0.007 & 0.001 \\
Madagascar               & 0.054 & 0.056 & 0.006 & 0.063 & 0.033 & 0.030 \\
Malawi                   & 0.020 & 0.067 & -0.022 & 0.059 & 0.037 & -0.002 \\
Nigeria                  & 0.000 & 0.043 & -0.010 & 0.048 & 0.026 & -0.001 \\
Rwanda                   & 0.075 & 0.067 & 0.051 & 0.060 & 0.038 & 0.129 \\
Seychelles               & 0.034 & 0.027 & -0.039 & 0.013 & 0.015 & -0.005 \\
Sierra Leone             & 0.079 & 0.050 & 0.270 & 0.058 & 0.030 & 0.065 \\
Somalia                  & 0.020 & 0.071 & -0.041 & 0.061 & 0.458 & -0.003 \\
Sudan                    & 0.012 & 0.055 & -0.101 & 0.060 & 0.033 & -0.010 \\
Tanzania                 & 0.108 & 0.057 & 0.179 & 0.068 & 0.033 & 0.148 \\
Togo                     & 0.060 & 0.056 & 0.048 & 0.058 & 0.033 & 0.026 \\
Zambia                   & 0.067 & 0.049 & 0.082 & 0.063 & 0.030 & 0.083 \\
Zimbabwe                 & 0.071 & 0.045 & 0.183 & 0.055 & 0.028 & 0.106 \\
\midrule
\# Controls              & 17 & 20 & 11 & 18 & 20 & 11 \\
\bottomrule
\end{tabular}\label{Tab:Empirical_Results_3}
\begin{justify}
 \footnotesize
{Synthetic Control Weights: S-DiD, SC, and ASC estimators for Log and Level GDP per capita. Entries are reported to three decimal places. The last row counts the number of countries with non-zero weights. Estimated non-zero S-DiD pre-treatment period weights are highly concentrated in a small subset of years. For log GDP, the estimated lambda weights are 0.045 in 1960, 0.029 in 1963, 0.012 in 1978, 0.011 in 1979, and 0.903 in 1989. For GDP in levels, the only non-zero pre-treatment weight is 1989, with estimated weight 1.0. In both specifications, all other pre-treatment years receive zero estimated weight.}   
\end{justify}
\end{table}

From Panel A, one observes a fairly narrow range of the $ATT_{\omega,T}$ estimates by transition window. One observes that had Benin not democratised, her economy would have been $5.9\% - 7.8\% $ smaller on average over the period spanning the early 1990s up to 2018. To have a sense of the economic significance of these effects, consider comparisons with the average World GDP per capita growth rate over the 1991-2018 period.\footnote{The average World GDP per capita growth rate over the 1991-2018 period is 3.597\% -- source: authors' calculation using World Bank data sourced from \url{www.macrotrends.net}.} The $ATT_{\omega,T}$ estimates are economically significant and interesting  -- Benin's economy would have, on average over the period spanning the early 1990s up to 2018, ``shrunk" by a factor of $1.6 - 2.2$ times the average World GDP per capita growth rate over the 1991-2018 period had she not democratised. Moreover, the estimates are statistically significant at all conventional levels.

Panel B uses Benin's GDP per capita (in constant 2015 US dollars) of $ 1041.653 $ as a benchmark for interpreting results on the Level GDP per capita.\footnote{The corresponding row is labelled ``2015 BMK".} Thus, one observes that Benin's annual average GDP per capita over the period spanning the early 1990s through 2018 as a percentage of her $2015$ GDP per capita (in 2015 constant USD) would have been $4.1\% - 5.4\% $ smaller had she not democratised. These values are economically significant and interesting as the margin is economically non-trivial.\footnote{Level $ATT_{\omega,T}$ estimates are expressed in 2015 USD in Benin and hence are generally dependent on 2015 price levels and exchange rates in Benin. To further aid interpretation, taking the 2015 ratio of Benin's Purchasing Power Parity (PPP) adjusted GDP per capita in constant 2021 USD to the GDP per capita in constant 2015 USD of $2.9$ (source: World Bank Indicators), assuming this ratio is stable over the post-treatment period, the 2015 constant USD PPP-adjusted $ATT_{\omega,T}$ estimates of Panel B read: $ 131.9819, \ 163.4295, \  155.8257, \text{ and } 120.6922 $.} Moreover, all $ATT_{\omega,T}$ level estimates are statistically significant at all conventional levels. Thus, across different specifications of the transition window, the effects of democratisation on economic growth are positive and statistically significant. That the range of estimates $4.1\% - 5.4\% $ is narrow indicates the robustness of the results to the specification of the transition window.

The estimates on $X_{t-1}$ are indicative of a high level of persistence in the series $X_t$ which is the difference between the annual GDP per capita of Benin and Togo. The Augmented Dickey-Fuller (ADF) test suggests the unit-root null hypothesis is rejected at the $ 5\%$ level across all specifications in both panels A and B. The Kwiatkowski-Phillips-Schmidt-Shin (KPSS) test of trend stationarity suggests a failure to detect trend non-stationarity at the $5\%$ level, while the Durbin-Watson test indicates a failure to reject the null of zero auto-correlation (against the two-sided alternative) at any conventional nominal level.\footnote{The DW test also serves as a specification test in this setting. Suppose the true model had an MA(1) error instead, i.e., $ X_t = \beta_0 + \widetilde{ATT}_n^{w,\psi}  \mathbbm{1}\{t \geq 1\} + \mathcal{E}_t + \beta_1 \mathcal{E}_{t-1} $, then the resulting errors from the ``mis-specified" model (with the lagged $X_t$ above ) would be auto-correlated. This is the violation the DW test is apt to detect.} Taken together, the diagnostic tests suggest the reliability of the T-DiD results in \Cref{Tab:Empirical_Results}.

To compare the proposed \( T \)-DiD estimator with alternative methods in this empirical setting, the S-DID, SC, and ASC results reported in \Cref{Tab:Empirical_Results_2} use a donor pool comprising \emph{all} African countries \emph{without} missing constant 2015 USD GDP per capita data, whose average democracy index remained below the 0.5 threshold throughout the 1960--2018 period. The resulting donor pool consists of 20 countries, along with the corresponding SC weights, as detailed in \Cref{Tab:Empirical_Results_3}. 1990-1992 is maintained as the transition window. The efficient T-DiD following \Cref{Prop:Eff_ATT} with Cameroon and Togo as controls is compared with the Synthetic DiD of \citet{arkhangelsky-etal-2021}, the Synthetic Control method of \citet{abadie-gardeazabal-2003}, and the Augmented Synthetic Control method of \citet{ben-feller-rothstein-2021}. Relative to the T-DiD, these estimates are not precisely estimated. This is not surprising as the minimum attainable p-value (or test size under the null) in a permutation test applied to the SC and ASC is $1/21 \approx 0.0476 $.\footnote{Besides, the discussion in \Cref{SubSect:Persistence} suggests that the SC estimators without adjusting for auto-correlation in the outcome and its fitted counterfactual may be targeting the long-run (multiplier) effect instead.}

\section{Conclusion}\label{Sect:Con}

While being instrumental in answering the empirical question at hand, the T-DiD estimator equally contributes to the larger applied and theoretical literature on causal inference in large-$\T,T$ settings with a fixed number of cross-sectional units. The $ATT_{\omega,T}$ parameter of interest indexed over a space of convex weighting schemes is asymptotically identified under identification conditions where even both asymptotic parallel trends and limited anticipation identification conditions may be individually violated. Under NED—a fairly general form of temporal dependence—and weak dominance regularity conditions, the T-DiD estimator is asymptotically normal, thereby paving the way for standard inference using, e.g., $t$-tests.

This paper further extends the baseline theory to accommodate possible intricacies underpinned by stochastic and deterministic trends in $X_t:= Y_{1,t} - Y_{0,t} $ and multiple control units. This paper exploits over-identification in the multiple-control setting to propose an identification test with more desirable statistical properties relative to the pre-trends test. For example, the proposed identification test can detect violations of identification in post-treatment periods, unlike pre-trends tests. 

Sustained interest in the economic gains of democracy spawns a vast array of interesting contributions that fail to reach a clear consensus. This issue is largely due to the immense heterogeneity in effects, which renders the generalisability of results quite difficult. This paper departs from the extant literature by adopting a novel perspective -- estimating country-specific effects. This paper equally meets the methodological challenge that this new approach engenders --- proposing the T-DiD, which is adaptable to as few as a single treated unit and a single control unit with large pre- and post-treatment periods. The estimated effect of democracy on economic growth in this paper is not only statistically significant at all conventional levels; it is also economically interesting. Benin's economy would have been $5.9\%-7.8\% $ smaller on average had she not democratised. 

\section*{Acknowledgements}
\noindent This paper benefited from the invaluable feedback of Al-mouksit Akim, Brantly Callaway, Alecia Cassidy, Traviss Cassidy, Daniel Henderson, Jonathan Hall, Feiyu Jiang, Eric Kadio, D\'esir\'e K\'edagni, Junsoo Lee, Xiaochun Liu, Tigran Melkonyan, and Julius Owusu. This manuscript benefited from the use of ChatGPT (OpenAI, GPT-5.3) for language editing and presentation.

\section*{Declaration of Interests}
\noindent The authors have no interests to declare.

\section*{Data Availability Statement}
\noindent Data used in the study are publicly available. The replication package, including the dataset, is available on author E.S.T's website.

\begingroup
\setstretch{1.25}
\setlength\bibitemsep{0.5 pt}
\printbibliography
\endgroup
\end{refsection}

\newpage
\setcounter{page}{1}

\appendix
\renewcommand{\thetable}{S.\arabic{table}}
\renewcommand{\thefigure}{S.\arabic{figure}}
\renewcommand{\thesection}{S.\arabic{section}}

\begin{center}
    {\large \bf Supplementary Material} \\[0.2cm]
    
    {\Large Difference-in-differences with as few as two cross-sectional units\\ -- A new perspective to the democracy-growth debate} \\[0.2cm]
   
    {\large Gilles Boevi Koumou \quad \quad Emmanuel Selorm Tsyawo}
    \hrule height 1pt
\end{center}

\vspace{1cm}
The supplementary material provides detailed proofs and extensions of the theoretical results presented in the main text. Specifically, \Cref{App_Sect:Proofs_main} contains proofs of all theoretical results in the main text. \Cref{App_Sect:Extensions} generalises the baseline model in \eqref{eqn:ATT_lin_mod} to accommodate both stochastic and deterministic trends, and extends the two-unit framework to settings with multiple treated units. Supporting lemmata and propositions referenced in the main text are provided in \Cref{App_Sect:Supp_Lem,App_Sect:Useful_Prop}, respectively. Further discussions on related literature, pre-trends testing adapted to the baseline setting of this paper, and the empirical setting are presented in \Cref{App_Sect:Discussions}, while simulation results appear in \Cref{Sect:Sim}.

\begin{refsection}
\section{Proofs of main results}\label{App_Sect:Proofs_main}

\subsection{Proof of \Cref{Prop:PT_DGP}}

  \begin{namedproposition}{\ref{Prop:PT_DGP}}
      \Paste{Key:Prop:PT_DGP}
  \end{namedproposition}

\textbf{Proof:} 

The assumption $\E[e_t\mid D=d]=0, \ d \in \{0,1\} $ and \eqref{eqn:PT_DGP1} imply 
\begin{align*}
    \E[Y_t(0)\mid D=d] = \alpha_0 + (\alpha_1-\alpha_0)d + \varphi(t) + \nu_{0}(t,d;\eta).
\end{align*}
Thus, 
 \[
 \E[Y_t(0)-Y_\tau(0)\mid D=d] = \varphi(t) - \varphi(\tau) + \nu_{0}(t,d;\eta) - \nu_{0}(\tau,d;\eta)  
 \]
for any pair $(\tau,t) \in [-\T] \times [T] $, and 
\begin{align*}
    \E[Y_t(0)-Y_\tau(0)\mid D=1] - \E[Y_t(0)-Y_\tau(0)\mid D=0] &= \big( \nu_{0}(t,1;\eta) - \nu_{0}(\tau,1;\eta) \big) - \big( \nu_{0}(t,0;\eta) - \nu_{0}(\tau,0;\eta) \big) \\
    &= \frac{1}{4}\big| 1 + 0.5\sin(t) \big|t^{-\eta} + \frac{1}{4}\big| 1 + 0.5\sin(-\tau) \big|(-\tau)^{-\eta}.
\end{align*}It then follows that
\begin{align*}
    \frac{1}{\T T}\sum_{-\tau=1}^{\T}\sum_{t=1}^T &\big( \E[Y_t(0)-Y_\tau(0)\mid D=1] - \E[Y_t(0)-Y_\tau(0)\mid D=0] \big) \\
    &= \frac{1}{4T} \sum_{t=1}^T \big| 1 + 0.5\sin(t) \big|t^{-\eta} + \frac{1}{4\T }\sum_{-\tau=1}^{\T} \big| 1 + 0.5\sin(-\tau) \big|(-\tau)^{-\eta}.
\end{align*}

Since \( \big| 1 + 0.5\sin(t) \big| \leq (3/2) \, \forall \, t\geq 1 \), then for $\eta>0$,
\begin{align*}
    \Bigg| \frac{1}{T} \sum_{t=1}^T \big| 1 + 0.5\sin(t) \big|t^{-\eta}\Bigg| &\leq \frac{3}{2T}\sum_{t=1}^T t^{-\eta}\\
    &\leq \frac{3}{2T}\Big(1+\int_{1}^{T} t^{-\eta} dt\Big)\\
   &= \frac{3}{2}\begin{cases} 
      T^{-1} \Big(1 - \frac{1}{1-\eta}\Big) + T^{-\eta}\frac{1}{(1-\eta)}, & \quad \eta \neq 1  \\
      T^{-1} + T^{-1}\log(T), & \quad \eta = 1
   \end{cases}.
\end{align*}
\noindent By the same token, 
\begin{align*}
    \Bigg| \frac{1}{\T }\sum_{-\tau=1}^{\T} \big| 1 + 0.5\sin(-\tau) \big|(-\tau)^{-\eta} \Bigg| \leq \frac{3}{2}\begin{cases} 
      \T^{-1} \Big(1 - \frac{1}{1-\eta}\Big) + \T^{-\eta}\frac{1}{(1-\eta)}, & \quad \eta \neq 1  \\
      \T^{-1} + \T^{-1}\log(\T), & \quad \eta = 1
   \end{cases} 
\end{align*}

Since \( T^{-1}\log(T) = o(T^{-\bar{\eta}}) \, \forall \, \bar{\eta} < 1 \) and $ \eta > 1/2$, it follows by the triangle inequality that 
\[
 \frac{1}{\T T}\sum_{-\tau=1}^{\T}\sum_{t=1}^T \big( \E[Y_t(0)-Y_\tau(0)\mid D=1] - \E[Y_t(0)-Y_\tau(0)\mid D=0] \big) = \mathcal{O}\Big((\T \wedge T)^{-(1/2 + (\dot{\eta} -1/2))}\Big), \ \dot{\eta} \in (1/2, \, 1).
 \]
 \noindent Set $ \gamma = \dot{\eta} - (1/2) $ and observe that $\gamma\in (0,1/2)$, which implies $\gamma>0$. This completes the proof.
 
\qed

\subsection{Proof of \Cref{Theorem:Identification}}
\begin{namedtheorem}{\ref{Theorem:Identification}}
    \Paste{Key:Theorem:Identification}
\end{namedtheorem}

\textbf{Proof:}
\paragraph{Part (a):} Recall \( R_{\tau,t} = \E[Y_t(0)-Y_\tau(0)\mid D=1] - \E[Y_t(0)-Y_\tau(0)\mid D=0] - \E[Y_\tau(1)-Y_\tau(0)\mid D=1] \), then for any pair $(\tau,t) \in [-\T] \times [T] $,
\begin{align}\label{eqn:ATT_decomp}
    ATT(t) =& \E[Y_t(1)-Y_t(0)\mid D=1] \nonumber \\
    =& \E[Y_t(1)-Y_\tau(0)\mid D=1] - \E[Y_t(0)-Y_\tau(0)\mid D=1] \nonumber \\
    =& \E[Y_t(1)-Y_\tau(0)\mid D=1] - \E[Y_t(0)-Y_\tau(0)\mid D=0] \nonumber \\
    &- \big(\E[Y_t(0)-Y_\tau(0)\mid D=1] - \E[Y_t(0)-Y_\tau(0)\mid D=0]\big) \nonumber \\
    =& \E[Y_t(1)-Y_\tau(1)\mid D=1] - \E[Y_t(0)-Y_\tau(0)\mid D=0] \nonumber \\
    & - \big\{\E[Y_t(0)-Y_\tau(0)\mid D=1] - \E[Y_t(0)-Y_\tau(0)\mid D=0] - \E[Y_\tau(1)-Y_\tau(0)\mid D=1]\big\} \nonumber \\
    =& \E[Y_{1,t} - Y_{1,\tau}] - \E[Y_{0,t} - Y_{0,\tau}] \nonumber \\
    & - \big(\E[Y_t(0)-Y_\tau(0)\mid D=1] - \E[Y_t(0)-Y_\tau(0)\mid D=0] - \E[Y_\tau(1)-Y_\tau(0)\mid D=1]\big) \nonumber \\
    =&: ATT_{\tau,t} - R_{\tau,t}.
\end{align}The first equality follows by definition; the second subtracts and adds $\E[Y_\tau(0)\mid D=1]$; the third subtracts and adds $\E[Y_t(0)-Y_\tau(0)\mid D=0]$; the fourth subtracts and adds $ \E[Y_\tau(1)\mid D=1] $; the fifth uses the fact that treated potential outcomes are observed for the treated unit and untreated potential outcomes are observed for the untreated unit; and the last line follows by definition.

Since pre-treatment and post-treatment weights each sum to one by construction, it follows from the above decomposition of $ATT(t)$ that
\begin{align*}
    ATT_{\omega,T}:&= \sum_{t=1}^{T} w_T(t) ATT(t) = \Big(\sum_{-\tau=1}^{\T } \psi_{\T }(-\tau) \Big) \sum_{t=1}^{T} w_T(t) ATT(t) \\
    &= \sum_{-\tau=1}^{\T }\sum_{t=1}^T \psi_{\T }(-\tau) w_T(t) ATT_{\tau,t} - \sum_{-\tau=1}^{\T }\sum_{t=1}^T \psi_{\T }(-\tau) w_T(t) R_{\tau,t}.
\end{align*}
By the triangle inequality and \Cref{ass:parallel_trends},
\begin{align*}
    \Big|\sum_{-\tau=1}^{\T }\sum_{t=1}^T \psi_{\T }(-\tau) & w_T(t) \Big( \E[Y_t(0)-Y_\tau(0)\mid D=1] - \E[Y_t(0)-Y_\tau(0)\mid D=0] \Big) \Big|\\ 
    \leq & \Big(T \sup_{w\in \W }\max_{ t \in [T] }w_T(t)\Big) \times \Big(\T\sup_{\psi \in \W }\max_{ \tau \in [-\T] }\psi_{\T }(-\tau)\Big)\\
    &\times \Big|\frac{1}{\T T}\sum_{-\tau=1}^{\T}\sum_{t=1}^T \Big( \E[Y_t(0)-Y_\tau(0)\mid D=1] - \E[Y_t(0)-Y_\tau(0)\mid D=0] \Big) \Big|\\
    =& \, \mathcal{O}\big((\T \wedge T)^{-(1/2+\gamma)}\big).
\end{align*}
\noindent In a similar vein, 
\begin{align*}
    \Big|\sum_{-\tau=1}^{\T }\psi_{\T }(-\tau) \big( \E[Y_\tau(1)-Y_\tau(0)\mid D=1] \big) \Big| &\leq \Big(\T \sup_{\psi \in \W }\max_{\tau\in [-\T]}\psi_{\T }(-\tau)\Big)\Big|\frac{1}{\T }\sum_{-\tau=1}^{\T} \big( \E[Y_\tau(1)-Y_\tau(0)\mid D=1] \big) \Big| \\
    &= \, \mathcal{O}(\T ^{-(1/2+\delta)})
\end{align*} holds by \Cref{ass:limited_anticip}. Thanks to the above and \Cref{ass:TTratio_lambda_n},
\begin{align}\label{eqn:wR_t_uniform_o}
    \displaystyle R_n^{w,\psi}: = \sum_{-\tau=1}^{\T }\sum_{t=1}^T \psi_{\T }(-\tau) w_T(t)R_{\tau,t}  &= \mathcal{O}\big((\T \wedge T)^{-(1/2+\gamma)}\big) +  \mathcal{O}(\T ^{-(1/2+\delta)}) \nonumber \\
    &= \mathcal{O}(n^{-(1/2+\gamma\wedge\delta)}) \text{ uniformly in } \{w,\psi\} \in \W^2.
\end{align}
\noindent To see why \eqref{eqn:wR_t_uniform_o} holds, note that
\begin{align*}
    \mathcal{O}\big((\T \wedge T)^{-(1/2+\gamma)}\big) +  \mathcal{O}(\T ^{-(1/2+\delta)}) &\asymp \frac{1}{(\T \wedge T) ^{(1/2+\gamma)}} + \frac{1}{\T ^{(1/2+\delta)}}\\
    &\lesssim \frac{1}{\T ^{(1/2+\gamma\wedge\delta)} \wedge T^{(1/2+\gamma)} } \lesssim \frac{1}{\T ^{(1/2+\gamma\wedge\delta)} \wedge T^{(1/2+\gamma\wedge\delta)} }\\
    &= \Big(\frac{1}{\T \wedge T}\Big)^{(1/2+\gamma\wedge\delta)} = \Big(\frac{n^{-1}}{(1-\lambda_n)\wedge \lambda_n}\Big)^{(1/2+\gamma\wedge\delta)}\\
    &\asymp n^{-(1/2+\gamma\wedge\delta)} 
\end{align*}
\noindent uniformly in $n$ subject to \Cref{ass:TTratio_lambda_n}.

Averaging both sides of \eqref{eqn:ATT_decomp} and applying the foregoing gives
\begin{align}\label{eqn:ATT_identif}
    \sum_{t=1}^T w_T(t) ATT(t) &= \sum_{-\tau=1}^{\T }\sum_{t=1}^T \psi_{\T }(-\tau) w_T(t)ATT_{\tau,t} - \sum_{-\tau=1}^{\T }\sum_{t=1}^T \psi_{\T }(-\tau) w_T(t)R_{\tau,t} \nonumber\\ 
    &= \widetilde{ATT}_n^{w,\psi} - R_n^{w,\psi} \nonumber\\
    &=\widetilde{ATT}_n^{w,\psi} + \mathcal{O}(n^{-(1/2+\gamma\wedge\delta)})
\end{align} uniformly in $ \{w,\psi\} \in \W^2  $. $\widetilde{ATT}_n^{w,\psi}$ is identified from the data sampling process. Identification therefore holds \emph{modulo} $ R_n^{w,\psi} = \mathcal{O}(n^{-(1/2+\gamma\wedge\delta)}) $.

\paragraph{Part (b):} Fix $w\in\W $ and observe that since pre-treatment weights must sum to one and post-treatment weights must also sum to one,
\begin{align*}
    \sum_{-\tau=1}^{\T }\sum_{t=1}^T\psi_{\T }(-\tau) w_T(t)ATT(t) &= \Big(\sum_{-\tau=1}^{\T }\psi_{\T }(-\tau) \Big) \sum_{t=1}^{T}w_T(t)ATT(t) \\ 
    &= \sum_{t=1}^{T}w_T(t)ATT(t) =: ATT_{\omega,T}
\end{align*}
\noindent, i.e., $ATT_{\omega,T}$ does not depend on $\psi\in\W$ for $\tau\leq -1$. Thus, for any $ \{\psi,\phi\} \in \W^2 $,
\begin{align*}
    \widetilde{ATT}_n^{w,\psi} - \widetilde{ATT}_n^{w,\phi} =& \sum_{-\tau=1}^{\T }\sum_{t=1}^Tw_T(t)\big(\psi_{\T }(-\tau)ATT_{\tau,t} - ATT(t) + ATT(t) - \phi_{\T }(-\tau)ATT_{\tau,t}\big)\\
    =& \sum_{-\tau=1}^{\T }\sum_{t=1}^T\psi_{\T }(-\tau) w_T(t)\big(ATT_{\tau,t} - ATT(t)\big) \\
    &- \sum_{-\tau=1}^{\T }\sum_{t=1}^T\phi_{\T }(-\tau) w_T(t)\big(ATT_{\tau,t} - ATT(t)\big) \\
    =& \sum_{-\tau=1}^{\T }\sum_{t=1}^T\psi_{\T }(-\tau) w_T(t)R_{\tau,t} - \sum_{-\tau=1}^{\T }\sum_{t=1}^T\phi_{\T }(-\tau) w_T(t)R_{\tau,t}\\
    =& R_n^{w,\psi} - R_n^{w,\phi}.
\end{align*}

\noindent The conclusion follows from the triangle inequality and \eqref{eqn:wR_t_uniform_o} in Part (a) above.

\qed

\subsection{Proof of \Cref{Theorem:AsympN}}

\begin{namedtheorem}{\ref{Theorem:AsympN}}
    \Paste{Key:Theorem:AsympN}
\end{namedtheorem}
\textbf{Proof:}

Under the conditions of \Cref{Theorem:Identification}, \eqref{eqn:ATT_decomp}, and \eqref{eqn:ATT_identif}, 
\begin{align*}
    ATT_{\omega,T}
    &= \sum_{-\tau=1}^{\T }\sum_{t=1}^T\psi_{\T }(-\tau) w_T(t)\Big( \E[Y_{1,t} - Y_{1,\tau}] - \E[Y_{0,t} - Y_{0,\tau}] \Big) -  \sum_{-\tau=1}^{\T }\sum_{t=1}^T\psi_{\T }(-\tau) w_T(t)R_{\tau,t}\\
    &= \sum_{t=-\T }^{T} \omega_n(t)\E[X_t] - \sum_{-\tau=1}^{\T }\sum_{t=1}^T\psi_{\T }(-\tau) w_T(t) R_{\tau,t}\\
    &= \sum_{t=-\T }^{T} \omega_n(t)\E[X_t] + \mathcal{O}(n^{-(1/2 + \gamma\wedge \delta)})
\end{align*}

\noindent uniformly in $\W^2 $. In addition to \eqref{eqn:DiD_estimator2},
\begin{equation}\label{eqn:ATT_expand}
    \widehat{ATT}_{\omega,T} - ATT_{\omega,T} = \sum_{t=-\T }^{T} \omega_n(t)(X_t-\E[X_t]) + \mathcal{O}(n^{-(1/2+\gamma\wedge \delta)})
\end{equation}
\noindent uniformly in $\W^2$. Further,
\begin{align*}
    s_{\omega,n}^{-1}(\widehat{ATT}_{\omega,T} - ATT_{\omega,T}) &= s_{\omega,n}^{-1}\sum_{t=-\T }^{T} \omega_n(t)(X_t-\E[X_t]) + \mathcal{O}((\sqrt{n}s_{\omega,n})^{-1}n^{-(\gamma\wedge \delta)})\\
    &= \sum_{t=-\T }^{T} X_{nt} + \mathcal{O}(n^{-(\gamma\wedge \delta)})
\end{align*}

\noindent by \Cref{ass:bound_sn} uniformly in $\W^2 $. The conclusion follows from \Cref{lem:CLT}.

\qed

\subsection{Proof of \Cref{Theorem:AsympN_ext}}
\begin{namedcorollary}{\ref{Theorem:AsympN_ext}}
    \Paste{Key:Theorem:AsympN_ext}
\end{namedcorollary}

\textbf{Proof:}
The following expansion follows from \eqref{eqn:ATT_expand} uniformly in $\W^2$:
    \begin{align*}
        \uptau_J'(\widehat{ATT}_{\omega,\cdot T} - \mathbbm{1}_J ATT_{\omega,T})  = \sum_{t=-\T }^{T} \omega_n(t) \uptau_J'(X_{\cdot t} -\E[X_{\cdot t}]) + \mathcal{O}(n^{-(1/2+\gamma\wedge \delta)})
    \end{align*}
\noindent where $X_{\cdot t}:=(X_{1 t},\ldots,X_{J t})'$ and $X_{j t}, \ j \in [J] $ is the difference between the outcomes of the treated unit and the $j$'th control unit. Further,
\begin{align*}
&(\uptau_J'S_{\omega,n}\uptau_J)^{-1/2}\uptau_J'(\widehat{ATT}_{\omega,\cdot T} - \mathbbm{1}_J ATT_{\omega,T})\\  
&= (\uptau_J'S_{\omega,n}\uptau_J)^{-1/2}\sum_{t=-\T }^{T} \omega_n(t) \uptau_J'(X_{\cdot t} -\E[X_{\cdot t}]) + (n\uptau_J'S_{\omega,n}\uptau_J)^{-1/2} \times \mathcal{O}(n^{-(\gamma\wedge \delta)})\\
&= (\uptau_J'S_{\omega,n}\uptau_J)^{-1/2}\sum_{t=-\T }^{T} \omega_n(t) \uptau_J'(X_{\cdot t} -\E[X_{\cdot t}]) + \mathcal{O}(n^{-(\gamma\wedge \delta)})
    \end{align*}
since $\displaystyle \inf_{\uptau_J\in \mathbb{S}^J} n\uptau_J'S_{\omega,n}\uptau_J =: \rho_{\mathrm{min}}(nS_{\omega,n}) \geq \epsilon^2 $ by \Cref{ass:bound_sn_ext}. The conclusion follows similarly from the proof of \Cref{lem:CLT} and is therefore omitted.
\qed

\subsection{Proof of \Cref{Prop:Eff_ATT}}
\begin{namedproposition}{\ref{Prop:Eff_ATT}}
    \Paste{Key:Prop:Eff_ATT}
\end{namedproposition}
\textbf{Proof:}
Consider the classical minimum distance problem \eqref{eqn:Eff_ATT} (see, e.g., \citet[Sect. 14.5]{wooldridge-2010}) with an arbitrary positive definite weight matrix $H_n$. The solution is given by \Cref{eqn:Eff_ATT}. Under the conditions of \Cref{Theorem:AsympN_ext}, 
    \[
    \V[\widehat{ATT}_{\omega,T}^{h_J}]^{-1/2}(\widehat{ATT}_{\omega,T}^{h_J} - ATT_{\omega,T}) = (n\V[\widehat{ATT}_{\omega,T}^{h_J}])^{-1/2}\sqrt{n}(\widehat{ATT}_{\omega,T}^{h_J} - ATT_{\omega,T}) \xrightarrow{d} \mathcal{N}(0,1) 
    \] where $\V[\widehat{ATT}_{\omega,T}^{h_J}] = (\mathbbm{1}_J'H_n\mathbbm{1}_J)^{-2}\mathbbm{1}_J'H_nS_{\omega,n}H_n\mathbbm{1}_J$. Since $P_n:= S_{\omega,n}^{1/2}H_n\mathbbm{1}_J(\mathbbm{1}_J'H_nS_{\omega,n}H_n\mathbbm{1}_J)^{-1}\mathbbm{1}_J'H_nS_{\omega,n}^{1/2} $ is idempotent,
    \begin{align*}
        &\mathbbm{1}_J'S_{\omega,n}^{-1}\mathbbm{1}_J - (\mathbbm{1}_J'H_n\mathbbm{1}_J)(\mathbbm{1}_J'H_nS_{\omega,n}H_n\mathbbm{1}_J)^{-1}(\mathbbm{1}_J'H_n\mathbbm{1}_J)\\ 
        &= \mathbbm{1}_J'(S_{\omega,n}^{-1} - H_n\mathbbm{1}_J(\mathbbm{1}_J'H_nS_{\omega,n}H_n\mathbbm{1}_J)^{-1}\mathbbm{1}_J'H_n)\mathbbm{1}_J\\
        &= \mathbbm{1}_J'S_{\omega,n}^{-1/2}S_{\omega,n}^{1/2}(S_{\omega,n}^{-1} - H_n\mathbbm{1}_J(\mathbbm{1}_J'H_nS_{\omega,n}H_n\mathbbm{1}_J)^{-1}\mathbbm{1}_J'H_n)S_{\omega,n}^{1/2}S_{\omega,n}^{-1/2}\mathbbm{1}_J\\
        &= \mathbbm{1}_J'S_{\omega,n}^{-1/2}(\mathrm{I} - S_{\omega,n}^{1/2}H_n\mathbbm{1}_J(\mathbbm{1}_J'H_nS_{\omega,n}H_n\mathbbm{1}_J)^{-1}\mathbbm{1}_J'H_nS_{\omega,n}^{1/2})S_{\omega,n}^{-1/2}\mathbbm{1}_J\\
        &= ||(\mathrm{I} - P_n)S_{\omega,n}^{-1/2}\mathbbm{1}_J||^2.
    \end{align*}

From the foregoing, 
\[
\Big(\mathbbm{1}_J'S_{\omega,n}^{-1}\mathbbm{1}_J - (\mathbbm{1}_J'H_n\mathbbm{1}_J)(\mathbbm{1}_J'H_nS_{\omega,n}H_n\mathbbm{1}_J)^{-1}(\mathbbm{1}_J'H_n\mathbbm{1}_J) \Big)/n = ||(\mathrm{I} - P_n)(nS_{\omega,n})^{-1/2}\mathbbm{1}_J||^2 \geq 0.
\]
\noindent Thus, the optimal choice of $H_n$ is $S_{\omega,n}^{-1}$, which yields the optimal linear combination $\displaystyle h_J^*:= \frac{1}{\mathbbm{1}_J'S_{\omega,n}^{-1}\mathbbm{1}_J}S_{\omega,n}^{-1}\mathbbm{1}_J$ in $\Hyp^J$. This proves the assertion as claimed.
\qed

\subsection{Proof of \Cref{Theorem:Test_Implication} }
\begin{namedtheorem}{\ref{Theorem:Test_Implication}}
    \Paste{Key:Theorem:Test_Implication}
\end{namedtheorem}
\textbf{Proof:}
Recalling $h_J'\mathbbm{1}_J=1$ for all $h_J\in \Hyp^J$, the following expansion holds:
\begin{align*}
    \widehat{ATT}_{\omega,\cdot T} - \mathbbm{1}_J\widehat{ATT}_{\omega,T}^* & = (\widehat{ATT}_{\omega,\cdot T} - \mathbbm{1}_JATT_{\omega,T}) - \mathbbm{1}_J(\widehat{ATT}_{\omega,T}^* - ATT_{\omega,T}) \\
    & = (\widehat{ATT}_{\omega,\cdot T} - \mathbbm{1}_JATT_{\omega,T}) - \mathbbm{1}_J\hat h_{J}^{*\prime} (\widehat{ATT}_{\omega,\cdot T} - \mathbbm{1}_J ATT_{\omega,T}) \\
    &= [\mathrm{I}_J - \mathbbm{1}_J\hat h_{J}^{*\prime}](\widehat{ATT}_{\omega,\cdot T} - \mathbbm{1}_J ATT_{\omega,T}) \\
    &= [\mathrm{I}_J - \mathbbm{1}_J(\mathbbm{1}_J'\widehat{S}_{\omega,n}^{-1}\mathbbm{1}_J)^{-1} \mathbbm{1}_J'\widehat{S}_{\omega,n}^{-1}](\widehat{ATT}_{\omega,\cdot T} - \mathbbm{1}_JATT_{\omega,T})\\
    &= [\widehat{S}_{\omega,n}^{1/2} - \mathbbm{1}_J(\mathbbm{1}_J'\widehat{S}_{\omega,n}^{-1}\mathbbm{1}_J)^{-1} \mathbbm{1}_J'\widehat{S}_{\omega,n}^{-1/2}]\widehat{S}_{\omega,n}^{-1/2}(\widehat{ATT}_{\omega,\cdot T} - \mathbbm{1}_JATT_{\omega,T})
\end{align*}

\noindent where $ \Hyp^J \ni h_{J}^*:= \widehat{S}_{\omega,n}^{-1} \mathbbm{1}_J (\mathbbm{1}_J'\widehat{S}_{\omega,n}^{-1}\mathbbm{1}_J)^{-1} $. Pre-multiplying the above by $\widehat{S}_{\omega,n}^{-1/2}$, it follows that
\begin{align*}
    \widehat{S}_{\omega,n}^{-1/2}(\widehat{ATT}_{\omega,\cdot T} - \mathbbm{1}_J\widehat{ATT}_{\omega,T}^*) &= [\mathrm{I}_J - \widehat{S}_{\omega,n}^{-1/2}\mathbbm{1}_J(\mathbbm{1}_J'\widehat{S}_{\omega,n}^{-1}\mathbbm{1}_J)^{-1} \mathbbm{1}_J'\widehat{S}_{\omega,n}^{-1/2}]\widehat{S}_{\omega,n}^{-1/2}(\widehat{ATT}_{\omega,\cdot T} - \mathbbm{1}_JATT_{\omega,T})\\
    &=: (\mathrm{I}_J - \widehat{P}_{n})\widehat{S}_{\omega,n}^{-1/2}(\widehat{ATT}_{\omega,\cdot T} - \mathbbm{1}_JATT_{\omega,T}).
\end{align*}

\noindent Observe that $\widehat{P}_{n}:= \widehat{S}_{\omega,n}^{-1/2}\mathbbm{1}_J(\mathbbm{1}_J'\widehat{S}_{\omega,n}^{-1}\mathbbm{1}_J)^{-1} \mathbbm{1}_J'\widehat{S}_{\omega,n}^{-1/2} $ is a projection matrix, thus $\mathrm{I}_J - \widehat{P}_{n}$ is idempotent. Denote its probability limit, in view of \Cref{ass:consis_Sn} and the continuous mapping theorem, as: $\widetilde{P}_{n}:= \widetilde{S}_{\omega,n}^{-1/2}\mathbbm{1}_J(\mathbbm{1}_J'\widetilde{S}_{\omega,n}^{-1}\mathbbm{1}_J)^{-1} \mathbbm{1}_J'\widetilde{S}_{\omega,n}^{-1/2} $.  Moreover, since $\mathrm{rank}[\widetilde{S}_{\omega,n}^{-1/2}\mathbbm{1}_J]=1$, $\mathrm{rank}[\mathrm{I}_J - \widetilde{P}_{n}]=J-1$. Also, denoting $\widehat{\mathrm{I}}_{n}: = \widehat{S}_{\omega,n}^{-1/2} \widetilde{S}_{\omega,n}^{1/2} $, the following expansion follows from \eqref{eqn:ATT_expand}:
\begin{align*}
        \widehat{S}_{\omega,n}^{-1/2}(\widehat{ATT}_{\omega,\cdot T} - \mathbbm{1}_J ATT_{\omega,T})  &= \widehat{\mathrm{I}}_{n}\widetilde{S}_{\omega,n}^{-1/2}\sum_{t=-\T }^{T} \omega_n(t)(X_{\cdot t} -\E[X_{\cdot t}]) + \widehat{\mathrm{I}}_{n}(n\widetilde{S}_{\omega,n})^{-1/2} \sqrt{n}R_{\cdot n}.
\end{align*}

\noindent Under \Cref{ass:consis_Sn}, $||\widehat{\mathrm{I}}_{n} - \mathrm{I}_J|| \xrightarrow{p} 0 $. In addition to the conditions of \Cref{lem:CLT}, the continuous mapping theorem, and the Cram\'er-Wold device,
$$
\widehat{\mathrm{I}}_{n}\widetilde{S}_{\omega,n}^{-1/2}\sum_{t=-\T }^{T} \omega_n(t)(X_{\cdot t} -\E[X_{\cdot t}]) \xrightarrow{d} \mathcal{N}(0,\mathrm{I}_J). 
$$

\paragraph{Part (a):} Under $\Hyp_o$, $||\widehat{\mathrm{I}}_{n}\widetilde{S}_{\omega,n}^{-1/2}R_{\cdot n}^{w,\psi}|| \xrightarrow{p} 0$ thus $\widehat{S}_{\omega,n}^{-1/2}(\widehat{ATT}_{\omega,\cdot T} - \mathbbm{1}_J ATT_{\omega,T}) \xrightarrow{d} \mathcal{N}(0,\mathrm{I}_J) $ and the conclusion follows.

\paragraph{Part (b):} Under $\Hyp_{an}$,
\begin{align*}
    \widehat{\mathrm{I}}_{n}(n\widetilde{S}_{\omega,n})^{-1/2} \sqrt{n}R_{\cdot n}^{w,\psi} = \widehat{\mathrm{I}}_{n}(n\widetilde{S}_{\omega,n})^{-1/2} \uptau_J^R ||\sqrt{n}R_{\cdot n}^{w,\psi}||
\end{align*}

\noindent where $\uptau_J^R= \sqrt{n}R_{\cdot n}^{w,\psi} ||\sqrt{n}R_{\cdot n}^{w,\psi}||^{-1}$. Under the condition that $\uptau_J^R$ does not lie in the null space of $(n\widetilde{S}_{\omega,n})^{-1/2}[\mathrm{I}_J - \widetilde{P}_{n}](n\widetilde{S}_{\omega,n})^{-1/2}$ as $n\rightarrow \infty$, the conclusion for this part follows.

\paragraph{Part (c):} Under $\Hyp_a$ and the condition that $\uptau_J^R$ does not lie in the null space of $(n\widetilde{S}_{\omega,n})^{-1/2}[\mathrm{I}_J - \widetilde{P}_{n}](n\widetilde{S}_{\omega,n})^{-1/2}$ as $n\rightarrow \infty$, $||[\mathrm{I}_J - \widetilde{P}_{n}](n\widetilde{S}_{\omega,n})^{-1/2}\uptau_J^R||^2\cdot ||\sqrt{n}R_{\cdot n}^{w,\psi}||^2 \rightarrow \infty$ as $n\rightarrow \infty$.

\qed

\section{Extensions}\label{App_Sect:Extensions}

\subsection{Extension -- Unit roots and high persistence}\label{subsec:non_stationarity}

In spite of the relatively weak conditions, e.g., for identification (\Cref{ass:parallel_trends,ass:limited_anticip}), and sampling (\Cref{ass:Sampling,ass:technical_CLT}) that are allowed under the baseline model, unit roots and high persistence in $X_t$ can lead to a violation of the dominance condition in \Cref{ass:dominance_Y} and therefore hamper inference using \Cref{Theorem:AsympN}. To this end, this paper extends the baseline theory to accommodate unit root processes and deterministic time trends.

\subsubsection{Unit roots in $X_t$}

The incidence of unit-root non-stationary processes is quite common in economics. Consider the simple linear model \eqref{eqn:ATT_lin_mod}. $X_t$ becomes a unit-root process if $\beta_1=1$ in $U_t = \beta_1 X_{t-1} + \mathcal{E}_t $ and $\mathcal{E}_t$ is some white-noise process. In this first instance, i.e., $X_t = \beta_0 + \widetilde{ATT}_n^{w,\psi}  \indicator{t \geq 1} + X_{t-1} + \mathcal{E}_t$, applying the first difference removes the unit root, namely, $\Delta X_t = \beta_0 + \widetilde{ATT}_n^{w,\psi}  \indicator{t \geq 1} + \mathcal{E}_t$ where $\Delta X_t:=X_t - X_{t-1}$. Thus, handling unit roots in the T-DiD context is straightforward.

\subsubsection{Persistence in $X_t$}\label{SubSect:Persistence}
    Economic time series data, e.g., GDP data, are often highly persistent. Very high $|\beta_1|<1$ in $U_t = \beta_1 X_{t-1} + \mathcal{E}_t $ can pose a problem for inference, especially in small samples. Thus, following, e.g., \citet[Sect. 16.15]{hansen_2021econometrics} and \citet{baillie2025robust}, one may consider substituting in the term $U_t$ into \eqref{eqn:ATT_lin_mod}:
\begin{equation}\label{eqn:ATT_lin_mod2}
    X_t = \beta_0 + \widetilde{ATT}_n^{w,\psi}  \indicator{t \geq 1} + \beta_1 X_{t-1} + \mathcal{E}_t.
\end{equation} $\widetilde{ATT}_n^{w,\psi}$ in the above model is an average (contemporaneous) treatment effect on the treated. Including the lagged $X_t$ in \eqref{eqn:ATT_lin_mod2} ensures the resulting residuals are less serially correlated, thereby improving the efficiency of the estimator \citep{baillie2025robust}. 

Quite importantly, the presence of auto-correlation also has ramifications from an identification perspective. Suppose the true data-generating process followed \eqref{eqn:ATT_lin_mod2}, then for $\mu:= \beta_0/(1-\beta_1) $ where $ |\beta_1|<1 $,  $X_t$ can be written in the form
\begin{equation*}
    X_t = \mu + \beta_1^{t+\T+1}\big( X_{-\T-1} -\mu \big) + \underbrace{\widetilde{ATT}_n^{w,\psi} \Big(\frac{1-\beta_1^t}{1-\beta_1}\Big)}_{\widetilde{A}_{n}(t)} \indicator{t \geq 1} + \underbrace{\sum_{s=-\T}^t \beta_1^{t-s}\mathcal{E}_s}_{\widetilde{\mathcal{E}}_t}.
\end{equation*}From the above, ignoring persistence in $X_t$ has consequences for identification; the probability limit of the estimator is the multiplier effect \( \displaystyle \frac{1}{T} \sum_{t=1}^T \widetilde{A}_{n}(t) = \widetilde{ATT}_n^{w,\psi}\Big(1 - \frac{\beta_1(1-\beta_1^T)}{T(1-\beta_1)} \Big)\frac{1}{1-\beta_1} \approx \widetilde{ATT}_n^{w,\psi} + \widetilde{ATT}_n^{w,\psi} \frac{\beta_1}{1-\beta_1}  \) where the effective estimand $\widetilde{ATT}_n^{w,\psi}$ is properly viewed as the average contemporaneous effect of treatment over the post-treatment window (under the conditions of \Cref{Theorem:Identification}) and the extra term $ \widetilde{ATT}_n^{w,\psi} \frac{\beta_1}{1-\beta_1} $ occurs via the dynamic propagation through $X_{t-1}$. The resulting error term \( \widetilde{\mathcal{E}}_t:=\sum_{s=-\T}^t \beta_1^{t-s}\mathcal{E}_s \) is non-trivially serially correlated, with implications for CLT-based inference and efficiency. For example, including moving average terms to reduce the serial correlation increases the number of parameters to estimate and reduces the precision of the ATT estimate.

\subsubsection{Serially Correlated $ U_t $ }
Suppose that $U_t$ follows an MA(q) process, namely $U_t = \sum_{s=0}^q \beta_s\mathcal{E}_{t-s} $. Ignoring this form of temporal dependence in $U_t$ when estimating \eqref{eqn:ATT_lin_mod} does not lead to inconsistency of the estimator. Moreover, HAC estimators of the standard errors are consistent, although there is a loss in efficiency in not explicitly accounting for this dependence structure. In the context of the SC or factor models, \citet{gonccalves-ng-2024imputation} demonstrates improved coverage rates and reduced bias from controlling for lagged errors in imputing counterfactual outcomes under stationarity conditions.

\subsection{Extension -- Deterministic time trend}\label{App_Sect:Det_Trend}
\Cref{Rem:Prop_1_DGP}, for example, shows that common and possibly unbounded time trends get differenced away in $X_t:= Y_{1,t} - Y_{0,t} $ and do not pose a problem for the asymptotic identification and inference on $ATT_{\omega,T}$. However, if time trends differ between the treated and untreated units, the T-DiD estimate captures these differing trends, which is undesirable from an identification point of view. This subsection follows \citet{Li-2020-statistical} and \citet[Chap. 16]{Hamilton-1994-series} in dealing with deterministic time trends.

Consider the leading case of a deterministic trend in $X_t$ with $U_t=\beta_1 \tilde{t}(t) + \mathcal{E}_t$ in \eqref{eqn:ATT_lin_mod} namely,
\begin{equation}\label{eqn:ATT_lin_mod_dtrend}
    \begin{split}
        X_t &= \beta_0 + \widetilde{ATT}_n^{w,\psi}  \indicator{t\geq 1} + \beta_1 \tilde{t}(t) + \mathcal{E}_t\\
        &= \mathcal{X}_t\widetilde{\beta}_n^{w,\psi} + \mathcal{E}_t\\
        &= \mathcal{X}_t\Gamma_n^{-1/2}\Gamma_n^{1/2}\widetilde{\beta}_n^{w,\psi} + \mathcal{E}_t\\
        &= \mathcal{X}_{\Gamma,t}\widetilde{\beta}_{\Gamma,n}^{w,\psi} + \mathcal{E}_t
    \end{split}
\end{equation}where $\tilde{t}(t):= \T + t + 1$, $\mathcal{X}_t:= \big( 1,\indicator{t\geq 1},\tilde{t}(t) \big)$, and $\mathcal{X}_{\Gamma,t}:=\mathcal{X}_t\Gamma_n^{-1/2}$. To simplify the exposition, the following sampling assumption is imposed on $\mathcal{E}_t$ in \eqref{eqn:ATT_lin_mod_dtrend}. 
\begin{namedassumption}{\ref{ass:Sampling}-Ext}\label{ass:Sampling_dtrend}
    $\mathcal{E}_t$ in \eqref{eqn:ATT_lin_mod_dtrend} is a square-integrable martingale difference sequence with $\E[\mathcal{E}_t|\mathcal{F}_{t-1}]=0 $ and $\E[\mathcal{E}_t^2|\mathcal{F}_{t-1}] = \sigma_{\varepsilon}^2>0$ where $ \mathcal{F}_{t-1}:= \sigma\big(X_{-\T}, \ldots, X_{t-1}\big) $ denotes the natural filtration.
\end{namedassumption}

$\Gamma_n^{-1/2}$ is a diagonal matrix such that 
$$ 
\sum_{t=-\T}^T \widetilde{w}_n(t) \mathcal{X}_{\Gamma,t}'\mathcal{X}_{\Gamma,t}= Q_{\widetilde{w},n} + \mathcal{O}(n^{-1}),  
$$
and $Q_{\widetilde{w},n}$ is a positive-definite matrix  under \Cref{ass:TTratio_lambda_n}. For concreteness and ease of exposition, consider the uniform (regression) weighting scheme $\widetilde{w}(t) = T^{-1}\indicator{t\geq 1} + \T^{-1}\indicator{t\leq -1}$. The appropriate scaling matrix is $\Gamma_{\omega,n}^{1/2} = \mathrm{diag}(1,1,n)=:\Gamma_n^{1/2} $. The following is a linear representation of the ordinary least squares estimator of $\widetilde{\beta}_{\Gamma,n}^{w,\psi}$ in \eqref{eqn:ATT_lin_mod_dtrend}: 
\begin{equation*}
    (\widehat{\beta}_{\Gamma,n}^{w,\psi} - \widetilde{\beta}_{\Gamma,n}^{w,\psi}):= \Big( \sum_{t=-\T}^T \widetilde{w}_n(t) \mathcal{X}_{\Gamma,t}'\mathcal{X}_{\Gamma,t} \Big)^{-1} \sum_{t=-\T}^T \widetilde{w}_n(t) \mathcal{X}_{\Gamma,t}'\mathcal{E}_t.
\end{equation*}

Define $Q(\lambda_n):=Q_A(\lambda_n)^{-1}Q_B(\lambda_n)Q_A(\lambda_n)^{-1}$ where 

\begin{align}
   \label{eqn:Q_A} Q_A(\lambda_n):&= \begin{bmatrix}
        2 &  1 & \Big(\frac{3-2\lambda_n}{2}\Big) \\
         1 &  1 & \Big(\frac{2-\lambda_n}{2}\Big) \\
  \Big(\frac{3-2\lambda_n}{2}\Big) & \Big(\frac{2-\lambda_n}{2}\Big) &   \frac{1}{3} (2\lambda_n^2 - 5\lambda_n + 4)
\end{bmatrix}, \\
  \label{eqn:Q_B}  Q_B(\lambda_n):&= \sigma_{\varepsilon}^2
    \begin{bmatrix}
        \frac{1}{\lambda_n(1-\lambda_n)} &  \frac{1}{\lambda_n} & \frac{1}{\lambda_n} \\
         \frac{1}{\lambda_n} & \frac{1}{\lambda_n} & \frac{1}{2}\Big(\frac{2-\lambda_n}{\lambda_n}\Big) \\
  \frac{1}{\lambda_n} & \frac{1}{2}\Big(\frac{2-\lambda_n}{\lambda_n}\Big) &   \frac{2}{\lambda_n}\big( 3 - 2\lambda_n\big) 
\end{bmatrix},          
\end{align}and $ \sigma_{\varepsilon}^2:= \V[\mathcal{E}_t] $ under \Cref{ass:Sampling_dtrend}. Part (c) of \Cref{lem:Q_A_Q_B} shows that $Q_A(\lambda)$, $Q_B(\lambda)$, and $Q(\lambda)$ are positive-definite uniformly in $\lambda \in [\epsilon, \ 1-\epsilon] \subset (0,1)$.\footnote{Also see \Cref{Fig:min_eig}.} Recall $\mathbb{S}^J$ is the space of all $J\times 1$ vectors with unit Euclidean norm. Let $e_2:=(0,1,0)'$ denote a standard basis vector, $\widetilde{e}_2:= ||Q_B(\lambda_n)^{1/2}e_2||^{-1}Q_B(\lambda_n)^{1/2}e_2$, and $Q(\lambda_n)^{1/2}:=Q_A(\lambda_n)^{-1} Q_B(\lambda_n)^{1/2}$. The following states an asymptotic normality result under deterministic time trends.
\begin{theorem}\label{Theorem:AsympN_dTrend}
\Copy{Key:Theorem:AsympN_dTrend}{Suppose \Cref{ass:parallel_trends,ass:limited_anticip,ass:TTratio_lambda_n,ass:Sampling_dtrend} hold, then

    (a)
    \begin{equation*}
    \uptau_3'Q(\lambda_n)^{-1/2}\sqrt{n}(\widehat{\beta}_{\Gamma,n}^{w,\psi} - \widetilde{\beta}_{\Gamma,n}^{w,\psi}) \xrightarrow{d} \mathcal{N}(0,1)
\end{equation*}uniformly in $\uptau_3 \in \mathbb{S}^3 $, and

     (b) \begin{align*}
    \frac{\sqrt{n}(\widehat{ATT}_{\omega,T} - ATT_{\omega,T})}{\sqrt{e_2'Q(\lambda_n)e_2}} = \widetilde{e}_2'Q(\lambda_n)^{-1/2}\sqrt{n}(\widehat{\beta}_{\Gamma,n}^{w,\psi} - \widetilde{\beta}_{\Gamma,n}^{w,\psi}) + o(1) \xrightarrow{d} \mathcal{N}(0,1)
\end{align*}
\noindent as $n\rightarrow\infty$ under the uniform weighting scheme.}
\end{theorem}
\noindent See \Cref{Proof:Theorem:AsympN_dTrend} for proof.

For the general weighting scheme, one has
\begin{align*}
    &\sum_{t=-\T}^T \widetilde{w}_n(t) \mathcal{X}_{\Gamma,t}'\mathcal{X}_{\Gamma,t}\\
    &= \Gamma_n^{-1/2} \sum_{t=-\T}^T
    \begin{bmatrix}
        \widetilde{w}_n(t) &  \widetilde{w}_n(t)\indicator{t\geq 1} &  \widetilde{w}_n(t) \tilde{t}(t)\\
         \widetilde{w}_n(t)\indicator{t\geq 1}
&  \widetilde{w}_n(t)\indicator{t\geq 1} &  \widetilde{w}_n(t)\indicator{t\geq 1}\tilde{t}(t)\\
  \widetilde{w}_n(t) \tilde{t}(t) &   \widetilde{w}_n(t)\indicator{t\geq 1}\tilde{t}(t) &   \widetilde{w}_n(t) \tilde{t}(t)^2
\end{bmatrix}
    \Gamma_n^{-1/2}\\
    &=    
    \begin{bmatrix}
        2 &  1 &  \sum_{t=-\T}^T \widetilde{w}_n(t) \Big(\frac{\tilde{t}(t)}{n}\Big) \\
         1 &  1 &  \sum_{t=1}^T\widetilde{w}_n(t)\Big(\frac{\tilde{t}(t)}{n}\Big) \\
  \sum_{t=-\T}^T\widetilde{w}_n(t) \Big(\frac{\tilde{t}(t)}{n}\Big) &   \sum_{t=1}^T\widetilde{w}_n(t)\Big(\frac{\tilde{t}(t)}{n}\Big) &   \sum_{t=-\T}^T\widetilde{w}_n(t) \Big(\frac{\tilde{t}(t)}{n}\Big)^2
\end{bmatrix}
\end{align*}where the constants in the above matrix follow from the requirement that weights for pre-treatment and post-treatment periods each sum to one and $\tilde{t}(t):=\T + t + 1$. Thus, generally, the above matrix is well behaved; it is positive definite and all entries are bounded and bounded away from zero under \Cref{ass:TTratio_lambda_n}.\footnote{ $\mathcal{X}_t\Gamma_n^{-1/2} = (1,\indicator{t\geq 1},\tilde{t}(t)/n) $ has full column rank, the entry $\tilde{t}(t)/n \in (0,1] $ uniformly in $\tilde{t}(\cdot) \in [n]$, and $\sum_{t=-\T}^T \widetilde{w}_n(t) \indicator{t\leq -1} = \sum_{t=-\T}^T \widetilde{w}_n(t) \indicator{t\geq 1} = 1$.}

\subsection{Multiple treated units}\label{App_Sect:Mult_Treat}
The index $i$ is included in the subscripts for emphasis and to render the respective quantities specific to a treated unit $i \in [I] $. Define the $I\times 1$ vector $\widehat{ATT}_{\omega,\cdot T}^*:= (\widehat{ATT}_{\omega,1T}^*,\ldots,\widehat{ATT}_{\omega,IT}^*)'$, the $I\times 1$ vector $ATT_{\omega,\cdot T}:= (ATT_{\omega,1T},\ldots,ATT_{\omega,IT})'$, and the $I\times I$ matrix $S_{\omega,n}^*:= \E[(\widehat{ATT}_{\omega,\cdot T}^* - ATT_{\omega,\cdot T})(\widehat{ATT}_{\omega,\cdot T}^* - ATT_{\omega,\cdot T})'] $.

\begin{corollary}\label{Coro:2}
Let \Cref{ass:parallel_trends,ass:limited_anticip,ass:Sampling,ass:technical_CLT,ass:TTratio_lambda_n,ass:dominance_Y} for each unique treated-control pair $(i,j)$ such that $ j\in [J_i] $ and $i\in [I]$ hold. In addition to \Cref{ass:bound_sn_ext},
\noindent (a)
\[(s_{\omega,in}^*)^{-1}(\widehat{ATT}_{\omega,iT}^* - ATT_{\omega,iT}) \xrightarrow{d} \mathcal{N}(0,1) \]
\noindent and  (b)
\[
(\uptau_I'S_{\omega,n}^* \uptau_I)^{-1}\uptau_I'(\widehat{ATT}_{\omega,\cdot T}^* - ATT_{\omega,\cdot T}) \xrightarrow{d} \mathcal{N}(0,1) 
\]
\noindent for all $\uptau_I \in \mathbb{S}^I$ uniformly in $\W^2 $.
\end{corollary}

\begin{proof}
\textbf{Part (a):} The proof follows from that of \Cref{Theorem:AsympN_ext} given the specific linear combination in $ \widehat{ATT}_{iT}^* = {h_{iJ}^*}'\widehat{ATT}_{i\cdot T} = (\mathbbm{1}_J'S_{in}^{-1}\mathbbm{1}_J)^{-1} \mathbbm{1}_J'S_{in}^{-1}\widehat{ATT}_{i\cdot T}$, set  $\uptau_J = h_{iJ}^* ||h_{iJ}^*||^{-1} $.

\textbf{Part (b):} The proof of this part is similar to \Cref{Theorem:AsympN_ext} and is hence omitted.
\end{proof}

\begin{remark}
    Part (b) of \Cref{Coro:2} is based on efficient estimators of $ATT_{\omega,iT}$. This, however, need not be the case. More general weighting schemes are admissible, including one that simply selects a single control unit out of the set $ [J_i]$ of valid controls for unit $i$.
\end{remark}

\subsection{Proof of \Cref{Theorem:AsympN_dTrend}}\label{Proof:Theorem:AsympN_dTrend}

\begin{namedtheorem}{\ref{Theorem:AsympN_dTrend}}
    \Paste{Key:Theorem:AsympN_dTrend}
\end{namedtheorem}

\textbf{Proof:}

\paragraph{Part (a):}

Under the conditions of \Cref{lem:Q_A_Q_B}, the continuous mapping theorem and the continuity of the inverse at a non-singular matrix,
\begin{align*}
    \sqrt{n}(\widehat{\beta}_{\Gamma,n}^{w,\psi} - \widetilde{\beta}_{\Gamma,n}^{w,\psi})&= \Big( \sum_{t=-\T}^T \widetilde{w}_n(t) \mathcal{X}_{\Gamma,t}'\mathcal{X}_{\Gamma,t} \Big)^{-1} \sqrt{n}\sum_{t=-\T}^T \widetilde{w}_n(t) \mathcal{X}_{\Gamma,t}'\mathcal{E}_t\\
    &= Q_A(\lambda_n)^{-1} \sqrt{n}\sum_{t=-\T}^T \widetilde{w}_n(t) \mathcal{X}_{\Gamma,t}'\mathcal{E}_t + \mathcal{O}_p(n^{-1}).
\end{align*}

\noindent Thus, under the conditions of \Cref{lem:Q_A_Q_B},
\begin{align*}
    \V\Big[\uptau_3'\sqrt{n}(\widehat{\beta}_{\Gamma,n}^{w,\psi} - \widetilde{\beta}_{\Gamma,n}^{w,\psi})\Big] = \uptau_3'Q(\lambda_n)\uptau_3 + \mathcal{O}(n^{-1}).
\end{align*}The rest of the proof is conducted in the following steps.

First, the scalar-valued random variable $ \Upsilon_t:= \sqrt{n}\uptau_3'\widetilde{w}_n(t) \mathcal{X}_{\Gamma,t}'\mathcal{E}_t$ constitutes a square-integrable martingale difference sequence under \Cref{ass:Sampling_dtrend} since the following hold: (1) $\E[\Upsilon_t|\mathcal{F}_{\Upsilon,t-1}]=0 $ where $\mathcal{F}_{\Upsilon,t} := \sigma(\{\Upsilon_{t'}, t' \leq t\}) $ is the natural filtration and (2) 
\begin{align*}
    \sigma_{\Upsilon,t}^2:= \E[\Upsilon_t^2|\mathcal{F}_{t-1}] &= \sigma_{\varepsilon}^2n\widetilde{w}_n(t)^2\uptau_3' \mathcal{X}_{\Gamma,t}'\mathcal{X}_{\Gamma,t}\uptau_3\\ 
    &= n\sigma_{\varepsilon}^2||\widetilde{w}_n(t)\mathcal{X}_{\Gamma,t}\uptau_3||^2\\ 
    &\leq n\sigma_{\varepsilon}^2\frac{1}{\T^2 \wedge T^2}||\mathcal{X}_{\Gamma,t}||^2\\ 
    &= n\sigma_{\varepsilon}^2\frac{1}{\T^2 \wedge T^2}||(1,\indicator{t\geq 1},\tilde{t}(t)/n)||^2\\ 
    &\leq 3n\sigma_{\varepsilon}^2\frac{1}{\T^2 \wedge T^2}\\
    &=3\sigma_{\varepsilon}^2\frac{1}{n}\frac{1}{(1-\lambda_n)^2 \wedge \lambda_n^2}=:c_{\Upsilon,n}^2
\end{align*}

 \noindent Observe that $c_{\Upsilon,n}^2$ does not vary with $t$. First, the above implies $\{\Upsilon_t^2/c_{\Upsilon,n}^2\}$ is uniformly integrable. 
 
 Second, define 
 $$ 
 s_{\Upsilon,n}^2:= \sum_{t=-\T}^T \E[\Upsilon_t^2] = n\sigma_{\varepsilon}^2  \sum_{t=-\T}^T ||\widetilde{w}_n(t)\mathcal{X}_{\Gamma,t}\uptau_3||^2 =  ||Q_B(\lambda_n)^{1/2}\uptau_3||^2 + \mathcal{O}(n^{-1}).
 $$ From the above, the lower bound can be characterised as 
 \begin{align*}
     s_{\Upsilon,n}^2 = n\sigma_{\varepsilon}^2  \sum_{t=-\T}^T ||\widetilde{w}_n(t)\mathcal{X}_{\Gamma,t}\uptau_3||^2 \geq \sigma_{\varepsilon}^2\frac{1}{(1-\lambda_n)^2 \vee \lambda_n^2} \frac{1}{n}\sum_{t=-\T}^T ||\mathcal{X}_{\Gamma,t}||^2 \geq \sigma_{\varepsilon}^2\frac{1}{(1-\lambda_n)^2 \vee \lambda_n^2}.
 \end{align*}

 \noindent Thus,
 $$
 \sup_{n} n \frac{c_{\Upsilon,n}^2}{s_{\Upsilon,n}^2} \leq 3\frac{(1-\lambda_n)^2 \vee \lambda_n^2 }{(1-\lambda_n)^2 \wedge \lambda_n^2} = 3\Big(\frac{(1-\lambda_n) \vee \lambda_n }{(1-\lambda_n) \wedge \lambda_n}\Big)^2 < 3\Big(\frac{1-\epsilon}{\epsilon}\Big)^2 < \infty  
 $$  under \Cref{ass:TTratio_lambda_n}. 

 Taken together, it follows from \citet[Theorem 25.4]{davidson2021stochastic} that conditions 25.3(a) and 25.3(b) of \citet[Theorem 25.3]{davidson2021stochastic} are satisfied for $\Upsilon_{nt}:= \Upsilon_t/s_{\Upsilon,n}$. It follows from \citet[Theorem 25.3]{davidson2021stochastic} that $ \sum_{t=-\T}^T \Upsilon_{nt} \xrightarrow{d} \mathcal{N}(0,1) $ as $n\rightarrow \infty$ under \Cref{ass:TTratio_lambda_n}.

Noticing that $ \mathbb{S}^3 \ni \widetilde{\uptau}_3:= ||Q_B(\lambda_n)^{1/2}\uptau_3||^{-1}Q_B(\lambda_n)^{1/2}\uptau_3 $, for $ \uptau_3 \in \mathbb{S}^3 $,
\begin{equation}\label{eqn:expand_dTrend_CLT}
    \begin{split}
    &\widetilde{\uptau}_3'Q(\lambda_n)^{-1/2}\sqrt{n}(\widehat{\beta}_{\Gamma,n}^{w,\psi} - \widetilde{\beta}_{\Gamma,n}^{w,\psi}) = \widetilde{\uptau}_3'Q(\lambda_n)^{-1/2}Q_A(\lambda_n)^{-1} \sqrt{n}\sum_{t=-\T}^T \widetilde{w}_n(t) \mathcal{X}_{\Gamma,t}'\mathcal{E}_t + \mathcal{O}_p(n^{-1})\\
    &=\widetilde{\uptau}_3'Q(\lambda_n)^{-1/2} \underbrace{Q_A(\lambda_n)^{-1} Q_B(\lambda_n)^{1/2}}_{=: Q(\lambda_n)^{1/2}}\Big( Q_B(\lambda_n)^{-1/2} \sqrt{n}\sum_{t=-\T}^T \widetilde{w}_n(t) \mathcal{X}_{\Gamma,t}'\mathcal{E}_t \Big) + \mathcal{O}_p(n^{-1})\\
    &=\widetilde{\uptau}_3'Q_B(\lambda_n)^{-1/2} \sqrt{n}\sum_{t=-\T}^T \widetilde{w}_n(t) \mathcal{X}_{\Gamma,t}'\mathcal{E}_t + \mathcal{O}_p(n^{-1})\\
    &:= ||Q_B(\lambda_n)^{1/2}\uptau_3||^{-1} \uptau_3'\sqrt{n}\sum_{t=-\T}^T \widetilde{w}_n(t) \mathcal{X}_{\Gamma,t}'\mathcal{E}_t + \mathcal{O}_p(n^{-1})\\
    &=: \sum_{t=-\T}^T \Upsilon_{nt} + \mathcal{O}_p(n^{-1}) \xrightarrow{d} \mathcal{N}(0,1).
\end{split}
\end{equation}

\paragraph{Part (b):}

Noticing that $ \mathbb{S}^3 \ni \widetilde{e}_2:= ||Q_B(\lambda_n)^{1/2}e_2||^{-1}Q_B(\lambda_n)^{1/2}e_2 $, for $ (0,1,0)' =:e_2 \in \mathbb{S}^3 $, it follows from \eqref{eqn:expand_dTrend_CLT} that 
\begin{align*}
    \widetilde{e}_2'Q(\lambda_n)^{-1/2}\sqrt{n}(\widehat{\beta}_{\Gamma,n}^{w,\psi} - \widetilde{\beta}_{\Gamma,n}^{w,\psi}) &= ||Q_B(\lambda_n)^{1/2}e_2||^{-1} e_2'\sqrt{n}\sum_{t=-\T}^T \widetilde{w}_n(t) \mathcal{X}_{\Gamma,t}'\mathcal{E}_t + \mathcal{O}_p(n^{-1})\\
    &= \frac{\sqrt{n}(\widehat{ATT}_{\omega,T} - \widetilde{ATT}_n^{w,\phi} )}{\sqrt{e_2'Q(\lambda_n)e_2}}\\
    &= \frac{\sqrt{n}(\widehat{ATT}_{\omega,T} - ATT_{\omega,T})}{\sqrt{e_2'Q(\lambda_n)e_2}} + \frac{\sqrt{n}R_n^{w,\phi}}{\sqrt{e_2'Q(\lambda_n)e_2}}\\
    &=\frac{\sqrt{n}(\widehat{ATT}_{\omega,T} - ATT_{\omega,T})}{\sqrt{e_2'Q(\lambda_n)e_2}} + o(1).
\end{align*}The last equality follows from the conditions of \Cref{Theorem:Identification} and the continuous mapping theorem. Thus, invoking part (a) above,
\begin{align*}
    \frac{\sqrt{n}(\widehat{ATT}_{\omega,T} - ATT_{\omega,T})}{\sqrt{e_2'Q(\lambda_n)e_2}} = \widetilde{e}_2'Q(\lambda_n)^{-1/2}\sqrt{n}(\widehat{\beta}_{\Gamma,n}^{w,\psi} - \widetilde{\beta}_{\Gamma,n}^{w,\psi}) + o(1) \xrightarrow{d} \mathcal{N}(0,1),
\end{align*} and this proves the assertion as claimed.
\qed

\section{Supporting Lemmata}\label{App_Sect:Supp_Lem}

\begin{lemma}\label{lem:X_Sig_sn_LrBound}
    Suppose \Cref{ass:dominance_Y,ass:bound_sn,ass:TTratio_lambda_n} hold, then $\displaystyle\frac{||X_{nt}||_r}{c_{nt}}\lesssim 1 $ uniformly in $(n,t)$ and $\W^2 $ where $r>2$.
\end{lemma}

\begin{proof}
    \begin{align*}
        &||\omega_n(t)(X_t - \E[X_t])||_r = w_T(t)||X_t-\E[X_t]||_r \indicator{t\geq 1} + \psi_{\T }(-t)||X_t-\E[X_t]||_r\indicator{t\leq -1}\\
        &= w_T(t)||(Y_{1,t}-Y_{0,t})-\E[(Y_{1,t}-Y_{0,t})]||_r \indicator{t\geq 1} + \psi_{\T }(-t)||(Y_{1,t}-Y_{0,t})-\E[(Y_{1,t}-Y_{0,t})]||_r\indicator{t\leq -1}\\
        &\leq C \big( w_T(t)\indicator{t\geq 1} + \psi_{\T }(-t) \indicator{t\leq -1}\big) \\
        &\leq \frac{C}{\lambda_n n}\sup_{\phi\in\W } \Big\{ T \max_{t \in [T]}\phi_T(t) \Big\}\indicator{t\geq 1} + \frac{C}{(1-\lambda_n)n}\sup_{\phi\in\W } \Big\{\T \max_{\tau \in [-\T] } \phi_{\T }(-\tau) \Big\} \indicator{t\leq -1}\\
        &\lesssim n^{-1}.
    \end{align*}
\noindent The first inequality follows from \Cref{ass:dominance_Y}. The last line uses the definition of $\W$ in \eqref{eqn:dfn_W} and \Cref{ass:TTratio_lambda_n}. It follows from the above and \Cref{ass:bound_sn} that 
\begin{align*}
    \frac{||X_{nt}||_r}{c_{nt}} &= \frac{||\omega_n(t)(X_t - \E[X_t])||_r}{s_{\omega,n}(\sigma_{nt}\vee s_{\omega,n})} \lesssim \frac{1}{\sqrt{n}s_{\omega,n}(\sqrt{n}\sigma_{nt} \vee \sqrt{n}s_{\omega,n})} \leq \frac{1}{\epsilon \sqrt{n}\sigma_{nt} \vee \epsilon^2} = \frac{1}{\epsilon \sqrt{n}\sigma_{nt}} \wedge \frac{1}{\epsilon^2} \leq \frac{1}{\epsilon^2}
\end{align*}
\noindent uniformly in $(n,t)$ and $\W^2 $.
\end{proof}

\begin{lemma}\label{lem:CLT}
    Suppose \Cref{ass:dominance_Y,ass:bound_sn,ass:TTratio_lambda_n,ass:Sampling,ass:technical_CLT} hold, then $\displaystyle \sum_{t=-\T }^{T} X_{nt} \xrightarrow{d} \mathcal{N}(0,1) $ uniformly in $\W^2 $.
\end{lemma}

\begin{proof}

The proof proceeds by the following steps.

(a) \[ \Big|\Big|\sum_{t=-\T }^{T}X_{nt}\Big|\Big|_2 =  s_{\omega,n}^{-1}\Big|\Big|\sum_{t=-\T }^{T}\omega_n(t)(X_t - \E[X_t])\Big|\Big|_2=1. \]

(b) By \Cref{lem:X_Sig_sn_LrBound},
$$
 \Big|\Big|\frac{X_{nt}}{c_{nt}}\Big|\Big|_r = \Big|\Big|\frac{X_{nt}}{\sigma_{nt}\vee s_{\omega,n}}\Big|\Big|_r \lesssim 1  
$$
\noindent uniformly in $(n,t)$ and $\W^2 $ hence, $\displaystyle \frac{X_{nt}}{c_{nt}}$ is $L_r$-bounded uniformly in $(n,t)$ and $\W^2 $ under \Cref{ass:dominance_Y,ass:bound_sn,ass:TTratio_lambda_n}.

(c) The conditions of \Cref{lem:X_Sig_sn_LrBound} imply $\displaystyle \frac{d_{nt}}{c_{nt}} \leq \bar{d}\frac{||X_{nt}||_2}{\sigma_{nt}\vee s_{\omega,n}} \lesssim 1 $ uniformly in $(n,t)$ and $\W^2 $. 

Steps (a)-(c) above and \Cref{ass:Sampling,ass:technical_CLT} verify the conditions of \citet[Theorem 25.12]{davidson2021stochastic} thus $\sum_{t=-\T }^{T} X_{nt} \xrightarrow{d} \mathcal{N}(0,1) $ uniformly in $\W^2 $ as claimed.
\end{proof}

\begin{lemma}\label{lem:Q_A_Q_B}
    Fix $\widetilde{w}_n(t) = T^{-1}\indicator{t\geq 1} + \T^{-1}\indicator{t\leq -1}$, i.e., the uniform regression weighting scheme, then under \Cref{ass:TTratio_lambda_n,ass:Sampling_dtrend}, 
    
    (a)
    \begin{align*}
        \sum_{t=-\T}^T\widetilde{w}(t)\mathcal{X}_{\Gamma,t}'\mathcal{X}_{\Gamma,t} = Q_A(\lambda_n) + \mathcal{O}(n^{-1});
    \end{align*} 
    
    (b)
    \begin{align*}
        \V\Big[\sqrt{n}\sum_{t=-\T}^T \widetilde{w}_n(t) \mathcal{X}_{\Gamma,t}'\mathcal{E}_t\Big] = Q_B(\lambda_n) + \mathcal{O}(n^{-1})
    \end{align*}
    where the expressions of $Q_A(\lambda_n)$ and $Q_B(\lambda_n)$ are given in \eqref{eqn:Q_A} and \eqref{eqn:Q_B}, respectively; and 
    
    (c) $Q_A(\lambda), \ Q_B(\lambda)$, and $Q(\lambda):=Q_A(\lambda)^{-1} Q_B(\lambda) Q_A(\lambda)^{-1} $ are (each) positive definite uniformly in $\lambda \in [\epsilon, \ 1-\epsilon] \subset (0,1) $.
\end{lemma}

\begin{proof}
\textbf{Part (a)}

    As pre- and post-treatment weights must sum to one, one has, using the uniform regression weighting scheme $\widetilde{w}_n(t) = T^{-1}\indicator{t\geq 1} + \T^{-1}\indicator{t\leq -1}$ that,
\begin{align*}
    \sum_{t=-\T}^T\widetilde{w}(t)\mathcal{X}_t'\mathcal{X}_t 
    &= \sum_{t=-\T}^T
    \begin{bmatrix}
        \widetilde{w}_n(t) &  \widetilde{w}_n(t)\indicator{t\geq 1} &  \widetilde{w}_n(t) \tilde{t}(t)\\
         \widetilde{w}_n(t)\indicator{t\geq 1}
&  \widetilde{w}_n(t)\indicator{t\geq 1} &  \widetilde{w}_n(t)\indicator{t\geq 1}\tilde{t}(t)\\
  \widetilde{w}_n(t) \tilde{t}(t) &   \widetilde{w}_n(t)\indicator{t\geq 1}\tilde{t}(t) &   \widetilde{w}_n(t)\tilde{t}(t)^2
\end{bmatrix}\\
    &= \begin{bmatrix}
        2 &  1 &  1 + n\Big(\frac{3-2\lambda_n}{2}\Big) \\
         1 &  1 &  \frac{1}{2} + n\Big(\frac{2-\lambda_n}{2}\Big) \\
  1 + n\Big(\frac{3-2\lambda_n}{2}\Big) & \frac{1}{2} + n\Big(\frac{2-\lambda_n}{2}\Big) &   q_{A,n}^*
\end{bmatrix}
\end{align*}where  
$$
q_{A,n}^* = \frac{1}{6\lambda_n} \Big( 2\big((2\lambda_n-1)(1-\lambda_n)^2 + 1\big)n^2 + 3((1-\lambda_n)(2\lambda_n-1) + 1)n + 2\lambda_n   \Big) 
$$ which is well-defined for $\lambda_n>0$  thanks to \Cref{ass:TTratio_lambda_n}.

The non-trivial entries of the above matrix are computed in the following steps.
\noindent First, 
\begin{align*}
    \sum_{t=-\T}^T\widetilde{w}_n(t) \indicator{t\geq 1} \tilde{t}(t) = \frac{1}{T} \sum_{t=\T+1}^{n}t = \frac{1}{T}\frac{T}{2}(\T + 1 + n) = \frac{1}{2}(\T + 1 + n) = \frac{1}{2} + n\Big(\frac{2-\lambda_n}{2}\Big).
\end{align*}

\noindent Second, 
\begin{align*}
    \sum_{t=-\T}^T\widetilde{w}_n(t) \tilde{t}(t) &= \frac{1}{\T} \sum_{t=1}^{\T} t + \frac{1}{T} \sum_{t=\T+1}^{n} t\\
    &= \frac{1}{2}(\T+1) + \frac{1}{2}(\T+1+n) = 1 + n\Big(\frac{3-2\lambda_n}{2}\Big).
\end{align*}

\noindent Third,
\begin{align*}
    \sum_{t=-\T}^T \widetilde{w}_n(t)\tilde{t}(t)^2 &= \frac{1}{\T} \sum_{t=1}^{\T} t^2 + \frac{1}{T} \sum_{t=\T+1}^{n} t^2\\
    &= \frac{1}{\T} \sum_{t=1}^{\T} t^2 + \frac{1}{T} \Big( \sum_{t=1}^{n} t^2 - \sum_{t=1}^{\T} t^2\Big)\\
    &= \frac{1}{T}  \sum_{t=1}^{n} t^2 + \Big(\frac{1}{\T} - \frac{1}{T}\Big)\sum_{t=1}^{\T} t^2\\
    &= \frac{1}{T}\Big(\frac{1}{6} n(n+1)(2n+1) \Big) + \Big(\frac{1}{\T} - \frac{1}{T}\Big)\Big(\frac{1}{6} \T(\T+1)(2\T+1) \Big)\\
    &= \frac{1}{6\lambda_n} \Big( 2\big((2\lambda_n-1)(1-\lambda_n)^2 + 1\big)n^2 + 3((1-\lambda_n)(2\lambda_n-1) + 1)n + 2\lambda_n   \Big)\\
    &=: q_{A,n}^*.
\end{align*}

Next, fix $\Gamma_{\omega,n}^{-1/2} = \mathrm{diag}(1,1,n^{-1})$, then, 
\begin{align*}
    \sum_{t=-\T}^T\widetilde{w}(t)\mathcal{X}_{\Gamma,t}'\mathcal{X}_{\Gamma,t} &= 
    \begin{bmatrix}
        2 &  1 & \Big(\frac{3-2\lambda_n}{2}\Big) \\
         1 &  1 & \Big(\frac{2-\lambda_n}{2}\Big) \\
  \Big(\frac{3-2\lambda_n}{2}\Big) & \Big(\frac{2-\lambda_n}{2}\Big) &   \frac{1}{3} (2\lambda_n^2 - 5\lambda_n + 4)
\end{bmatrix} + \mathcal{O}(n^{-1})\\ 
&=:Q_A(\lambda_n) + \mathcal{O}(n^{-1}).
\end{align*}

\textbf{Part (b)}

Now consider 
\begin{align*}
    \sum_{t=-\T}^T \widetilde{w}_n(t) \mathcal{X}_{\Gamma,t}'\mathcal{E}_t = 
    \begin{bmatrix}
        \sum_{t=-\T}^T \widetilde{w}_n(t) \mathcal{E}_t\\
        \sum_{t=-\T}^T \widetilde{w}_n(t) \indicator{t\geq 1}\mathcal{E}_t\\
        \sum_{t=-\T}^T \widetilde{w}_n(t) (\tilde{t}(t)/n) \mathcal{E}_t
    \end{bmatrix}.
\end{align*} Under \Cref{ass:Sampling_dtrend} on $\mathcal{E}_t$, 
\begin{align*}
    \V[\sqrt{n}\sum_{t=-\T}^T \widetilde{w}_n(t) \mathcal{X}_{\Gamma,t}'\mathcal{E}_t] &= \sigma_{\varepsilon}^2n\sum_{t=-\T}^T \widetilde{w}_n(t)^2\mathcal{X}_{\Gamma,t}'\mathcal{X}_{\Gamma,t}\\ 
    &= \sigma_{\varepsilon}^2n\sum_{t=-\T}^T
    \begin{bmatrix}
        \widetilde{w}_n(t)^2 &  \widetilde{w}_n(t)^2\indicator{t\geq 1} &  \widetilde{w}_n(t)^2 \tilde{t}(t)/n\\
         \widetilde{w}_n(t)^2\indicator{t\geq 1}
&  \widetilde{w}_n(t)^2\indicator{t\geq 1} &  \widetilde{w}_n(t)^2\indicator{t\geq 1}\tilde{t}(t)/n\\
  \widetilde{w}_n(t)^2 \tilde{t}(t)/n &   \widetilde{w}_n(t)^2\indicator{t\geq 1}\tilde{t}(t)/n &   \widetilde{w}_n(t)^2\big(\tilde{t}(t)/n\big)^2
\end{bmatrix}\\
    &= \sigma_{\varepsilon}^2
    \begin{bmatrix}
        \frac{1}{\lambda_n(1-\lambda_n)} &  \frac{1}{\lambda_n} & \frac{1}{\lambda_n} \\
         \frac{1}{\lambda_n} & \frac{1}{\lambda_n} & \frac{1}{2}\Big(\frac{2-\lambda_n}{\lambda_n}\Big) \\
  \frac{1}{\lambda_n} & \frac{1}{2}\Big(\frac{2-\lambda_n}{\lambda_n}\Big) &   \frac{2}{\lambda_n}\big( 3 - 2\lambda_n\big) 
\end{bmatrix} + \mathcal{O}(n^{-1})\\
&= Q_B(\lambda_n) + \mathcal{O}(n^{-1}).
\end{align*}

The entries of the above matrix are computed in the following steps.

\noindent First,
\begin{align*}
    n\sum_{t=-\T}^T \widetilde{w}_n(t)^2 \indicator{t\geq 1} = \frac{n}{T^2}\sum_{t=-\T}^T \indicator{t\geq 1} =  \frac{n}{T} = \frac{1}{\lambda_n}.
\end{align*}

\noindent Second,
\begin{align*}
    n\sum_{t=-\T}^T \widetilde{w}_n(t)^2 = \frac{n}{\T^2}\sum_{t=-\T}^T \indicator{t\leq -1}  + \frac{n}{T^2}\sum_{t=-\T}^T \indicator{t\geq 1} =  \frac{n}{\T} + \frac{n}{T} = \frac{1}{\lambda_n(1-\lambda_n)}.
\end{align*}

\noindent Third, 
\begin{align*}
    n\sum_{t=-\T}^T \widetilde{w}_n(t)^2\indicator{t\geq 1}\tilde{t}(t)/n &= \sum_{t=-\T}^T \widetilde{w}_n(t)^2\indicator{t\geq 1}\tilde{t}(t)\\
    &= \frac{1}{T^2}\sum_{t=\T+1}^n t = \frac{1}{2T}( \T + 1 + n )\\
    &= \frac{1}{2}\frac{(2-\lambda_n)}{\lambda_n}  + \frac{1}{2\lambda_n n}\\
    &= \frac{1}{2}\frac{(2-\lambda_n)}{\lambda_n}  + \mathcal{O}(n^{-1})
\end{align*} thanks to \Cref{ass:TTratio_lambda_n}.

\noindent Fourth,
\begin{align*}
    n\sum_{t=-\T}^T \widetilde{w}_n(t)^2 \tilde{t}(t)/n &= \sum_{t=-\T}^T \widetilde{w}_n(t)^2 \tilde{t}(t) = \frac{1}{\T^2} \sum_{t=1}^{\T} t + \frac{1}{T^2} \sum_{t=\T+1}^n t \\ 
    &= \frac{1}{2\T}( \T + 1) + \frac{1}{2T}( \T + 1 + n )\\
    &= \frac{1}{2}\Big(1 + \frac{2-\lambda_n}{\lambda_n}\Big) + \frac{1}{2n}\Big( \frac{1}{(1-\lambda_n)} + \frac{1}{\lambda_n} \Big)\\
    &= \frac{1}{2}\Big(1 + \frac{2-\lambda_n}{\lambda_n}\Big) + \mathcal{O}(n^{-1})\\
    &= \frac{1}{\lambda_n} + \mathcal{O}(n^{-1})
\end{align*} thanks to \Cref{ass:TTratio_lambda_n}.

\noindent Fifth, 
\begin{align*}
    &n\sum_{t=-\T}^T\widetilde{w}_n(t)^2(\tilde{t}(t)/n)^2 = \frac{1}{n} \sum_{t=-\T}^T \widetilde{w}_n(t)^2\tilde{t}(t)\\ 
    &= \frac{1}{n} \Big( \frac{1}{\T^2}\sum_{t=1}^{\T} t^2 + \frac{1}{T^2}\sum_{t=\T+1}^n t^2 \Big) \\
    &= \frac{1}{n} \Big( \frac{1}{T^2}\sum_{t=1}^n t^2 + \Big( \frac{1}{\T^2} - \frac{1}{T^2} \Big) \sum_{t=1}^{\T} t^2 \Big)\\
    &= \frac{1}{6nT^2}n(n+1)(2n+1) + \frac{1}{6n}\Big( \frac{T^2 - \T^2}{\T^2T^2}\Big)\big(\T(\T+1)(2\T+1) \big)\\
    &= \frac{1}{\lambda_n^2n^3} \Big( 2n^3 + 3n^2 + n \Big) + \Big( \frac{2\lambda_n - 1}{\lambda_n^2(1-\lambda_n)^2n^3} \Big)\big( 2(1-\lambda_n)^3n^3 + 3(1-\lambda_n)^2n^2 + (1-\lambda_n)n \big) \\
    &= \frac{2}{\lambda_n^2} + \frac{2}{\lambda_n^2}(2\lambda_n-1)(1-\lambda_n) + \frac{1}{\lambda_n^2} \Big( \frac{3}{n} + \frac{1}{n^2} \Big) + \Big( \frac{2\lambda_n - 1}{\lambda_n^2(1-\lambda_n)^2} \Big)\Big( \frac{3(1-\lambda_n)^2}{n} + \frac{1-\lambda_n}{n^2} \Big)  \\
    &= \frac{2}{\lambda_n^2}\big( 1 + (2\lambda_n-1)(1-\lambda_n) \big) + \mathcal{O}(n^{-1})\\
    &= \frac{2}{\lambda_n}\big( 3 - 2\lambda_n\big) + \mathcal{O}(n^{-1}).
\end{align*}

\noindent All the above entries are well-defined thanks to \Cref{ass:TTratio_lambda_n}.

\textbf{Part (c)}
 
Let
 \begin{align*}
     Q_A(\lambda) = \begin{bmatrix}
        2 &  1 & \Big(\frac{3-2\lambda}{2}\Big) \\
         1 &  1 & \Big(\frac{2-\lambda}{2}\Big) \\
  \Big(\frac{3-2\lambda}{2}\Big) & \Big(\frac{2-\lambda}{2}\Big) &   \frac{1}{3} (2\lambda^2 - 5\lambda + 4)
\end{bmatrix} 
=: \begin{bmatrix}
A_{1:2,1:2}(\lambda) & A_{1:2,3}(\lambda) \\    
A_{1:2,3}(\lambda)' & A_{3,3}(\lambda)
\end{bmatrix},
 \end{align*}where $ A_{1:2,1:2}(\lambda):= \begin{bmatrix}
        2 &  1 \\
         1 &  1 \\
\end{bmatrix} $ is the $2\times 2$ leading principal sub-matrix. The eigenvalues of $A_{1:2,1:2}(\lambda)$ are $ \big((3+\sqrt{5})/2, \, (3-\sqrt{5})/2 \big)' $, thus $A_{1:2,1:2}(\lambda)$ is positive definite uniformly in $ \lambda \in [\epsilon, \, 1-\epsilon ] $. Further, consider the Schur complement of $A_{1:2,1:2}(\lambda)$ in $Q_A(\lambda)$, namely 
\begin{align*}
    \mathrm{SC}_{A}(\lambda) :&= A_{3,3}(\lambda) - A_{1:2,3}(\lambda)' A_{1:2,1:2}^{-1}(\lambda) A_{1:2,3}(\lambda)\\ 
    &= \frac{(2\lambda^2 - 5\lambda + 4)}{3}  - \frac{2\lambda^2 - 6\lambda + 5}{4}\\
    &= \frac{1}{6}\big(\lambda - \tfrac{1}{2}\big)^2 + \frac{1}{24}.
\end{align*} Over the open interval $(0,1)$, $\mathrm{SC}_{A}(\lambda)$ is a parabola, convex with a global minimum at $\lambda^*=1/2$ which equals $\mathrm{SC}_{A}(1/2) = 1/24 > 0 $. It follows from the above and \citet[Proposition 8.2.4(v)]{bernstein2009matrix} that $Q_A(\lambda)$ is positive definite uniformly in $ \lambda \in [\epsilon, 1-\epsilon], \ \epsilon \in (0,1/2] $.

Next, consider 
\begin{align*}
    \frac{\lambda}{ \sigma_{\varepsilon}^{2} }Q_B(\lambda) = 
    \begin{bmatrix}
        \frac{1}{(1-\lambda)} &  1 & 1 \\
         1 & 1 & \frac{(2-\lambda)}{2} \\
  1 & \frac{(2-\lambda)}{2} &   2( 3 - 2\lambda) 
\end{bmatrix},
\end{align*}and notice that since $\sigma_{\varepsilon}^{2}>0$, $Q_B(\lambda)$ is positive definite in $\lambda \in [\epsilon, 1-\epsilon]$ if and only if $\frac{\lambda}{ \sigma_{\varepsilon}^{2} }Q_B(\lambda)$ is positive definite. 

Let $B_{1:k,1:k}(\lambda)$ denote the upper-left $k\times k$ sub-matrix of $\frac{\lambda}{ \sigma_{\varepsilon}^{2} }Q_B(\lambda)$. 
$ \operatorname{det}\big(B_{1,1}(\lambda)\big) = \frac{1}{(1-\lambda)} \geq \frac{1}{(1-\epsilon)} > 0 $ for all $\lambda \in [\epsilon, 1-\epsilon] $. $ \operatorname{det}(B_{1:2,1:2}(\lambda)) = \frac{1}{1-\lambda} - 1 \geq \frac{1}{(1-\epsilon)} - 1 = \frac{\epsilon}{(1-\epsilon)} > 0 $. This confirms $B_{1:2,1:2}(\lambda)$ is positive definite uniformly in $\lambda \in [\epsilon, 1-\epsilon] $ by Sylvester's criterion. Compute the Schur complement of $B_{1:2,1:2}(\lambda)$ in $\frac{\lambda}{ \sigma_{\varepsilon}^{2} }Q_B(\lambda)$:
\begin{align*}
    \mathrm{SC}_B(\lambda) :&= B_{3,3}(\lambda) - B_{1:2,3}(\lambda)' B_{1:2,1:2}^{-1}(\lambda) B_{1:2,3}(\lambda)\\ 
    &= 2( 3 - 2\lambda) - (1 - (3/4)\lambda)\\
    &= 4^{-1}(20-13\lambda)
\end{align*}$ \mathrm{SC}_B(\lambda) \geq 4^{-1}(7+13\epsilon) > 0 $ uniformly in $\lambda \in [\epsilon, 1-\epsilon] $. Again, invoking \citet[Proposition 8.2.4(v)]{bernstein2009matrix}, it follows from the above that $Q_B(\lambda)$ is positive definite uniformly in $ \lambda \in [\epsilon, 1-\epsilon], \ \epsilon \in (0,1/2] $.

Lastly, observe that for all $\uptau \in \mathbb{S}^3$, $\uptau'Q(\lambda)\uptau = \uptau Q_A(\lambda)^{-1} Q_B(\lambda) Q_A(\lambda)^{-1}\uptau = \widetilde{\uptau}'Q_B(\lambda) \widetilde{\uptau} = || Q_B^{1/2}(\lambda) \widetilde{\uptau}|| > 0  $ since $ \widetilde{\uptau}:= Q_A(\lambda)^{-1}\uptau \neq 0 $ by the positive definiteness of $Q_A(\lambda)$ and $Q_B(\lambda)$ uniformly in $\lambda \in [\epsilon, \, 1-\epsilon]$ from the foregoing. 

This completes the proof for this part. \Cref{Fig:min_eig} illustrates the result using plots of the minimum eigenvalues of $Q_A(\lambda)$, $Q_B(\lambda)$ and $Q(\lambda)$, respectively, as a function of $\lambda \in [\epsilon_{\lambda}, \ 1-\epsilon_{\lambda}], \ \epsilon_{\lambda} = 10^{-6}$.

\begin{figure}[!htbp]
\centering 
\caption{Minimum Eigenvalues of $Q_A(\lambda)$, $Q_B(\lambda)$, and $Q(\lambda)$.}
\begin{subfigure}{0.32\textwidth}
\centering
\includegraphics[width=1\textwidth]{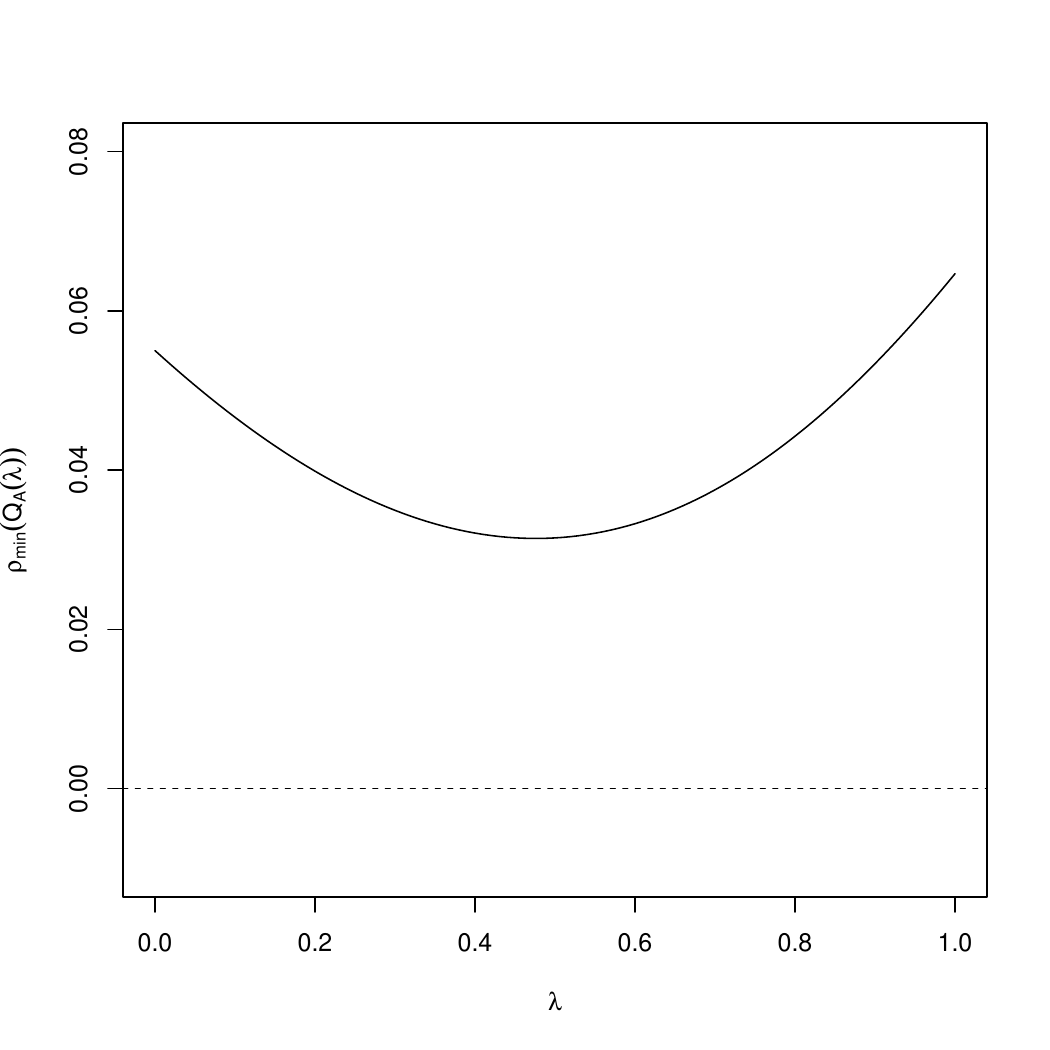}
\caption{$Q_A(\lambda)$}
\end{subfigure}
\begin{subfigure}{0.32\textwidth}
\centering
\includegraphics[width=1\textwidth]{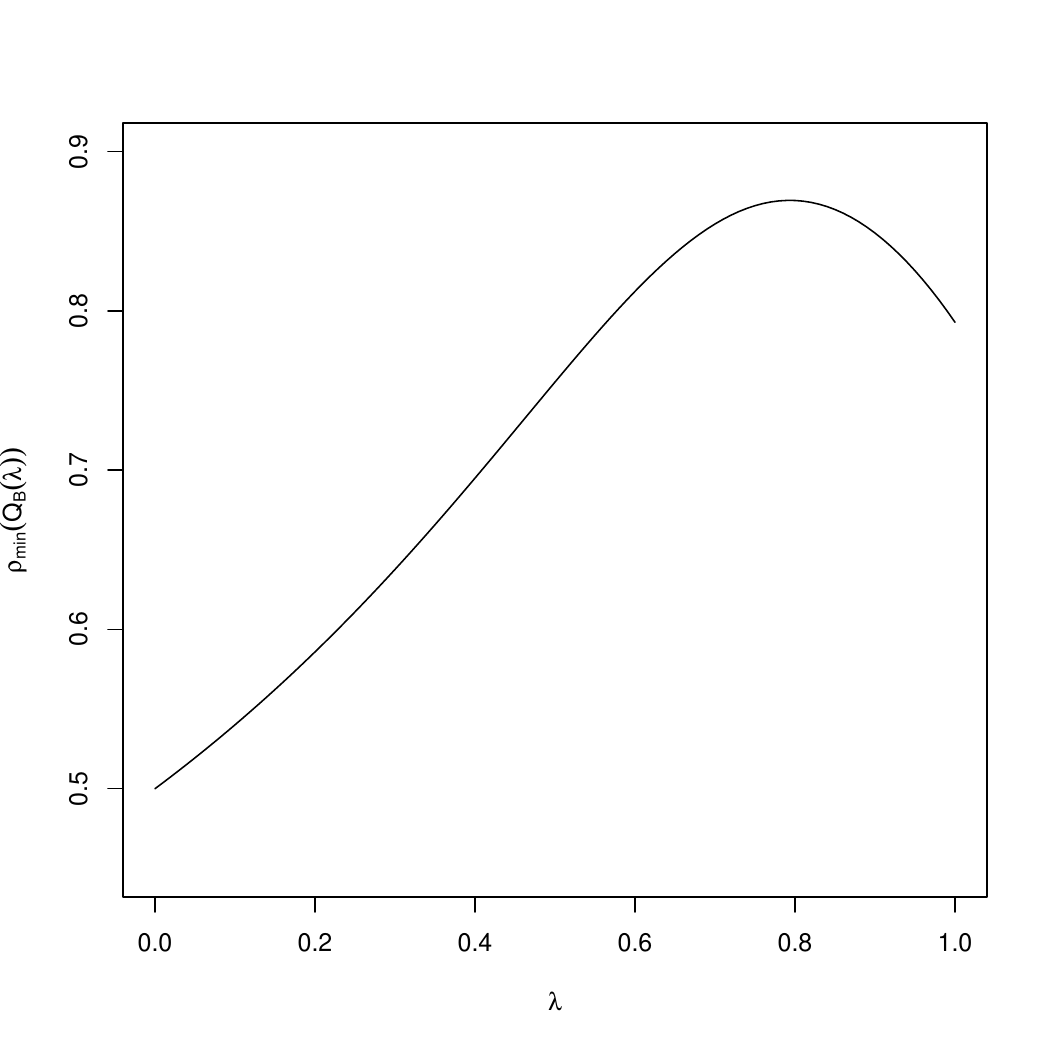}
\caption{$Q_B(\lambda)$}
\end{subfigure}
\begin{subfigure}{0.32\textwidth}
\centering
\includegraphics[width=1\textwidth]{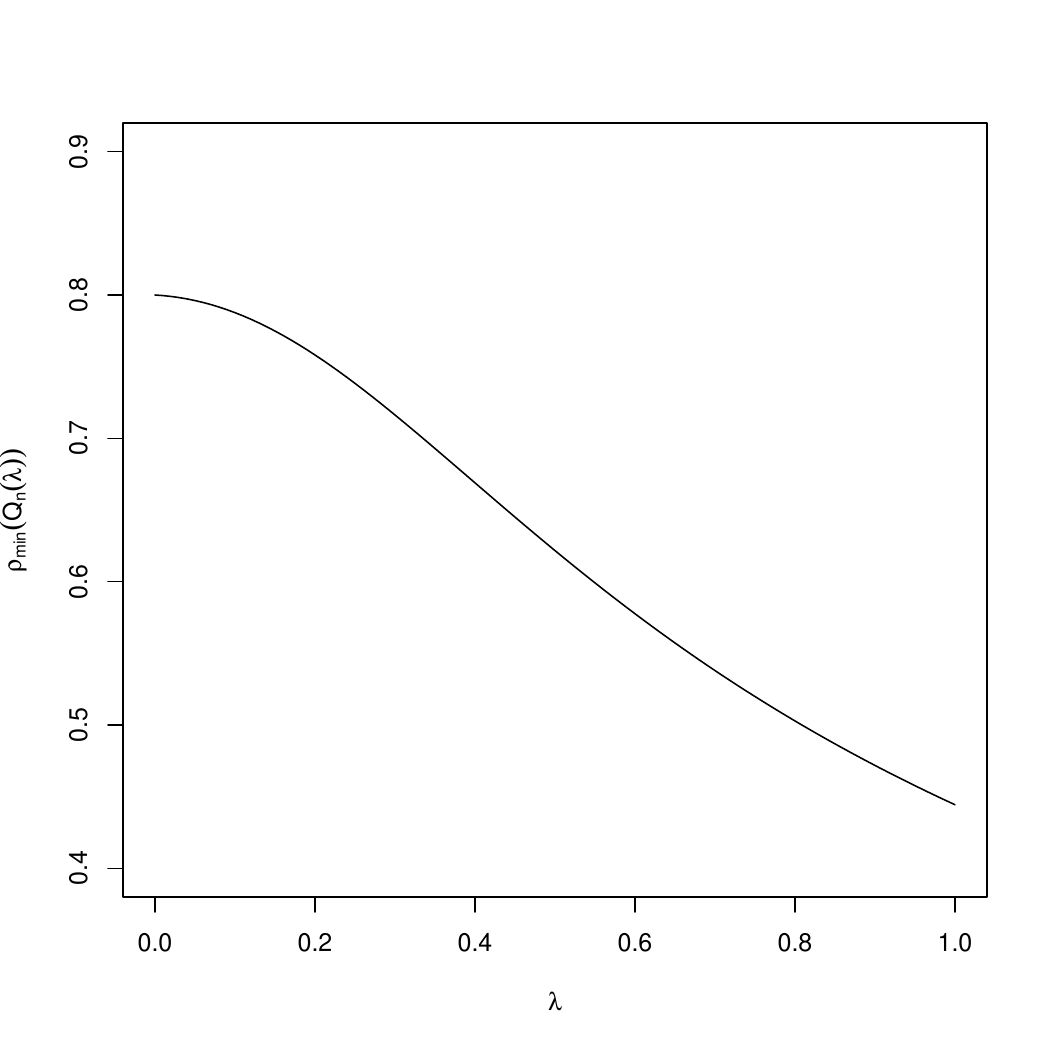}
\caption{$Q(\lambda)$}
\end{subfigure}
\label{Fig:min_eig}

{\footnotesize
\textit{Notes:} The continuous lines plot the minimum eigenvalue of the respective $3\times 3$ matrices as a function of $\lambda \in [\epsilon_{\lambda}, \ 1-\epsilon_{\lambda}], \ \epsilon_{\lambda} = 10^{-6} $. The horizontal dashed line in the first plot corresponds to zero.
}
\end{figure}
\end{proof}

\section{Useful Propositions}\label{App_Sect:Useful_Prop}

\subsection{T-DiD as a least-squares estimator}
\begin{proposition}\label{Prop:ATT_equiv_wreg}
The vector of minimisers of the least squares criterion
\begin{align*}
        S_{\widetilde{w},n}(\beta_0,B): = \sum_{t=-\T}^T \widetilde{w}_n(t)(X_t - \beta_0 - B  \indicator{t\geq 1})^2
    \end{align*} is given by $\displaystyle (\widehat{\beta}_0,\widehat{B})' = \argmin_{(\beta_0,B)'} S_{\widetilde{w},n}(\beta_0,B)$ where $\widehat{B} = \widehat{ATT}_{\omega,T}$.
\end{proposition}
\begin{proof}
    The first order conditions with respect to $\beta_0$ and $B$ are given by 
    \begin{align}
     \label{eqn:FOC1}   \frac{\partial S_{\widetilde{w},n}(\widehat{\beta}_0,\widehat{B})}{\partial \widehat{\beta}_0}:&= -2\sum_{t=-\T}^T \widetilde{w}_n(t)(X_t - \widehat{\beta}_0 - \widehat{B}  \indicator{t\geq 1}) = 0; \text{ and}\\
     \label{eqn:FOC2}   \frac{\partial S_{\widetilde{w},n}(\widehat{\beta}_0,\widehat{B}))}{ \partial \widehat{B}}:&= -2\sum_{t=-\T}^T \widetilde{w}_n(t)\indicator{t\geq 1}(X_t - \widehat{\beta}_0 - \widehat{B}\indicator{t\geq 1}) = 0.
    \end{align} Since $\sum_{t=-\T}^T \widetilde{w}_n(t)\indicator{t\leq -1} =  \sum_{t=-\T}^T \widetilde{w}_n(t)\indicator{t\geq 1}  = 1 $ and $\sum_{t=-\T}^T \widetilde{w}_n(t)=2$ by definition, it follows from \eqref{eqn:FOC1} that 
    \begin{align*}
        \widehat{\beta}_0 = \frac{1}{2} \Big(\sum_{t=-\T}^T \widetilde{w}_n(t) X_t - \widehat{B}\Big)
    \end{align*} which, substituted into \eqref{eqn:FOC2}, gives
    \begin{align*}
        \widehat{B} &= 2\sum_{t=-\T}^T \widetilde{w}_n(t)\indicator{t\geq 1}X_t - \sum_{t=-\T}^T \widetilde{w}_n(t)X_t \\
        &= 2\sum_{t=-\T}^T \widetilde{w}_n(t)\indicator{t\geq 1}X_t - \sum_{t=-\T}^T \widetilde{w}_n(t)\big( \indicator{t\leq -1} + \indicator{t\geq 1} \big)X_t\\
        &= \sum_{t=-\T}^T \widetilde{w}_n(t)\indicator{t\geq 1}X_t - \sum_{t=-\T}^T \widetilde{w}_n(t)\indicator{t\leq -1}X_t\\
        &= \sum_{t=1}^T w_T(t)X_t - \sum_{-\tau= 1}^{\T} \psi(-\tau)X_t=: \widehat{ATT}_{\omega,T}
    \end{align*}by \eqref{eqn:DiD_estimator}.
\end{proof}

\subsection{Deriving the T-DiD from SC-type arguments}\label{App_Sub_Sect:TDiD_SC}
Let $\widehat{Y}_{1,t}(0), \, t \in [-\T] \cup [T] $ denote the synthetic control. Then for the demeaned convex-weighted SC with a single control unit, the SC weight is trivially one: 
\begin{align*}
    &\Big(\widehat{Y}_{1,\tau}(0) - \frac{1}{\T} \sum_{-\tau'=1}^{\T} Y_{1,\tau'}\Big) = \underbrace{1}_{\widehat{b}_{\T}} \times 
 \Big(Y_{0,\tau} - \frac{1}{\T} \sum_{-\tau'=1}^{\T} Y_{0,\tau'}\Big) \text{ whence}\\
    & \widehat{Y}_{1,\tau}(0) = \underbrace{\frac{1}{\T} \sum_{-\tau'=1}^{\T} \big(Y_{1,\tau'} - Y_{0,\tau'}\big)}_{\widehat{a}_{\T}} + \underbrace{1}_{\widehat{b}_{\T}} \times Y_{0,\tau} =: \widehat{a}_{\T} + \widehat{b}_{\T}Y_{0,\tau}
\end{align*}for each $\tau \in [-\T] $.

Using the above to recover untreated potential outcomes of the treated unit in the post-treatment period, one has for each $t\in [T]$, $\widehat{Y}_{1,t}(0) = \widehat{a}_{\T} + \widehat{b}_{\T}Y_{0,t}$. Thus, an estimator of $ATT_{\omega,T}$ can be expressed as
\begin{equation*}
    \widehat{ATT}_{\omega,T}^S:= \sum_{t=1}^T w_T(t)\big(Y_{1,t}(1) - \widehat{Y}_{1,t}(0)\big) = \sum_{t=1}^T w_T(t)\big(Y_{1,t} - Y_{0,t}\big) - \frac{1}{\T} \sum_{-\tau'=1}^{\T} \big(Y_{1,\tau'} - Y_{0,\tau'}\big).
\end{equation*}From the expression of the T-DiD in \eqref{eqn:DiD_estimator}, $\widehat{ATT}_{\omega,T}^S$ is a special case of $\widehat{ATT}_{\omega,T}$ with uniform weighting in pre-treatment periods, i.e., $\psi_{\T }(-\tau) = 1/\T $.

\section{Discussions}\label{App_Sect:Discussions}

\subsection{Further discussion of the literature}

The canonical DiD derives identifying variation from pooled cross-sectional units. It is one of the most commonly used econometric methods for estimating the causal effects of policy interventions in various social sciences, including economics and political science. For example, \citet{Giavazzi2005}, \citet{Persson2005}, \citet{Persson2006}, \citet{Papaioannou2008} and \citet{Persson2009} use more or less the canonical DiD (C-DiD hereafter) to estimate the effect of democracy on economic growth. However, the C-DiD and its refinements, e.g., \citet{ferman-pinto-2019-inference,callaway-santanna-2021,chan2021pcdid,Bellego-Benatia-Dortet-2024-chained,galindo-Some-Tchuente2018fuzzy,picchetti-Pinto-Shinoki-2024difference} are unsuited for fixed-$N$ large-$T$ settings (such as the one under consideration in this paper) where there may be temporal dependence, unit roots, deterministic trends, or cross-sectional dependence between units arising from, e.g., business cycles, persistent series, and correlated shocks. 

Pooled regression analyses often rule out such cross-sectional dependence, which is common in fixed-$N$ large-$T$ settings. For example, it is plausible that common waves of democratisation or de-democratisation affect multiple countries, thereby inducing cross-sectional dependence. Consider, for instance, the wave of democratisation in West Africa in the early 1990s and the Arab Spring uprisings of the early 2010s. Estimating ATT parameters using two-way fixed effects models faces challenges similar to those in pooled regression analyses, where significant heterogeneity in treatment effects is often masked. For these reasons, the literature largely turns to Synthetic Control (SC) methods in such settings -- see e.g., \citet{abadie-gardeazabal-2003,abadie-diamond-hainmueller-2010,xu-2017,abadie-2021,ferman-pinto-2021-synthetic,Sun-Ben-Michael-Feller-2023-using}. SC formalises case studies, and it is adapted to exploiting the temporal variation in the data to estimate treatment effects. 

There are, however, documented impediments to using the SC from a practical perspective. The SC requires a donor pool of multiple control units. For example, bias from the pre-treatment fit when the number of control units is fixed may not shrink to zero as the number of pre-treatment periods increases \citep{ferman-pinto-2021-synthetic}. Additionally, the survey in Appendix A of \citet{lei-sudijono-2024-inference} indicates that unit-level placebo inference, which several SC methods use for inference and robustness checks, may be impractical at small nominal levels --- see also \citet{fry-2024-method}. In fact, the power of such tests is increasing in the number of selected control units, whereas the smallest attainable $p$-value is decreasing in the number of selected controls. Per-period treatment effects on the treated unit, which are typical target SC estimands, are not consistently estimable under fixed-$N$ settings. Consequently, the resulting confidence or predictive intervals can be wide and uninformative \citep{chernozhukov-Wuthrich-Zhu-2024t}. Further, a good pre-treatment fit of the treated unit may not carry over to control units in the donor pool thus complicating the use of permutation methods for SC-based inference \citep{abadie-2021}. The setting of \citet{Li-2020-statistical}, namely several pre- and post-treatment periods with a fixed number of units is close to the one under consideration in this paper. There are important differences to highlight, however. (1) \citet{Li-2020-statistical} uses the SC framework and inference therein is sub-sampling-based and non-standard, whereas the T-DiD is asymptotically normal. (2) The \citet{Li-2020-statistical} model requires multiple (albeit a finite number of) controls for the pre-treatment fitting procedure to be meaningful while the baseline (building block) of the proposed T-DiD requires only one control unit. (3) \citet[Assumption 1]{Li-2020-statistical} imposes stationarity, while it is not required in the current paper.

There are drawbacks to using the SC from an econometric viewpoint. (1) The SC requires a perfect pre-treatment fit, which is often hard to achieve or verify, especially when (potential) outcomes are non-stationary. Inadequate pre-treatment fit, which is more common with longer time series, can introduce bias \citep{ferman-pinto-2021-synthetic,ben-feller-rothstein-2021}. When the pre-treatment fit is good, it may be attributable to non-stationary common shocks or time trends -- see, e.g., \citet[p. 1216]{ferman-pinto-2021-synthetic}. (2) SC methods draw on cross-sectional dependence typically modelled via linear factor models, e.g., \citet{ferman-pinto-2021-synthetic,fry-2024-method,Sun-Ben-Michael-Feller-2023-using}. Without strong cross-sectional dependence, the SC can suffer identification challenges and pre-treatment fit bias. (3) When there is only one suitable control unit, as in the current running empirical example, the convex-weighted SC \emph{trivially} reduces to the difference between post-treatment treated and control outcomes. The post-treatment control outcome serves as the post-treatment untreated potential outcome for the treated unit. The convex-weighted demeaned SC is algebraically shown in \Cref{App_Sub_Sect:TDiD_SC} to be a form of the proposed T-DiD. SC permutation-based inference is, however, infeasible in this setting as there is only one control unit. This implies a stronger identification condition relative to the proposed T-DiD -- the (weighted) average of untreated potential outcomes of both the treated and untreated units must be equal. This rules out \emph{characteristically} different means of the untreated potential outcomes of both units. Although the T-DiD estimator requires at least one valid control unit while the SC only needs a donor pool of peers, it is worth emphasising that control validity in the DiD framework can be weaker than that of the SC. For example, the outcome of a valid DiD control unit does not need to be in \emph{any way} dependent on that of the treated unit in pre-treatment periods since a pre-post asymptotic parallel trends condition suffices. Thus, control units that may be individually SC-irrelevant can be DiD-valid.

The preceding paragraphs highlight two important points. First, an empirical setting with as few as a single treated unit and a single control unit easily arises and becomes empirically relevant once a researcher is interested in unit-specific effects \emph{in lieu of} treatment effects from pooled regression analyses where heterogeneity can be substantially masked. Second, existing methods viz. the C-DiD, SC, and factor models are inadequate in this specific setting. Hence, this paper proposes the T-DiD, which exploits temporal variation in the data for consistent estimation and reliable inference in the presence \emph{or absence} of cross-sectional dependence. Drawing on recent advances in the DiD literature on treatment effect heterogeneity, e.g., \citet{callaway-santanna-2021,chaisemartin-dhaultfoeuille-2020,borusyak-jaravel-spiess-2024}, this paper defines convex-weighted treatment effects on the treated across time and proposes a consistent and asymptotically normal DiD estimator for it.\footnote{Period-specific individual treatment effects on the treated are not consistently estimable in fixed-$N$ settings.} Thus, this paper introduces the DiD into a space within the body of literature that has hitherto been the preserve of Synthetic Control (SC) methods, factor models, and comparative studies. This paper provides easily interpretable asymptotic parallel trends and asymptotically limited anticipation assumptions under which convex-weighted treatment effects on the treated, namely $ATT_{\omega,T}:= \sum_{t=1}^T\omega_T(t)ATT(t) $, are \emph{asymptotically} identified with as few as a single treated and a single control unit and several pre- and post-treatment periods. 

This paper makes a methodological contribution to the literature on the estimation of treatment effects in time series settings. \citet{arkhangelsky-etal-2021} introduces the Synthetic Difference-in-Differences that combines attractive features of both the SC and C-DiD. It is unsuited for the two-unit baseline case considered in this paper, as it requires the number of untreated units to go to infinity -- see Assumption 2 therein. \citet{gonccalves-ng-2024imputation} provides improved predictors of, e.g., SC counterfactuals by exploiting serial correlation in the error. \citet{Pesaran-Smith-2016counterfactual,Angrist-Jorda-Kuersteiner-2018semiparametric} construct counterfactual outcomes from pre-treatment covariates of the treated unit under the identification assumption that the latter are invariant to treatment. A special case of the aforementioned (without pre-treatment covariates) is the before-after (BA) difference in means estimator, see e.g., \citet[p. 366]{carvalho-Masini-Ricardo-2018arco} for a discussion. The BA uses the average pre-treatment outcome of the treated unit as its post-treatment untreated potential outcome. The BA estimator, like the \citet{Pesaran-Smith-2016counterfactual,Angrist-Jorda-Kuersteiner-2018semiparametric,botosaru-giacomini-weidner-2023forecasted} estimators, cannot disentangle common shocks from treatment effects, especially if these occur post-treatment.

A related strand of literature concerns the difference-in-discontinuities (diff-in-disc) estimator \citep[e.g.,][]{galindo-Some-Tchuente2018fuzzy, picchetti-Pinto-Shinoki-2024difference}, which combines the features of regression discontinuity designs and difference-in-differences. This approach is conceptually distinct from the T-DiD, as it identifies local treatment effects defined around a threshold of a running variable (such as time). In contrast, the T-DiD identifies an estimand that pools across all post-treatment periods. Thus, the two estimators are complementary: the diff-in-disc utilises cross-sectional variation (requiring large $N$) to estimate an ATT within an asymptotically shrinking bandwidth, whereas the T-DiD leverages a fixed number of units with large $T$ to estimate a weighted average of ATTs globally. Consequently, while the identification conditions of the diff-in-disc are local, those of the T-DiD are global.

\subsection{Pre-trends testing}\label{App_Sect:Pre_Trends}
While tests of pre-treatment trends are known to only provide \emph{suggestive evidence}, it is instructive to tease out hypotheses they \emph{actually} test. The idea of a pre-test in the current context can be fashioned as follows. Choose a pre-treatment period $\T _o$ such that $\T_o/\T \in (0,1)$ and define $ATT_{\omega,\T_o}:=\sum_{-\tau=1}^{\T _o}w(\tau)ATT(\tau)$. Extending the use of notation in an obvious way, one has the decomposition $ATT_{\omega,\T_o} = \widetilde{ATT}_{\T}^{w,\psi} - R_{\T }^{w,\psi} $ where as before, $\widetilde{ATT}_{\T}^{w,\psi}$ is identified but $ATT_{\omega,\T_o}$ is not, and $R_{\T }^{w,\psi}$ is the difference between a pre-treatment trend bias $TB_{\omega,\T }$ and anticipation bias $ATT_{\omega,\T  - \T _o}$. Decompose $R_{\T }^{w,\psi}$ further as $R_{\T }^{w,\psi}:= TB_{\omega,\T } - ATT_{\omega,\T  - \T _o} $. Thus,
\[
    \widetilde{ATT}_{\T}^{w,\psi} = TB_{\omega,\T } + (ATT_{\omega,\T_o} - ATT_{\T  - \T _o}).
\]
\noindent In line with the convergence rate-based hypotheses adopted in \Cref{Sect:Tests}, a pre-test adapted to the current setting can be formulated using the representation $|\widetilde{ATT}_{\T}^{w,\psi}| = C_{\omega,\T} \T^{1/2-\check{\delta}} $ where $ \{C_{\omega,\T}: \T \geq 1 \}$ is a sequence of bounded positive constants.
\begin{align*}
    \Hyp_o^{pt}: \check{\delta} > 1/2, \quad \Hyp_{an}^{pt}: \check{\delta} = 1/2, \quad \Hyp_a^{pt}: \check{\delta} < 1/2.
\end{align*}

\noindent The pre-test can be implemented using the $t$-statistic
\begin{equation*}
    \widehat{t}_{\omega,\T} = \frac{\widehat{ATT}_{\omega,\T_o}}{\widehat{s}_{\omega,\T}} \text{ where } \widehat{ATT}_{\omega,\T_o} = \sum_{-\tau =1 }^{\T} \omega_{\T}(\tau)X_\tau,
\end{equation*} $\omega_{\T}(\tau):= w_{\T_o}(\tau)\indicator{\tau > \T_o} - \psi_{\T - \T_o}(\tau)\indicator{\tau < \T_o}  $, and $\widehat{s}_{\omega,\T}$ is consistent for $s_{\omega,\T}:= \E[(\widehat{ATT}_{\omega,\T_o}-\E[\widehat{ATT}_{\omega,\T_o}])^2] $ in the sense $\displaystyle \plim_{\T\rightarrow \infty} \widehat{s}_{\omega,\T}/s_{\omega,\T} = 1 $.

$ TB_{\omega,\T } = o(\T^{1/2}) $ neither implies nor is implied by \Cref{ass:parallel_trends}. Similarly, the rate condition on anticipation bias-related terms -- $(ATT_{\omega,\T_o} - ATT_{\omega,\T  - \T _o}) = o(\T^{1/2})$ -- neither implies nor is implied by \Cref{ass:limited_anticip}. Thus, parallel trends in pre-treatment periods offer no guarantee of parallel trends after $t=0$ in the absence of treatment \citep{kahn2020promise,dette-2024-testing}. The foregoing extends the arguments of \citet{kahn2020promise} on pre-tests in the C-DiD setting to the current setting with a large time dimension, temporal dependence, cross-sectional dependence, and anticipation bias. 

Besides possibly distorting inference, pre-tests are documented in the literature to have poor power performance \citep{freyaldenhoven-hansen-shapiro-2019,roth-2022-pretest}. \citet{dette-2024-testing} innovatively reverses the burden of proof by testing for pre-treatment parallel trends as the alternative hypothesis. Thus, the test proposed therein under the alternative provides power in favour of pre-treatment parallel trends. While this approach addresses the power deficiency of commonly used tests of pre-treatment parallel trends, it does not cure the inherent problem underscored in the preceding paragraph of pre-tests (in general) being uninformative of the validity (or absence thereof) of parallel trends, especially when violations occur in the post-treatment period.

Pre-tests can lead to biased inference if used as a basis for estimating treatment effects on the treated \citep{roth-2022-pretest}. In a similar vein, the foregoing discussion points out why inference on $ATT_{\omega,T}$ based on passing a pre-test risks being misleading because ``passing" a pre-test neither implies nor is implied by the identification of $ATT_{\omega,T}$. In contrast, the over-identifying restrictions test in \Cref{Theorem:Test_Implication} is valuable as it tests identification save on a space of distributions where the test has no power under the alternative hypotheses. For example, the proposed test can detect identification failure stemming from violations of \Cref{ass:parallel_trends} or \Cref{ass:limited_anticip} that might ``pass" a pre-test. Moreover, pre-tests are only based on pre-treatment data, whereas the proposed test is based on the entire data set whence its ability to detect violations of \Cref{ass:parallel_trends} in post-treatment periods and have more power as evidenced in simulations -- see \Cref{App_Sub:Id_test}.

\subsection{Further discussion on the empirical analysis}\label{App_Sect:Emp_Supplement}

The empirical application of interest in this paper examines the impact of democracy on economic growth. Given the current global decline in both democracy and satisfaction with democratic systems \citep{Coppedge2022, Foa2020}, it is apt to re-examine whether democracy drives economic growth and well-being. Several efforts have been devoted to quantifying the effect of democracy on economic growth -- see \citet{Doucouliagos2008}, \citet{Colagrossi2020} and \citet{Knutsen2021} for reviews. However, empirical findings on the relationship between democracy and economic growth are conflicting. Some studies, such as \citet{Acemoglu2019}, report a positive relationship, while others, like \citet{Gerring2005}, find a negative effect, and still others, such as \citet{Murtin2014}, observe no significant effect. As \citet{Knutsen2021} points out, these conflicting findings stem not only from differences in modelling choices and data quality but also from the varying economic performance of regimes with similar levels of democracy, particularly at the autocratic end of the spectrum. For example, \citet{Chen-Stengos-2024-threshold} finds heterogeneity in effects by regime type, institutional quality, and education levels -- see also \citet{Acemoglu2019}. Thus, one's answer to the democracy-growth question is largely an artefact of the composition of units in the pooled regression analyses. 

\pgfplotsset{my personal style/.style=
{font=\footnotesize},width=10cm,height=7.5cm} 

\begin{figure}[h!]
\begin{center}
\caption{Democracy Indices}
\label{fig1}

\begin{tikzpicture}[]
\begin{axis}[my personal style,minor x tick num=1,
xlabel=Period,
xticklabels={{},{},1960,1970,1980,1990,2000,2010,2020},
ymin=0,
ymax=0.7,
ylabel=, legend style={
at={(0.4,0.9)},
anchor=north east}]

\addplot[mark=*,color=black,line width=1pt,mark size=1pt] table[col sep=tab,x=Year,y=Benin] {figures/DemocracyIndexAll.txt};

\addplot[mark=*,color=blue,line width=1pt,solid,mark size=1pt] table[col sep=tab,x=Year,y=Togo] {figures/DemocracyIndexAll.txt};

\addplot[mark=*,color=gray,line width=1pt,solid,mark size=1pt] table[col sep=tab,x=Year,y=Niger] {figures/DemocracyIndexAll.txt};

\addplot[mark=*,color=red,line width=1pt,solid,mark size=1pt] table[col sep=tab,x=Year,y=Nigeria] {figures/DemocracyIndexAll.txt};

\addplot[mark=*,color=green,line width=1pt,solid,mark size=1pt] table[col sep=tab,x=Year,y=Burkina] {figures/DemocracyIndexAll.txt};

\addplot[draw=gray] coordinates {(1960,0.5) (2016,0.5) };

\legend{Benin, Togo, Niger, Nigeria, Burkina Faso, Threshold (0.5)}
\end{axis}
\end{tikzpicture}
\end{center}
\begin{justify}
{\footnotesize
\textit{Notes:} The plots above depict democracy indices for West African countries neighbouring Benin from 1960 through 2018.
}
\end{justify}
\end{figure}

For the empirical application of the T-DiD, this paper focuses on Benin, a country that experienced strong democratisation in the last three decades. The ideal pool of candidate controls includes neighbouring countries such as Togo, Burkina Faso, Niger, and Nigeria. However, Burkina Faso and Niger have experienced democratic regimes over some periods as shown in \Cref{fig1} using the democracy index based on the V-Dem project; see \cref{Sect:Data} for more details. Nigeria has economic characteristics that are dissimilar to those of
Benin. For example, Nigeria has an independent monetary policy while Benin does not. In sum, Togo remains the only suitable control unit among the aforementioned candidate controls. 

Benin is a former French colony, and Togo is a former French protectorate. Both countries attained independence in 1960. They share a border of almost $651$ km and are culturally similar. Both countries are members of the same monetary union -- the West African Economic and Monetary Union (UEMOA) -- even after independence, and they have the same monetary policy. Politically -- see, e.g., \citet{Kohnert2021} -- the two countries were autocracies until 1990 when Benin initiated a process of democratisation. The former French president François Mitterand, in the wake of the fall of the Berlin Wall in 1989 and the implosion of the Soviet Union, encouraged francophone African countries to adopt democracy.\footnote{The statement was made in his speech at Baule during the 16th Conference of Heads of State of Africa and France -- see \url{https://www.vie-publique.fr/discours/127621-allocution-de-m-francois-mitterrand-president-de-la-republique-sur-la}.} Benin and Togo have both undergone democratisation processes, but their outcomes differ: Benin becomes a model democracy for the whole of Sub-Saharan Africa, while Togo remains largely perceived as an autocratic country. The choice of Togo as a control unit for Benin weakens the asymptotic parallel trends and limited anticipation identification conditions. Both identification conditions only need to hold conditional not only on the aforementioned observed characteristics but also on unobservable common characteristics.

\section{Simulations}\label{Sect:Sim}
This section presents simulation results. Section \ref{App_Sub:Bias_Size} presents the bias and empirical rejection rates of competing estimators, and Section \ref{App_Sub:Pow_Curve_ATT} presents power curves corresponding to estimator-based $t$-tests. The power curves serve to examine the ability of the $t$-tests to detect significant $ATT_{\omega,T}$. Section \Cref{App_Sub:Id_test} examines the size and power performance of the proposed identification tests, and Section \Cref{App_Sub:Id_test} extends the results to heterogeneous (through time) treatment effects.

\subsection{DGPs}
This section examines the small-sample performance of the T-DiD compared to alternatives such as the SC and the BA under several interesting data-generating processes (DGPs). Consider the following DGP of untreated potential outcomes
\begin{align*}
    Y_{d,t}(0) = \alpha_0 + d(\alpha_1-\alpha_0) + \alpha_2Y_{d,t-1}(0) + \varphi(t) + d\nu_t + \big(1-d(1-1/\sqrt{2})\big)\big(e_{0,t} + \alpha_3 e_{0,t-1}),
\end{align*}
\noindent with the observed outcome generated as 
\begin{align}\label{eqn:DGP_Sim2}
    Y_{d,t} = Y_{d,t}(0) + d \Big(ATT(t) + \alpha_2\big(Y_{d,t-1} - Y_{d,t-1}(0)\big) + \alpha_4t + (1/\sqrt{2})\big(e_{1,t} + \alpha_3 e_{1,t-1}\big) \Big)
\end{align}where $Y_{d,t}(0)$ denotes the untreated potential outcome of unit $d\in\{0,1\}$. $ATT(t) = 0.0$ for all periods for all DGPs. Variations of $(\alpha_0, \alpha_1, \alpha_2, \alpha_3, \alpha_4)'$, $ \varphi(t)$, $\nu_t$, and $ e_{d,t}$ generate different DGPs described below. $\alpha_1\neq \alpha_0$ induces a location shift between the untreated potential outcomes of the treated and untreated units; the SC fails under this scenario. $Y_{d,t}$ and $Y_{d,t}(0)$ are first-order auto-regressive AR(1) and first-order moving average MA(1) processes when $|\alpha_2| \in (0,1) $ and $|\alpha_3| \in (0,1) $, respectively. $Y_{d,t}$ and $Y_{d,t}(0)$ become unit root processes when $\alpha_2=1$. $\alpha_4\neq 0$ leads to uncommon trend non-stationarity in $Y_{1,t}$, unlike in $Y_{0,t}$. The errors of the observed and untreated potential outcomes of both the treated and untreated units at each period $t\in [-\T] \cup [T] $ are correlated, thus there is cross-sectional dependence across units with correlation coefficient $\mathrm{cor}[e_{0t}, (1/\sqrt{2})(e_{0t}+e_{1t})] = 1/\sqrt{2} \approx 0.707$. $\varphi(\cdot)$ introduces a common trend while the random variable $\nu_t$ introduces violations of the standard parallel trends and no-anticipation assumptions. Quite importantly, the BA fails whenever $\varphi(\cdot)$ in the post-treatment period is not cancelled out (at least asymptotically) by that of the pre-treatment period. For inference, the convex-weighted SC and BA use a straightforward adaptation of the asymptotic normality result in this paper.

The following specifies forms of the common trend $\varphi(t)$ and the error $e_{dt}$:
\begin{align*}
   \varphi(t) = 
   \begin{cases}
    \varphi_{NT}(t):= 0.0 & \text{ no trends}\\
    \varphi_{BST}(t):= \sqrt{2}\indicator{t\geq 4} & \text{ bounded binary trend} \\
    \varphi_{CST}(t):= \cos(t) & \text{ bounded continuous trend} \\
    \varphi_{NSQT}(t):= t + t^2/500 & \text{ unbounded quadratic trend}
    \end{cases}
\end{align*}and
\begin{align*}
    e_{dt} = 
    \begin{cases}
        \varepsilon_{dt}\varepsilon_{dt-1} & \text{ MDS}\\
        \sigma_{dt}\varepsilon_{dt} , \ \sigma_{dt}^2 = 0.4 + 0.3e_{dt-1}^2 + 0.3\sigma_{dt-1}^2 & \text{ GARCH(1,1)}\\
        \big(\varepsilon_{dt} + \varepsilon_{dt-1}\varepsilon_{dt-2}\big)/\sqrt{2} & \text{ White Noise}
    \end{cases}
\end{align*}
\noindent where $\varepsilon_{0t} \overset{iid}{\sim} (\chi^2-1)/\sqrt{2}$ and $  \varepsilon_{1t} \overset{iid}{\sim} \mathbb{T}_T/\sqrt{T/(T-2)}$ where $\mathbb{T}_T$ denotes the Student $t$-distribution with $T$ degrees of freedom. Simulation results are based on $10,000 $ Monte Carlo samples. Eleven DGPs are considered in total:
\begin{enumerate}
\item[1.] DGP SC-BA: $ \alpha_1 = \alpha_0 = 0.5 $, $\alpha_2 = \alpha_3 = 0.0$, $ \varphi(t) = 0 $, $\nu_t = 0$, $e_{dt}$ is MDS; 
			
\item[2.] DGP BA: $ \alpha_1 = -\alpha_0 = 0.5 $, $\alpha_2 = \alpha_3 = 0.0$, $ \varphi(t) = 0 $, $\nu_t = 0$,  $e_{dt}$ is MDS; 
			
\item[3.] DGP SC: $ \alpha_1 = \alpha_0 = 0.5 $, $\alpha_2 = \alpha_3 = 0.0$, $ \varphi(t) = \varphi_{BST}(t) $, $\nu_t = 0$,  $e_{dt}$ is MDS; 
			
\item[4.] (A) DGP PT-NA: $ \alpha_1 = \alpha_0 = 0.5 $, $\alpha_2 = \alpha_3 = 0.0$, $ \varphi(t) = 0.0 $, $\nu_t \sim \mathcal{N}\big(0.5\indicator{t\neq 0}\mathrm{sgn}(t)|t|^{-0.9},1\big) $,  $e_{dt}$ is MDS; 

\item[4.] (B) DGP PT-NA: $ \alpha_1 = \alpha_0 = 0.5 $, $\alpha_2 = \alpha_3 = 0.0$, $ \varphi(t) = 0.0 $, $\nu_t \sim \mathcal{N}\big(0.5\indicator{t\neq 0}|t|^{-0.25},1\big) $,  $e_{dt}$ is MDS;   

\item[5.] DGP GARCH(1,1): $ \alpha_1 = \alpha_0 = 0.5 $, $\alpha_2 = \alpha_3 = 0.0$, $ \varphi(t) = 0.0 $, $\nu_t = 0.0 $, $e_{dt}$ is GARCH(1,1); 

\item[6.] DGP MA(1): $ \alpha_1 = \alpha_0 = 0.5 $, $\alpha_2 = 0.0$, $\alpha_3 = 0.25$, $ \varphi(t) = 0.0 $, $\nu_t = 0.0 $,  $e_{dt}$ is White Noise; 

\item[7.] DGP AR(1): $ \alpha_0 = \alpha_1 = \alpha_2 = 0.5$, $\alpha_3 = 0.0$, $ \varphi(t) = \cos(t) $, $\nu_t = 0.0 $,  $e_{dt}$ is White Noise; 

\item[8.] DGP U-R: $ \alpha_0 = \alpha_1 = 0.5$, $\alpha_2 = 1.0$, $\alpha_3 = 0.0$, $ \varphi(t) = 0.0 $, $\nu_t = 0.0 $, $e_{dt}$ is White Noise;

\item[9.] DGP Q-T: $ \alpha_0 = \alpha_1 = 0.5$, $\alpha_2 = \alpha_3 = 0.0$, $ \varphi(t) = \varphi_{NSQT}(t) $, $\nu_t = 0.0 $,  $e_{dt}$ is White Noise;  and

\item[10.] DGP T-T: $ \alpha_1 = \alpha_0 = 0.5 $, $\alpha_2 = \alpha_3 = 0.0$, $ \varphi(t) = 0.0 $, $\alpha_4=1.0$, $\nu_t = 0$,  $e_{dt}$ is MDS.
\end{enumerate}

$\alpha_0\neq \alpha_1$ in DGP BA, thus the SC estimator is not expected to perform well. Likewise, $\varphi(t)=\varphi_{BST}(t)$ induces a temporal location shift in DGP SC, thus it is not expected to be favourable to the BA estimator. $\nu_t$ in DGP (A) PT-NA induces violations of the standard parallel trends and no anticipation assumptions although it satisfies \Cref{ass:parallel_trends,ass:limited_anticip}. Unlike in DGP PT-NA (A), weaker rates on the violations of \Cref{ass:parallel_trends,ass:limited_anticip} are allowable in DGP PT-NA (B) as both biases tend to cancel out asymptotically between the pre-treatment and post-treatment periods. In view of \Cref{rem:weak_suff_idcond}, \Cref{ass:parallel_trends,ass:limited_anticip} are both violated in DGP PT-NA (B) but identification holds because the biases tend to cancel out asymptotically. A common quadratic term via $\varphi(t) = t + t^2/500$ induces trend non-stationarity in DGP Q-T; this DGP is not favourable to the BA estimator. The trend in DGP T-T is linear and affects only the treated unit; the trend leads to identification failure if not removed.

\subsection{Bias and size control}\label{App_Sub:Bias_Size}

\begin{table}[!htbp]
\caption{Simulation - $ATT(t) = 0.0,\ t\geq 1 $ }
\centering
\begin{tabular}{rlccccccccc}
\cmidrule[0.5pt](l){2-11}
 & $T,\T $ & \multicolumn{1}{c}{DiD} & \multicolumn{1}{c}{SC} & \multicolumn{1}{c}{BA} & \multicolumn{1}{c}{DiD} & \multicolumn{1}{c}{SC} & \multicolumn{1}{c}{BA} & \multicolumn{1}{c}{DiD} & \multicolumn{1}{c}{SC} & \multicolumn{1}{c}{BA} \\ \cmidrule[0.2pt](l){2-11}
\multirow{12}{*}{ \STAB{\rotatebox[origin=c]{90}{\underline{DGP SC-BA}}}} 
&  & \multicolumn{3}{c}{MB} & \multicolumn{3}{c}{MAD} & \multicolumn{3}{c}{RMSE} \\ \cmidrule[0.2pt](l){3-5} \cmidrule[0.2pt](l){6-8} \cmidrule[0.2pt](l){9-11}
&25    &0.002 &0.002 &0.001 &0.14  &0.096 &0.175 &0.216 &0.151 &0.280  \\ 
       &50     &-0.001 &0.000      &-0.001 &0.101  &0.07   &0.129  &0.153  &0.107  &0.202  \\ 
       &100    &0.000      &0.000      &-0.001 &0.072  &0.051  &0.094  &0.108  &0.077  &0.142  \\ 
      &200   &0.001 &0.001 &0.002 &0.051 &0.036 &0.066 &0.076 &0.053 &0.100   \\ 
      &400   &0.000     &0.000     &0.000     &0.037 &0.025 &0.047 &0.054 &0.038 &0.070  \\ 
\cmidrule[0.2pt](l){3-5} \cmidrule[0.2pt](l){6-8} \cmidrule[0.2pt](l){9-11}

& & \multicolumn{3}{c}{Rej. 1\%} & \multicolumn{3}{c}{Rej. 5\%} & \multicolumn{3}{c}{Rej. 10\%} \\ \cmidrule[0.2pt](l){3-5} \cmidrule[0.2pt](l){6-8} \cmidrule[0.2pt](l){9-11}

&25    &0.012 &0.012 &0.013 &0.057 &0.056 &0.056 &0.116 &0.110  &0.114 \\ 
      &50    &0.012 &0.009 &0.011 &0.052 &0.053 &0.055 &0.105 &0.104 &0.110  \\ 
      &100   &0.008 &0.010  &0.010  &0.050  &0.053 &0.053 &0.101 &0.105 &0.106 \\ 
      &200   &0.008 &0.009 &0.010  &0.049 &0.045 &0.049 &0.100   &0.096 &0.101 \\ 
      &400   &0.011 &0.010  &0.009 &0.049 &0.049 &0.047 &0.099 &0.099 &0.097 \\  \cmidrule[0.5pt](l){2-11}

\multirow{12}{*}{  \STAB{\rotatebox[origin=c]{90}{\underline{DGP BA}}}} 
&  & \multicolumn{3}{c}{MB} & \multicolumn{3}{c}{MAD} & \multicolumn{3}{c}{RMSE} \\ \cmidrule[0.2pt](l){3-5} \cmidrule[0.2pt](l){6-8} \cmidrule[0.2pt](l){9-11}
       &25     &0.002  &-0.998 &0.001  &0.14   &0.997  &0.175  &0.216  &1.009  &0.280   \\ 
       &50     &-0.001 &-1.000     &-0.001 &0.101  &1.000      &0.129  &0.153  &1.006  &0.202  \\ 
       &100    &0.000      &-1.000     &-0.001 &0.072  &1.000      &0.094  &0.108  &1.003  &0.142  \\ 
       &200    &0.001  &-0.999 &0.002  &0.051  &0.999  &0.066  &0.076  &1.001  &0.100    \\ 
      &400   &0.000     &-1.000    &0.000     &0.037 &1.000     &0.047 &0.054 &1.001 &0.070  \\
\cmidrule[0.2pt](l){3-5} \cmidrule[0.2pt](l){6-8} \cmidrule[0.2pt](l){9-11}

& & \multicolumn{3}{c}{Rej. 1\%} & \multicolumn{3}{c}{Rej. 5\%} & \multicolumn{3}{c}{Rej. 10\%} \\ \cmidrule[0.2pt](l){3-5} \cmidrule[0.2pt](l){6-8} \cmidrule[0.2pt](l){9-11}

      &25    &0.012 &0.990 &0.013 &0.057 &0.997 &0.056 &0.116 &0.999 &0.114 \\ 
      &50    &0.012 &1.000     &0.011 &0.052 &1.000     &0.055 &0.105 &1.000     &0.110  \\ 
      &100   &0.008 &1.000     &0.010  &0.050  &1.000     &0.053 &0.101 &1.000     &0.106 \\ 
      &200   &0.008 &1.000     &0.010  &0.049 &1.000     &0.049 &0.100   &1.000     &0.101 \\ 
      &400   &0.011 &1.000     &0.009 &0.049 &1.000     &0.047 &0.099 &1.000     &0.097 \\ 
      \cmidrule[0.5pt](l){2-11}

\multirow{12}{*}{ \STAB{\rotatebox[origin=c]{90}{\underline{DGP SC}}}} 
&  & \multicolumn{3}{c}{MB} & \multicolumn{3}{c}{MAD} & \multicolumn{3}{c}{RMSE} \\ \cmidrule[0.2pt](l){3-5} \cmidrule[0.2pt](l){6-8} \cmidrule[0.2pt](l){9-11}
      &25    &0.002 &0.002 &1.246 &0.140 &0.096 &1.248 &0.216 &0.151 &1.277 \\ 
       &50     &-0.001 &0.000      &1.328  &0.101  &0.070  &1.329  &0.153  &0.107  &1.343  \\ 
      &100   &0.000     &0.000     &1.371 &0.072 &0.051 &1.372 &0.108 &0.077 &1.378 \\ 
      &200   &0.001 &0.001 &1.395 &0.051 &0.036 &1.393 &0.076 &0.053 &1.398 \\ 
      &400   &0.000     &0.000     &1.404 &0.037 &0.025 &1.404 &0.054 &0.038 &1.406 \\
\cmidrule[0.2pt](l){3-5} \cmidrule[0.2pt](l){6-8} \cmidrule[0.2pt](l){9-11}

& & \multicolumn{3}{c}{Rej. 1\%} & \multicolumn{3}{c}{Rej. 5\%} & \multicolumn{3}{c}{Rej. 10\%} \\ \cmidrule[0.2pt](l){3-5} \cmidrule[0.2pt](l){6-8} \cmidrule[0.2pt](l){9-11}

      &25    &0.012 &0.012 &0.919 &0.057 &0.056 &0.965 &0.116 &0.110  &0.979 \\ 
      &50    &0.012 &0.009 &0.992 &0.052 &0.053 &0.996 &0.105 &0.104 &0.998 \\ 
      &100   &0.008 &0.010  &1.000     &0.050  &0.053 &1.000     &0.101 &0.105 &1.000     \\ 
      &200   &0.008 &0.009 &1.000     &0.049 &0.045 &1.000     &0.100   &0.096 &1.000     \\ 
      &400   &0.011 &0.010  &1.000     &0.049 &0.049 &1.000     &0.099 &0.099 &1.000     \\
      \cmidrule[0.5pt](l){2-11}
\end{tabular}\label{Tab:Sim_I}
\end{table}

\begin{table}[!htbp]
\caption{Simulation - $ATT(t) = 0.0,\ t\geq 1 $}
\centering
\begin{tabular}{rlccccccccc}
\cmidrule[0.5pt](l){2-11}
 & $T,\T $ & \multicolumn{1}{c}{DiD} & \multicolumn{1}{c}{SC} & \multicolumn{1}{c}{BA} & \multicolumn{1}{c}{DiD} & \multicolumn{1}{c}{SC} & \multicolumn{1}{c}{BA} & \multicolumn{1}{c}{DiD} & \multicolumn{1}{c}{SC} & \multicolumn{1}{c}{BA} \\ \cmidrule[0.2pt](l){2-11}
\multirow{12}{*}{ \STAB{\rotatebox[origin=c]{90}{\underline{DGP PT-NA (A)}}} } 
&  & \multicolumn{3}{c}{MB} & \multicolumn{3}{c}{MAD} & \multicolumn{3}{c}{RMSE} \\ \cmidrule[0.2pt](l){3-5} \cmidrule[0.2pt](l){6-8} \cmidrule[0.2pt](l){9-11}
      &25    &0.177 &0.089 &0.176 &0.270 &0.179 &0.294 &0.397 &0.265 &0.436 \\ 
      &50    &0.109 &0.055 &0.109 &0.185 &0.127 &0.201 &0.273 &0.187 &0.302 \\ 
      &100   &0.063 &0.030 &0.062 &0.128 &0.087 &0.139 &0.190 &0.129 &0.212 \\ 
      &200   &0.037 &0.018 &0.037 &0.089 &0.061 &0.099 &0.132 &0.090 &0.147 \\ 
      &400   &0.023 &0.011 &0.023 &0.063 &0.043 &0.068 &0.092 &0.064 &0.102 \\ 
\cmidrule[0.2pt](l){3-5} \cmidrule[0.2pt](l){6-8} \cmidrule[0.2pt](l){9-11}

& & \multicolumn{3}{c}{Rej. 1\%} & \multicolumn{3}{c}{Rej. 5\%} & \multicolumn{3}{c}{Rej. 10\%} \\ \cmidrule[0.2pt](l){3-5} \cmidrule[0.2pt](l){6-8} \cmidrule[0.2pt](l){9-11}

      &25    &0.030  &0.029 &0.027 &0.096 &0.086 &0.092 &0.162 &0.142 &0.155 \\ 
      &50    &0.020  &0.021 &0.017 &0.077 &0.074 &0.076 &0.140 &0.130 &0.133 \\ 
      &100   &0.019 &0.018 &0.016 &0.069 &0.059 &0.070  &0.126 &0.111 &0.128 \\ 
      &200   &0.015 &0.012 &0.013 &0.065 &0.056 &0.062 &0.117 &0.110  &0.119 \\ 
      &400   &0.012 &0.012 &0.011 &0.056 &0.054 &0.056 &0.107 &0.105 &0.108 \\  \cmidrule[0.5pt](l){2-11}

\multirow{12}{*}{\STAB{\rotatebox[origin=c]{90}{\underline{DGP PT-NA (B)}}}} 
&  & \multicolumn{3}{c}{MB} & \multicolumn{3}{c}{MAD} & \multicolumn{3}{c}{RMSE} \\ \cmidrule[0.2pt](l){3-5} \cmidrule[0.2pt](l){6-8} \cmidrule[0.2pt](l){9-11}
     &25    &0.001 &0.287 &0.000     &0.241 &0.294 &0.269 &0.356 &0.381 &0.399 \\ 
      &50    &0.001 &0.246 &0.001 &0.168 &0.248 &0.187 &0.250  &0.304 &0.282 \\ 
       &100    &-0.002 &0.205  &-0.002 &0.121  &0.206  &0.134  &0.180   &0.241  &0.202  \\ 
       &200    &-0.001 &0.175  &-0.001 &0.085  &0.175  &0.095  &0.126  &0.196  &0.142  \\ 
      &400   &0.001 &0.148 &0.001 &0.06  &0.148 &0.067 &0.089 &0.161 &0.099 \\
\cmidrule[0.2pt](l){3-5} \cmidrule[0.2pt](l){6-8} \cmidrule[0.2pt](l){9-11}

& & \multicolumn{3}{c}{Rej. 1\%} & \multicolumn{3}{c}{Rej. 5\%} & \multicolumn{3}{c}{Rej. 10\%} \\ \cmidrule[0.2pt](l){3-5} \cmidrule[0.2pt](l){6-8} \cmidrule[0.2pt](l){9-11}

      &25    &0.016 &0.111 &0.016 &0.065 &0.246 &0.061 &0.118 &0.347 &0.119 \\ 
      &50    &0.013 &0.145 &0.011 &0.056 &0.308 &0.052 &0.105 &0.416 &0.107 \\ 
      &100   &0.012 &0.187 &0.010  &0.053 &0.382 &0.057 &0.102 &0.502 &0.108 \\ 
      &200   &0.010  &0.283 &0.011 &0.051 &0.508 &0.049 &0.103 &0.628 &0.102 \\ 
      &400   &0.010  &0.415 &0.009 &0.048 &0.648 &0.047 &0.099 &0.758 &0.101 \\ 
      \cmidrule[0.5pt](l){2-11}

\multirow{12}{*}{\STAB{\rotatebox[origin=c]{90}{\underline{DGP GARCH(1,1)}}}} 
&  & \multicolumn{3}{c}{MB} & \multicolumn{3}{c}{MAD} & \multicolumn{3}{c}{RMSE} \\ \cmidrule[0.2pt](l){3-5} \cmidrule[0.2pt](l){6-8} \cmidrule[0.2pt](l){9-11}
       &25     &0.004  &0.001  &-0.003 &0.141  &0.098  &0.177  &0.214  &0.154  &0.275  \\ 
      &50    &0.001 &0.000     &0.002 &0.102 &0.072 &0.130  &0.152 &0.107 &0.199 \\ 
      &100   &0.000     &0.000     &0.000     &0.071 &0.051 &0.092 &0.108 &0.077 &0.141 \\ 
      &200   &0.000     &0.000     &0.000     &0.052 &0.037 &0.066 &0.077 &0.054 &0.099 \\ 
      &400   &0.000     &0.000     &0.000     &0.036 &0.026 &0.047 &0.054 &0.038 &0.071 \\ 
\cmidrule[0.2pt](l){3-5} \cmidrule[0.2pt](l){6-8} \cmidrule[0.2pt](l){9-11}

& & \multicolumn{3}{c}{Rej. 1\%} & \multicolumn{3}{c}{Rej. 5\%} & \multicolumn{3}{c}{Rej. 10\%} \\ \cmidrule[0.2pt](l){3-5} \cmidrule[0.2pt](l){6-8} \cmidrule[0.2pt](l){9-11}

      &25    &0.017 &0.019 &0.014 &0.064 &0.072 &0.058 &0.120  &0.123 &0.111 \\ 
      &50    &0.012 &0.012 &0.012 &0.055 &0.055 &0.055 &0.108 &0.107 &0.106 \\ 
      &100   &0.011 &0.013 &0.009 &0.056 &0.055 &0.051 &0.106 &0.107 &0.100   \\ 
      &200   &0.011 &0.011 &0.009 &0.053 &0.051 &0.046 &0.099 &0.103 &0.096 \\ 
      &400   &0.009 &0.010  &0.010  &0.049 &0.052 &0.048 &0.098 &0.105 &0.100   \\ 
      \cmidrule[0.5pt](l){2-11}
\end{tabular}\label{Tab:Sim_II}
\end{table}

\begin{table}[!htbp]
\caption{Simulation - $ATT(t) = 0.0,\ t\geq 1 $}
\centering
\begin{tabular}{rlccccccccc}
\cmidrule[0.5pt](l){2-11}
 & $T,\T $ & \multicolumn{1}{c}{DiD} & \multicolumn{1}{c}{SC} & \multicolumn{1}{c}{BA} & \multicolumn{1}{c}{DiD} & \multicolumn{1}{c}{SC} & \multicolumn{1}{c}{BA} & \multicolumn{1}{c}{DiD} & \multicolumn{1}{c}{SC} & \multicolumn{1}{c}{BA} \\ \cmidrule[0.2pt](l){2-11}
 \multirow{12}{*}{ \STAB{\rotatebox[origin=c]{90}{\underline{DGP MA(1)}}}} 
&  & \multicolumn{3}{c}{MB} & \multicolumn{3}{c}{MAD} & \multicolumn{3}{c}{RMSE} \\ \cmidrule[0.2pt](l){3-5} \cmidrule[0.2pt](l){6-8} \cmidrule[0.2pt](l){9-11}
       & 25  &  0.005 &  0.003 & -0.002 &  0.174 &  0.119 &  0.218 &  0.269 &  0.190 &  0.347 \\ 
& 50  & -0.001 &  0.000 & -0.002 &  0.128 &  0.088 &  0.164 &  0.192 &  0.134 &  0.251 \\ 
& 100 & -0.001 &  0.000 & -0.001 &  0.091 &  0.065 &  0.118 &  0.135 &  0.096 &  0.178 \\ 
& 200 &  0.001 &  0.001 &  0.001 &  0.064 &  0.045 &  0.083 &  0.095 &  0.067 &  0.125 \\ 
& 400 &  0.000 &  0.000 &  0.000 &  0.045 &  0.032 &  0.058 &  0.067 &  0.048 &  0.088 \\
\cmidrule[0.2pt](l){3-5} \cmidrule[0.2pt](l){6-8} \cmidrule[0.2pt](l){9-11}

& & \multicolumn{3}{c}{Rej. 1\%} & \multicolumn{3}{c}{Rej. 5\%} & \multicolumn{3}{c}{Rej. 10\%} \\ \cmidrule[0.2pt](l){3-5} \cmidrule[0.2pt](l){6-8} \cmidrule[0.2pt](l){9-11}

      &25    &0.024 &0.022 &0.017 &0.078 &0.074 &0.068 &0.142 &0.131 &0.130  \\ 
      &50    &0.016 &0.014 &0.016 &0.071 &0.063 &0.068 &0.128 &0.121 &0.124 \\ 
      &100   &0.014 &0.013 &0.011 &0.061 &0.061 &0.057 &0.119 &0.116 &0.117 \\ 
      &200   &0.014 &0.012 &0.013 &0.057 &0.058 &0.056 &0.107 &0.109 &0.112 \\ 
      &400   &0.011 &0.012 &0.009 &0.052 &0.050  &0.052 &0.103 &0.104 &0.110  \\
      \cmidrule[0.5pt](l){2-11}
      
\multirow{12}{*}{ \STAB{\rotatebox[origin=c]{90}{\underline{DGP AR(1)}}}} 
&  & \multicolumn{3}{c}{MB} & \multicolumn{3}{c}{MAD} & \multicolumn{3}{c}{RMSE} \\ \cmidrule[0.2pt](l){3-5} \cmidrule[0.2pt](l){6-8} \cmidrule[0.2pt](l){9-11}
             &25     &0.003  &-0.021 &0.003  &0.162  &0.112  &0.208  &0.252  &0.183  &0.324  \\ 
       &50     &-0.001 &-0.012 &0.000      &0.111  &0.077  &0.145  &0.168  &0.119  &0.217  \\ 
       &100    &0.000      &-0.007 &-0.001 &0.076  &0.054  &0.100    &0.113  &0.082  &0.147  \\ 
       &200    &0.001  &-0.003 &0.001  &0.053  &0.037  &0.068  &0.078  &0.056  &0.102  \\ 
       &400    &0.000      &-0.002 &0.000      &0.037  &0.026  &0.047  &0.054  &0.039  &0.071  \\ 
\cmidrule[0.2pt](l){3-5} \cmidrule[0.2pt](l){6-8} \cmidrule[0.2pt](l){9-11}

& & \multicolumn{3}{c}{Rej. 1\%} & \multicolumn{3}{c}{Rej. 5\%} & \multicolumn{3}{c}{Rej. 10\%} \\ \cmidrule[0.2pt](l){3-5} \cmidrule[0.2pt](l){6-8} \cmidrule[0.2pt](l){9-11}

      &25    &0.024 &0.033 &0.021 &0.083 &0.088 &0.082 &0.141 &0.148 &0.143 \\ 
      &50    &0.018 &0.022 &0.02  &0.072 &0.069 &0.071 &0.129 &0.126 &0.133 \\ 
      &100   &0.014 &0.015 &0.014 &0.063 &0.066 &0.061 &0.115 &0.123 &0.117 \\ 
      &200   &0.013 &0.013 &0.015 &0.055 &0.059 &0.062 &0.105 &0.111 &0.114 \\ 
      &400   &0.011 &0.012 &0.012 &0.051 &0.053 &0.059 &0.100   &0.102 &0.116 \\  \cmidrule[0.5pt](l){2-11}

\multirow{12}{*}{ \STAB{\rotatebox[origin=c]{90}{\underline{DGP U-R}}}} 
&  & \multicolumn{3}{c}{MB} & \multicolumn{3}{c}{MAD} & \multicolumn{3}{c}{RMSE} \\ \cmidrule[0.2pt](l){3-5} \cmidrule[0.2pt](l){6-8} \cmidrule[0.2pt](l){9-11}
       &25     &0.003  &0.002  &-0.002 &0.143  &0.097  &0.178  &0.220   &0.156  &0.286  \\ 
       &50     &-0.001 &0.000      &-0.002 &0.103  &0.071  &0.131  &0.155  &0.108  &0.203  \\ 
       &100    &-0.001 &0.000      &-0.001 &0.073  &0.052  &0.094  &0.109  &0.078  &0.143  \\ 
      &200   &0.001 &0.000     &0.001 &0.051 &0.036 &0.066 &0.077 &0.054 &0.100   \\ 
      &400   &0.000     &0.000     &0.000     &0.036 &0.026 &0.047 &0.054 &0.038 &0.071 \\ 
\cmidrule[0.2pt](l){3-5} \cmidrule[0.2pt](l){6-8} \cmidrule[0.2pt](l){9-11}

& & \multicolumn{3}{c}{Rej. 1\%} & \multicolumn{3}{c}{Rej. 5\%} & \multicolumn{3}{c}{Rej. 10\%} \\ \cmidrule[0.2pt](l){3-5} \cmidrule[0.2pt](l){6-8} \cmidrule[0.2pt](l){9-11}

      &25    &0.013 &0.015 &0.010  &0.056 &0.056 &0.047 &0.110  &0.105 &0.101 \\ 
      &50    &0.010  &0.010  &0.009 &0.052 &0.048 &0.051 &0.108 &0.099 &0.101 \\ 
      &100   &0.010  &0.010  &0.008 &0.050  &0.052 &0.044 &0.102 &0.103 &0.097 \\ 
      &200   &0.011 &0.009 &0.009 &0.050  &0.052 &0.048 &0.097 &0.100   &0.097 \\ 
      &400   &0.009 &0.010  &0.008 &0.045 &0.047 &0.047 &0.097 &0.096 &0.101 \\ 
      \cmidrule[0.5pt](l){2-11}
\end{tabular}\label{Tab:Sim_III}
\end{table}

\begin{table}[!htbp]
\caption{Simulation - $ATT(t) = 0.0,\ t\geq 1 $}
\centering
\begin{tabular}{rlccccccccc}
\cmidrule[0.5pt](l){2-11}
 & $T,\T $ & \multicolumn{1}{c}{DiD} & \multicolumn{1}{c}{SC} & \multicolumn{1}{c}{BA} & \multicolumn{1}{c}{DiD} & \multicolumn{1}{c}{SC} & \multicolumn{1}{c}{BA} & \multicolumn{1}{c}{DiD} & \multicolumn{1}{c}{SC} & \multicolumn{1}{c}{BA} \\ \cmidrule[0.2pt](l){2-11}

\multirow{12}{*}{ \STAB{\rotatebox[origin=c]{90}{\underline{DGP Q-T}}}} 
&  & \multicolumn{3}{c}{MB} & \multicolumn{3}{c}{MAD} & \multicolumn{3}{c}{RMSE} \\ \cmidrule[0.2pt](l){3-5} \cmidrule[0.2pt](l){6-8} \cmidrule[0.2pt](l){9-11}
      &25     &0.003  &0.002  &25.998 &0.141  &0.096  &25.998 &0.217  &0.153  &25.999 \\ 
       &50     &-0.001 &0.000      &50.998 &0.103  &0.07   &50.998 &0.154  &0.107  &50.998 \\ 
        &100     &-0.001  &0.000       &100.998 &0.073   &0.052   &100.998 &0.108   &0.077   &100.999 \\ 
        &200     &0.001   &0.001   &201.001 &0.050    &0.036   &201     &0.076   &0.054   &201.001 \\ 
        &400     &0.000       &0.000       &401     &0.036   &0.026   &401.001 &0.054   &0.038   &401     \\ 
\cmidrule[0.2pt](l){3-5} \cmidrule[0.2pt](l){6-8} \cmidrule[0.2pt](l){9-11}

& & \multicolumn{3}{c}{Rej. 1\%} & \multicolumn{3}{c}{Rej. 5\%} & \multicolumn{3}{c}{Rej. 10\%} \\ \cmidrule[0.2pt](l){3-5} \cmidrule[0.2pt](l){6-8} \cmidrule[0.2pt](l){9-11}

      &25    &0.013 &0.015 &1.000     &0.060  &0.058 &1.000     &0.113 &0.105 &1.000     \\ 
      &50    &0.012 &0.011 &1.000     &0.052 &0.049 &1.000     &0.104 &0.100   &1.000     \\ 
      &100   &0.010  &0.010  &1.000     &0.051 &0.051 &1.000     &0.103 &0.103 &1.000     \\ 
      &200   &0.012 &0.010  &1.000     &0.051 &0.051 &1.000     &0.098 &0.100   &1.000     \\ 
      &400   &0.010  &0.010  &1.000     &0.046 &0.047 &1.000     &0.097 &0.096 &1.000     \\
      \cmidrule[0.5pt](l){2-11}

\multirow{12}{*}{ \STAB{\rotatebox[origin=c]{90}{\underline{DGP T-T}}}} 
&  & \multicolumn{3}{c}{MB} & \multicolumn{3}{c}{MAD} & \multicolumn{3}{c}{RMSE} \\ \cmidrule[0.2pt](l){3-5} \cmidrule[0.2pt](l){6-8} \cmidrule[0.2pt](l){9-11}
      & 25  &  0.003 &  25.005 &  0.008 &  0.270 &  25.008 &  0.344 &  0.426 &  25.007 &  0.571 \\ 
    & 50  & -0.005 &  49.995 &  0.005 &  0.200 &  49.997 &  0.251 &  0.306 &  49.996 &  0.400 \\ 
    & 100 & -0.002 & 100.000 & -0.002 &  0.140 & 100.002 &  0.183 &  0.214 & 100.000 &  0.282 \\ 
    & 200 &  0.003 & 200.003 &  0.002 &  0.100 & 200.004 &  0.130 &  0.151 & 200.003 &  0.199 \\ 
    & 400 &  0.000 & 400.000 & -0.001 &  0.073 & 399.999 &  0.093 &  0.108 & 400.000 &  0.140 \\ 
\cmidrule[0.2pt](l){3-5} \cmidrule[0.2pt](l){6-8} \cmidrule[0.2pt](l){9-11}

& & \multicolumn{3}{c}{Rej. 1\%} & \multicolumn{3}{c}{Rej. 5\%} & \multicolumn{3}{c}{Rej. 10\%} \\ \cmidrule[0.2pt](l){3-5} \cmidrule[0.2pt](l){6-8} \cmidrule[0.2pt](l){9-11}

      & 25  & 0.015 & 1.000 & 0.016 & 0.066 & 1.000 & 0.068 & 0.124 & 1.000 & 0.132 \\ 
& 50  & 0.013 & 1.000 & 0.014 & 0.060 & 1.000 & 0.062 & 0.118 & 1.000 & 0.118 \\ 
& 100 & 0.012 & 1.000 & 0.011 & 0.052 & 1.000 & 0.053 & 0.102 & 1.000 & 0.105 \\ 
& 200 & 0.009 & 1.000 & 0.011 & 0.048 & 1.000 & 0.049 & 0.099 & 1.000 & 0.101 \\ 
& 400 & 0.010 & 1.000 & 0.009 & 0.050 & 1.000 & 0.049 & 0.102 & 1.000 & 0.103 \\
      \cmidrule[0.5pt](l){2-11}
      
\end{tabular}\label{Tab:Sim_IV}
\end{table}

\Cref{Tab:Sim_I,Tab:Sim_II,Tab:Sim_III,Tab:Sim_IV} present the mean bias (MB), the median absolute deviation (MAD), the root-mean-squared error (RMSE), and the empirical rejection rates of the null hypothesis, $ \Hyp_o: ATT_{\omega,T} = 0.0 $ at conventional nominal levels $1\%$, $5\%$, and $10\%$ for the DiD, SC, and BA estimators across all 11 DGPs. Results are based on 10,000 Monte Carlo replications for each sample size $\T=T \in \{ 25,50,100,200,400\}$. Standard errors used throughout the text are heteroskedasticity and auto-correlation robust using the \citet{newey-west-1987-simple} procedure. For all competing estimators, lagged $X_t$ are controlled for in DGP AR(1) while first differences are applied to $X_t$ before estimation in DGP U-R. The uniform weighting scheme is used throughout.

All three estimators perform reasonably under scenarios where they are expected to. The poor performance of the SC in DGP BA confirms its sensitivity to differences in pre-treatment average outcomes, whereas the poor performance of the BA in DGP SC and Q-Trend confirms its sensitivity to common shocks, which the BA cannot disentangle from treatment effects. Although all estimators appear to suffer size distortion at small sample sizes under DGPs PT-NA (A), MA(1), and AR(1), these improve with the sample. The SC only uses data from the post-treatment periods so it fails to harness identifying variation from pre-treatment periods to counter bias occurring in the post-treatment periods. This explains the poor performance of the SC in DGP PT-NA (B) while both the T-DiD and BA perform reasonably. On the whole, one observes the robustness of the proposed T-DiD to several interesting and empirically relevant settings through the simulation results presented in \Cref{Tab:Sim_I,Tab:Sim_II,Tab:Sim_III,Tab:Sim_IV}.

\subsection{Detecting significant $ATT_{\omega,T}$s}\label{App_Sub:Pow_Curve_ATT}
To ensure that the good size control of the T-DiD is not achieved at the expense of power, the DGP in \Cref{Sect:Sim} (see \eqref{eqn:DGP_Sim2}) with $ATT(t)=ATT, \ ATT \in [0.0,1.0] $ is used. The hypotheses are $\Hyp_o: ATT_{\omega,T} = 0.0 \text{ vs } \Hyp_a: ATT_{\omega,T} \neq 0.0 $. Thus, the goal is to ensure that hypothesis tests based on the DiD estimate control size under the null and have the power to detect non-trivial effects of treatment or, generally, deviations away from a posited null hypothesis on the size of treatment effects.

\begin{figure}[!htbp]
\centering 
\begin{subfigure}{0.32\textwidth}
\centering
\caption{10\%}
\includegraphics[width=1\textwidth]{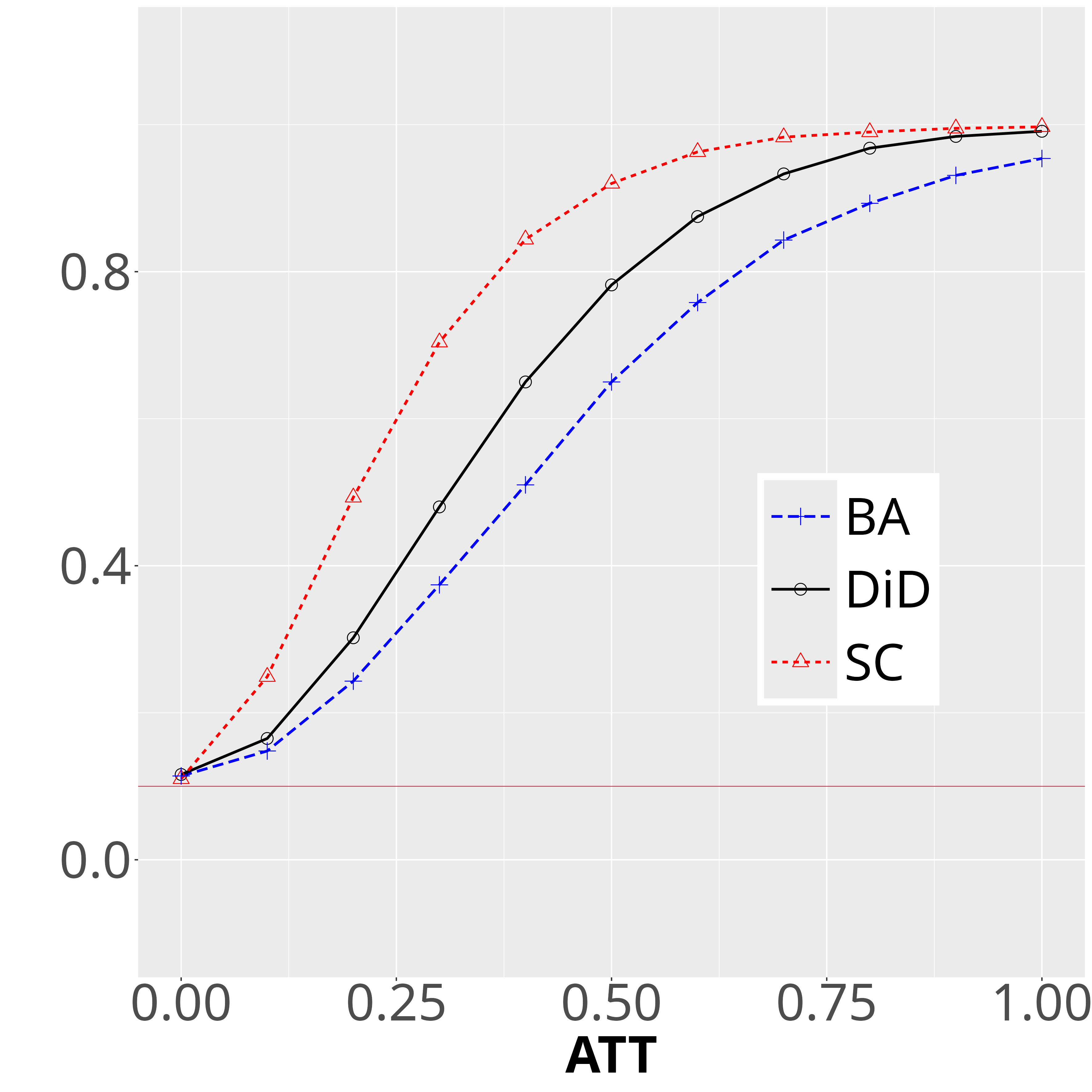}
\end{subfigure}
\begin{subfigure}{0.32\textwidth}
\centering
\caption{5\%}
\includegraphics[width=1\textwidth]{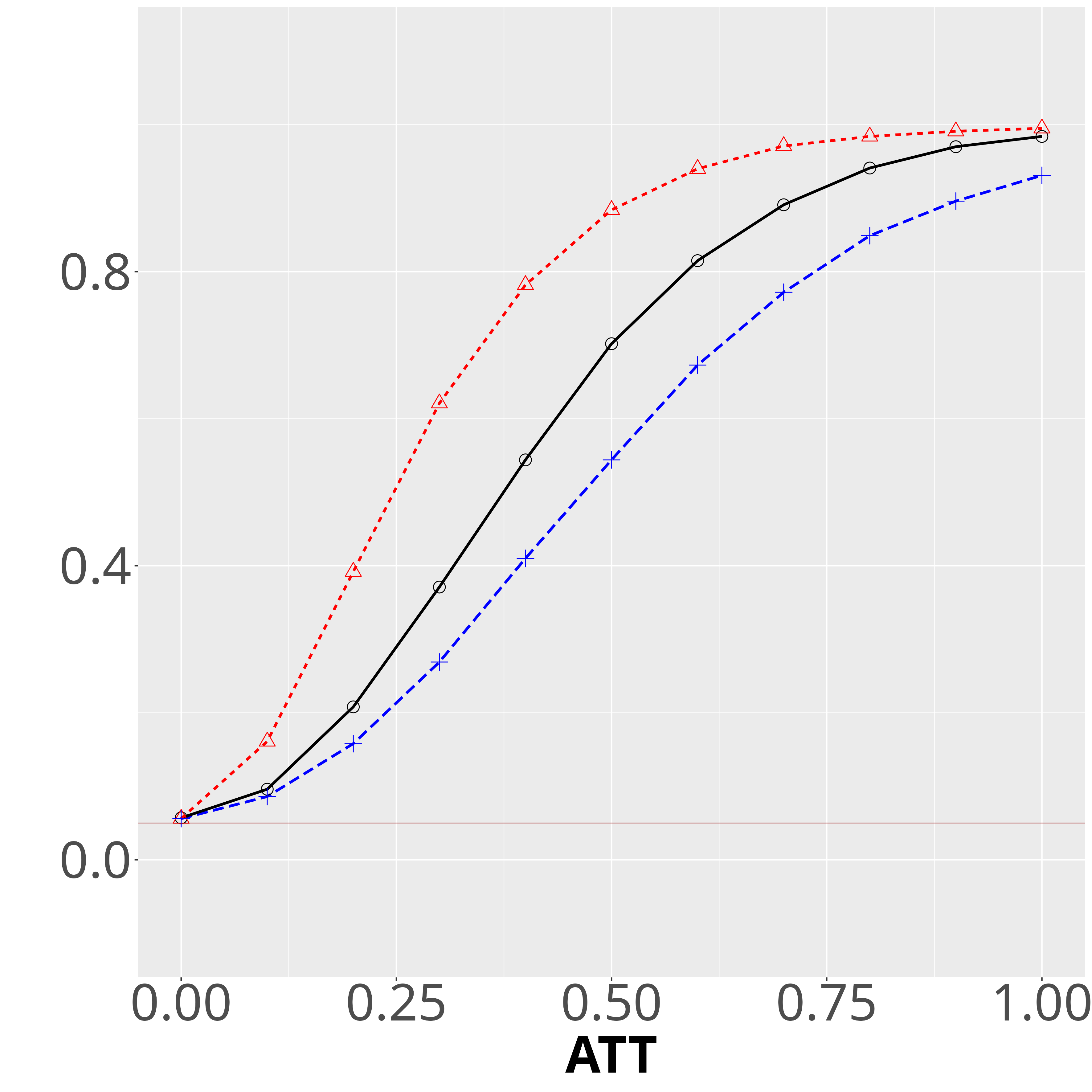}
\end{subfigure}
\begin{subfigure}{0.32\textwidth}
\centering
\caption{1\%}
\includegraphics[width=1\textwidth]{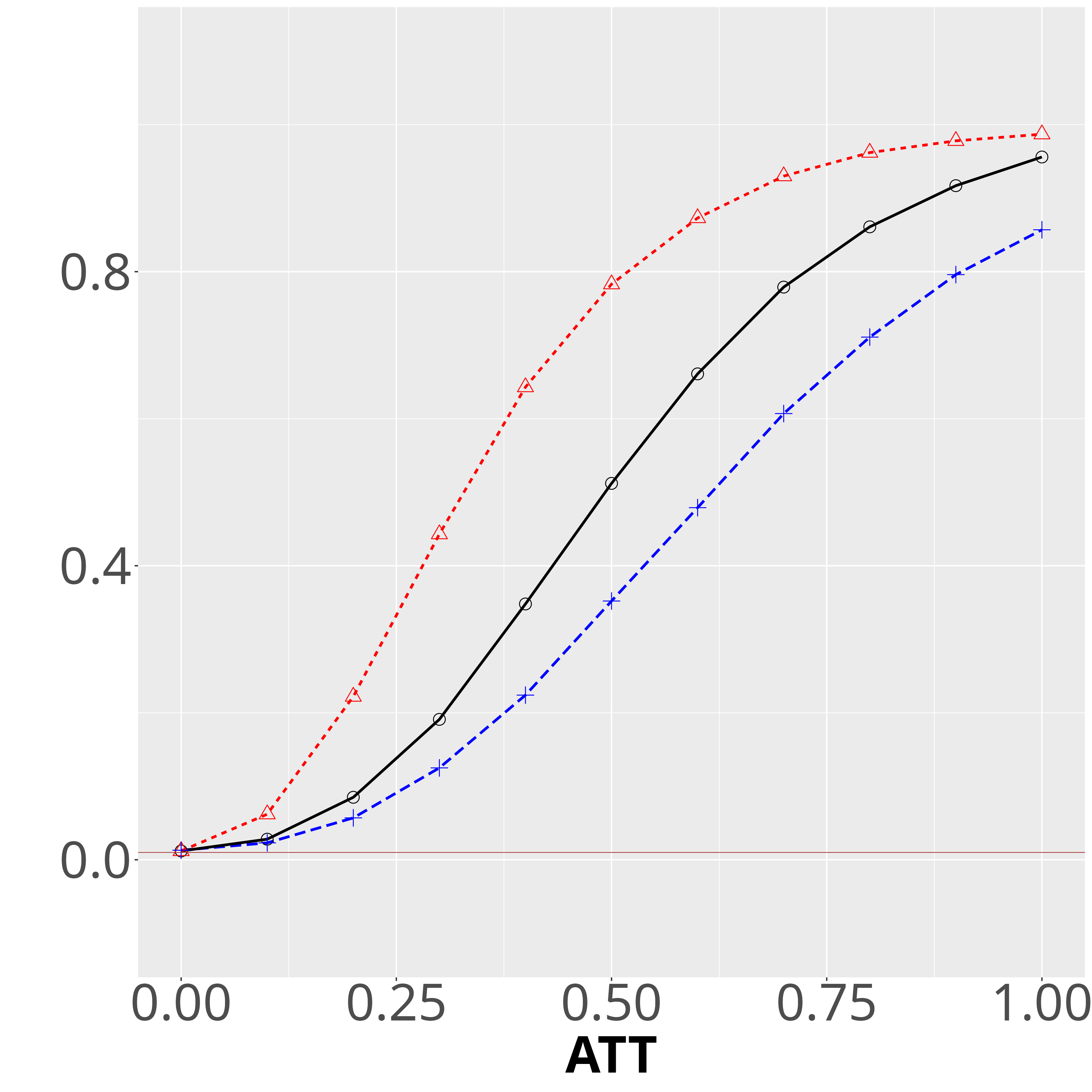}
\end{subfigure}
\caption{Power Curves - DGP SC-BA, $\T,T=25$.}
\label{Fig:DGP_SC_BA}
\end{figure}

\begin{figure}[!htbp]
\centering 
\begin{subfigure}{0.32\textwidth}
\centering
\caption{10\%}
\includegraphics[width=1\textwidth]{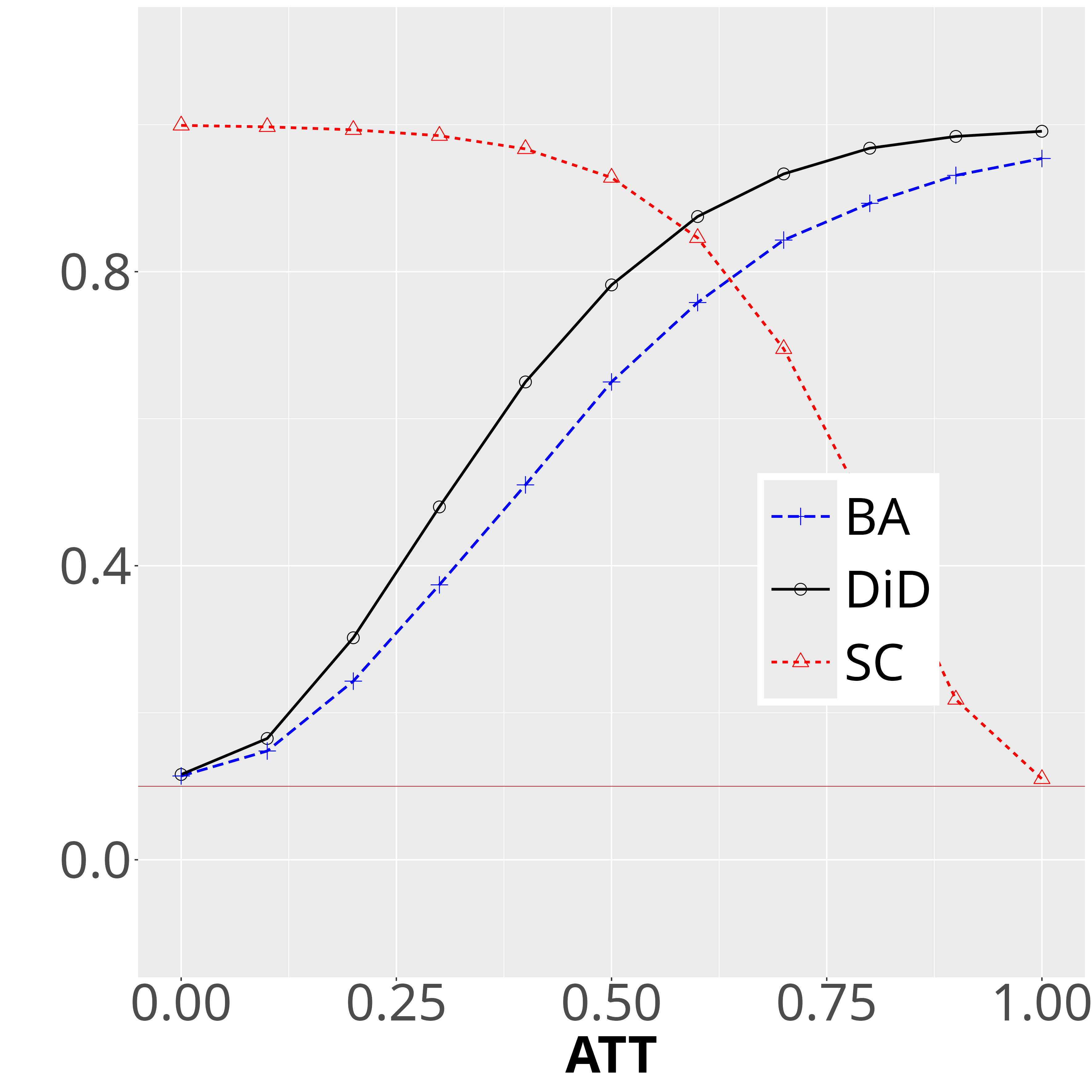}
\end{subfigure}
\begin{subfigure}{0.32\textwidth}
\centering
\caption{5\%}
\includegraphics[width=1\textwidth]{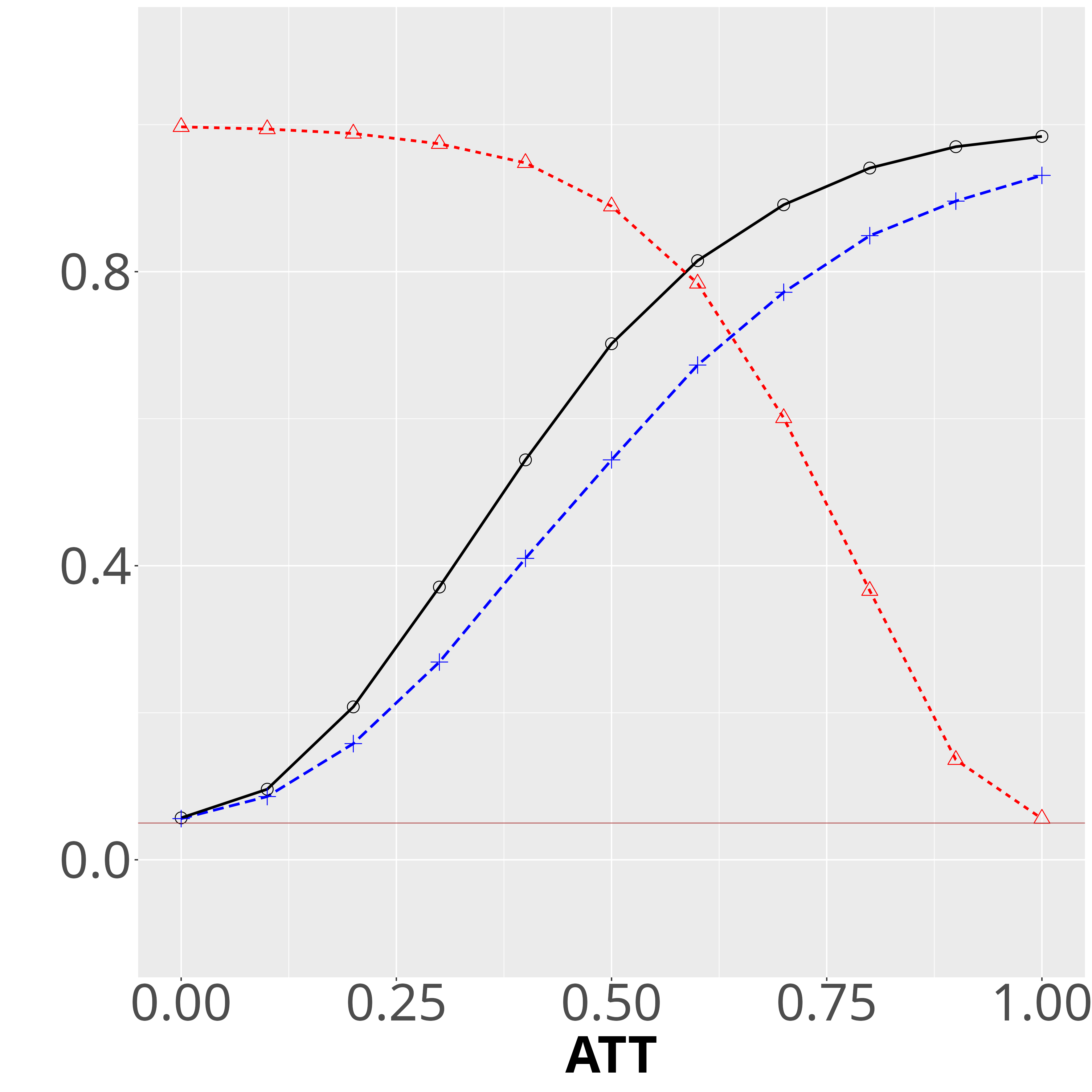}
\end{subfigure}
\begin{subfigure}{0.32\textwidth}
\centering
\caption{1\%}
\includegraphics[width=1\textwidth]{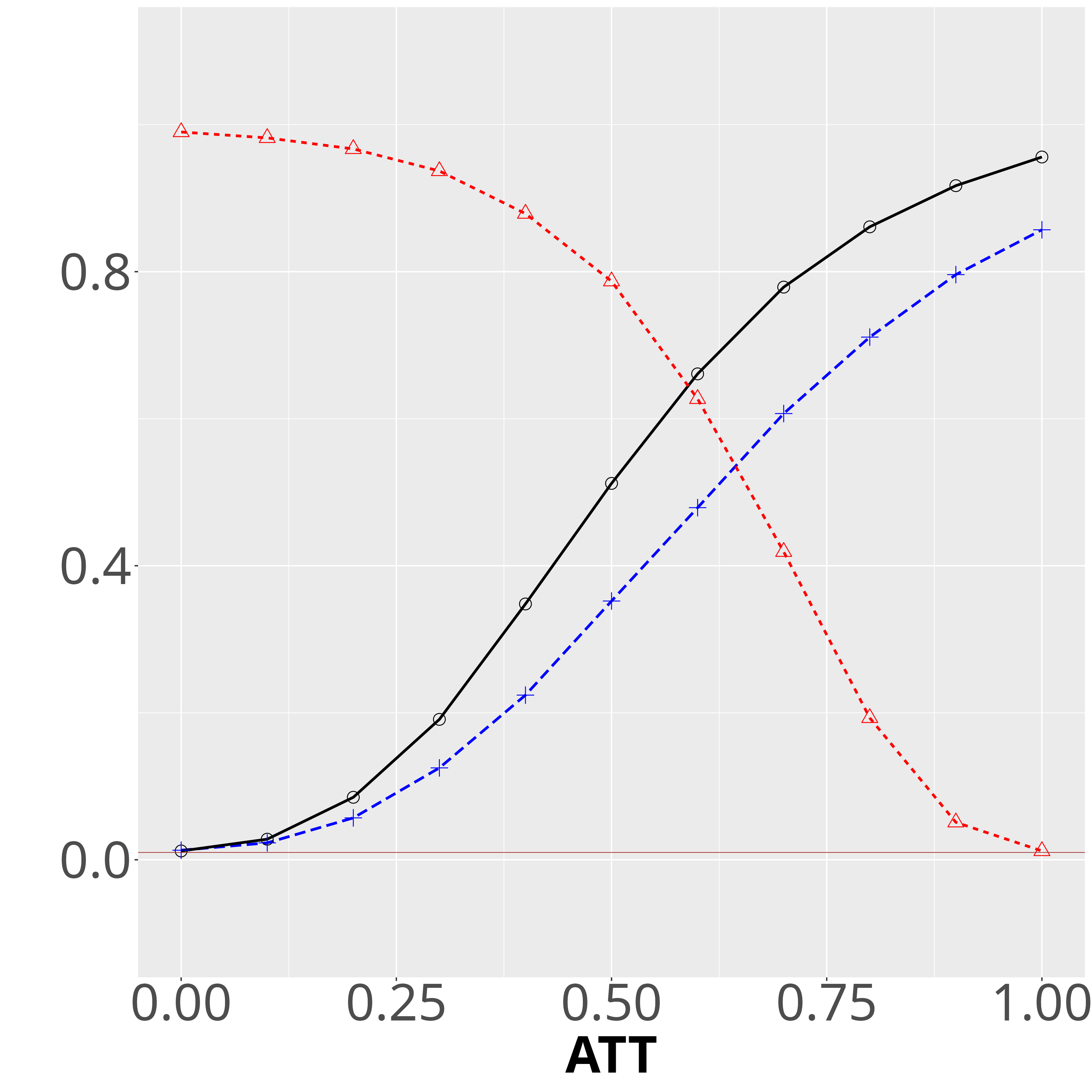}
\end{subfigure}
\caption{Power Curves - DGP BA, $\T,T=25$.}
\label{Fig:DGP_BA}
\end{figure}

\begin{figure}[!htbp]
\centering 
\begin{subfigure}{0.32\textwidth}
\centering
\caption{10\%}
\includegraphics[width=1\textwidth]{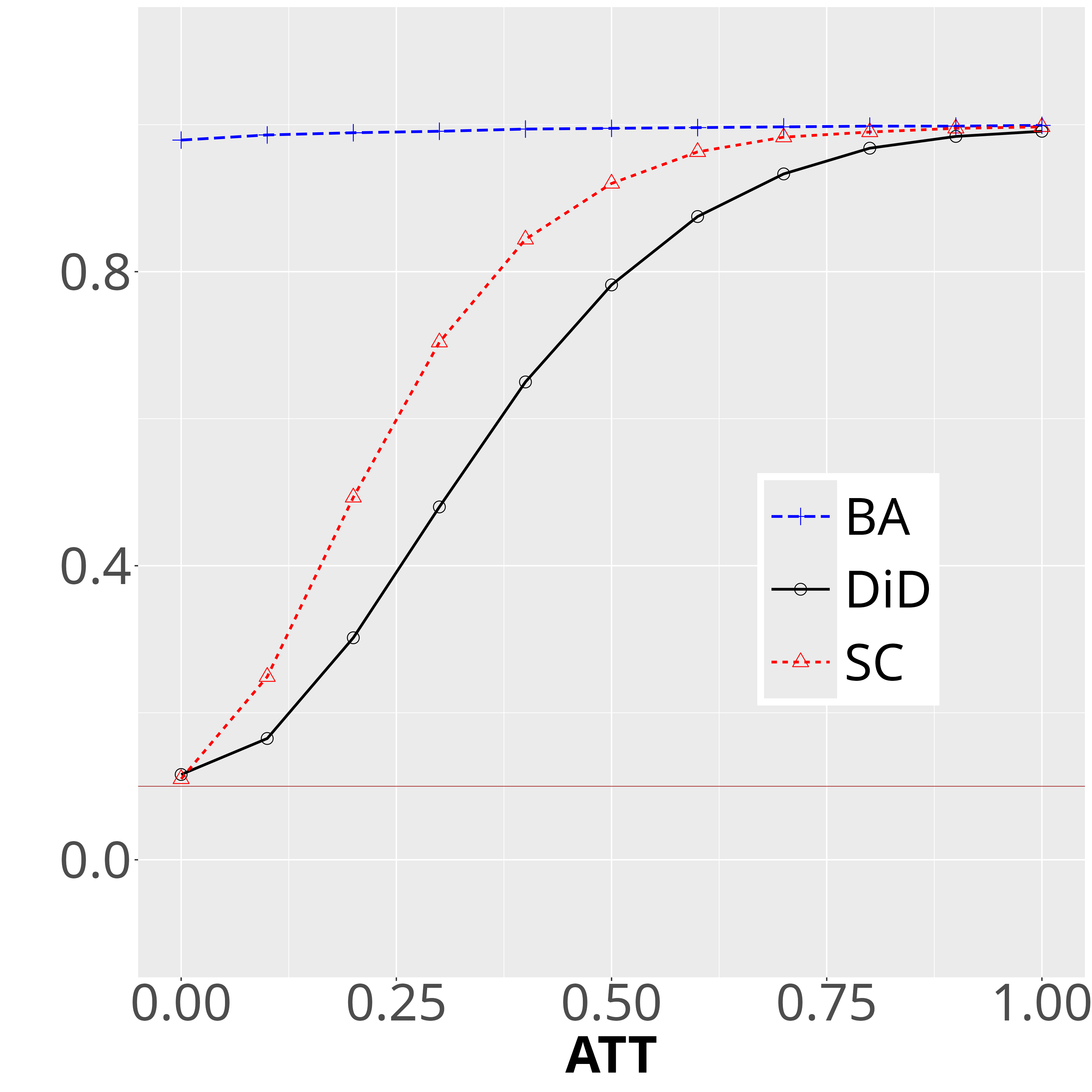}
\end{subfigure}
\begin{subfigure}{0.32\textwidth}
\centering
\caption{5\%}
\includegraphics[width=1\textwidth]{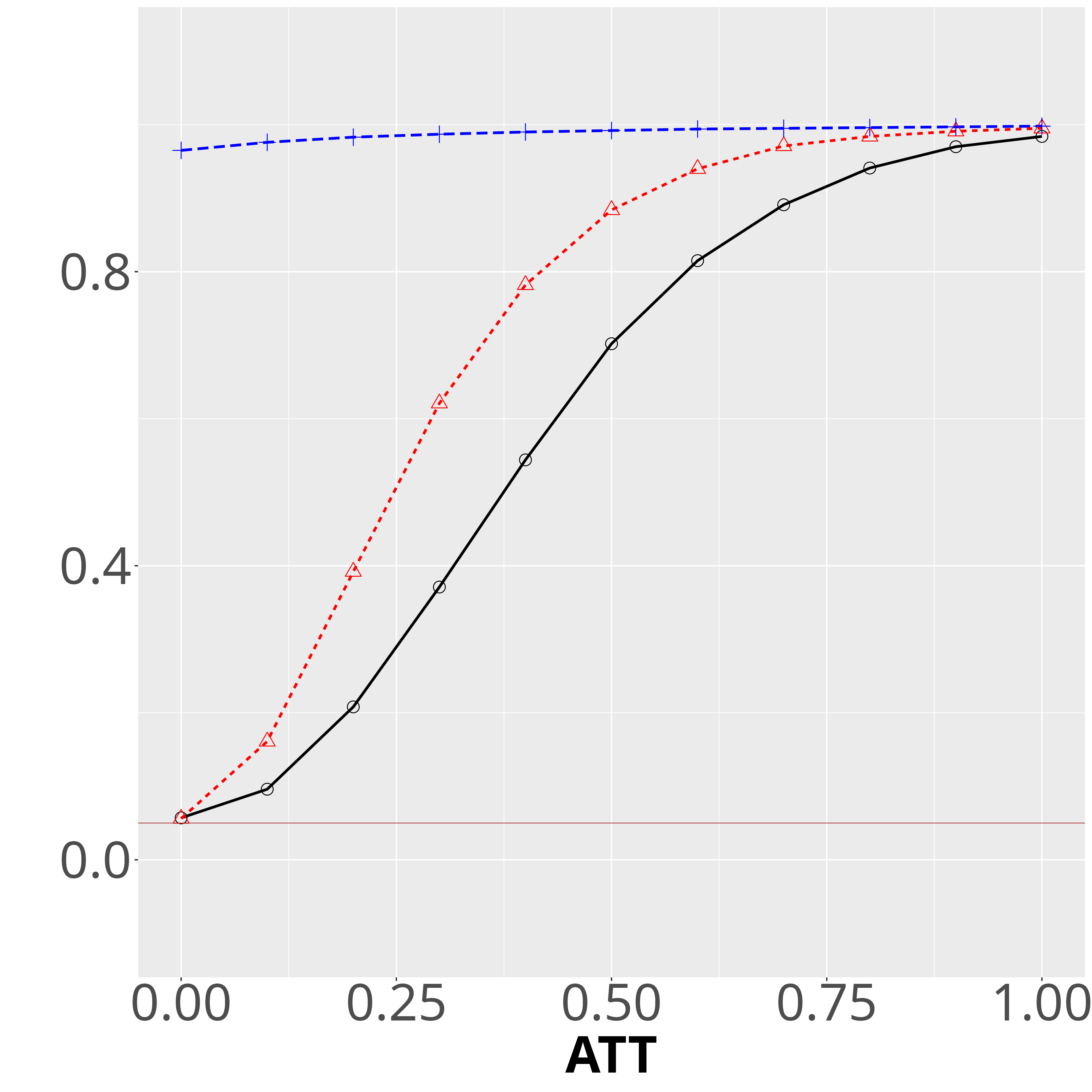}
\end{subfigure}
\begin{subfigure}{0.32\textwidth}
\centering
\caption{1\%}
\includegraphics[width=1\textwidth]{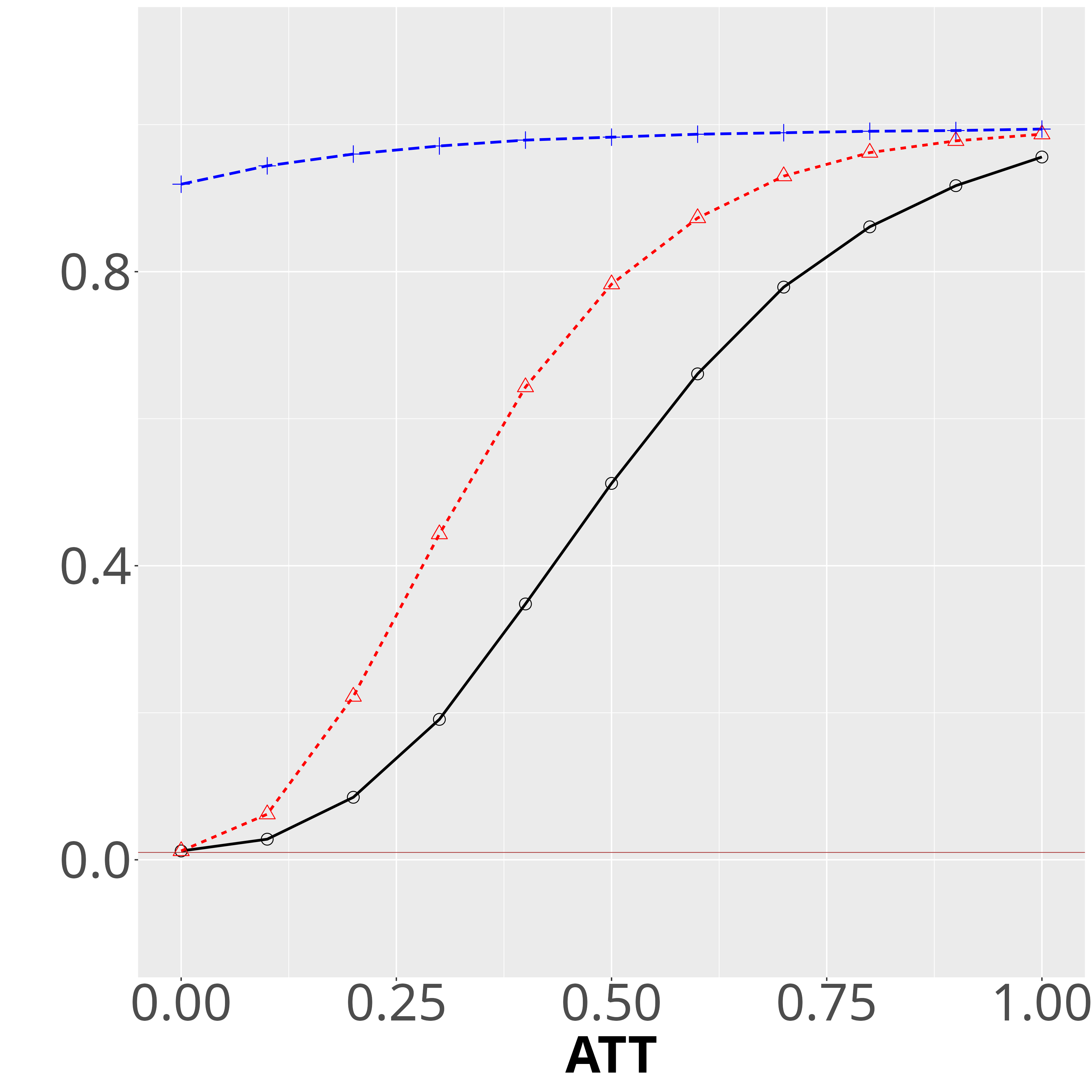}
\end{subfigure}
\caption{Power Curves - DGP SC, $\T,T=25$.}
\label{Fig:DGP_SC}
\end{figure}

\begin{figure}[!htbp]
\centering 
\begin{subfigure}{0.32\textwidth}
\centering
\caption{10\%}
\includegraphics[width=1\textwidth]{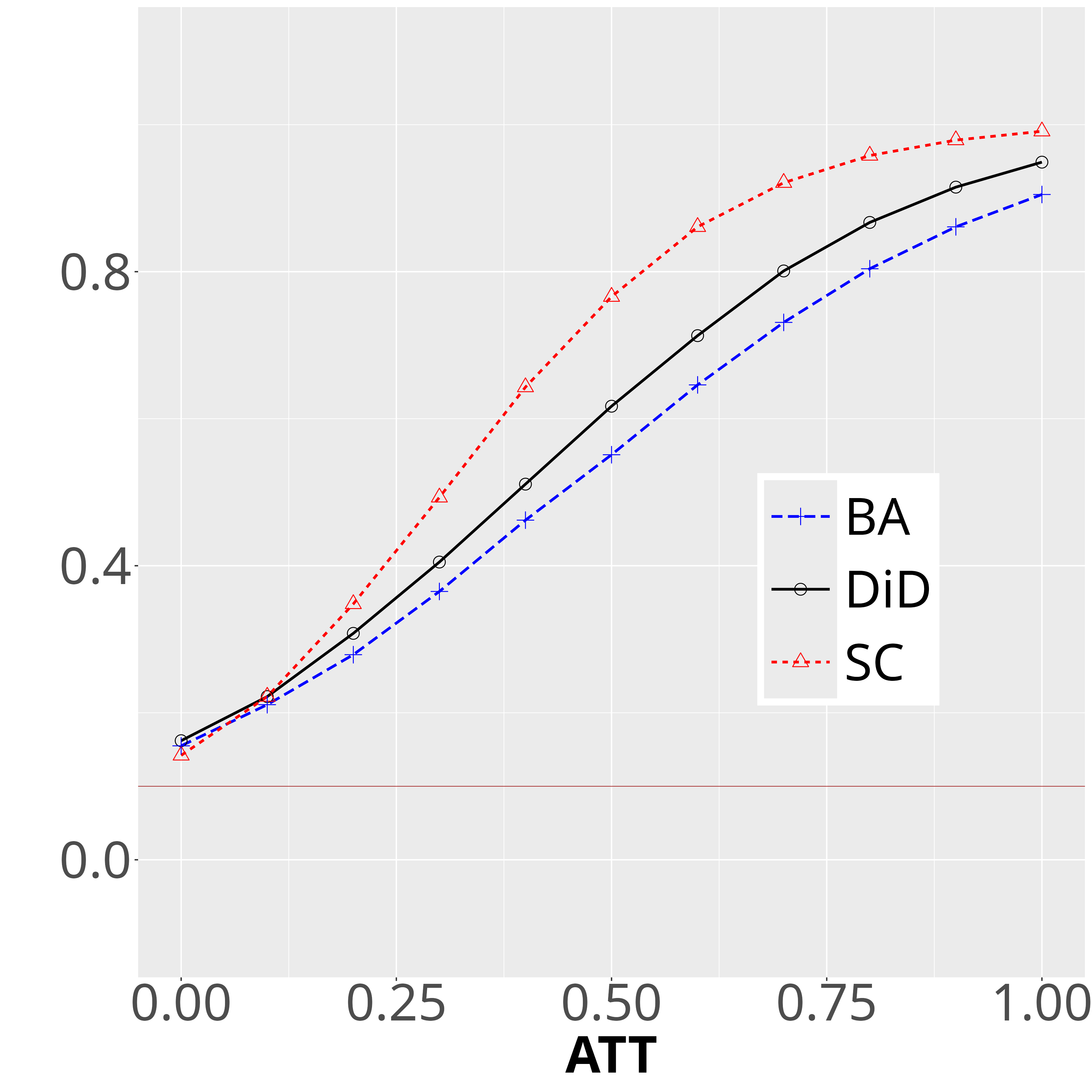}
\end{subfigure}
\begin{subfigure}{0.32\textwidth}
\centering
\caption{5\%}
\includegraphics[width=1\textwidth]{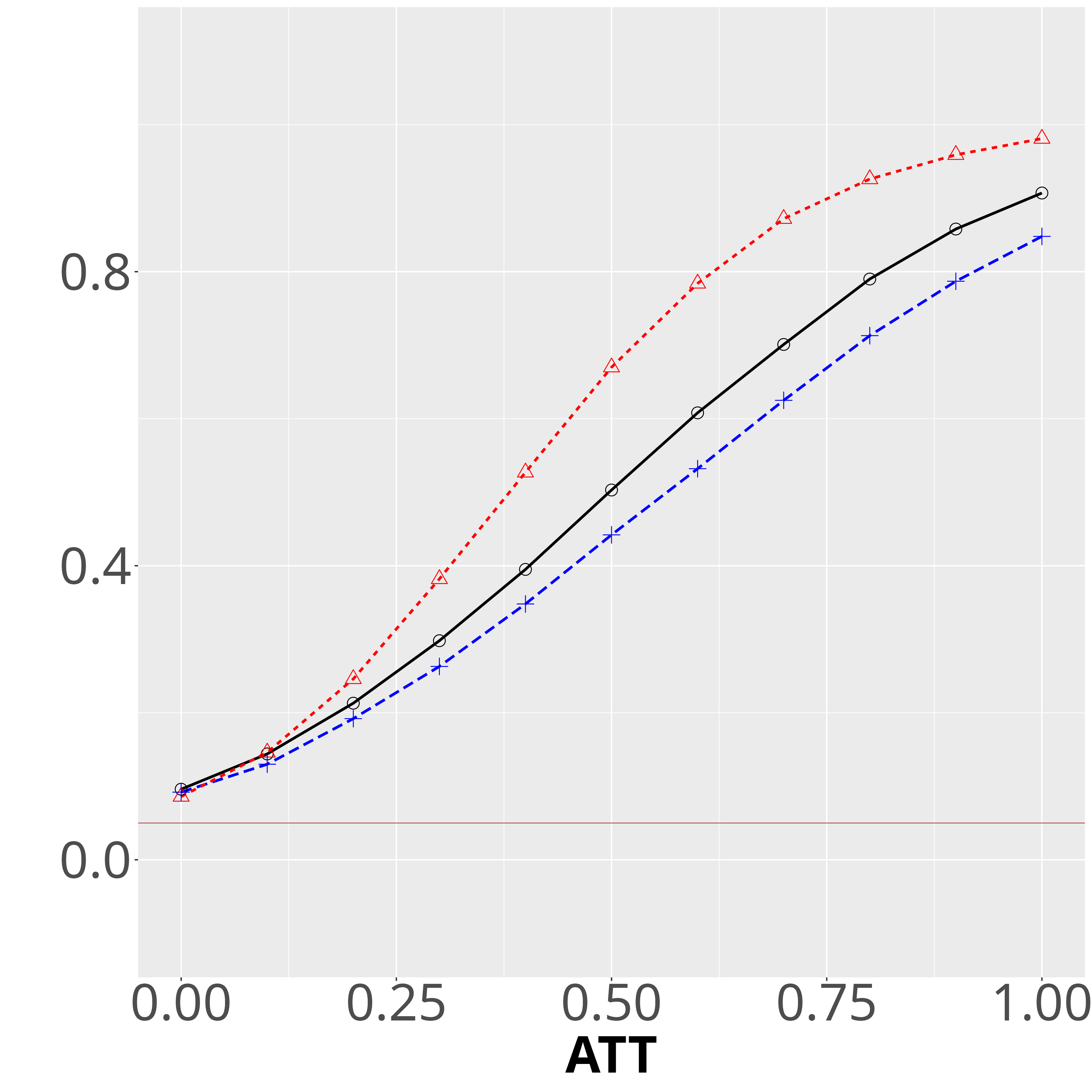}
\end{subfigure}
\begin{subfigure}{0.32\textwidth}
\centering
\caption{1\%}
\includegraphics[width=1\textwidth]{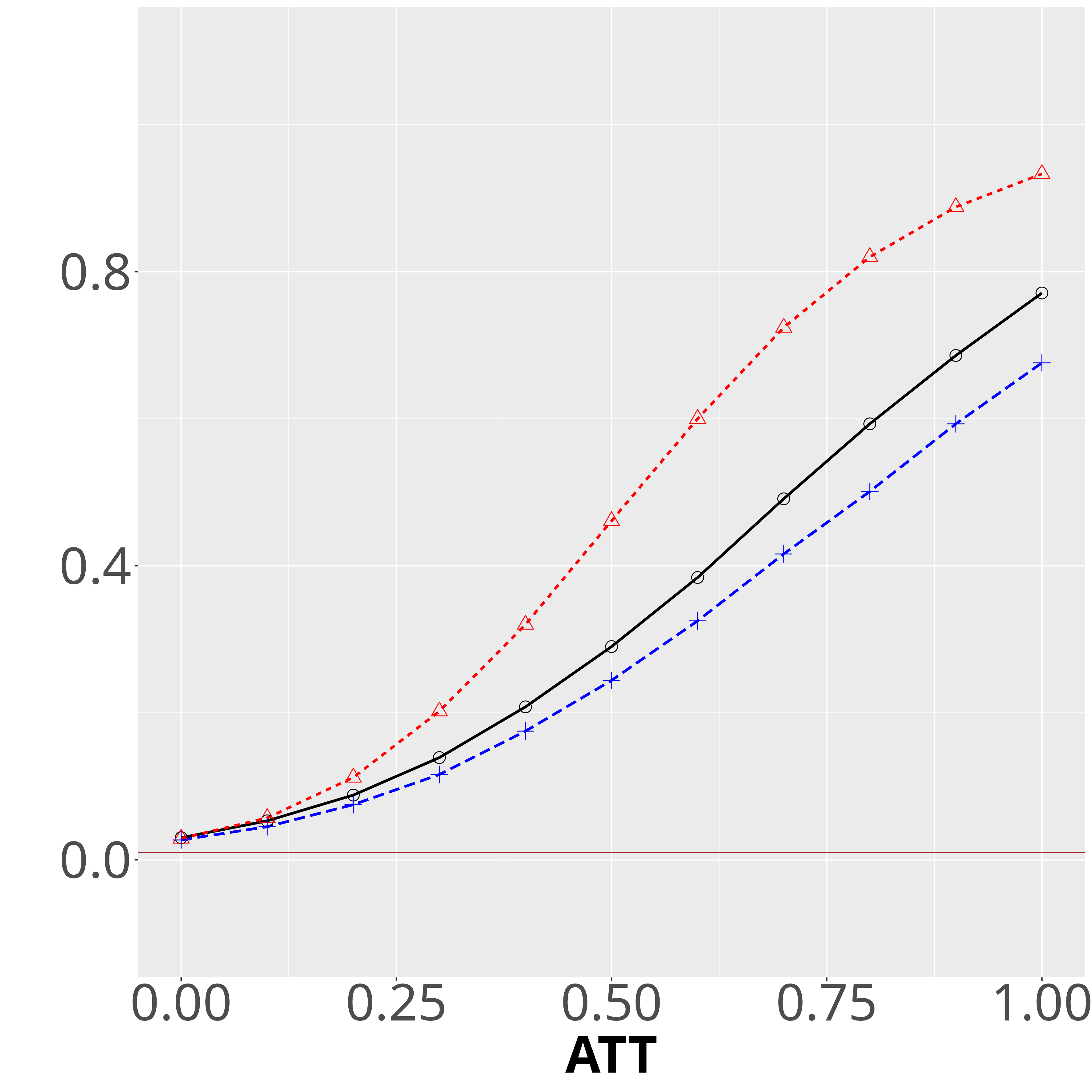}
\end{subfigure}
\caption{Power Curves - DGP PT-NA (A), $\T,T=25$.}
\label{Fig:DGP_PT_NA_A}
\end{figure}

\begin{figure}[!htbp]
\centering 
\begin{subfigure}{0.32\textwidth}
\centering
\caption{10\%}
\includegraphics[width=1\textwidth]{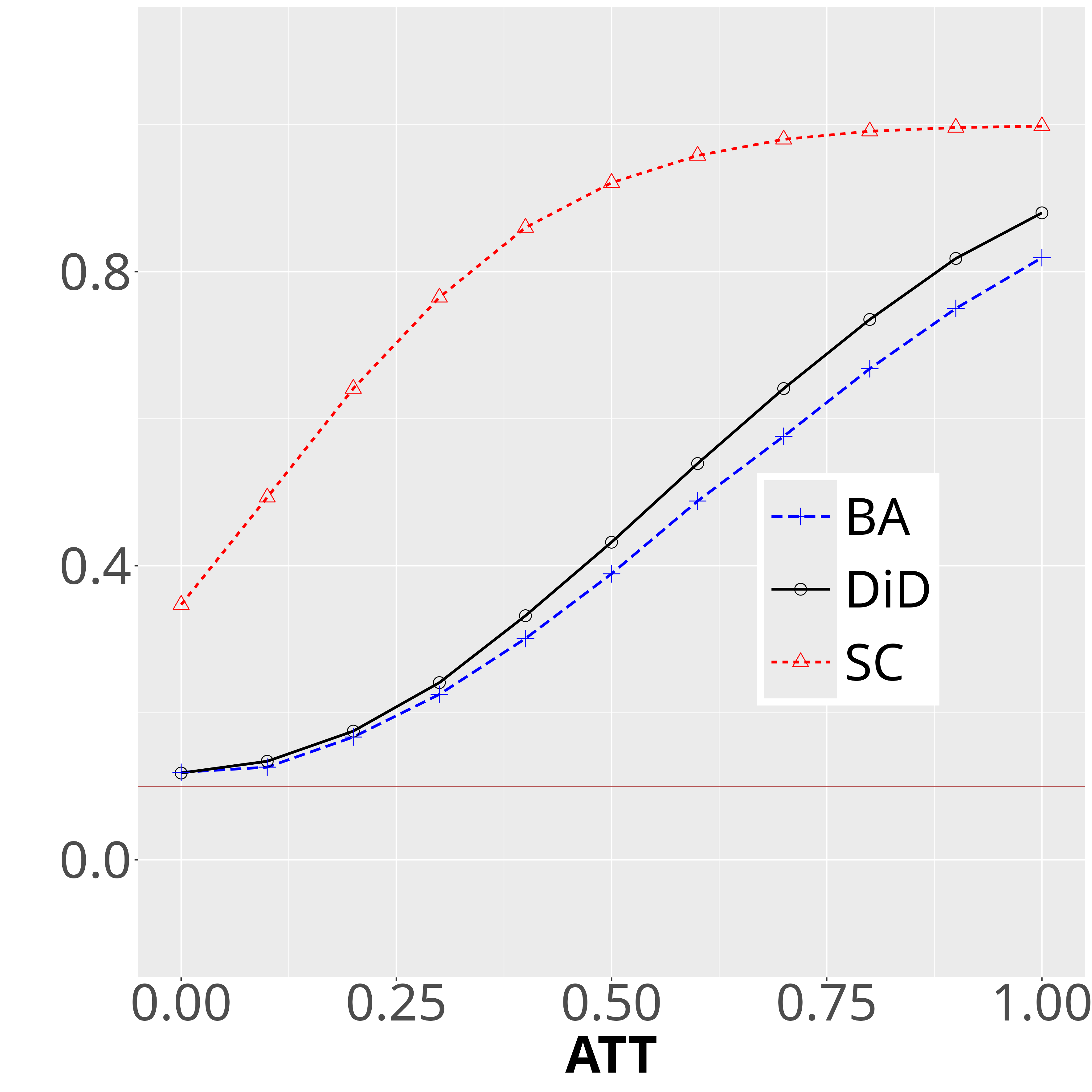}
\end{subfigure}
\begin{subfigure}{0.32\textwidth}
\centering
\caption{5\%}
\includegraphics[width=1\textwidth]{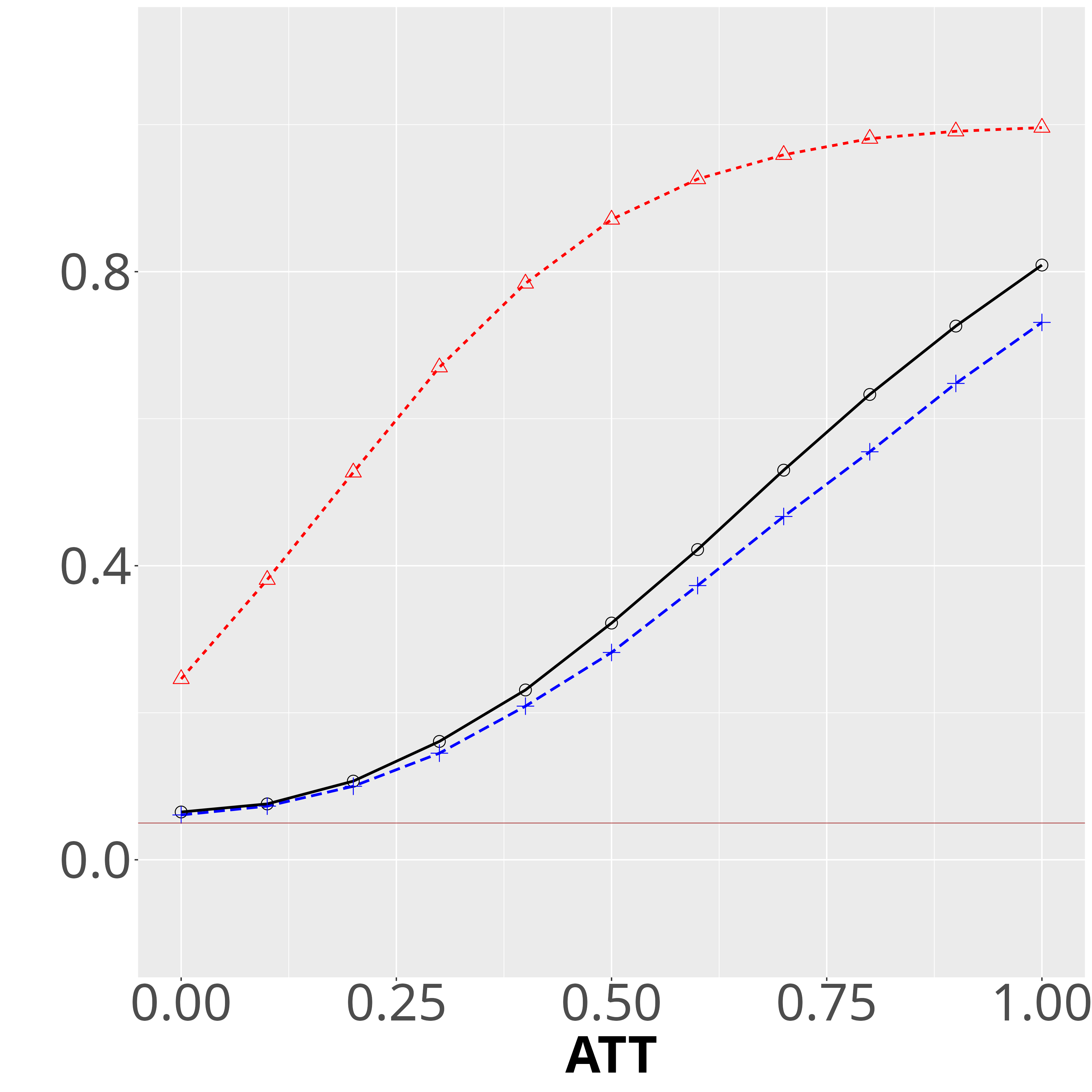}
\end{subfigure}
\begin{subfigure}{0.32\textwidth}
\centering
\caption{1\%}
\includegraphics[width=1\textwidth]{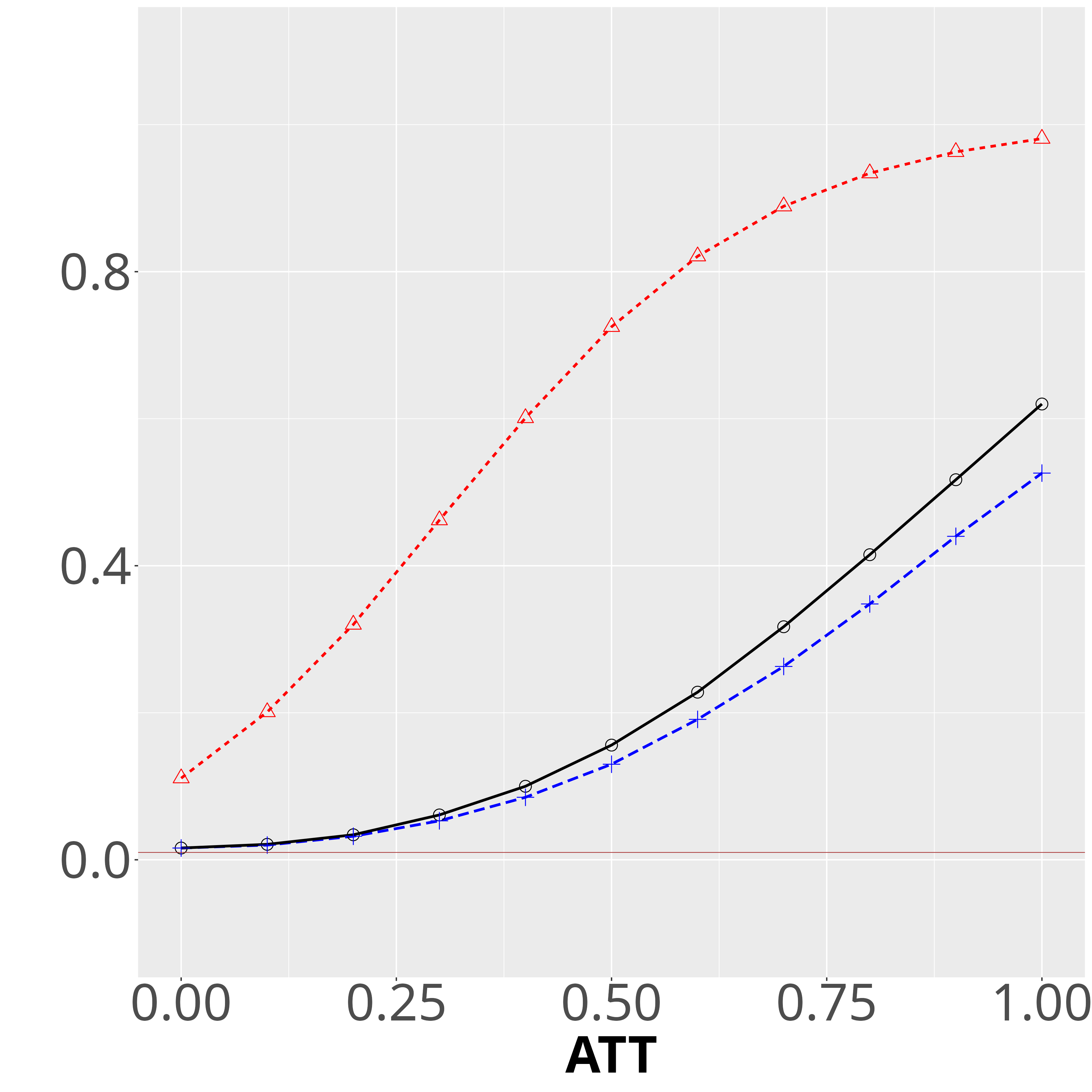}
\end{subfigure}
\caption{Power Curves - DGP PT-NA (B), $\T,T=25$.}
\label{Fig:DGP_PT_NA_B}
\end{figure}

\begin{figure}[!htbp]
\centering 
\begin{subfigure}{0.32\textwidth}
\centering
\caption{10\%}
\includegraphics[width=1\textwidth]{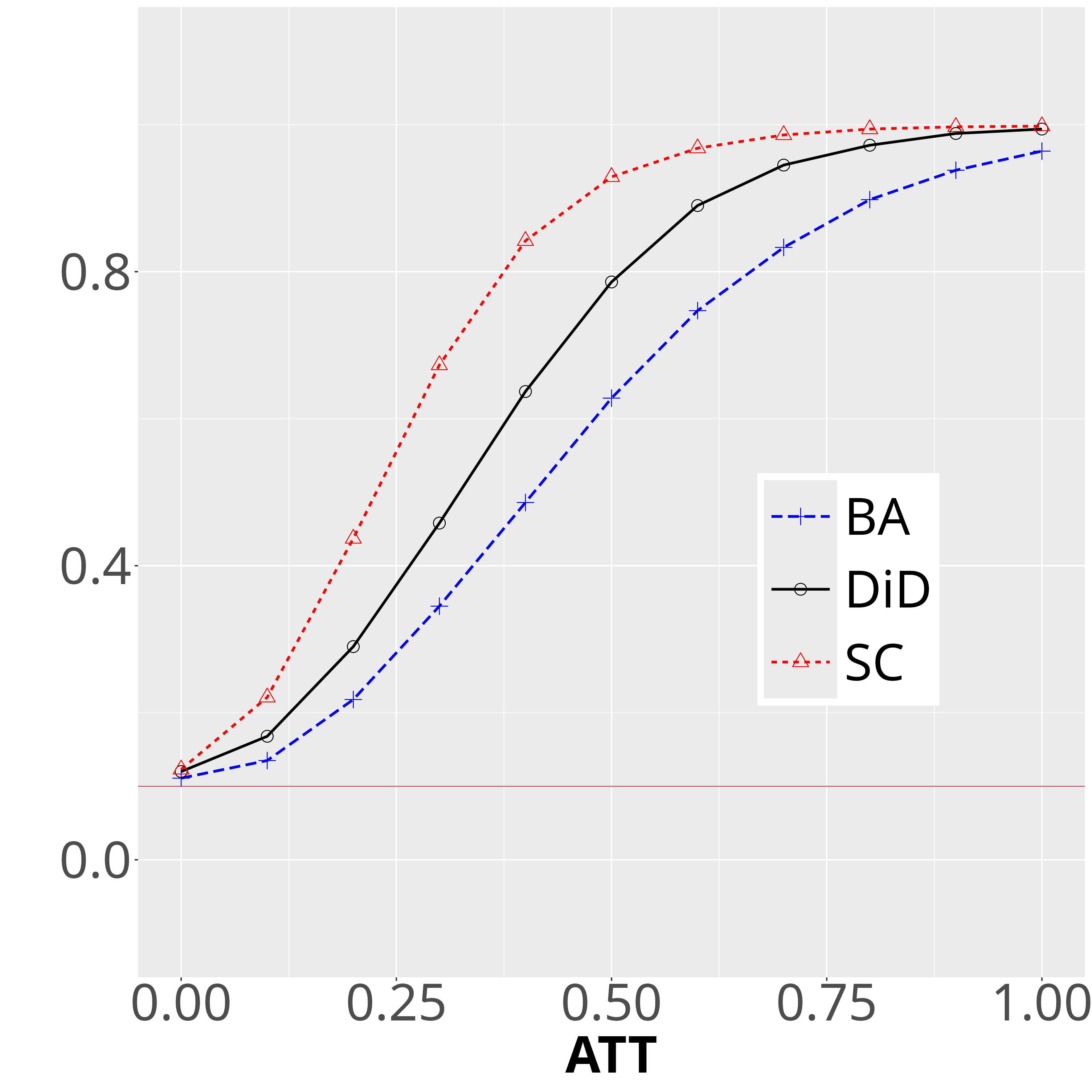}
\end{subfigure}
\begin{subfigure}{0.32\textwidth}
\centering
\caption{5\%}
\includegraphics[width=1\textwidth]{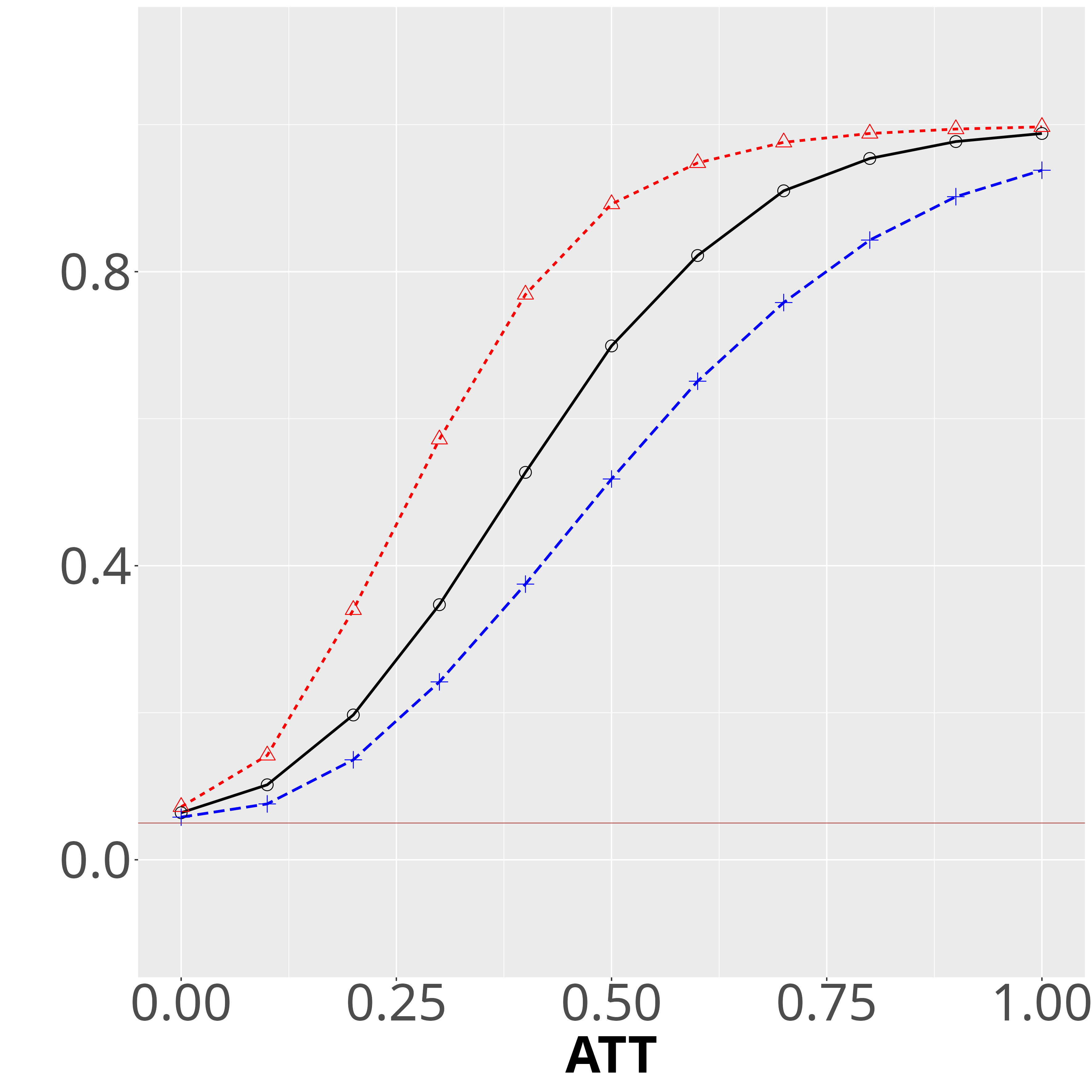}
\end{subfigure}
\begin{subfigure}{0.32\textwidth}
\centering
\caption{1\%}
\includegraphics[width=1\textwidth]{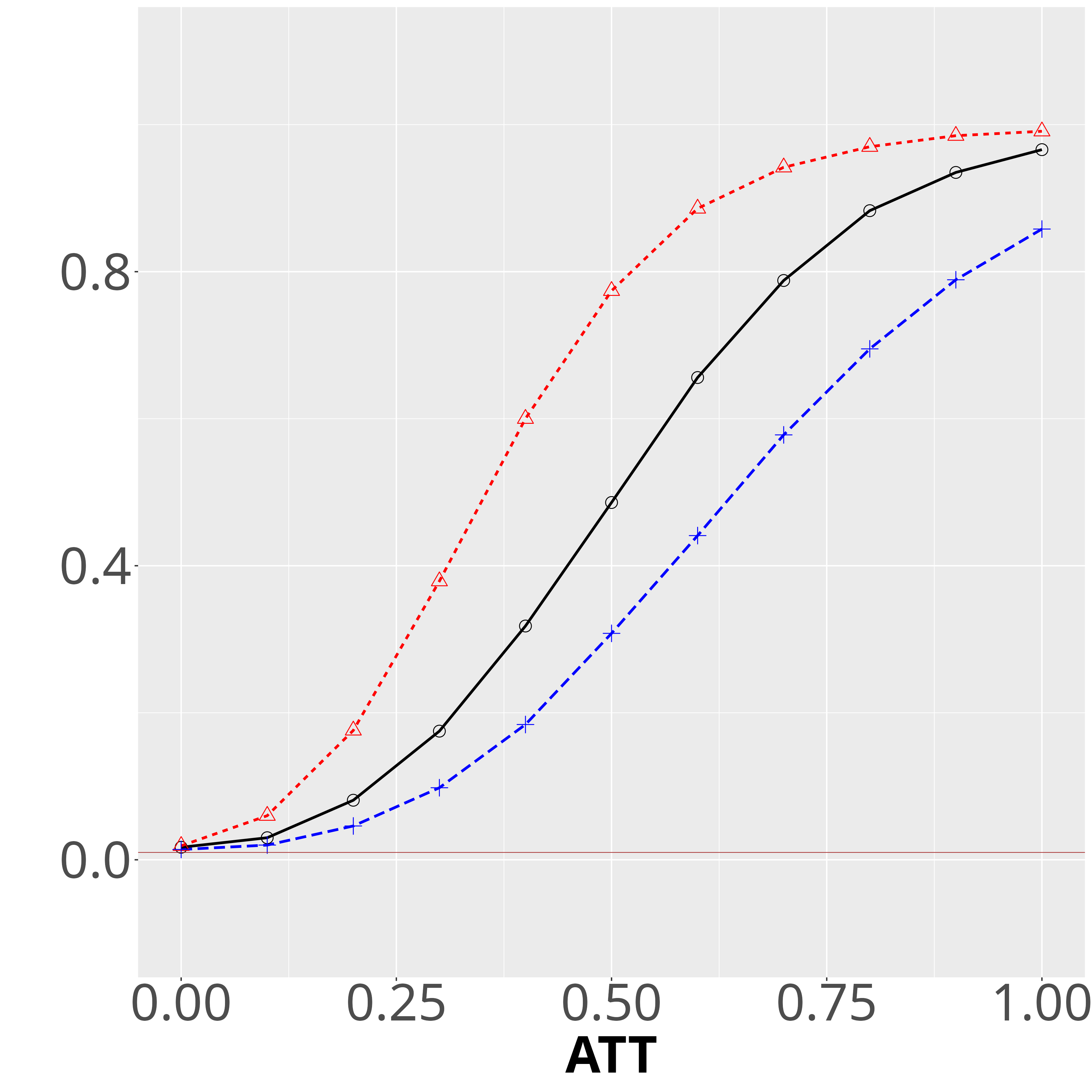}
\end{subfigure}
\caption{Power Curves - DGP GARCH(1,1), $\T,T=25$.}
\label{Fig:DGP_GARCH}
\end{figure}

\begin{figure}[!htbp]
\centering 
\begin{subfigure}{0.32\textwidth}
\centering
\caption{10\%}
\includegraphics[width=1\textwidth]{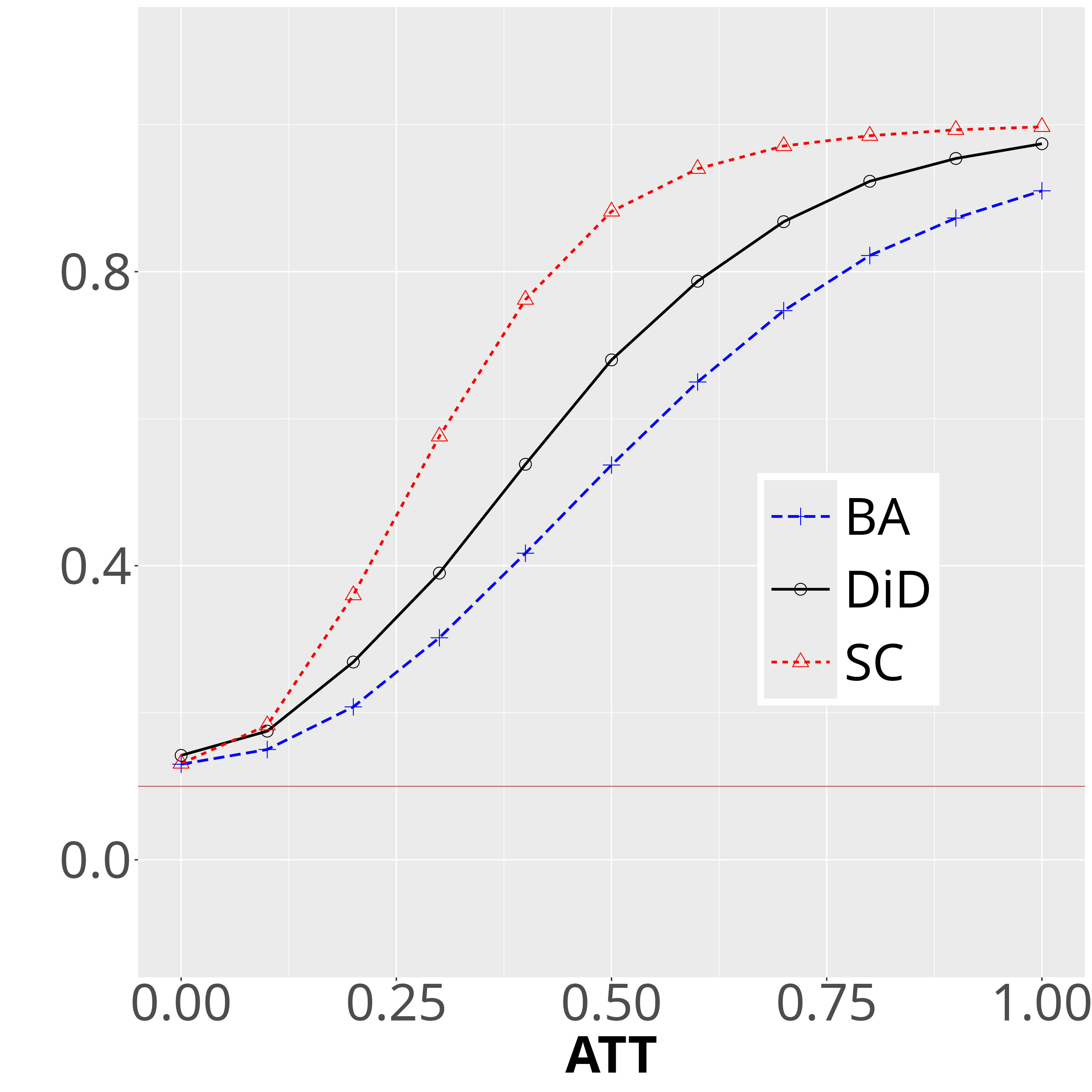}
\end{subfigure}
\begin{subfigure}{0.32\textwidth}
\centering
\caption{5\%}
\includegraphics[width=1\textwidth]{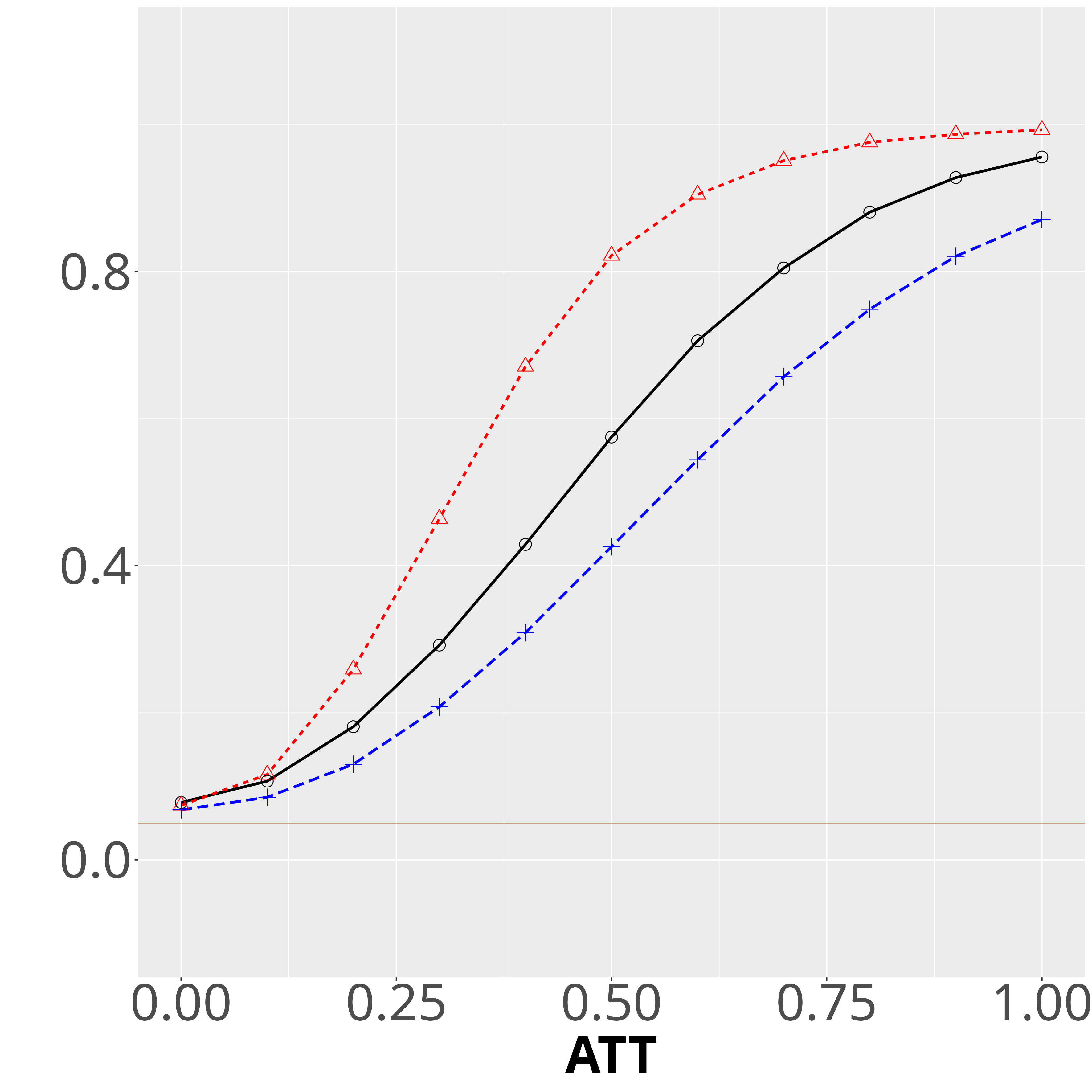}
\end{subfigure}
\begin{subfigure}{0.32\textwidth}
\centering
\caption{1\%}
\includegraphics[width=1\textwidth]{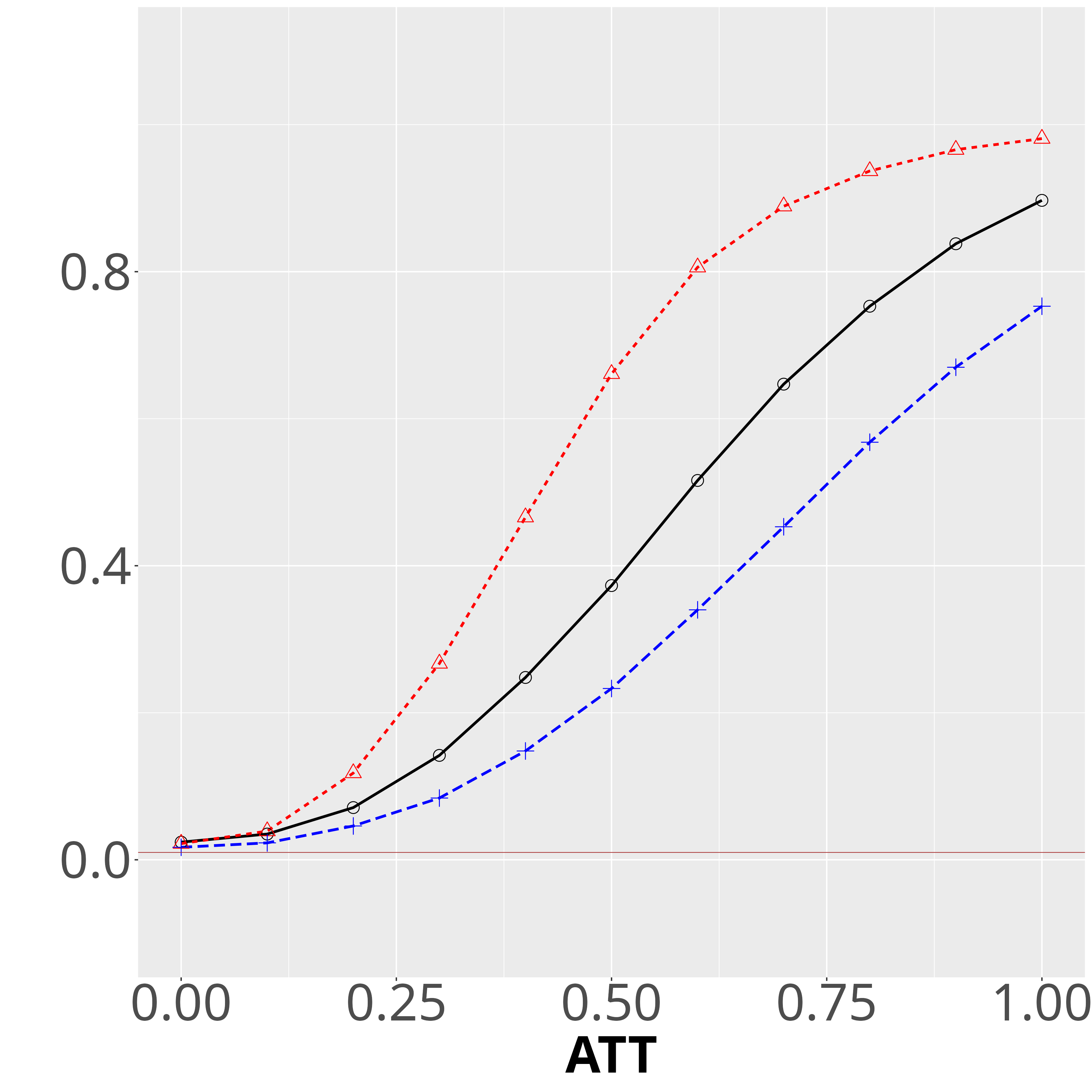}
\end{subfigure}
\caption{Power Curves - DGP MA(1), $\T,T=25$.}
\label{Fig:DGP_MA}
\end{figure}

\begin{figure}[!htbp]
\centering 
\begin{subfigure}{0.32\textwidth}
\centering
\caption{10\%}
\includegraphics[width=1\textwidth]{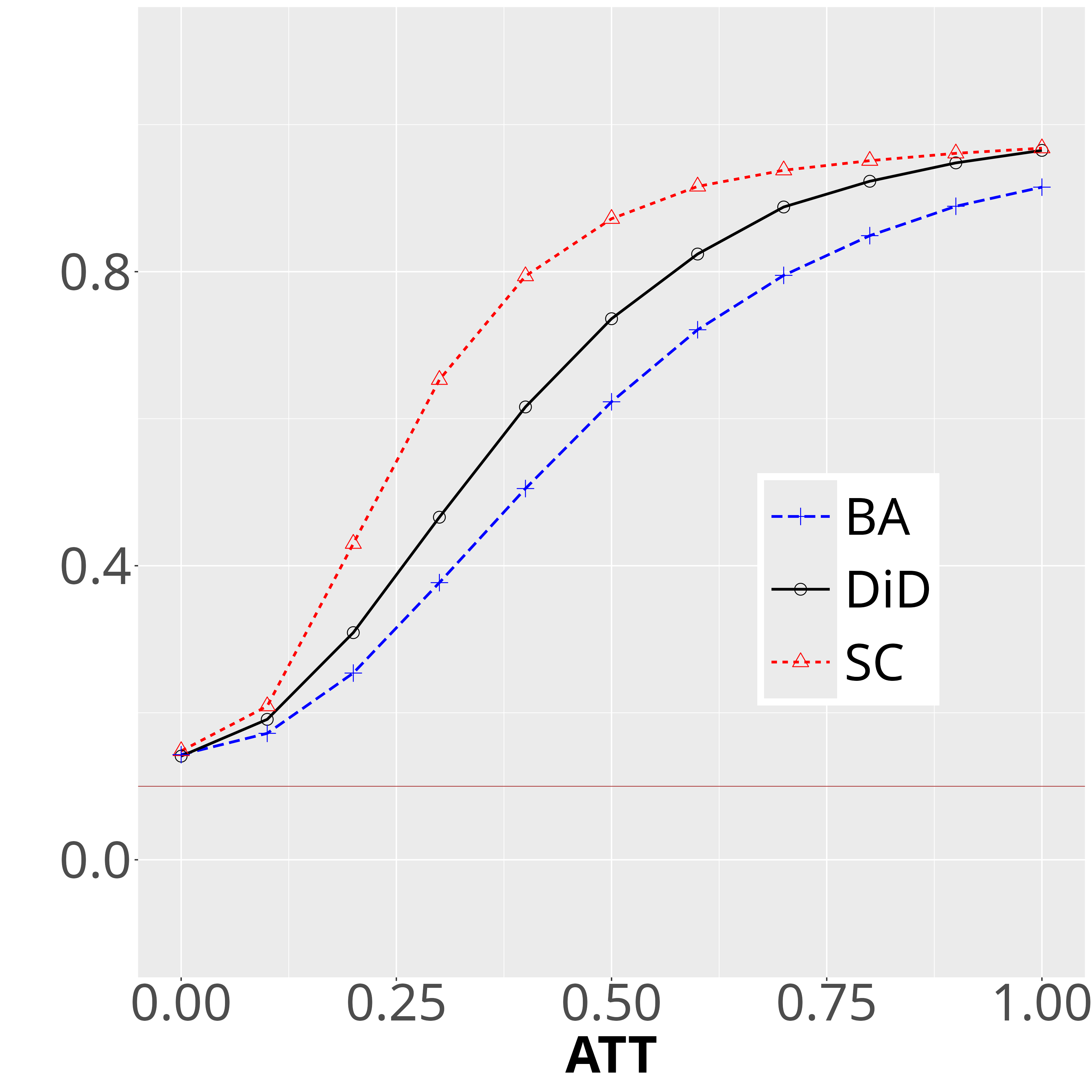}
\end{subfigure}
\begin{subfigure}{0.32\textwidth}
\centering
\caption{5\%}
\includegraphics[width=1\textwidth]{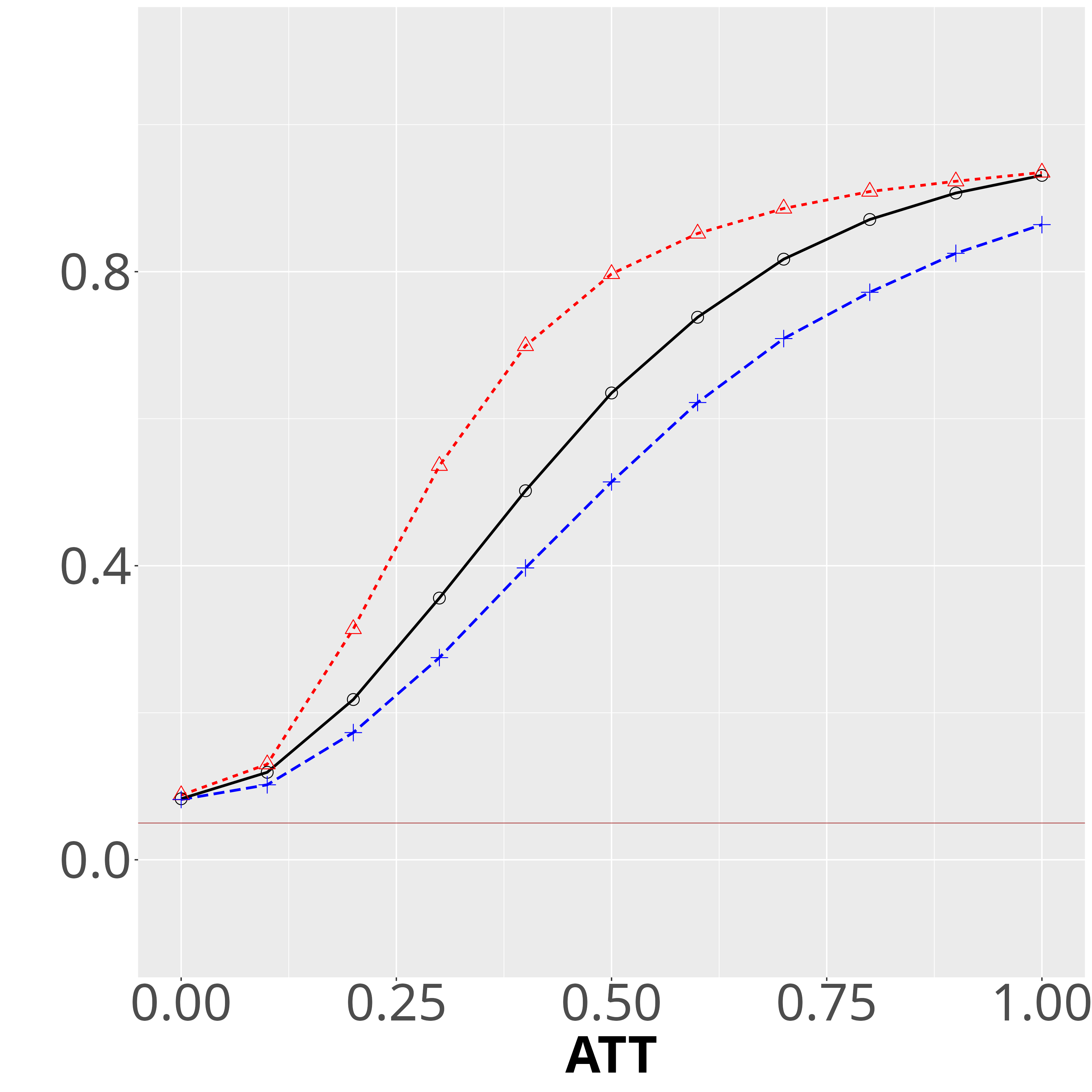}
\end{subfigure}
\begin{subfigure}{0.32\textwidth}
\centering
\caption{1\%}
\includegraphics[width=1\textwidth]{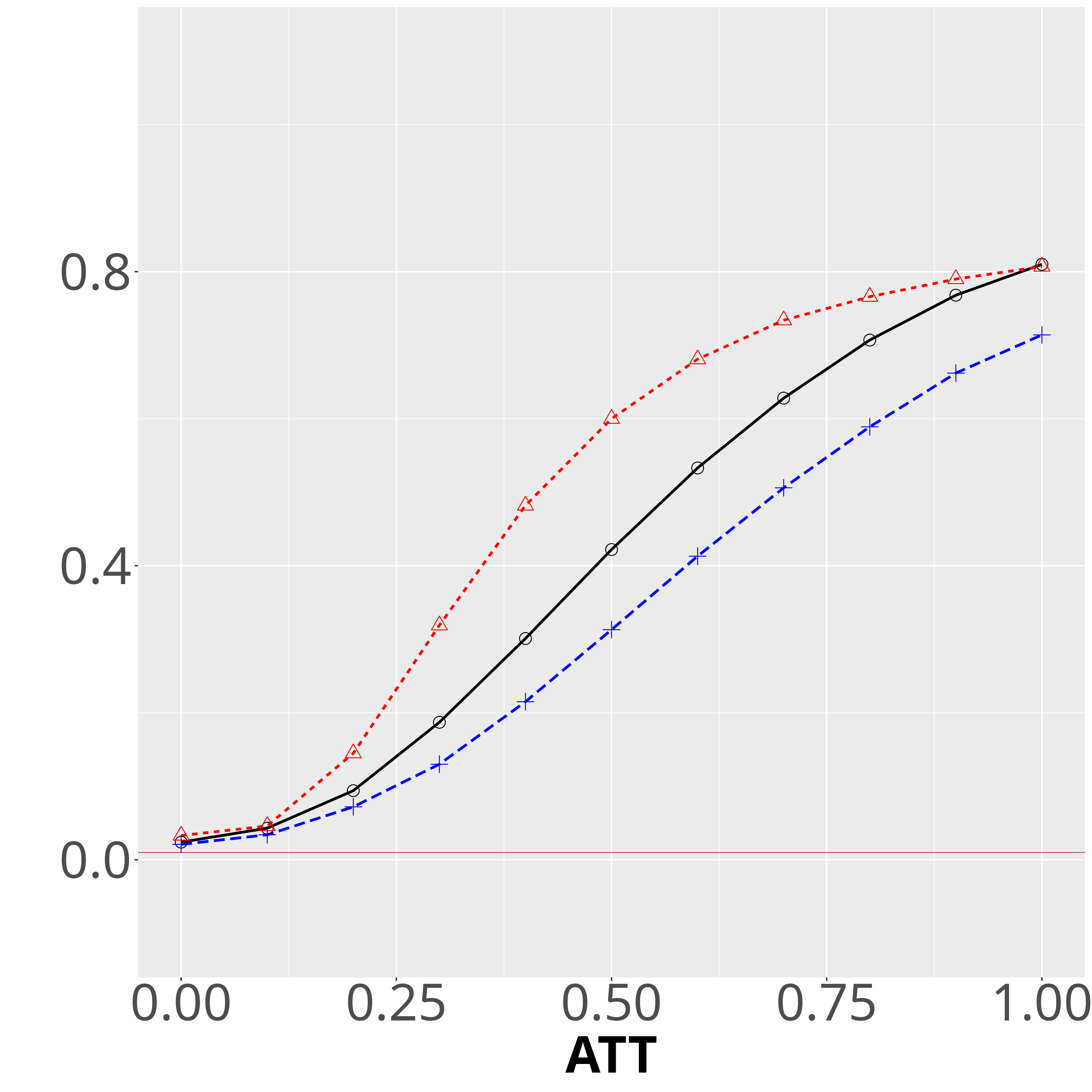}
\end{subfigure}
\caption{Power Curves - DGP AR(1), $\T,T=25$.}
\label{Fig:DGP_AR}
\end{figure}

\begin{figure}[!htbp]
\centering 
\begin{subfigure}{0.32\textwidth}
\centering
\caption{10\%}
\includegraphics[width=1\textwidth]{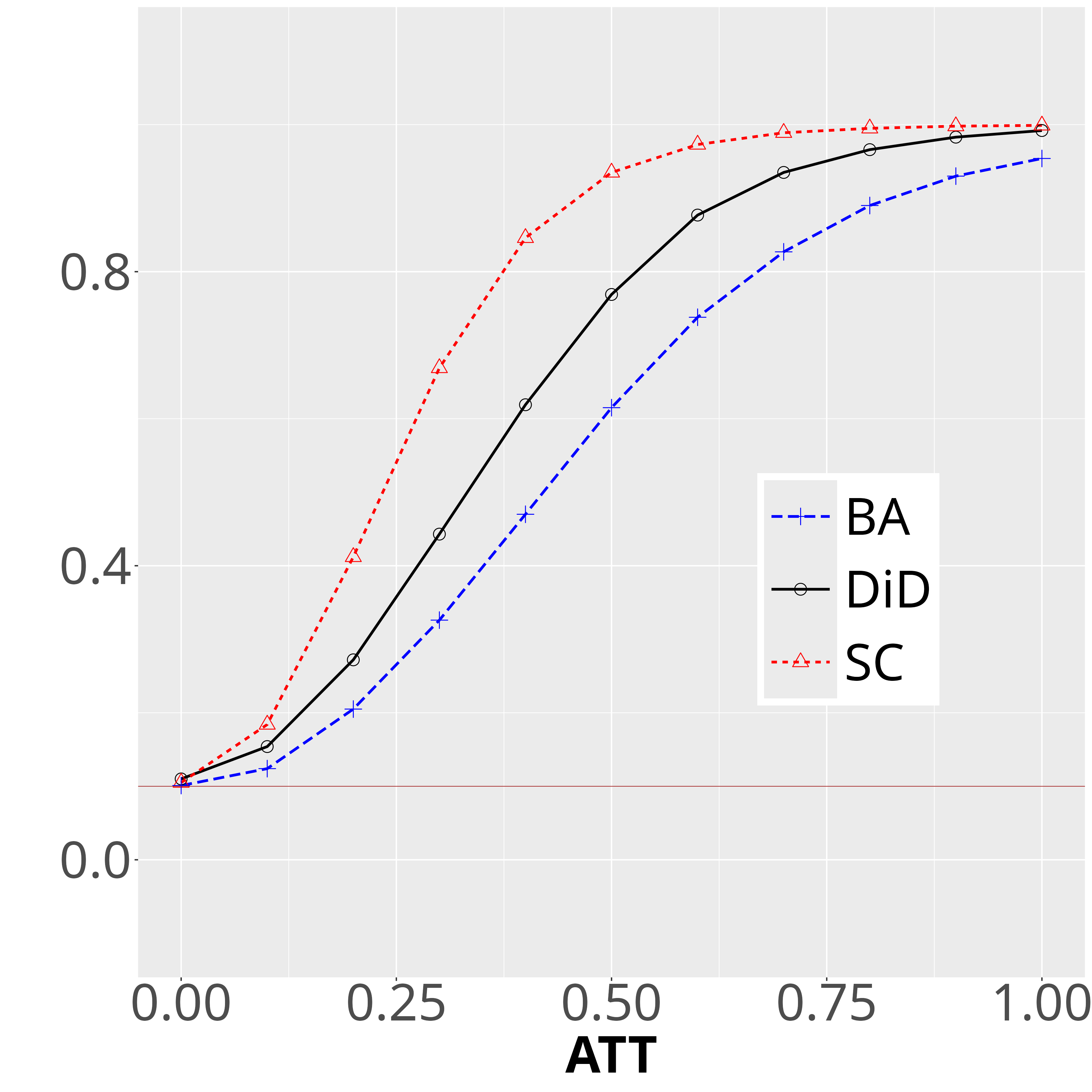}
\end{subfigure}
\begin{subfigure}{0.32\textwidth}
\centering
\caption{5\%}
\includegraphics[width=1\textwidth]{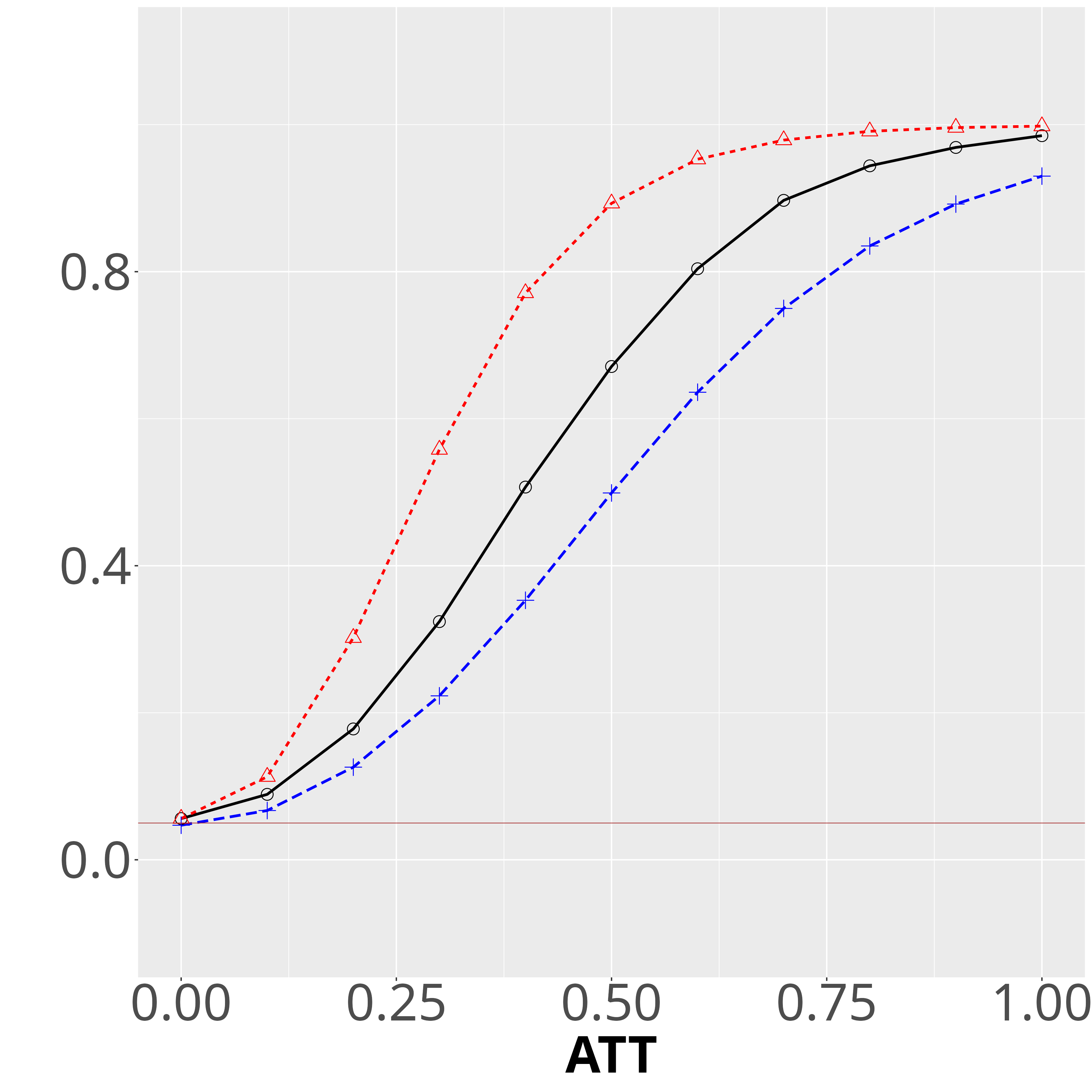}
\end{subfigure}
\begin{subfigure}{0.32\textwidth}
\centering
\caption{1\%}
\includegraphics[width=1\textwidth]{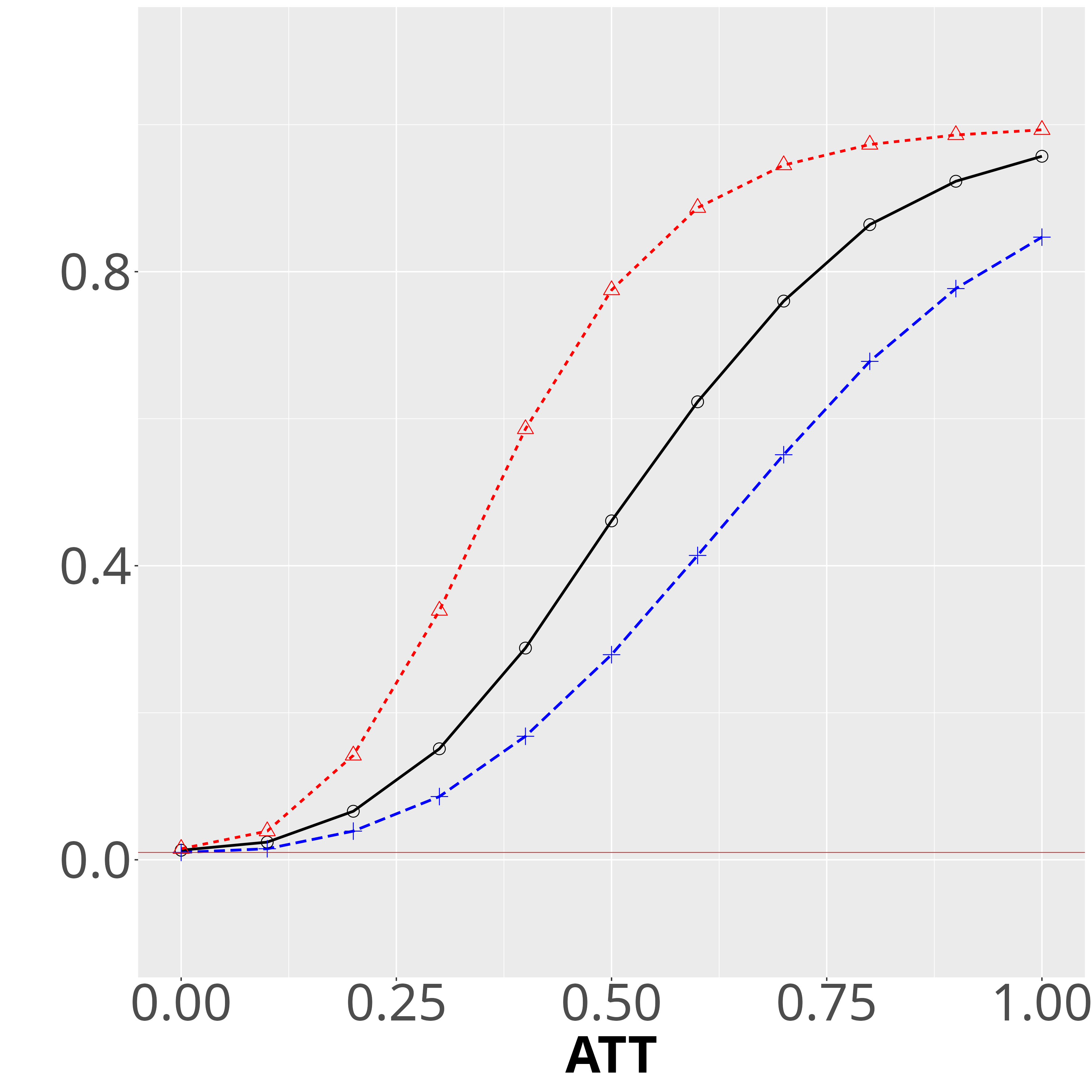}
\end{subfigure}
\caption{Power Curves - DGP U-R, $\T,T=25$.}
\label{Fig:DGP_Unit_Root}
\end{figure}

\begin{figure}[!htbp]
\centering 
\begin{subfigure}{0.32\textwidth}
\centering
\caption{10\%}
\includegraphics[width=1\textwidth]{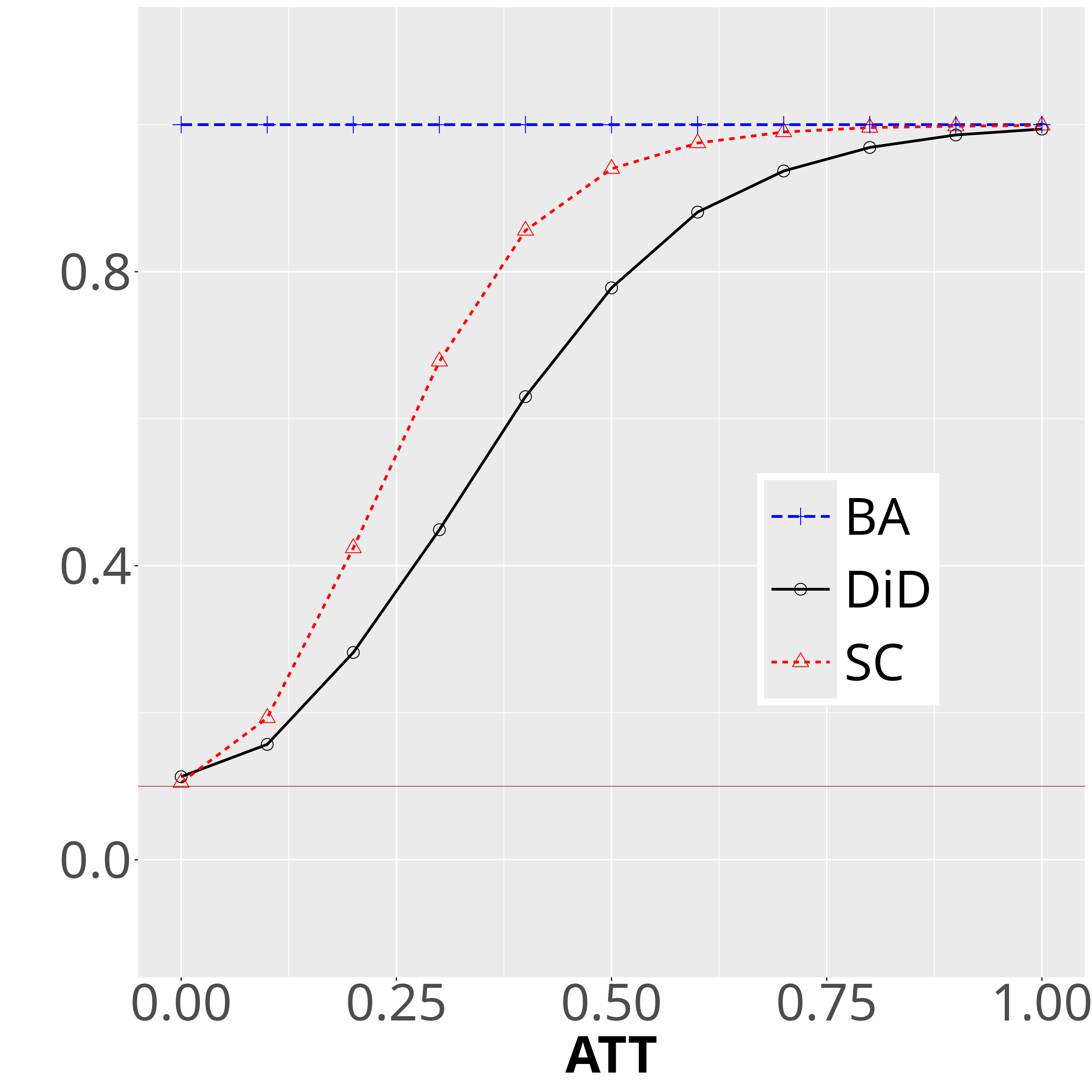}
\end{subfigure}
\begin{subfigure}{0.32\textwidth}
\centering
\caption{5\%}
\includegraphics[width=1\textwidth]{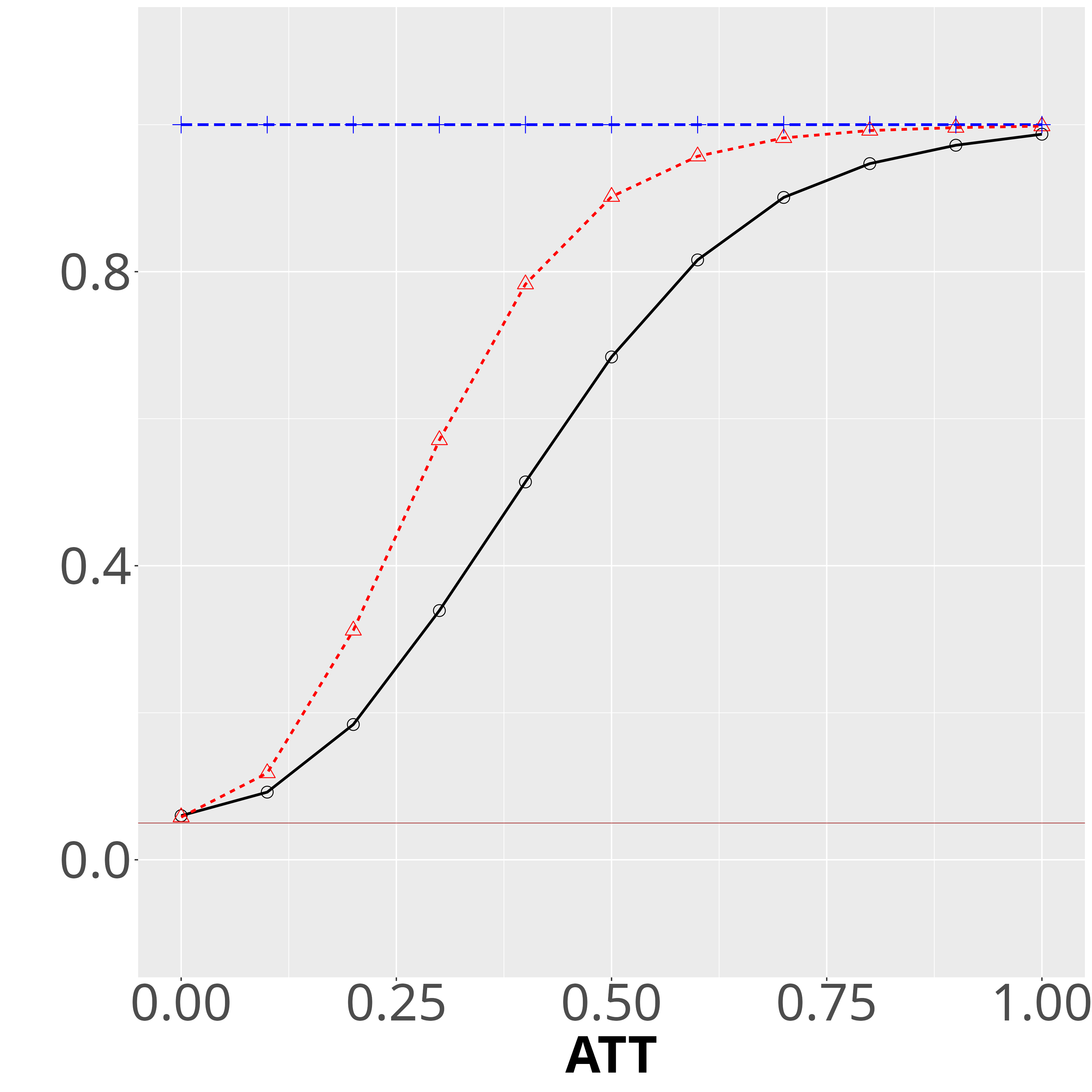}
\end{subfigure}
\begin{subfigure}{0.32\textwidth}
\centering
\caption{1\%}
\includegraphics[width=1\textwidth]{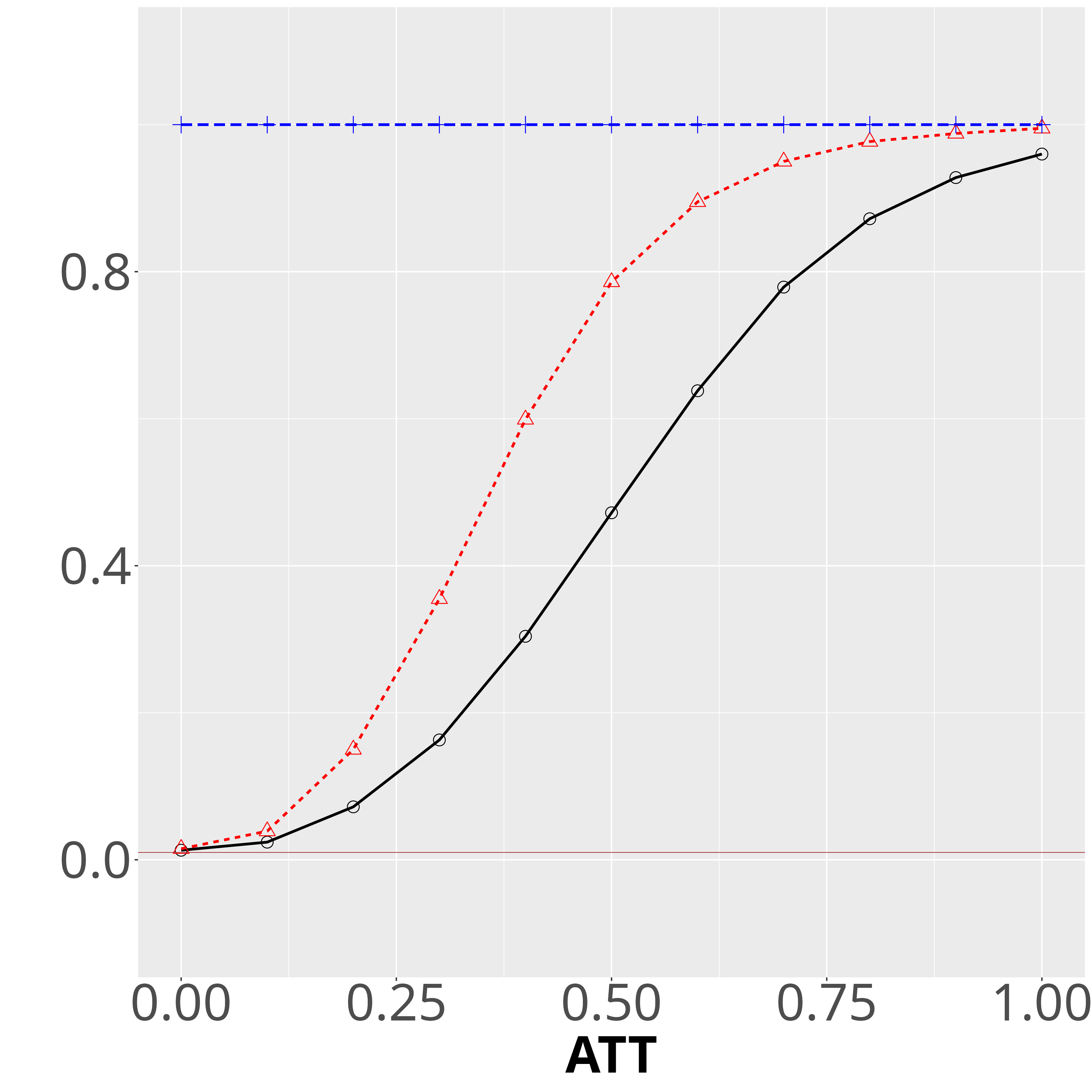}
\end{subfigure}
\caption{Power Curves - DGP Q-T, $\T,T=25$.}
\label{Fig:DGP_Q_Trend}
\end{figure}

\begin{figure}[!htbp]
\centering 
\begin{subfigure}{0.32\textwidth}
\centering
\caption{10\%}
\includegraphics[width=1\textwidth]{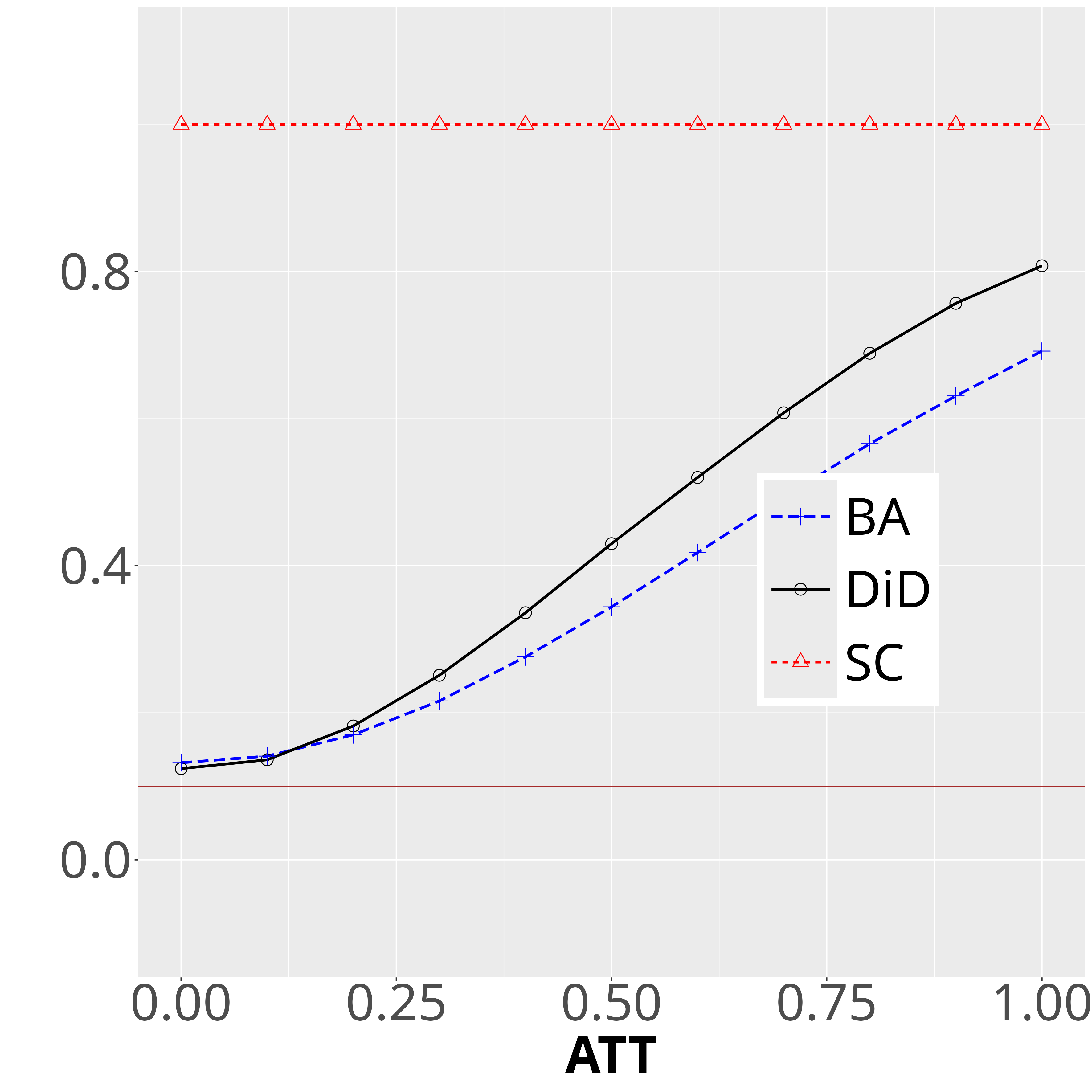}
\end{subfigure}
\begin{subfigure}{0.32\textwidth}
\centering
\caption{5\%}
\includegraphics[width=1\textwidth]{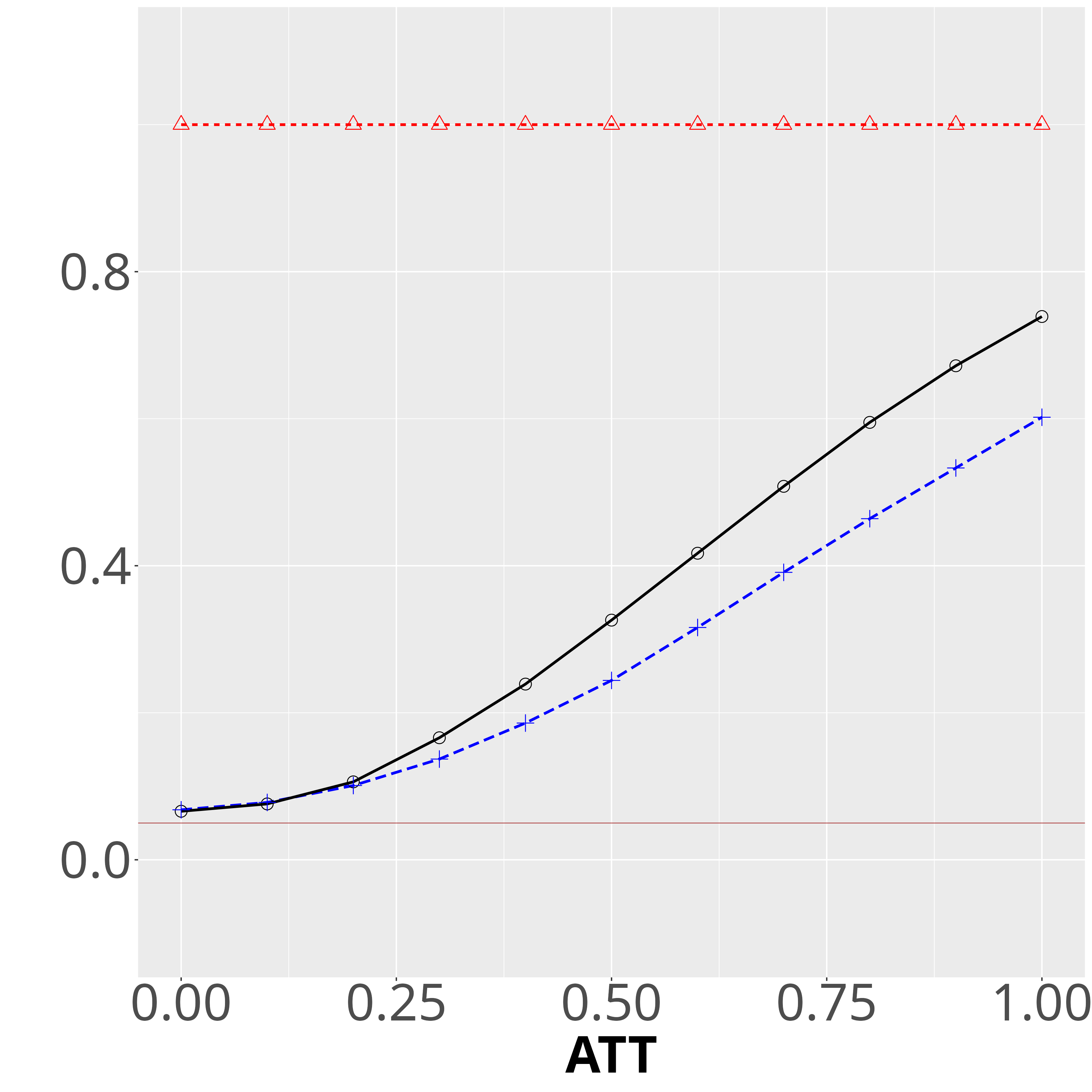}
\end{subfigure}
\begin{subfigure}{0.32\textwidth}
\centering
\caption{1\%}
\includegraphics[width=1\textwidth]{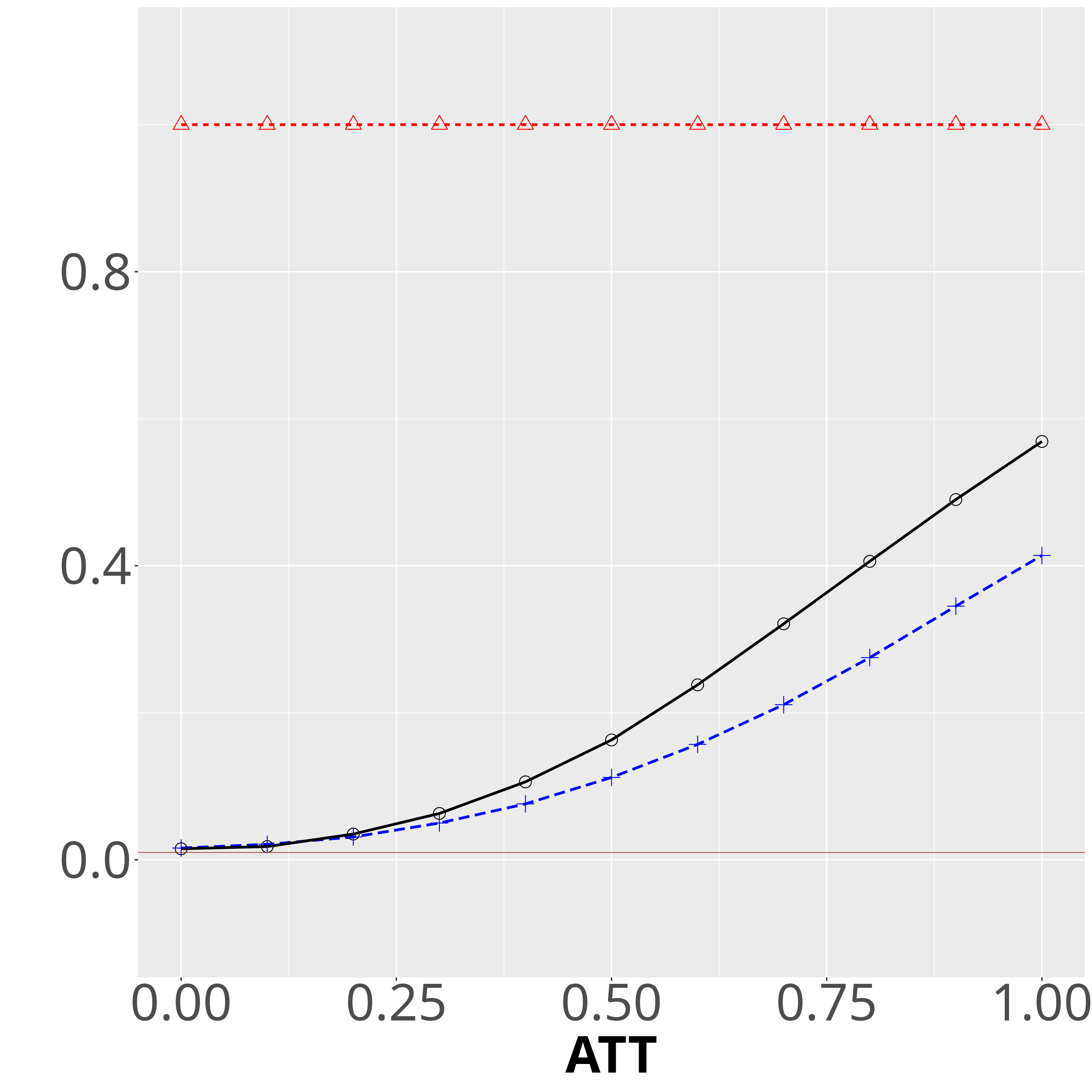}
\end{subfigure}
\caption{Power Curves - DGP T-T, $\T,T=25$.}
\label{Fig:DGP_T_Trend}
\end{figure}

\Cref{Fig:DGP_SC_BA,Fig:DGP_BA,Fig:DGP_SC,Fig:DGP_PT_NA_A,Fig:DGP_PT_NA_B,Fig:DGP_GARCH,Fig:DGP_MA,Fig:DGP_AR,Fig:DGP_Unit_Root,Fig:DGP_Q_Trend,Fig:DGP_T_Trend} present power curves corresponding to estimator-based $t$-tests of the hypotheses $ \Hyp_o: \ ATT_{\omega,T} = 0.0 \text{ vs. } \Hyp_a: ATT_{\omega,T} \neq 0.0 $. Overall, the power curves continue to reflect the robustness of the T-DiD across interesting settings and its ability to detect significant treatment effects. One also observes the failure of either the SC or BA to meaningfully control size under $\Hyp_o$, have power under $\Hyp_a$, or both whenever any of the stringent conditions under which they hold is violated.

\subsection{Testing identification}\label{App_Sub:Id_test}
This section examines the empirical sizes and power performance of the family of over-identifying restrictions and pre-trends tests. As both types of tests are based on estimates of $ATT$ on some (sub-)sample, this gives rise to over-identifying restrictions tests id.BA, id.DiD, and id.SC and pre-tests pt.BA, pt.DiD, and pt.SC, where the prefixes ``id" and ``pt" denote over-identifying restrictions and pre-trends tests, respectively, using the prefixed estimator SC, DiD, or BA. Pre-trends tests are conducted using approximately half of the pre-treatment periods as a ``pre-treatment" period and the other half as the ``post-treatment" period.

Modifications of DGPs PT-NA (A) and (B) via the term $\nu_t$ provide the following scenarios for examining over-identifying restrictions and pre-trends identification tests using simulations; 
\begin{enumerate}
    \item DGP idTest I : $\nu_t \sim \mathcal{N}\big(2.5h|t|^{-0.25},1\big) $ if $t\leq -1$ and $\nu_t=0$ otherwise, $\alpha_0 = -\alpha_1 = 0.5$; 
    \item DGP idTest II : $\nu_t \sim \mathcal{N}\big(2.5h|t|^{-0.25},1\big) $ if $t\geq 1$ and $\nu_t=0$ otherwise, $\alpha_0 = \alpha_1 = 0.5$, $ \varphi(t) = \sqrt{2}\indicator{t\geq 4} $; and
    \item DGP idTest III : $\nu_t \sim \mathcal{N}\big(2.5( 1 - h(1-\mathrm{sgn}(t))|t|^{-0.25},1\big) $, $\alpha_0 = -\alpha_1 = 0.5$, $ \varphi(t) = \sqrt{2}\indicator{t\geq 4} $.
\end{enumerate} $ATT=0.0$; thus, a ``treated" unit is as good as a ``control" unit. For the comparison of tests of identification, the $t$-test version of the over-identifying restrictions test is used since there are only two candidate controls. Thus, both units are candidate controls for the over-identifying restrictions tests; one unit serves as the  ``treated unit" while the other serves as a ``control unit" for pre-trends tests.  

In all three DGPs above, identification in the T-DiD context holds when $h=0$. Both id.DiD and pt.DiD should control size in DGPs idTest I and idTest II while only id.DiD is expected to control size in DGP idTest III. In fact, both \Cref{ass:parallel_trends,ass:limited_anticip} are individually violated in DGP idTest III, but DiD-identification holds -- see \Cref{rem:weak_suff_idcond}.  For $h\in (0,1]$ in (1) DGP idTest I, violations of DiD identification occur in the pre-treatment period (which can be picked up by the pt.DiD), (2) DGP idTest II, violations occur in the post-treatment period, which cannot be picked up by the pt.DiD, and (3) violations occur in both pre- and post-treatment periods.

\begin{figure}[!htbp]
\centering 
\begin{subfigure}{0.32\textwidth}
\centering
\caption{10\%}
\includegraphics[width=1\textwidth]{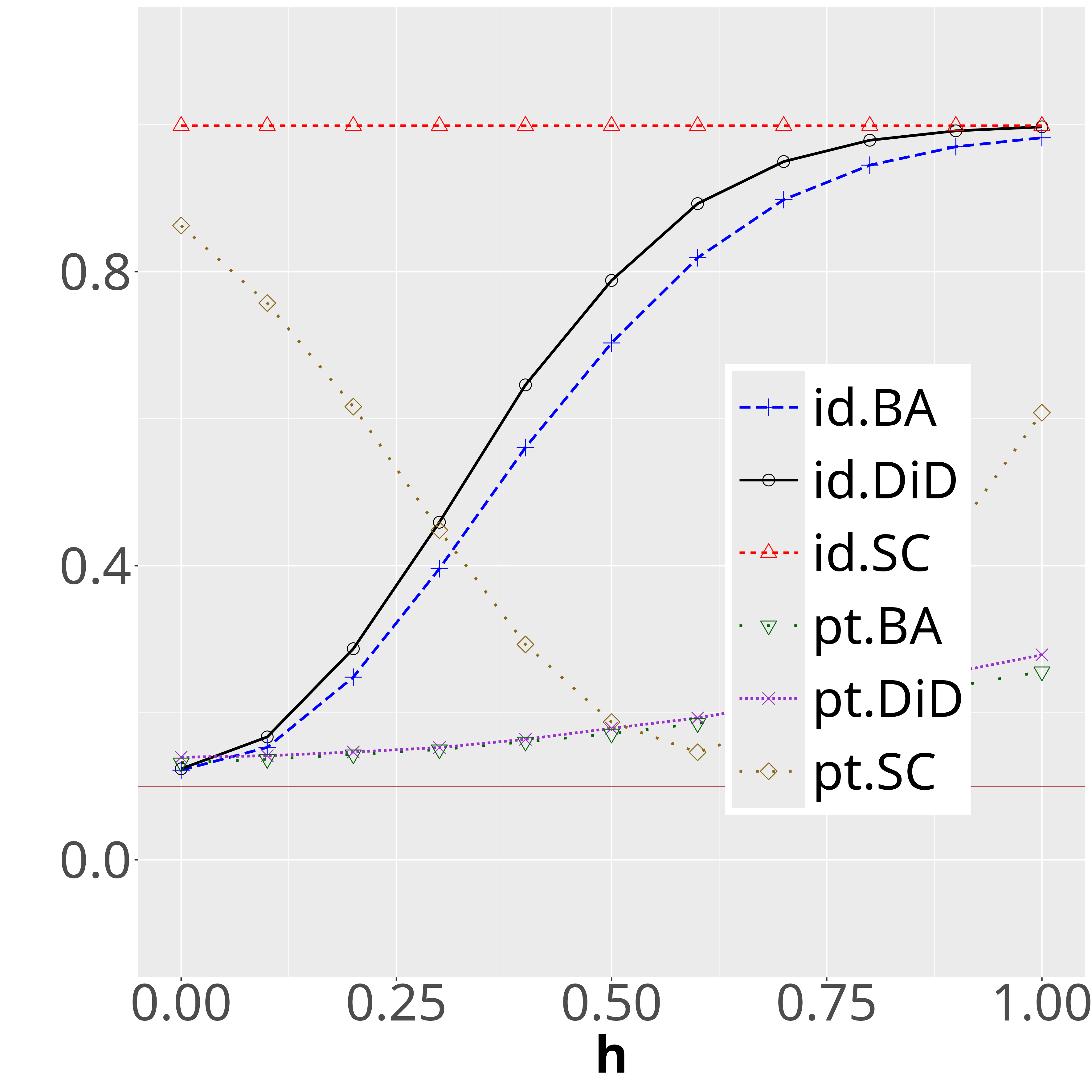}
\end{subfigure}
\begin{subfigure}{0.32\textwidth}
\centering
\caption{5\%}
\includegraphics[width=1\textwidth]{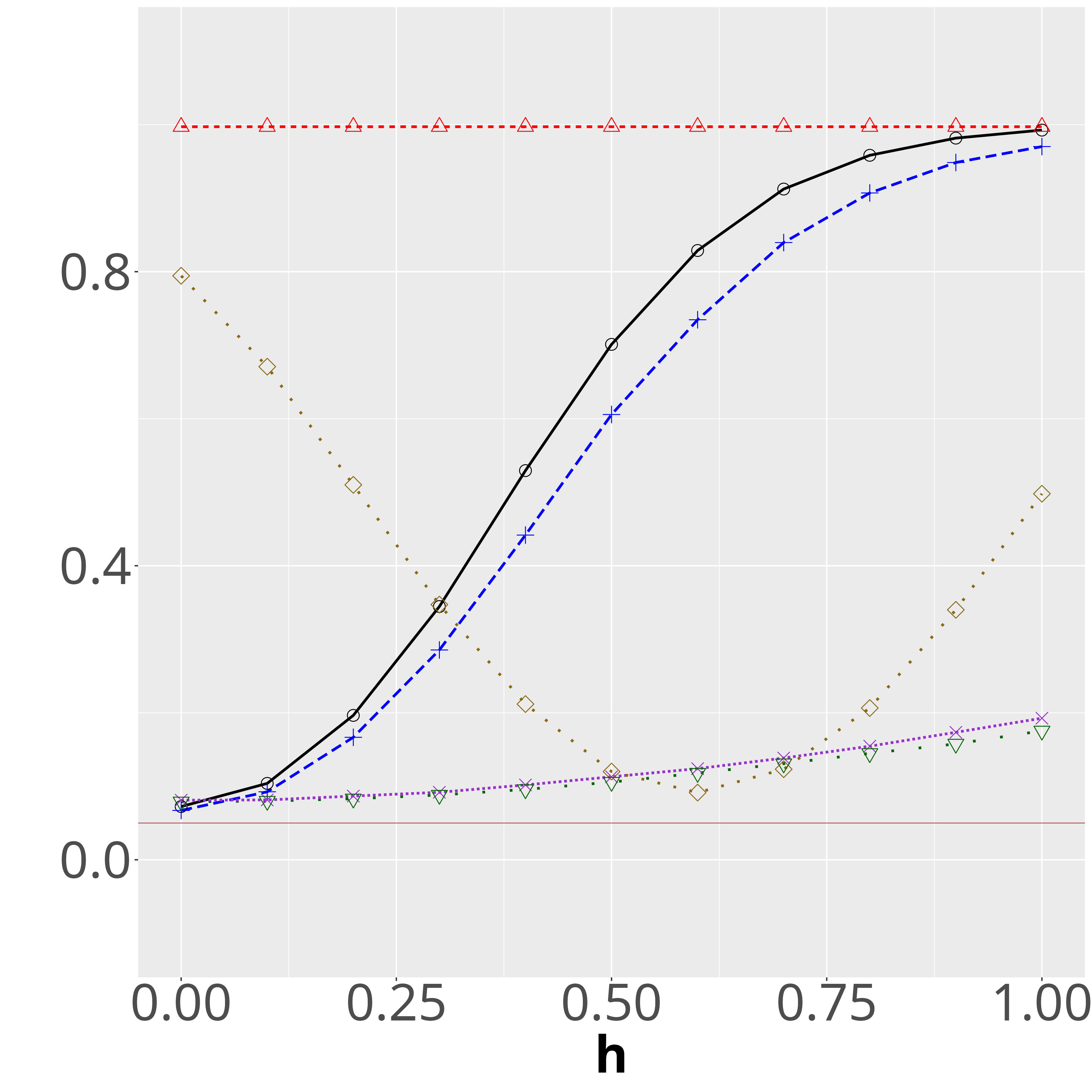}
\end{subfigure}
\begin{subfigure}{0.32\textwidth}
\centering
\caption{1\%}
\includegraphics[width=1\textwidth]{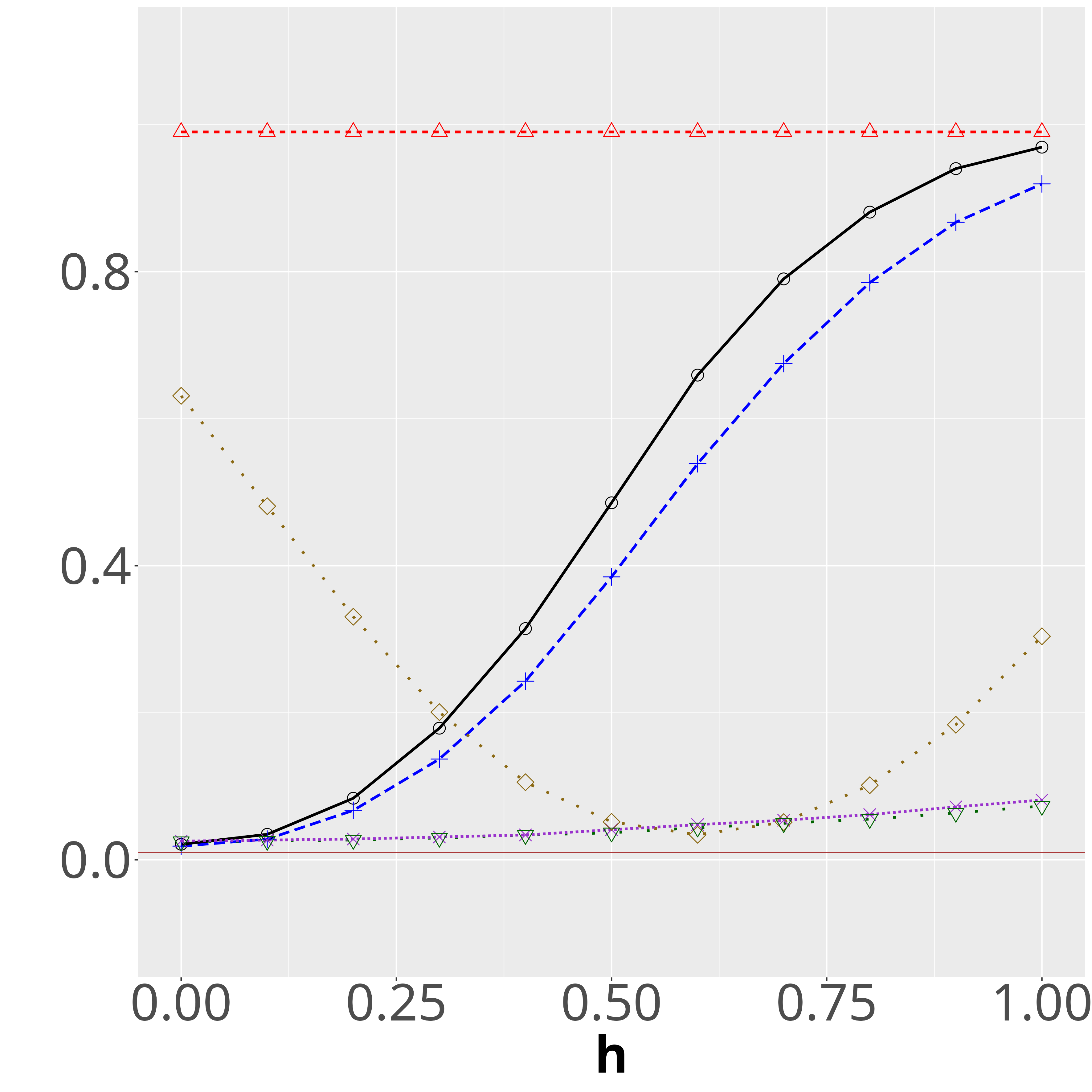}
\end{subfigure}
\caption{Power Curves - DGP idTest I, $\T,T=25$.}
\label{Fig:idtest_1}
\end{figure}

\begin{figure}[!htbp]
\centering 
\begin{subfigure}{0.32\textwidth}
\centering
\caption{10\%}
\includegraphics[width=1\textwidth]{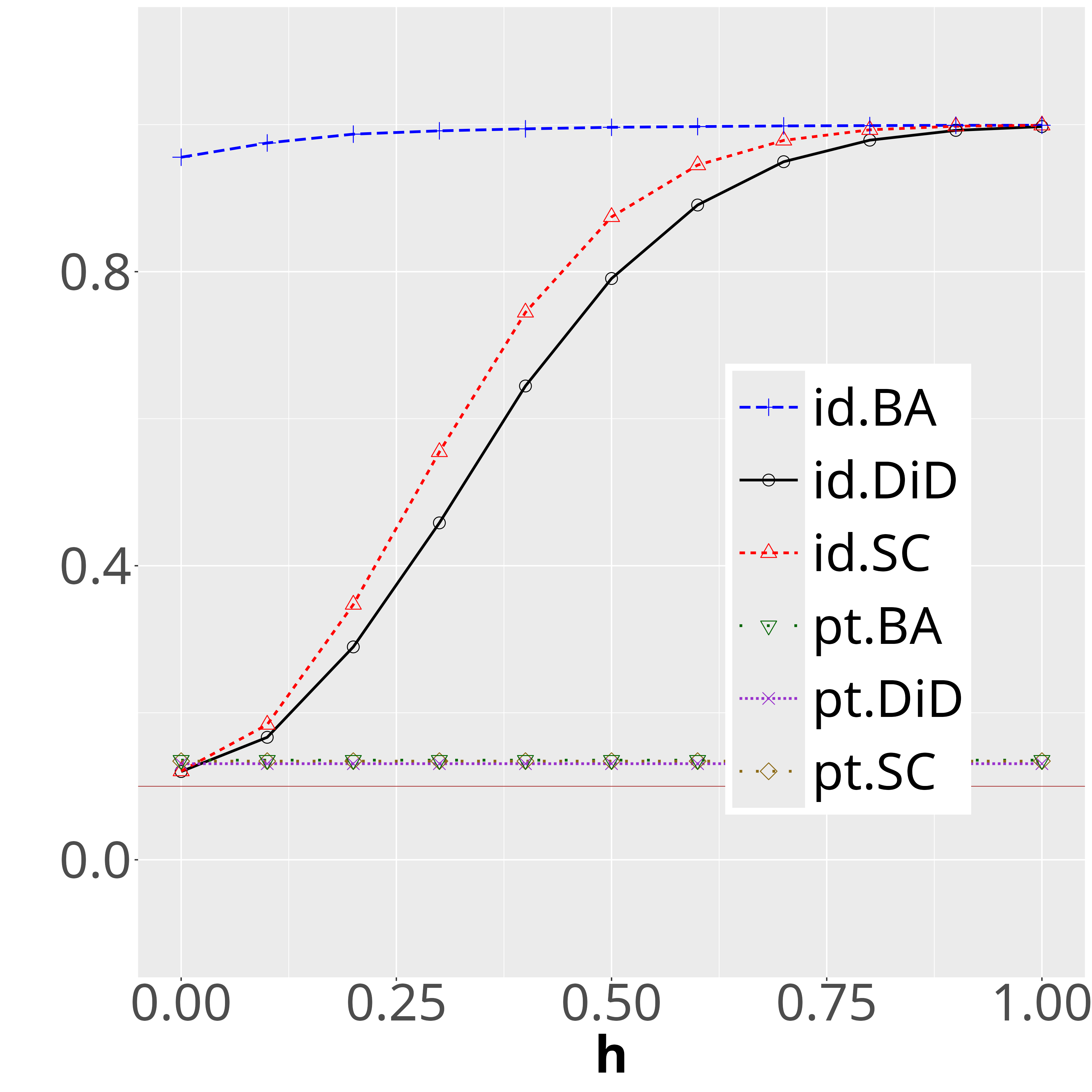}
\end{subfigure}
\begin{subfigure}{0.32\textwidth}
\centering
\caption{5\%}
\includegraphics[width=1\textwidth]{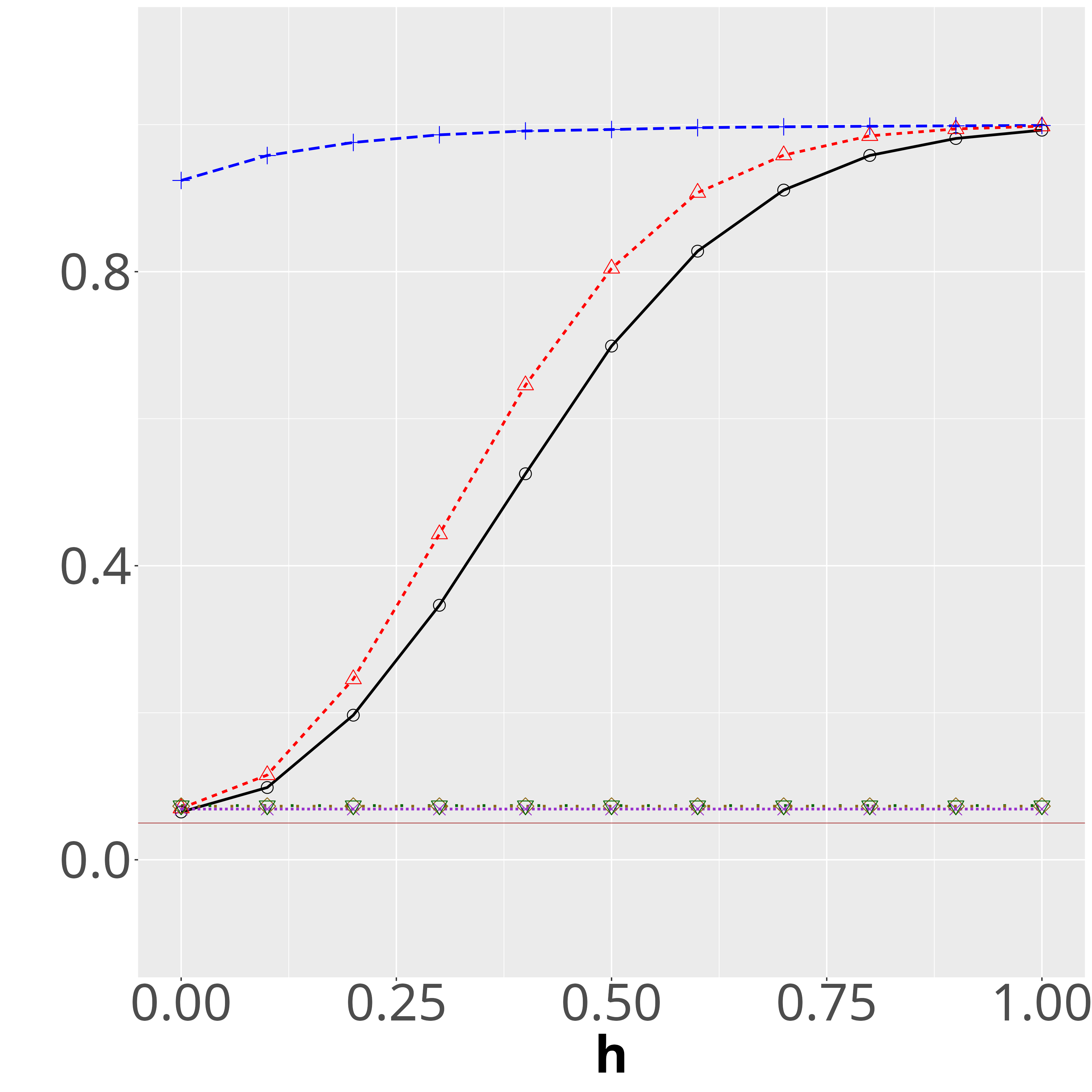}
\end{subfigure}
\begin{subfigure}{0.32\textwidth}
\centering
\caption{1\%}
\includegraphics[width=1\textwidth]{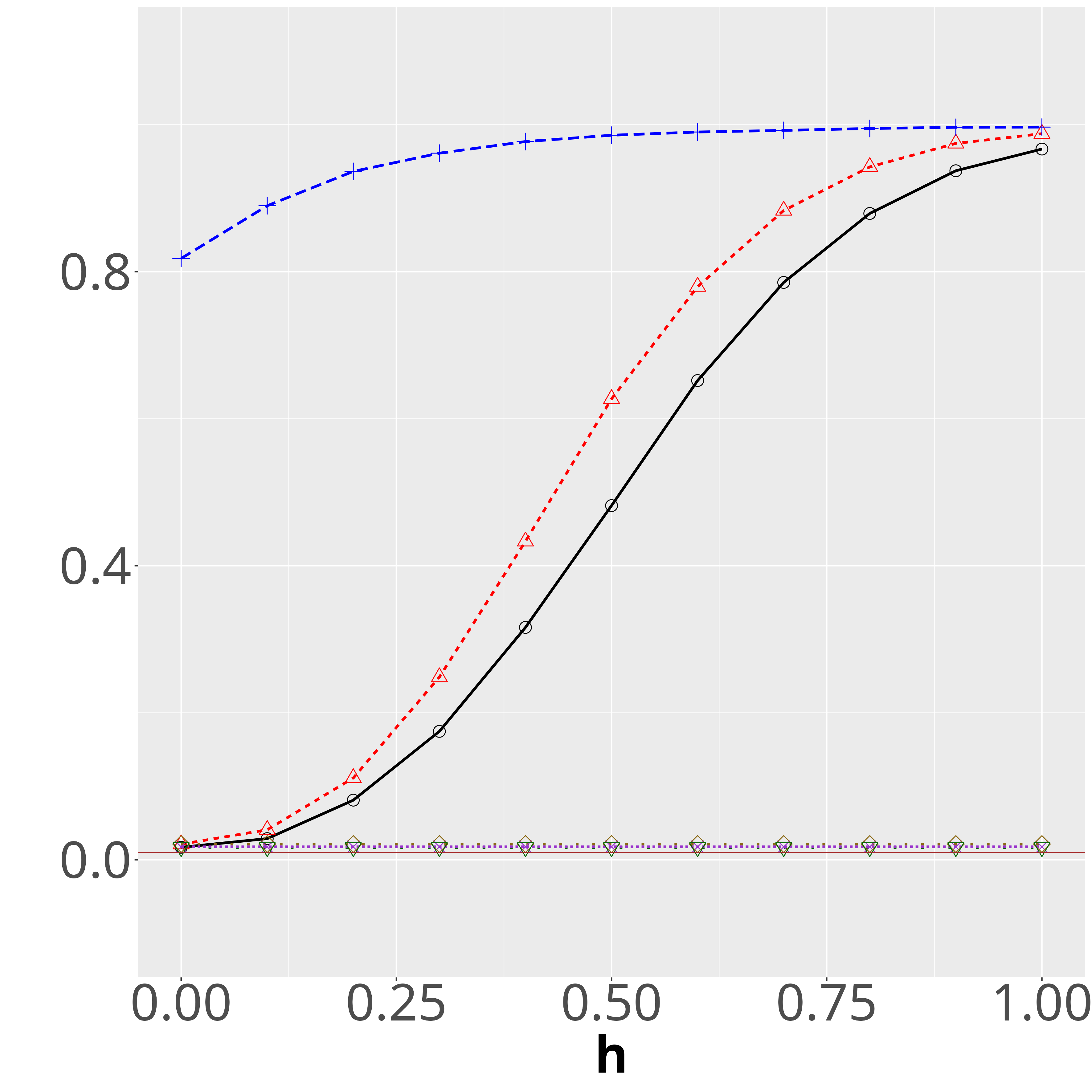}
\end{subfigure}
\caption{Power Curves - DGP idTest II, $\T,T=25$.}
\label{Fig:idtest_2}
\end{figure}

\begin{figure}[!htbp]
\centering 
\begin{subfigure}{0.32\textwidth}
\centering
\caption{10\%}
\includegraphics[width=1\textwidth]{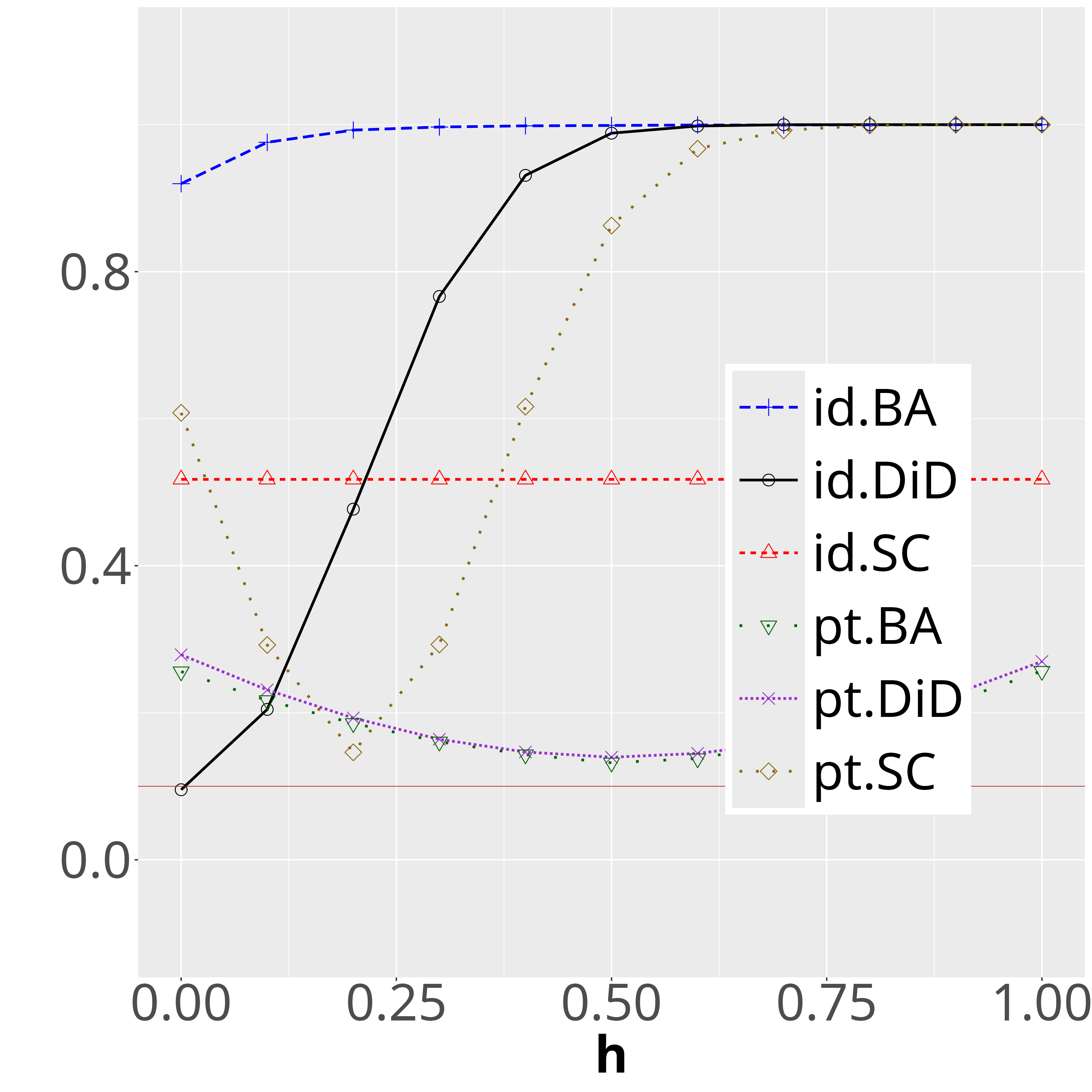}
\end{subfigure}
\begin{subfigure}{0.32\textwidth}
\centering
\caption{5\%}
\includegraphics[width=1\textwidth]{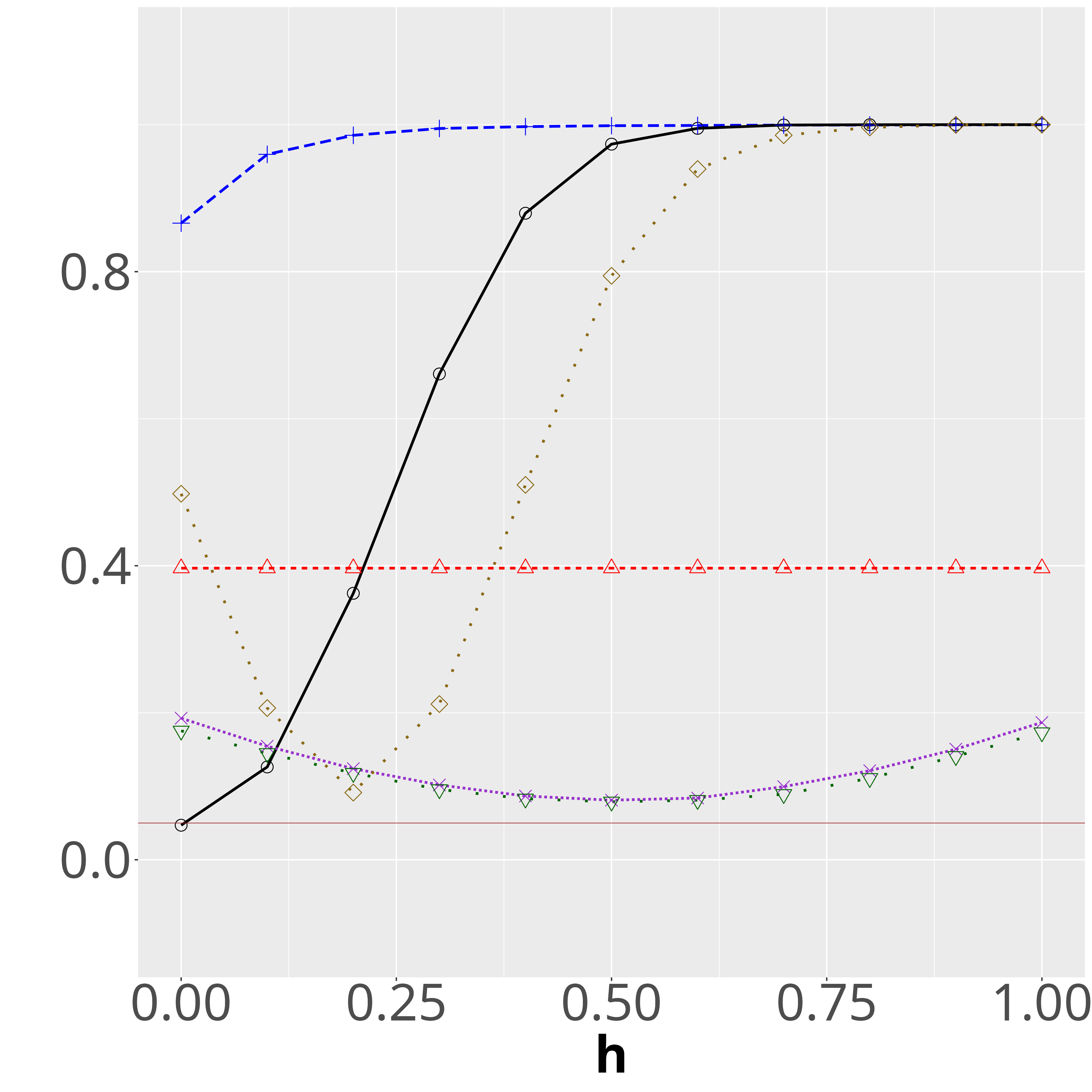}
\end{subfigure}
\begin{subfigure}{0.32\textwidth}
\centering
\caption{1\%}
\includegraphics[width=1\textwidth]{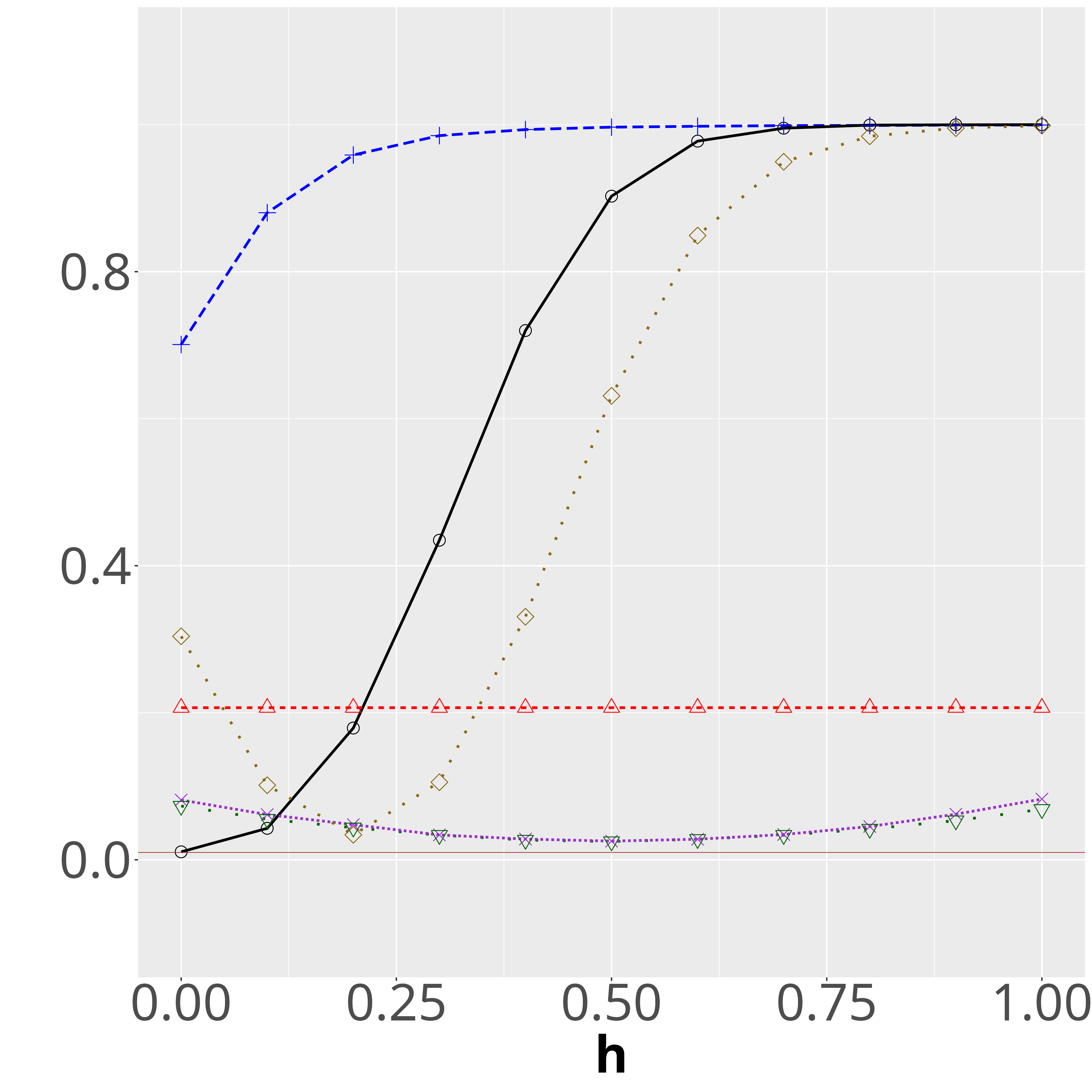}
\end{subfigure}
\caption{Power Curves - DGP idTest III, $\T,T=25$.}
\label{Fig:idtest_3}
\end{figure}

\Cref{Fig:idtest_1,Fig:idtest_2,Fig:idtest_3} present power curves to compare all six candidate tests of identification. As expected, over-identifying restrictions tests based on the BA and SC continue to bear the drawbacks of the respective estimators. For example, id.BA jointly controls size meaningfully and has non-trivial power in \Cref{Fig:idtest_1} but not in \Cref{Fig:idtest_2} whereas the opposite holds for id.SC. All pre-tests have trivial power, fail to control size meaningfully, or both in \Cref{Fig:idtest_1,Fig:idtest_2,Fig:idtest_3}. In sum, this short simulation exercise demonstrates the reliability of the DiD-based over-identifying restrictions test of identification even in scenarios, such as DGP idTest II where pre-tests do not help detect violations of ATT. 

\subsection{Heterogeneous Treatment Effects}\label{App_Sub:Het_Treat}
This section provides more simulation results. The DGPs remain the same as in \Cref{App_Sub:Bias_Size}, but $ATT(t)$ is heterogeneous in $t$ and $ATT_{\omega,T}$ varies with $T$:
\begin{equation*}
    ATT(t)=
    \begin{cases}
        \sin(t)/\pi & \text{ if } t\geq 1\\
        0 & \text{ otherwise}
    \end{cases}.
\end{equation*} With the choice of the uniform weighting scheme, the sequence of ATT parameters (indexed by $T$) is given by $ATT_{\omega,T} = (\pi T)^{-1}\sum_{t=1}^T \sin(t)$. This scenario helps to examine the performance of the DiD when the parameter of interest varies with $T$.

\begin{table}[!htbp]
\caption{Simulation - $ATT(t) = \sin(t)/\pi,\ t\geq 1 $}
\label{Tab:Sim_I_b}
\centering
\begin{tabular}{rlccccccccc}
\cmidrule[0.5pt](l){2-11}
 & $T,\T $ & \multicolumn{1}{c}{DiD} & \multicolumn{1}{c}{SC} & \multicolumn{1}{c}{BA} & \multicolumn{1}{c}{DiD} & \multicolumn{1}{c}{SC} & \multicolumn{1}{c}{BA} & \multicolumn{1}{c}{DiD} & \multicolumn{1}{c}{SC} & \multicolumn{1}{c}{BA} \\ \cmidrule[0.2pt](l){2-11}
\multirow{12}{*}{ \STAB{\rotatebox[origin=c]{90}{\underline{DGP SC-BA}}}} 
&  & \multicolumn{3}{c}{MB} & \multicolumn{3}{c}{MAD} & \multicolumn{3}{c}{RMSE} \\ \cmidrule[0.2pt](l){3-5} \cmidrule[0.2pt](l){6-8} \cmidrule[0.2pt](l){9-11}
      &25    &0.002 &0.002 &0.001 &0.140  &0.096 &0.175 &0.216 &0.151 &0.280  \\ 
       &50     &-0.001 &0.000      &-0.001 &0.101  &0.070   &0.129  &0.153  &0.107  &0.202  \\ 
       &100    &0.000      &0.000      &-0.001 &0.072  &0.051  &0.094  &0.108  &0.077  &0.142  \\ 
      &200   &0.001 &0.001 &0.002 &0.051 &0.036 &0.066 &0.076 &0.053 &0.100   \\ 
      &400   &0.000     &0.000     &0.000     &0.037 &0.025 &0.047 &0.054 &0.038 &0.070  \\ 
\cmidrule[0.2pt](l){3-5} \cmidrule[0.2pt](l){6-8} \cmidrule[0.2pt](l){9-11}

& & \multicolumn{3}{c}{Rej. 1\%} & \multicolumn{3}{c}{Rej. 5\%} & \multicolumn{3}{c}{Rej. 10\%} \\ \cmidrule[0.2pt](l){3-5} \cmidrule[0.2pt](l){6-8} \cmidrule[0.2pt](l){9-11}

      &25    &0.010  &0.007 &0.011 &0.048 &0.035 &0.050  &0.100   &0.077 &0.103 \\ 
      &50    &0.009 &0.004 &0.009 &0.043 &0.035 &0.049 &0.092 &0.077 &0.102 \\ 
      &100   &0.006 &0.005 &0.008 &0.042 &0.035 &0.048 &0.089 &0.08  &0.096 \\ 
      &200   &0.006 &0.004 &0.008 &0.041 &0.031 &0.044 &0.089 &0.071 &0.091 \\ 
      &400   &0.008 &0.005 &0.008 &0.041 &0.035 &0.043 &0.089 &0.076 &0.090  \\ \cmidrule[0.5pt](l){2-11}

\multirow{12}{*}{ \STAB{\rotatebox[origin=c]{90}{\underline{DGP BA}}}} 
&  & \multicolumn{3}{c}{MB} & \multicolumn{3}{c}{MAD} & \multicolumn{3}{c}{RMSE} \\ \cmidrule[0.2pt](l){3-5} \cmidrule[0.2pt](l){6-8} \cmidrule[0.2pt](l){9-11}
       &25     &0.002  &-0.998 &0.001  &0.140   &0.997  &0.175  &0.216  &1.009  &0.280   \\ 
       &50     &-0.001 &-1.000     &-0.001 &0.101  &1.000      &0.129  &0.153  &1.006  &0.202  \\ 
       &100    &0.000      &-1.000     &-0.001 &0.072  &1.000      &0.094  &0.108  &1.003  &0.142  \\ 
       &200    &0.001  &-0.999 &0.002  &0.051  &0.999  &0.066  &0.076  &1.001  &0.100    \\ 
      &400   &0.000     &-1.000    &0.000     &0.037 &1.000     &0.047 &0.054 &1.001 &0.070  \\
\cmidrule[0.2pt](l){3-5} \cmidrule[0.2pt](l){6-8} \cmidrule[0.2pt](l){9-11}

& & \multicolumn{3}{c}{Rej. 1\%} & \multicolumn{3}{c}{Rej. 5\%} & \multicolumn{3}{c}{Rej. 10\%} \\ \cmidrule[0.2pt](l){3-5} \cmidrule[0.2pt](l){6-8} \cmidrule[0.2pt](l){9-11}

      &25    &0.010  &0.989 &0.011 &0.048 &0.997 &0.050  &0.100   &0.999 &0.103 \\ 
      &50    &0.009 &1.000     &0.009 &0.043 &1.000     &0.049 &0.092 &1.000     &0.102 \\ 
      &100   &0.006 &1.000     &0.008 &0.042 &1.000     &0.048 &0.089 &1.000     &0.096 \\ 
      &200   &0.006 &1.000     &0.008 &0.041 &1.000     &0.044 &0.089 &1.000     &0.091 \\ 
      &400   &0.008 &1.000     &0.008 &0.041 &1.000     &0.043 &0.089 &1.000     &0.090  \\ 
      \cmidrule[0.5pt](l){2-11}

\multirow{12}{*}{ \STAB{\rotatebox[origin=c]{90}{\underline{DGP SC}}}} 
&  & \multicolumn{3}{c}{MB} & \multicolumn{3}{c}{MAD} & \multicolumn{3}{c}{RMSE} \\ \cmidrule[0.2pt](l){3-5} \cmidrule[0.2pt](l){6-8} \cmidrule[0.2pt](l){9-11}
      &25    &0.002 &0.002 &1.246 &0.140  &0.096 &1.248 &0.216 &0.151 &1.277 \\ 
       &50     &-0.001 &0.000      &1.328  &0.101  &0.07   &1.329  &0.153  &0.107  &1.343  \\ 
      &100   &0.000     &0.000     &1.371 &0.072 &0.051 &1.372 &0.108 &0.077 &1.378 \\ 
      &200   &0.001 &0.001 &1.395 &0.051 &0.036 &1.393 &0.076 &0.053 &1.398 \\ 
      &400   &0.000     &0.000     &1.404 &0.037 &0.025 &1.404 &0.054 &0.038 &1.406 \\
\cmidrule[0.2pt](l){3-5} \cmidrule[0.2pt](l){6-8} \cmidrule[0.2pt](l){9-11}

& & \multicolumn{3}{c}{Rej. 1\%} & \multicolumn{3}{c}{Rej. 5\%} & \multicolumn{3}{c}{Rej. 10\%} \\ \cmidrule[0.2pt](l){3-5} \cmidrule[0.2pt](l){6-8} \cmidrule[0.2pt](l){9-11}

      &25    &0.010  &0.007 &0.921 &0.048 &0.035 &0.965 &0.100   &0.077 &0.979 \\ 
      &50    &0.009 &0.004 &0.992 &0.043 &0.035 &0.996 &0.092 &0.077 &0.998 \\ 
      &100   &0.006 &0.005 &0.999 &0.042 &0.035 &1.000     &0.089 &0.080  &1.000     \\ 
      &200   &0.006 &0.004 &1.000     &0.041 &0.031 &1.000     &0.089 &0.071 &1.000     \\ 
      &400   &0.008 &0.005 &1.000     &0.041 &0.035 &1.000     &0.089 &0.076 &1.000     \\
      \cmidrule[0.5pt](l){2-11}
\end{tabular}
\end{table}

\begin{table}[!htbp]
\caption{Simulation - $ATT(t) = \sin(t)/\pi,\ t\geq 1 $}
\label{Tab:Sim_II_b}
\centering
\begin{tabular}{rlccccccccc}
\cmidrule[0.5pt](l){2-11}
 & $T,\T $ & \multicolumn{1}{c}{DiD} & \multicolumn{1}{c}{SC} & \multicolumn{1}{c}{BA} & \multicolumn{1}{c}{DiD} & \multicolumn{1}{c}{SC} & \multicolumn{1}{c}{BA} & \multicolumn{1}{c}{DiD} & \multicolumn{1}{c}{SC} & \multicolumn{1}{c}{BA} \\ \cmidrule[0.2pt](l){2-11}
\multirow{12}{*}{ \STAB{\rotatebox[origin=c]{90}{\underline{DGP PT-NA (A)}}}} 
&  & \multicolumn{3}{c}{MB} & \multicolumn{3}{c}{MAD} & \multicolumn{3}{c}{RMSE} \\ \cmidrule[0.2pt](l){3-5} \cmidrule[0.2pt](l){6-8} \cmidrule[0.2pt](l){9-11}
      &25    &0.177 &0.089 &0.176 &0.270  &0.179 &0.294 &0.397 &0.265 &0.436 \\ 
      &50    &0.109 &0.055 &0.109 &0.185 &0.127 &0.201 &0.273 &0.187 &0.302 \\ 
      &100   &0.063 &0.03  &0.062 &0.128 &0.087 &0.139 &0.19  &0.129 &0.212 \\ 
      &200   &0.037 &0.018 &0.037 &0.089 &0.061 &0.099 &0.132 &0.090  &0.147 \\ 
      &400   &0.023 &0.011 &0.023 &0.063 &0.043 &0.068 &0.092 &0.064 &0.102 \\ 
\cmidrule[0.2pt](l){3-5} \cmidrule[0.2pt](l){6-8} \cmidrule[0.2pt](l){9-11}

& & \multicolumn{3}{c}{Rej. 1\%} & \multicolumn{3}{c}{Rej. 5\%} & \multicolumn{3}{c}{Rej. 10\%} \\ \cmidrule[0.2pt](l){3-5} \cmidrule[0.2pt](l){6-8} \cmidrule[0.2pt](l){9-11}

      &25    &0.027 &0.025 &0.026 &0.089 &0.075 &0.086 &0.154 &0.132 &0.147 \\ 
      &50    &0.019 &0.017 &0.016 &0.072 &0.067 &0.072 &0.133 &0.12  &0.128 \\ 
      &100   &0.016 &0.014 &0.015 &0.064 &0.052 &0.068 &0.12  &0.099 &0.124 \\ 
      &200   &0.013 &0.010  &0.011 &0.061 &0.048 &0.058 &0.113 &0.100   &0.114 \\ 
      &400   &0.011 &0.010  &0.011 &0.054 &0.047 &0.053 &0.102 &0.096 &0.104 \\  \cmidrule[0.5pt](l){2-11}

\multirow{12}{*}{ \STAB{\rotatebox[origin=c]{90}{\underline{DGP PT-NA (B)}}}} 
&  & \multicolumn{3}{c}{MB} & \multicolumn{3}{c}{MAD} & \multicolumn{3}{c}{RMSE} \\ \cmidrule[0.2pt](l){3-5} \cmidrule[0.2pt](l){6-8} \cmidrule[0.2pt](l){9-11}
      &25    &0.001 &0.287 &0.000     &0.241 &0.294 &0.269 &0.356 &0.381 &0.399 \\ 
      &50    &0.001 &0.246 &0.001 &0.168 &0.248 &0.187 &0.25  &0.304 &0.282 \\ 
       &100    &-0.002 &0.205  &-0.002 &0.121  &0.206  &0.134  &0.18   &0.241  &0.202  \\ 
       &200    &-0.001 &0.175  &-0.001 &0.085  &0.175  &0.095  &0.126  &0.196  &0.142  \\ 
      &400   &0.001 &0.148 &0.001 &0.06  &0.148 &0.067 &0.089 &0.161 &0.099 \\
\cmidrule[0.2pt](l){3-5} \cmidrule[0.2pt](l){6-8} \cmidrule[0.2pt](l){9-11}

& & \multicolumn{3}{c}{Rej. 1\%} & \multicolumn{3}{c}{Rej. 5\%} & \multicolumn{3}{c}{Rej. 10\%} \\ \cmidrule[0.2pt](l){3-5} \cmidrule[0.2pt](l){6-8} \cmidrule[0.2pt](l){9-11}

      &25    &0.015 &0.1   &0.014 &0.06  &0.227 &0.058 &0.112 &0.327 &0.115 \\ 
      &50    &0.012 &0.13  &0.01  &0.051 &0.288 &0.05  &0.100   &0.397 &0.105 \\ 
      &100   &0.011 &0.17  &0.009 &0.05  &0.361 &0.055 &0.097 &0.482 &0.105 \\ 
      &200   &0.009 &0.259 &0.01  &0.048 &0.486 &0.045 &0.099 &0.610  &0.096 \\ 
      &400   &0.009 &0.392 &0.008 &0.045 &0.633 &0.045 &0.095 &0.749 &0.098 \\   \cmidrule[0.5pt](l){2-11}

\multirow{12}{*}{ \STAB{\rotatebox[origin=c]{90}{\underline{DGP GARCH(1,1)}}}} 
&  & \multicolumn{3}{c}{MB} & \multicolumn{3}{c}{MAD} & \multicolumn{3}{c}{RMSE} \\ \cmidrule[0.2pt](l){3-5} \cmidrule[0.2pt](l){6-8} \cmidrule[0.2pt](l){9-11}
      &25     &0.004  &0.001  &-0.003 &0.141  &0.098  &0.177  &0.214  &0.154  &0.275  \\ 
      &50    &0.001 &0.000     &0.002 &0.102 &0.072 &0.13  &0.152 &0.107 &0.199 \\ 
      &100   &0.000     &0.000     &0.000     &0.071 &0.051 &0.092 &0.108 &0.077 &0.141 \\ 
      &200   &0.000     &0.000     &0.000     &0.052 &0.037 &0.066 &0.077 &0.054 &0.099 \\ 
      &400   &0.000     &0.000     &0.000     &0.036 &0.026 &0.047 &0.054 &0.038 &0.071 \\  
\cmidrule[0.2pt](l){3-5} \cmidrule[0.2pt](l){6-8} \cmidrule[0.2pt](l){9-11}

& & \multicolumn{3}{c}{Rej. 1\%} & \multicolumn{3}{c}{Rej. 5\%} & \multicolumn{3}{c}{Rej. 10\%} \\ \cmidrule[0.2pt](l){3-5} \cmidrule[0.2pt](l){6-8} \cmidrule[0.2pt](l){9-11}

      &25    &0.013 &0.012 &0.010  &0.055 &0.051 &0.05  &0.106 &0.098 &0.102 \\ 
      &50    &0.010  &0.006 &0.009 &0.045 &0.037 &0.048 &0.094 &0.081 &0.097 \\ 
      &100   &0.008 &0.007 &0.008 &0.047 &0.039 &0.046 &0.094 &0.081 &0.091 \\ 
      &200   &0.008 &0.005 &0.007 &0.045 &0.036 &0.041 &0.087 &0.079 &0.088 \\ 
      &400   &0.006 &0.005 &0.008 &0.041 &0.035 &0.042 &0.086 &0.08  &0.092 \\
      \cmidrule[0.5pt](l){2-11}
\end{tabular}
\end{table}

\begin{table}[!htbp]
\caption{Simulation - $ATT(t) = \sin(t)/\pi,\ t\geq 1 $}
\label{Tab:Sim_III_b}
\centering
\begin{tabular}{rlccccccccc}
\cmidrule[0.5pt](l){2-11}
 & $T,\T $ & \multicolumn{1}{c}{DiD} & \multicolumn{1}{c}{SC} & \multicolumn{1}{c}{BA} & \multicolumn{1}{c}{DiD} & \multicolumn{1}{c}{SC} & \multicolumn{1}{c}{BA} & \multicolumn{1}{c}{DiD} & \multicolumn{1}{c}{SC} & \multicolumn{1}{c}{BA} \\ \cmidrule[0.2pt](l){2-11}
 \multirow{12}{*}{ \STAB{\rotatebox[origin=c]{90}{\underline{DGP MA(1)}}}} 
&  & \multicolumn{3}{c}{MB} & \multicolumn{3}{c}{MAD} & \multicolumn{3}{c}{RMSE} \\ \cmidrule[0.2pt](l){3-5} \cmidrule[0.2pt](l){6-8} \cmidrule[0.2pt](l){9-11}
       &25     &0.005  &0.003  &-0.002 &0.174  &0.119  &0.218  &0.269  &0.19   &0.347  \\ 
       &50     &-0.001 &0.000      &-0.002 &0.128  &0.088  &0.164  &0.192  &0.134  &0.251  \\ 
       &100    &-0.001 &0.000      &-0.001 &0.091  &0.065  &0.118  &0.135  &0.096  &0.178  \\ 
      &200   &0.001 &0.001 &0.001 &0.064 &0.045 &0.083 &0.095 &0.067 &0.125 \\ 
      &400   &0.000     &0.000     &0.000     &0.045 &0.032 &0.058 &0.067 &0.048 &0.088 \\ 
\cmidrule[0.2pt](l){3-5} \cmidrule[0.2pt](l){6-8} \cmidrule[0.2pt](l){9-11}

& & \multicolumn{3}{c}{Rej. 1\%} & \multicolumn{3}{c}{Rej. 5\%} & \multicolumn{3}{c}{Rej. 10\%} \\ \cmidrule[0.2pt](l){3-5} \cmidrule[0.2pt](l){6-8} \cmidrule[0.2pt](l){9-11}

      &25    &0.022 &0.018 &0.015 &0.072 &0.061 &0.065 &0.132 &0.113 &0.122 \\ 
      &50    &0.014 &0.009 &0.015 &0.066 &0.051 &0.064 &0.121 &0.102 &0.118 \\ 
      &100   &0.012 &0.010  &0.010  &0.058 &0.052 &0.054 &0.114 &0.103 &0.113 \\ 
      &200   &0.013 &0.010  &0.012 &0.054 &0.051 &0.054 &0.105 &0.101 &0.109 \\ 
      &400   &0.011 &0.011 &0.009 &0.051 &0.049 &0.053 &0.103 &0.100   &0.110  \\
      \cmidrule[0.5pt](l){2-11}
      
\multirow{12}{*}{ \STAB{\rotatebox[origin=c]{90}{\underline{DGP AR(1)}}}} 
&  & \multicolumn{3}{c}{MB} & \multicolumn{3}{c}{MAD} & \multicolumn{3}{c}{RMSE} \\ \cmidrule[0.2pt](l){3-5} \cmidrule[0.2pt](l){6-8} \cmidrule[0.2pt](l){9-11}
& 25  & -0.007 & -0.029 & -0.007 & 0.160 & 0.108 & 0.206 & 0.248 & 0.177 & 0.321 \\ 
& 50  & -0.006 & -0.016 & -0.005 & 0.110 & 0.075 & 0.144 & 0.166 & 0.116 & 0.216 \\ 
& 100 & -0.003 & -0.009 & -0.003 & 0.076 & 0.053 & 0.100 & 0.112 & 0.080 & 0.146 \\ 
& 200 & -0.001 & -0.004 &  0.000 & 0.052 & 0.036 & 0.068 & 0.078 & 0.055 & 0.101 \\ 
& 400 & -0.001 & -0.002 & -0.001 & 0.036 & 0.026 & 0.047 & 0.054 & 0.038 & 0.071 \\  
\cmidrule[0.2pt](l){3-5} \cmidrule[0.2pt](l){6-8} \cmidrule[0.2pt](l){9-11}

& & \multicolumn{3}{c}{Rej. 1\%} & \multicolumn{3}{c}{Rej. 5\%} & \multicolumn{3}{c}{Rej. 10\%} \\ \cmidrule[0.2pt](l){3-5} \cmidrule[0.2pt](l){6-8} \cmidrule[0.2pt](l){9-11}

& 25  & 0.019 & 0.021 & 0.017 & 0.068 & 0.065 & 0.071 & 0.124 & 0.115 & 0.130 \\ 
& 50  & 0.013 & 0.013 & 0.015 & 0.061 & 0.049 & 0.065 & 0.113 & 0.096 & 0.122 \\ 
& 100 & 0.012 & 0.010 & 0.012 & 0.054 & 0.046 & 0.055 & 0.103 & 0.093 & 0.111 \\ 
& 200 & 0.011 & 0.009 & 0.014 & 0.048 & 0.043 & 0.058 & 0.092 & 0.087 & 0.109 \\ 
& 400 & 0.009 & 0.007 & 0.011 & 0.044 & 0.038 & 0.056 & 0.089 & 0.078 & 0.111 \\ 
\cmidrule[0.5pt](l){2-11}

\multirow{12}{*}{ \STAB{\rotatebox[origin=c]{90}{\underline{DGP U-R}}}} 
&  & \multicolumn{3}{c}{MB} & \multicolumn{3}{c}{MAD} & \multicolumn{3}{c}{RMSE} \\ \cmidrule[0.2pt](l){3-5} \cmidrule[0.2pt](l){6-8} \cmidrule[0.2pt](l){9-11}
       &25     &-0.008 &-0.009 &-0.013 &0.144  &0.099  &0.178  &0.22   &0.156  &0.286  \\ 
       &50     &-0.006 &-0.005 &-0.007 &0.103  &0.072  &0.131  &0.155  &0.108  &0.203  \\ 
       &100    &-0.003 &-0.003 &-0.004 &0.073  &0.053  &0.094  &0.109  &0.078  &0.143  \\ 
       &200    &-0.001 &-0.001 &-0.001 &0.051  &0.036  &0.066  &0.077  &0.054  &0.100    \\ 
       &400    &-0.001 &-0.001 &-0.001 &0.036  &0.026  &0.047  &0.054  &0.038  &0.071  \\  
\cmidrule[0.2pt](l){3-5} \cmidrule[0.2pt](l){6-8} \cmidrule[0.2pt](l){9-11}

& & \multicolumn{3}{c}{Rej. 1\%} & \multicolumn{3}{c}{Rej. 5\%} & \multicolumn{3}{c}{Rej. 10\%} \\ \cmidrule[0.2pt](l){3-5} \cmidrule[0.2pt](l){6-8} \cmidrule[0.2pt](l){9-11}

      &25    &0.011 &0.010  &0.008 &0.049 &0.043 &0.042 &0.097 &0.086 &0.095 \\ 
      &50    &0.008 &0.005 &0.007 &0.045 &0.035 &0.044 &0.096 &0.077 &0.094 \\ 
      &100   &0.007 &0.006 &0.006 &0.043 &0.037 &0.039 &0.091 &0.08  &0.089 \\ 
      &200   &0.009 &0.005 &0.008 &0.042 &0.038 &0.043 &0.086 &0.078 &0.091 \\ 
      &400   &0.007 &0.005 &0.007 &0.039 &0.034 &0.041 &0.085 &0.073 &0.094 \\ 
      \cmidrule[0.5pt](l){2-11}
\end{tabular}
\end{table}

\begin{table}[!htbp]
\caption{Simulation - $ATT(t) = \sin(t)/\pi,\ t\geq 1 $}
\label{Tab:Sim_IV_b}
\centering
\setlength{\tabcolsep}{3pt}
\begin{tabular}{rlccccccccc}
\cmidrule[0.5pt](l){2-11}
 & $T,\T $ & \multicolumn{1}{c}{DiD} & \multicolumn{1}{c}{SC} & \multicolumn{1}{c}{BA} & \multicolumn{1}{c}{DiD} & \multicolumn{1}{c}{SC} & \multicolumn{1}{c}{BA} & \multicolumn{1}{c}{DiD} & \multicolumn{1}{c}{SC} & \multicolumn{1}{c}{BA} \\ \cmidrule[0.2pt](l){2-11}

\multirow{12}{*}{ \STAB{\rotatebox[origin=c]{90}{\underline{DGP Q-T}}}} 
&  & \multicolumn{3}{c}{MB} & \multicolumn{3}{c}{MAD} & \multicolumn{3}{c}{RMSE} \\ \cmidrule[0.2pt](l){3-5} \cmidrule[0.2pt](l){6-8} \cmidrule[0.2pt](l){9-11}
       &25     &0.003  &0.002  &25.998 &0.141  &0.096  &25.998 &0.217  &0.153  &25.999 \\ 
       &50     &-0.001 &0.000      &50.998 &0.103  &0.070   &50.998 &0.154  &0.107  &50.998 \\ 
        &100     &-0.001  &0.000       &100.998 &0.073   &0.052   &100.998 &0.108   &0.077   &100.999 \\ 
        &200     &0.001   &0.001   &201.001 &0.050    &0.036   &201.000     &0.076   &0.054   &201.001 \\ 
        &400     &0.000       &0.000       &401.000     &0.036   &0.026   &401.001 &0.054   &0.038   &401.000     \\  
\cmidrule[0.2pt](l){3-5} \cmidrule[0.2pt](l){6-8} \cmidrule[0.2pt](l){9-11}

& & \multicolumn{3}{c}{Rej. 1\%} & \multicolumn{3}{c}{Rej. 5\%} & \multicolumn{3}{c}{Rej. 10\%} \\ \cmidrule[0.2pt](l){3-5} \cmidrule[0.2pt](l){6-8} \cmidrule[0.2pt](l){9-11}

      &25    &0.011 &0.009 &1.000     &0.050  &0.043 &1.000     &0.101 &0.083 &1.000     \\ 
      &50    &0.009 &0.005 &1.000     &0.046 &0.034 &1.000     &0.092 &0.074 &1.000     \\ 
      &100   &0.008 &0.005 &1.000     &0.043 &0.036 &1.000     &0.092 &0.079 &1.000     \\ 
      &200   &0.009 &0.005 &1.000     &0.043 &0.036 &1.000     &0.087 &0.076 &1.000     \\ 
      &400   &0.007 &0.005 &1.000     &0.038 &0.034 &1.000     &0.084 &0.073 &1.000     \\
      \cmidrule[0.5pt](l){2-11}

\multirow{12}{*}{ \STAB{\rotatebox[origin=c]{90}{\underline{DGP T-T}}}} 
&  & \multicolumn{3}{c}{MB} & \multicolumn{3}{c}{MAD} & \multicolumn{3}{c}{RMSE} \\ \cmidrule[0.2pt](l){3-5} \cmidrule[0.2pt](l){6-8} \cmidrule[0.2pt](l){9-11}
       &25     &0.075  &25.080  &0.080   &0.271  &25.083 &0.346  &0.432  &25.082 &0.576  \\ 
       &50     &0.032  &50.033 &0.042  &0.202  &50.035 &0.252  &0.308  &50.034 &0.403  \\ 
        &100     &0.017   &100.019 &0.017   &0.140    &100.021 &0.182   &0.215   &100.019 &0.283   \\ 
        &200     &0.012   &200.012 &0.011   &0.101   &200.013 &0.131   &0.151   &200.012 &0.200     \\ 
        &400     &0.002   &400.002 &0.001   &0.073   &400.001 &0.093   &0.108   &400.002 &0.140   \\   
\cmidrule[0.2pt](l){3-5} \cmidrule[0.2pt](l){6-8} \cmidrule[0.2pt](l){9-11}

& & \multicolumn{3}{c}{Rej. 1\%} & \multicolumn{3}{c}{Rej. 5\%} & \multicolumn{3}{c}{Rej. 10\%} \\ \cmidrule[0.2pt](l){3-5} \cmidrule[0.2pt](l){6-8} \cmidrule[0.2pt](l){9-11}

      &25    &0.012 &1.000     &0.016 &0.057 &1.000     &0.064 &0.112 &1.000     &0.124 \\ 
      &50    &0.011 &1.000     &0.011 &0.049 &1.000     &0.054 &0.104 &1.000     &0.108 \\ 
      &100   &0.009 &1.000     &0.009 &0.044 &1.000     &0.046 &0.091 &1.000     &0.097 \\ 
      &200   &0.007 &1.000     &0.009 &0.040  &1.000     &0.046 &0.087 &1.000     &0.093 \\ 
      &400   &0.008 &1.000     &0.008 &0.043 &1.000     &0.043 &0.090  &1.000     &0.096 \\
      \cmidrule[0.5pt](l){2-11}
\end{tabular}
\end{table}

Generally, \Cref{Tab:Sim_I_b,Tab:Sim_II_b,Tab:Sim_III_b,Tab:Sim_IV_b} confirm simulation results on homogeneous treatment effects. The DiD performs well across all settings considered in the DGPs, whereas the SC and BA are very sensitive to deviations from the strong identification conditions they require. $ATT(t)$ is heterogeneous across post-treatment periods, and the target parameter varies with $T$ and does not impact the performance of the DiD.

\subsection{Non-constant $ \lambda_n$ }
The next set of simulation results concerns the sensitivity to $ \lambda_n $. To this end, the sequence $ \lambda_n $ is specified as $ \lambda_n \in \{ 0.5,0.4,0.6,0.3,0.7 \} $ varying in $n \in \{50, 100, 200, 400, 800 \} $ over the closed interval $ \lambda_n \in [0.3, \, 0.7] $ with $T= \lambda_n n $, $ \T = n-T $, and $ATT(t)=0.0, \, t\geq 1  $. The results are presented in \Cref{Tab:Sim_I_e,Tab:Sim_II_e,Tab:Sim_III_e,Tab:Sim_IV_e}. Comparing these results to those in \Cref{Tab:Sim_I,Tab:Sim_II,Tab:Sim_III,Tab:Sim_IV}, one observes little difference in the performance of all estimators. The performance of the T-DiD remains robust across all the DGPs.

\begin{table}[!htbp]
\caption{Simulation – $ATT(t)=0.0,\ t\geq 1$, non-constant $\lambda_n$}
\centering
\begin{tabular}{rlcccccccccc}
\cmidrule[0.5pt](l){2-12}
 & $n$ & $\lambda_n$ & \multicolumn{1}{c}{DiD} & \multicolumn{1}{c}{SC} & \multicolumn{1}{c}{BA} &
   \multicolumn{1}{c}{DiD} & \multicolumn{1}{c}{SC} & \multicolumn{1}{c}{BA} &
   \multicolumn{1}{c}{DiD} & \multicolumn{1}{c}{SC} & \multicolumn{1}{c}{BA} \\ 
\cmidrule[0.2pt](l){2-12}

\multirow{12}{*}{\STAB{\rotatebox[origin=c]{90}{\underline{DGP SC–BA}}}}
&  &  & \multicolumn{3}{c}{MB} & \multicolumn{3}{c}{MAD} & \multicolumn{3}{c}{RMSE} \\ 
\cmidrule[0.2pt](l){4-6}\cmidrule[0.2pt](l){7-9}\cmidrule[0.2pt](l){10-12}
& 25  & 0.5 & 0.002 & 0.002 & 0.001 & 0.140 & 0.096 & 0.175 & 0.216 & 0.151 & 0.280 \\ 
& 50  & 0.4 & 0.000 & 0.000 & -0.002 & 0.103 & 0.079 & 0.132 & 0.157 & 0.121 & 0.204 \\ 
& 100 & 0.6 & 0.001 & 0.000 & -0.002 & 0.074 & 0.047 & 0.095 & 0.112 & 0.070 & 0.147 \\ 
& 200 & 0.3 & -0.001 & -0.001 & 0.000 & 0.055 & 0.046 & 0.071 & 0.084 & 0.070 & 0.109 \\ 
& 400 & 0.7 & 0.000 & 0.000 & 0.001 & 0.040 & 0.022 & 0.051 & 0.059 & 0.032 & 0.076 \\ 
\cmidrule[0.2pt](l){4-6}\cmidrule[0.2pt](l){7-9}\cmidrule[0.2pt](l){10-12}
&  &  & \multicolumn{3}{c}{Rej.\ 1\%} & \multicolumn{3}{c}{Rej.\ 5\%} & \multicolumn{3}{c}{Rej.\ 10\%} \\ 
\cmidrule[0.2pt](l){4-6}\cmidrule[0.2pt](l){7-9}\cmidrule[0.2pt](l){10-12}
& 25  & 0.5 & 0.012 & 0.012 & 0.013 & 0.057 & 0.056 & 0.056 & 0.116 & 0.110 & 0.114 \\ 
& 50  & 0.4 & 0.011 & 0.011 & 0.010 & 0.056 & 0.053 & 0.052 & 0.111 & 0.111 & 0.105 \\ 
& 100 & 0.6 & 0.010 & 0.010 & 0.011 & 0.055 & 0.052 & 0.055 & 0.110 & 0.104 & 0.110 \\ 
& 200 & 0.3 & 0.011 & 0.010 & 0.009 & 0.055 & 0.053 & 0.049 & 0.108 & 0.104 & 0.102 \\ 
& 400 & 0.7 & 0.011 & 0.010 & 0.009 & 0.048 & 0.050 & 0.050 & 0.094 & 0.101 & 0.098 \\ 
\cmidrule[0.5pt](l){2-12}

\multirow{12}{*}{\STAB{\rotatebox[origin=c]{90}{\underline{DGP BA}}}}
&  &  & \multicolumn{3}{c}{MB} & \multicolumn{3}{c}{MAD} & \multicolumn{3}{c}{RMSE} \\ 
\cmidrule[0.2pt](l){4-6}\cmidrule[0.2pt](l){7-9}\cmidrule[0.2pt](l){10-12}
& 25  & 0.5 & 0.002 & -0.998 & 0.001 & 0.140 & 0.997 & 0.175 & 0.216 & 1.009 & 0.280 \\ 
& 50  & 0.4 & 0.000 & -1.000 & -0.002 & 0.103 & 0.999 & 0.132 & 0.157 & 1.007 & 0.204 \\ 
& 100 & 0.6 & 0.001 & -1.000 & -0.002 & 0.074 & 0.998 & 0.095 & 0.112 & 1.002 & 0.147 \\ 
& 200 & 0.3 & -0.001 & -1.001 & 0.000 & 0.055 & 1.001 & 0.071 & 0.084 & 1.003 & 0.109 \\ 
& 400 & 0.7 & 0.000 & -1.000 & 0.001 & 0.040 & 1.000 & 0.051 & 0.059 & 1.001 & 0.076 \\ 
\cmidrule[0.2pt](l){4-6}\cmidrule[0.2pt](l){7-9}\cmidrule[0.2pt](l){10-12}
&  &  & \multicolumn{3}{c}{Rej.\ 1\%} & \multicolumn{3}{c}{Rej.\ 5\%} & \multicolumn{3}{c}{Rej.\ 10\%} \\ 
\cmidrule[0.2pt](l){4-6}\cmidrule[0.2pt](l){7-9}\cmidrule[0.2pt](l){10-12}
& 25  & 0.5 & 0.012 & 0.990 & 0.013 & 0.057 & 0.997 & 0.056 & 0.116 & 0.999 & 0.114 \\ 
& 50  & 0.4 & 0.011 & 0.999 & 0.010 & 0.056 & 1.000 & 0.052 & 0.111 & 1.000 & 0.105 \\ 
& 100 & 0.6 & 0.010 & 1.000 & 0.011 & 0.055 & 1.000 & 0.055 & 0.110 & 1.000 & 0.110 \\ 
& 200 & 0.3 & 0.011 & 1.000 & 0.009 & 0.055 & 1.000 & 0.049 & 0.108 & 1.000 & 0.102 \\ 
& 400 & 0.7 & 0.011 & 1.000 & 0.009 & 0.048 & 1.000 & 0.050 & 0.094 & 1.000 & 0.098 \\ 
\cmidrule[0.5pt](l){2-12}

\multirow{12}{*}{\STAB{\rotatebox[origin=c]{90}{\underline{DGP SC}}}}
&  &  & \multicolumn{3}{c}{MB} & \multicolumn{3}{c}{MAD} & \multicolumn{3}{c}{RMSE} \\ 
\cmidrule[0.2pt](l){4-6}\cmidrule[0.2pt](l){7-9}\cmidrule[0.2pt](l){10-12}
& 25  & 0.5 & 0.002 & 0.002 & 1.246 & 0.140 & 0.096 & 1.248 & 0.216 & 0.151 & 1.277 \\ 
& 50  & 0.4 & 0.000 & 0.000 & 1.306 & 0.103 & 0.079 & 1.303 & 0.157 & 0.121 & 1.322 \\ 
& 100 & 0.6 & 0.001 & 0.000 & 1.377 & 0.074 & 0.047 & 1.380 & 0.112 & 0.070 & 1.385 \\ 
& 200 & 0.3 & -0.001 & -0.001 & 1.379 & 0.055 & 0.046 & 1.377 & 0.084 & 0.070 & 1.383 \\ 
& 400 & 0.7 & 0.000 & 0.000 & 1.407 & 0.040 & 0.022 & 1.408 & 0.059 & 0.032 & 1.409 \\ 
\cmidrule[0.2pt](l){4-6}\cmidrule[0.2pt](l){7-9}\cmidrule[0.2pt](l){10-12}
&  &  & \multicolumn{3}{c}{Rej.\ 1\%} & \multicolumn{3}{c}{Rej.\ 5\%} & \multicolumn{3}{c}{Rej.\ 10\%} \\ 
\cmidrule[0.2pt](l){4-6}\cmidrule[0.2pt](l){7-9}\cmidrule[0.2pt](l){10-12}
& 25  & 0.5 & 0.012 & 0.012 & 0.919 & 0.057 & 0.056 & 0.965 & 0.116 & 0.110 & 0.979 \\ 
& 50  & 0.4 & 0.011 & 0.011 & 0.992 & 0.056 & 0.053 & 0.998 & 0.111 & 0.111 & 0.999 \\ 
& 100 & 0.6 & 0.010 & 0.010 & 0.999 & 0.055 & 0.052 & 1.000 & 0.110 & 0.104 & 1.000 \\ 
& 200 & 0.3 & 0.011 & 0.010 & 1.000 & 0.055 & 0.053 & 1.000 & 0.108 & 0.104 & 1.000 \\ 
& 400 & 0.7 & 0.011 & 0.010 & 1.000 & 0.048 & 0.050 & 1.000 & 0.094 & 0.101 & 1.000 \\ 
\cmidrule[0.5pt](l){2-12}
\end{tabular}
\label{Tab:Sim_I_e}
\end{table}

\begin{table}[!htbp]
\caption{Simulation -- $ATT(t)=0.0,\ t\geq 1$, non-constant $\lambda_n$}
\centering
\begin{tabular}{rlcccccccccc}
\cmidrule[0.5pt](l){2-12}
 & $n$ & $\lambda_n$ & \multicolumn{1}{c}{DiD} & \multicolumn{1}{c}{SC} & \multicolumn{1}{c}{BA} &
   \multicolumn{1}{c}{DiD} & \multicolumn{1}{c}{SC} & \multicolumn{1}{c}{BA} &
   \multicolumn{1}{c}{DiD} & \multicolumn{1}{c}{SC} & \multicolumn{1}{c}{BA} \\
\cmidrule[0.2pt](l){2-12}

\multirow{12}{*}{ \STAB{\rotatebox[origin=c]{90}{\underline{DGP PT-NA (A)}}} }
&  &  & \multicolumn{3}{c}{MB} & \multicolumn{3}{c}{MAD} & \multicolumn{3}{c}{RMSE} \\
\cmidrule[0.2pt](l){4-6} \cmidrule[0.2pt](l){7-9} \cmidrule[0.2pt](l){10-12}
& 25  & 0.5 & 0.177 & 0.089 & 0.176 & 0.270 & 0.179 & 0.294 & 0.397 & 0.265 & 0.436 \\
& 50  & 0.4 & 0.102 & 0.064 & 0.108 & 0.178 & 0.145 & 0.202 & 0.261 & 0.211 & 0.309 \\
& 100 & 0.6 & 0.075 & 0.027 & 0.063 & 0.150 & 0.080 & 0.145 & 0.225 & 0.118 & 0.217 \\
& 200 & 0.3 & 0.035 & 0.026 & 0.042 & 0.089 & 0.080 & 0.107 & 0.130 & 0.117 & 0.159 \\
& 400 & 0.7 & 0.034 & 0.008 & 0.026 & 0.095 & 0.036 & 0.076 & 0.140 & 0.054 & 0.112 \\
\cmidrule[0.2pt](l){4-6} \cmidrule[0.2pt](l){7-9} \cmidrule[0.2pt](l){10-12}
&  &  & \multicolumn{3}{c}{Rej.\ 1\%} & \multicolumn{3}{c}{Rej.\ 5\%} & \multicolumn{3}{c}{Rej.\ 10\%} \\
\cmidrule[0.2pt](l){4-6} \cmidrule[0.2pt](l){7-9} \cmidrule[0.2pt](l){10-12}
& 25  & 0.5 & 0.030 & 0.029 & 0.027 & 0.096 & 0.086 & 0.092 & 0.162 & 0.142 & 0.155 \\
& 50  & 0.4 & 0.024 & 0.023 & 0.022 & 0.080 & 0.079 & 0.077 & 0.141 & 0.134 & 0.136 \\
& 100 & 0.6 & 0.020 & 0.015 & 0.017 & 0.072 & 0.060 & 0.067 & 0.132 & 0.111 & 0.127 \\
& 200 & 0.3 & 0.014 & 0.013 & 0.013 & 0.064 & 0.059 & 0.061 & 0.120 & 0.110 & 0.117 \\
& 400 & 0.7 & 0.012 & 0.013 & 0.012 & 0.060 & 0.052 & 0.060 & 0.113 & 0.103 & 0.112 \\
\cmidrule[0.5pt](l){2-12}

\multirow{12}{*}{\STAB{\rotatebox[origin=c]{90}{\underline{DGP PT-NA (B)}}}}
&  &  & \multicolumn{3}{c}{MB} & \multicolumn{3}{c}{MAD} & \multicolumn{3}{c}{RMSE} \\
\cmidrule[0.2pt](l){4-6} \cmidrule[0.2pt](l){7-9} \cmidrule[0.2pt](l){10-12}
& 25  & 0.5 & 0.001 & 0.287 & 0.000 & 0.241 & 0.294 & 0.269 & 0.356 & 0.381 & 0.399 \\
& 50  & 0.4 & 0.023 & 0.258 & 0.021 & 0.175 & 0.264 & 0.191 & 0.259 & 0.327 & 0.290 \\
& 100 & 0.6 & -0.020 & 0.198 & -0.023 & 0.124 & 0.198 & 0.138 & 0.187 & 0.229 & 0.209 \\
& 200 & 0.3 & 0.035 & 0.197 & 0.036 & 0.097 & 0.197 & 0.106 & 0.142 & 0.227 & 0.158 \\
& 400 & 0.7 & -0.031 & 0.136 & -0.031 & 0.069 & 0.136 & 0.076 & 0.102 & 0.146 & 0.114 \\
\cmidrule[0.2pt](l){4-6} \cmidrule[0.2pt](l){7-9} \cmidrule[0.2pt](l){10-12}
&  &  & \multicolumn{3}{c}{Rej.\ 1\%} & \multicolumn{3}{c}{Rej.\ 5\%} & \multicolumn{3}{c}{Rej.\ 10\%} \\
\cmidrule[0.2pt](l){4-6} \cmidrule[0.2pt](l){7-9} \cmidrule[0.2pt](l){10-12}
& 25  & 0.5 & 0.016 & 0.111 & 0.016 & 0.065 & 0.246 & 0.061 & 0.118 & 0.347 & 0.119 \\
& 50  & 0.4 & 0.015 & 0.133 & 0.016 & 0.059 & 0.284 & 0.059 & 0.107 & 0.394 & 0.114 \\
& 100 & 0.6 & 0.016 & 0.210 & 0.013 & 0.060 & 0.416 & 0.057 & 0.113 & 0.539 & 0.110 \\
& 200 & 0.3 & 0.012 & 0.207 & 0.012 & 0.060 & 0.414 & 0.058 & 0.117 & 0.534 & 0.113 \\
& 400 & 0.7 & 0.016 & 0.494 & 0.013 & 0.064 & 0.728 & 0.059 & 0.122 & 0.817 & 0.114 \\
\cmidrule[0.5pt](l){2-12}

\multirow{12}{*}{\STAB{\rotatebox[origin=c]{90}{\underline{DGP GARCH(1,1)}}}}
&  &  & \multicolumn{3}{c}{MB} & \multicolumn{3}{c}{MAD} & \multicolumn{3}{c}{RMSE} \\
\cmidrule[0.2pt](l){4-6} \cmidrule[0.2pt](l){7-9} \cmidrule[0.2pt](l){10-12}
& 25  & 0.5 & 0.004 & 0.001 & -0.003 & 0.141 & 0.098 & 0.177 & 0.214 & 0.154 & 0.275 \\
& 50  & 0.4 & 0.002 & 0.001 & 0.001 & 0.103 & 0.079 & 0.131 & 0.156 & 0.120 & 0.203 \\
& 100 & 0.6 & 0.000 & 0.000 & 0.000 & 0.073 & 0.046 & 0.092 & 0.110 & 0.070 & 0.142 \\
& 200 & 0.3 & -0.001 & -0.001 & 0.000 & 0.056 & 0.047 & 0.071 & 0.083 & 0.070 & 0.109 \\
& 400 & 0.7 & 0.000 & 0.000 & 0.000 & 0.039 & 0.022 & 0.051 & 0.059 & 0.032 & 0.077 \\
\cmidrule[0.2pt](l){4-6} \cmidrule[0.2pt](l){7-9} \cmidrule[0.2pt](l){10-12}
&  &  & \multicolumn{3}{c}{Rej.\ 1\%} & \multicolumn{3}{c}{Rej.\ 5\%} & \multicolumn{3}{c}{Rej.\ 10\%} \\
\cmidrule[0.2pt](l){4-6} \cmidrule[0.2pt](l){7-9} \cmidrule[0.2pt](l){10-12}
& 25  & 0.5 & 0.017 & 0.019 & 0.014 & 0.064 & 0.072 & 0.058 & 0.120 & 0.123 & 0.111 \\
& 50  & 0.4 & 0.014 & 0.014 & 0.013 & 0.057 & 0.059 & 0.056 & 0.108 & 0.110 & 0.105 \\
& 100 & 0.6 & 0.011 & 0.010 & 0.008 & 0.054 & 0.053 & 0.049 & 0.102 & 0.105 & 0.098 \\
& 200 & 0.3 & 0.011 & 0.010 & 0.011 & 0.051 & 0.054 & 0.051 & 0.103 & 0.101 & 0.099 \\
& 400 & 0.7 & 0.011 & 0.009 & 0.011 & 0.052 & 0.053 & 0.050 & 0.101 & 0.101 & 0.102 \\
\cmidrule[0.5pt](l){2-12}
\end{tabular}
\label{Tab:Sim_II_e}
\end{table}

\begin{table}[!htbp]
\caption{Simulation -- $ATT(t)=0.0,\ t\geq 1$, non-constant $\lambda_n$}
\centering
\begin{tabular}{rlcccccccccc}
\cmidrule[0.5pt](l){2-12}
 & $n$ & $\lambda_n$ & \multicolumn{1}{c}{DiD} & \multicolumn{1}{c}{SC} & \multicolumn{1}{c}{BA} &
   \multicolumn{1}{c}{DiD} & \multicolumn{1}{c}{SC} & \multicolumn{1}{c}{BA} &
   \multicolumn{1}{c}{DiD} & \multicolumn{1}{c}{SC} & \multicolumn{1}{c}{BA} \\ 
\cmidrule[0.2pt](l){2-12}

\multirow{12}{*}{\STAB{\rotatebox[origin=c]{90}{\underline{DGP MA(1)}}}}
&  &  & \multicolumn{3}{c}{MB} & \multicolumn{3}{c}{MAD} & \multicolumn{3}{c}{RMSE} \\ 
\cmidrule[0.2pt](l){4-6}\cmidrule[0.2pt](l){7-9}\cmidrule[0.2pt](l){10-12}
& 25  & 0.5 &  0.005 &  0.003 & -0.002 & 0.174 & 0.119 & 0.218 & 0.269 & 0.190 & 0.347 \\
& 50  & 0.4 &  0.000 &  0.001 & -0.003 & 0.121 & 0.097 & 0.164 & 0.183 & 0.150 & 0.255 \\
& 100 & 0.6 &  0.001 &  0.000 & -0.002 & 0.105 & 0.059 & 0.118 & 0.158 & 0.088 & 0.182 \\
& 200 & 0.3 & -0.001 & -0.001 &  0.000 & 0.063 & 0.058 & 0.089 & 0.095 & 0.087 & 0.136 \\
& 400 & 0.7 &  0.000 &  0.000 &  0.001 & 0.068 & 0.028 & 0.063 & 0.101 & 0.040 & 0.096 \\
\cmidrule[0.2pt](l){4-6}\cmidrule[0.2pt](l){7-9}\cmidrule[0.2pt](l){10-12}
&  &  & \multicolumn{3}{c}{Rej.\ 1\%} & \multicolumn{3}{c}{Rej.\ 5\%} & \multicolumn{3}{c}{Rej.\ 10\%} \\ 
\cmidrule[0.2pt](l){4-6}\cmidrule[0.2pt](l){7-9}\cmidrule[0.2pt](l){10-12}
& 25  & 0.5 & 0.024 & 0.022 & 0.017 & 0.078 & 0.074 & 0.068 & 0.142 & 0.131 & 0.130 \\ 
& 50  & 0.4 & 0.017 & 0.014 & 0.016 & 0.070 & 0.062 & 0.066 & 0.130 & 0.123 & 0.127 \\ 
& 100 & 0.6 & 0.016 & 0.012 & 0.014 & 0.066 & 0.060 & 0.063 & 0.123 & 0.115 & 0.121 \\ 
& 200 & 0.3 & 0.015 & 0.012 & 0.013 & 0.059 & 0.060 & 0.058 & 0.114 & 0.113 & 0.112 \\ 
& 400 & 0.7 & 0.013 & 0.012 & 0.014 & 0.059 & 0.055 & 0.059 & 0.107 & 0.104 & 0.111 \\ 
\cmidrule[0.5pt](l){2-12}

\multirow{12}{*}{\STAB{\rotatebox[origin=c]{90}{\underline{DGP AR(1)}}}}
&  &  & \multicolumn{3}{c}{MB} & \multicolumn{3}{c}{MAD} & \multicolumn{3}{c}{RMSE} \\ 
\cmidrule[0.2pt](l){4-6}\cmidrule[0.2pt](l){7-9}\cmidrule[0.2pt](l){10-12}
& 25  & 0.5 & 0.003 & -0.021 & 0.003 & 0.162 & 0.112 & 0.208 & 0.252 & 0.183 & 0.324 \\
& 50  & 0.4 & 0.000 & -0.014 & -0.001 & 0.115 & 0.086 & 0.148 & 0.172 & 0.137 & 0.222 \\
& 100 & 0.6 & 0.001 & -0.005 & -0.001 & 0.077 & 0.049 & 0.101 & 0.116 & 0.074 & 0.151 \\
& 200 & 0.3 & -0.001 & -0.006 & 0.001 & 0.058 & 0.048 & 0.074 & 0.086 & 0.073 & 0.111 \\
& 400 & 0.7 & 0.000 & -0.001 & 0.001 & 0.040 & 0.022 & 0.051 & 0.059 & 0.033 & 0.077 \\
\cmidrule[0.2pt](l){4-6}\cmidrule[0.2pt](l){7-9}\cmidrule[0.2pt](l){10-12}
&  &  & \multicolumn{3}{c}{Rej.\ 1\%} & \multicolumn{3}{c}{Rej.\ 5\%} & \multicolumn{3}{c}{Rej.\ 10\%} \\ 
\cmidrule[0.2pt](l){4-6}\cmidrule[0.2pt](l){7-9}\cmidrule[0.2pt](l){10-12}
& 25  & 0.5 & 0.024 & 0.033 & 0.021 & 0.083 & 0.088 & 0.082 & 0.141 & 0.148 & 0.143 \\ 
& 50  & 0.4 & 0.018 & 0.024 & 0.019 & 0.071 & 0.079 & 0.073 & 0.131 & 0.135 & 0.132 \\ 
& 100 & 0.6 & 0.015 & 0.015 & 0.016 & 0.062 & 0.065 & 0.066 & 0.120 & 0.118 & 0.124 \\ 
& 200 & 0.3 & 0.014 & 0.015 & 0.016 & 0.056 & 0.063 & 0.064 & 0.107 & 0.115 & 0.117 \\ 
& 400 & 0.7 & 0.011 & 0.011 & 0.016 & 0.053 & 0.056 & 0.061 & 0.102 & 0.103 & 0.112 \\ 
\cmidrule[0.5pt](l){2-12}

\multirow{12}{*}{\STAB{\rotatebox[origin=c]{90}{\underline{DGP U--R}}}}
&  &  & \multicolumn{3}{c}{MB} & \multicolumn{3}{c}{MAD} & \multicolumn{3}{c}{RMSE} \\ 
\cmidrule[0.2pt](l){4-6}\cmidrule[0.2pt](l){7-9}\cmidrule[0.2pt](l){10-12}
& 25  & 0.5 & 0.003 & 0.002 & -0.002 & 0.143 & 0.097 & 0.178 & 0.220 & 0.156 & 0.286 \\
& 50  & 0.4 & 0.000 & 0.001 & -0.003 & 0.106 & 0.079 & 0.132 & 0.160 & 0.122 & 0.207 \\
& 100 & 0.6 & 0.001 & 0.000 & -0.002 & 0.074 & 0.048 & 0.094 & 0.112 & 0.071 & 0.147 \\
& 200 & 0.3 & -0.001 & -0.001 & 0.000 & 0.056 & 0.047 & 0.072 & 0.084 & 0.070 & 0.109 \\
& 400 & 0.7 & 0.000 & 0.000 & 0.001 & 0.039 & 0.022 & 0.050 & 0.059 & 0.032 & 0.077 \\
\cmidrule[0.2pt](l){4-6}\cmidrule[0.2pt](l){7-9}\cmidrule[0.2pt](l){10-12}
&  &  & \multicolumn{3}{c}{Rej.\ 1\%} & \multicolumn{3}{c}{Rej.\ 5\%} & \multicolumn{3}{c}{Rej.\ 10\%} \\ 
\cmidrule[0.2pt](l){4-6}\cmidrule[0.2pt](l){7-9}\cmidrule[0.2pt](l){10-12}
& 25  & 0.5 & 0.013 & 0.015 & 0.010 & 0.056 & 0.056 & 0.047 & 0.110 & 0.105 & 0.101 \\ 
& 50  & 0.4 & 0.011 & 0.010 & 0.011 & 0.055 & 0.051 & 0.052 & 0.111 & 0.103 & 0.102 \\ 
& 100 & 0.6 & 0.012 & 0.009 & 0.007 & 0.052 & 0.052 & 0.052 & 0.103 & 0.103 & 0.103 \\ 
& 200 & 0.3 & 0.012 & 0.009 & 0.009 & 0.052 & 0.051 & 0.048 & 0.100 & 0.099 & 0.102 \\ 
& 400 & 0.7 & 0.010 & 0.010 & 0.011 & 0.050 & 0.051 & 0.051 & 0.098 & 0.097 & 0.097 \\ 
\cmidrule[0.5pt](l){2-12}
\end{tabular}
\label{Tab:Sim_III_e}
\end{table}

\begin{table}[!htbp]
\caption{Simulation -- $ATT(t)=0.0,\ t\geq 1$, non-constant $\lambda_n$}
\centering
\begin{tabular}{rlcccccccccc}
\cmidrule[0.5pt](l){2-12}
 & $n$ & $\lambda_n$ & \multicolumn{1}{c}{DiD} & \multicolumn{1}{c}{SC} & \multicolumn{1}{c}{BA} &
   \multicolumn{1}{c}{DiD} & \multicolumn{1}{c}{SC} & \multicolumn{1}{c}{BA} &
   \multicolumn{1}{c}{DiD} & \multicolumn{1}{c}{SC} & \multicolumn{1}{c}{BA} \\ 
\cmidrule[0.2pt](l){2-12}

\multirow{12}{*}{\STAB{\rotatebox[origin=c]{90}{\underline{DGP Q--T}}}}
&  &  & \multicolumn{3}{c}{MB} & \multicolumn{3}{c}{MAD} & \multicolumn{3}{c}{RMSE} \\ 
\cmidrule[0.2pt](l){4-6}\cmidrule[0.2pt](l){7-9}\cmidrule[0.2pt](l){10-12}
& 25  & 0.5 & 0.003 & 0.002 & 25.998 & 0.141 & 0.096 & 25.998 & 0.217 & 0.153 & 25.999 \\
& 50  & 0.4 & 0.000 & 0.000 & 49.645 & 0.106 & 0.078 & 49.643 & 0.159 & 0.120 & 49.645 \\
& 100 & 0.6 & 0.000 & 0.000 & 106.371 & 0.073 & 0.047 & 106.373 & 0.111 & 0.070 & 106.371 \\
& 200 & 0.3 & -0.001 & -0.001 & 158.174 & 0.056 & 0.047 & 158.169 & 0.084 & 0.070 & 158.174 \\
& 400 & 0.7 & 0.000 & 0.000 & 571.987 & 0.039 & 0.022 & 571.990 & 0.059 & 0.032 & 571.987 \\
\cmidrule[0.2pt](l){4-6}\cmidrule[0.2pt](l){7-9}\cmidrule[0.2pt](l){10-12}
&  &  & \multicolumn{3}{c}{Rej.\ 1\%} & \multicolumn{3}{c}{Rej.\ 5\%} & \multicolumn{3}{c}{Rej.\ 10\%} \\ 
\cmidrule[0.2pt](l){4-6}\cmidrule[0.2pt](l){7-9}\cmidrule[0.2pt](l){10-12}
& 25  & 0.5 & 0.013 & 0.015 & 1.000 & 0.060 & 0.058 & 1.000 & 0.113 & 0.105 & 1.000 \\
& 50  & 0.4 & 0.010 & 0.009 & 1.000 & 0.053 & 0.049 & 1.000 & 0.109 & 0.099 & 1.000 \\
& 100 & 0.6 & 0.011 & 0.010 & 1.000 & 0.054 & 0.051 & 1.000 & 0.102 & 0.104 & 1.000 \\
& 200 & 0.3 & 0.012 & 0.009 & 1.000 & 0.052 & 0.052 & 1.000 & 0.101 & 0.100 & 1.000 \\
& 400 & 0.7 & 0.009 & 0.010 & 1.000 & 0.048 & 0.049 & 1.000 & 0.098 & 0.098 & 1.000 \\
\cmidrule[0.5pt](l){2-12}

\multirow{12}{*}{\STAB{\rotatebox[origin=c]{90}{\underline{DGP T--T}}}}
&  &  & \multicolumn{3}{c}{MB} & \multicolumn{3}{c}{MAD} & \multicolumn{3}{c}{RMSE} \\ 
\cmidrule[0.2pt](l){4-6}\cmidrule[0.2pt](l){7-9}\cmidrule[0.2pt](l){10-12}
& 25  & 0.5 & 0.003 & 25.005 & 0.008 & 0.270 & 25.008 & 0.344 & 0.426 & 25.007 & 0.571 \\
& 50  & 0.4 & -0.002 & 59.995 & 0.000 & 0.195 & 59.998 & 0.242 & 0.296 & 59.995 & 0.385 \\
& 100 & 0.6 & 0.004 & 80.002 & -0.005 & 0.139 & 80.002 & 0.176 & 0.212 & 80.002 & 0.275 \\
& 200 & 0.3 & -0.004 & 280.000 & -0.001 & 0.091 & 280.001 & 0.119 & 0.138 & 280.000 & 0.179 \\
& 400 & 0.7 & 0.000 & 240.000 & 0.000 & 0.065 & 240.000 & 0.085 & 0.098 & 240.000 & 0.127 \\
\cmidrule[0.2pt](l){4-6}\cmidrule[0.2pt](l){7-9}\cmidrule[0.2pt](l){10-12}
&  &  & \multicolumn{3}{c}{Rej.\ 1\%} & \multicolumn{3}{c}{Rej.\ 5\%} & \multicolumn{3}{c}{Rej.\ 10\%} \\ 
\cmidrule[0.2pt](l){4-6}\cmidrule[0.2pt](l){7-9}\cmidrule[0.2pt](l){10-12}
& 25  & 0.5 & 0.015 & 1.000 & 0.016 & 0.066 & 1.000 & 0.068 & 0.124 & 1.000 & 0.132 \\ 
& 50  & 0.4 & 0.012 & 1.000 & 0.014 & 0.060 & 1.000 & 0.064 & 0.117 & 1.000 & 0.121 \\ 
& 100 & 0.6 & 0.012 & 1.000 & 0.010 & 0.057 & 1.000 & 0.056 & 0.114 & 1.000 & 0.108 \\ 
& 200 & 0.3 & 0.012 & 1.000 & 0.010 & 0.053 & 1.000 & 0.055 & 0.107 & 1.000 & 0.104 \\ 
& 400 & 0.7 & 0.011 & 1.000 & 0.010 & 0.053 & 1.000 & 0.050 & 0.108 & 1.000 & 0.099 \\ 
\cmidrule[0.5pt](l){2-12}
\end{tabular}
\label{Tab:Sim_IV_e}
\end{table}

\newpage
\begingroup
\setstretch{1.25}
\setlength\bibitemsep{1.0 pt}
\printbibliography
\endgroup
\end{refsection}
 
\end{document}